\documentclass[oneside]{amsart}
\pdfoutput=1

\usepackage{geometry}
\usepackage[T1]{fontenc}
\usepackage[utf8]{inputenc}
\usepackage{lmodern}

\usepackage[dvipsnames]{xcolor}
\usepackage{tikz,tikz-3dplot}
\usetikzlibrary{positioning,arrows.meta,decorations.pathreplacing}

\makeatletter
 \tikzoption{canvas is plane}[]{\@setOxy#1}
 \def\@setOxy O(#1,#2,#3)x(#4,#5,#6)y(#7,#8,#9)%
   {\def\tikz@plane@origin{\pgfpointxyz{#1}{#2}{#3}}%
    \def\tikz@plane@x{\pgfpointxyz{#4}{#5}{#6}}%
    \def\tikz@plane@y{\pgfpointxyz{#7}{#8}{#9}}%
    \tikz@canvas@is@plane
   }
\makeatother  

\usepackage{hyperref}

\newsavebox{\smallclawbox}
\savebox{\smallclawbox}{%
\begin{tikzpicture}[x=1ex,y=1ex,baseline={([yshift=-.7ex]current bounding box.center)}]
    \coordinate (v) at (0,0);
    \coordinate (v1) at ([shift=(90:.8)]v);
    \coordinate (v2) at ([shift=(210:.8)]v);
    \coordinate (v3) at ([shift=(330:.8)]v);
    \draw (v) -- (v1);
    \draw (v) -- (v2);
    \draw (v) -- (v3);
    \filldraw (v) circle (.6pt);
    \filldraw (v1) circle (.6pt);
    \filldraw (v2) circle (.6pt);
    \filldraw (v3) circle (.6pt);
\end{tikzpicture}%
}
\newcommand{\smallclaw}{{\usebox{\smallclawbox}}}

\newsavebox{\solidedgebox}
\savebox{\solidedgebox}{%
\begin{tikzpicture}[baseline={([yshift=-0.7ex]0,0)}] 
    \coordinate (v1) at (0,0);
    \coordinate (v2) at (.7,0);
    \draw (v1) -- (v2); 
    \filldraw (v1) circle (1.3pt);
    \filldraw (v2) circle (1.3pt);
\end{tikzpicture}%
}
\newcommand{\solidedge}{{\usebox{\solidedgebox}}}

\newsavebox{\dashededgebox}
\savebox{\dashededgebox}{%
\begin{tikzpicture}[baseline={([yshift=-0.7ex]0,0)}] 
    \coordinate (v1) at (0,0);
    \coordinate (v2) at (.7,0);
    \draw[dashed] (v1) -- (v2); 
    \filldraw (v1) circle (1.3pt);
    \filldraw (v2) circle (1.3pt);
\end{tikzpicture}%
}
\newcommand{\dashededge}{{\usebox{\dashededgebox}}}

\newsavebox{\dottededgebox}
\savebox{\dottededgebox}{%
\begin{tikzpicture}[baseline={([yshift=-0.7ex]0,0)}] 
    \coordinate (v1) at (0,0);
    \coordinate (v2) at (.7,0);
    \draw[dotted] (v1) -- (v2); 
    \filldraw (v1) circle (1.3pt);
    \filldraw (v2) circle (1.3pt);
\end{tikzpicture}%
}
\newcommand{\dottededge}{{\usebox{\dottededgebox}}}

\newcommand{\CC}{{\mathbb C}}

\newcommand{\PP}{{\mathbb P}}
\newcommand{\QQ}{{\mathbb Q}}
\newcommand{\RR}{{\mathbb R}}

\newcommand{\ZZ}{{\mathbb Z}}
\newcommand{\ee}{\mathrm{e}}
\newcommand{\ii}{\mathrm{i}}

\newcommand{\sL}{\mathfrak{sl}}
\newcommand{\dd}{\mathrm{d}}
\newcommand{\Li}{\mathrm{Li}\,}
\newcommand{\aaa}{{\overline{a}}}
\newcommand{\ddd}{{\overline{d}}}
\newcommand{\ff}{{\overline{f}}}
\newcommand{\gggg}{{\overline{g}}}
\newcommand{\hh}{{\overline{h}}}
\newcommand{\kk}{{\overline{k}}}
\newcommand{\mm}{{\overline{m}}}

\newcommand{\sss}{{\overline{s}}}
\newcommand{\xx}{{\overline{x}}}
\newcommand{\zz}{{\overline{z}}}
\newcommand{\Gg}{{\overline{G}}}
\newcommand{\sC}{\mathcal{C}}
\newcommand{\sD}{\mathcal{D}}
\newcommand{\sE}{\mathcal{E}}
\newcommand{\sF}{\mathcal{F}}
\newcommand{\sG}{\mathcal{G}}
\newcommand{\sI}{\mathcal{I}}

\newcommand{\sM}{\mathcal{M}}
\newcommand{\sS}{\mathcal{S}}
\newcommand{\sV}{\mathcal{V}}

\newcommand{\sVGint}{{\sV_G^\mathrm{int}}}
\newcommand{\VGint}{{|\sVGint|}}
\newcommand{\sVGext}{{\sV_G^\mathrm{ext}}}
\newcommand{\VGext}{{|\sVGext|}}
\newcommand{\intsv}{\int_\mathrm{sv}}
\renewcommand{\Im}{\mathop\mathrm{Im}}
\renewcommand{\Re}{\mathop\mathrm{Re}}
\newcommand{\K}{\overline{K}}

\theoremstyle{plain}
\newtheorem{thm}{Theorem}
\newtheorem{lem}[thm]{Lemma}
\newtheorem{con}[thm]{Conjecture}
\newtheorem{cor}[thm]{Corollary}
\newtheorem{prop}[thm]{Proposition}
\newtheorem{quest}[thm]{Question}
\newtheorem{remark}[thm]{Remark}
\newtheorem{defn}[thm]{Definition}
\newtheorem{ex}[thm]{Example}

\newcommand*{\doilink}[2]{\href{https://doi.org/\detokenize{#2}}{#1}}
\newcommand*{\arxivlink}[2]{\href{https://arxiv.org/abs/\detokenize{#2}}{#1}}

\hypersetup{
  pdfauthor={Michael Borinsky and Oliver Schnetz},
  pdftitle={Graphical functions in even dimensions},
  pdfsubject={},
  pdfkeywords={},
  linkcolor  = RedViolet!85!black,
  citecolor  = BlueGreen!85!black,
  urlcolor   = YellowOrange!85!black,
  colorlinks = true,
}

\title{Graphical functions in even dimensions}
\date{}
\author{Michael Borinsky \and Oliver Schnetz}
\address{Michael Borinsky\\
Nikhef Theory Group\\
Science Park 105\\
1098 XG Amsterdam, The Netherlands}
\address{Oliver Schnetz\\
Department Mathematik\\
Cauerstra{\ss}e 11\\
91058 Erlangen, Germany}

\begin{document}
%\vspace*{-4\baselineskip}%
%\hspace*{\fill} \mbox{\footnotesize{\textsc{Nikhef-2021-007}}}
%\vspace*{4\baselineskip}%

%\begingroup
%\let\newpage\relax
\maketitle
\begin{abstract}
Graphical functions are special position space Feynman integrals, which can be used to calculate Feynman periods and one- or two-scale processes at high loop orders.
With graphical functions, renormalization constants have been calculated to loop orders seven and eight in four-dimensional $\phi^4$ theory and to order five in six-dimensional $\phi^3$ theory.
In this article we present the theory of graphical functions in even dimensions $\geq4$ with detailed reviews of known properties and full proofs whenever possible.
\end{abstract}

\section{Introduction}
\subsection{Motivation}
Graphical functions\footnote{This is an updated version of an 2021 preprint, incorporating corrections from the erratum \doilink{Commun.\ Number Theory Phys.\ 20 (2026) 461--462}{10.4310/CNTP.260606021019}. 
All main results of the original article remain valid. The erratum addresses subtle inaccuracies in certain statements that do not affect the conclusions.
Otherwise, this version is equivalent to the original article published in \doilink{Commun.\ Number Theory Phys.\ 16 (2022) 515--614}{10.4310/CNTP.2022.v16.n3.a3}.}
 are massless position space Feynman integrals in quantum field theory (QFT) which depend on three vectors $z_0$, $z_1$, $z_2$ in $D$-dimensional Euclidean space.
They are a powerful tool to perform calculations in perturbative QFT. Graphical functions were first defined in \cite{gf} but also play a prominent role in four-dimensional conformal
QFTs where they were independently introduced (see e.g.\ \cite{SYM}).

The theory of graphical functions has successfully been applied to four-dimensional $\phi^4$ theory (see e.g.\ \cite{ZZ,gf,numfunct}). Calculations of Feynman periods
in $\phi^4$ theory led to the discovery of the coaction conjectures in \cite{coaction} and the coaction principle (see e.g. \cite{Bcoact1,Bcoact2,motg2}). 

In this article, we provide an extension of graphical functions to even dimensions $\geq4$. This extension made six dimensional $\phi^3$ theory accessible to higher loop orders
(the loop order is the number of independent cycles in the underlying Feynman graph). Without the theory of graphical functions $\phi^3$ is a notoriously hard subject.
The authors calculated all primitive periods in $\phi^3$ theory up to loop order six \cite{phi3,Shlog}.
At loop order seven, 561 of the 607 primitive Feynman periods were computed. At loop orders eight and nine many sporadic results exist \cite{phi3,Shlog}.
These period calculations were the basis for a recent computation of the renormalization constants in $\phi^3$ theory up to loop order
five with applications to percolation theory \cite{5loopphi3} (see \cite{classphi3} for the independent classical calculation).
Recently, higher dimensional graphical functions also found applications in the context of graph complexes \cite{Bidf}.

Graphical functions also exist in odd dimensions $\geq3$. Their theory differs from the even dimensional case and is much less explored. 
The case of two dimensions is singular and needs special treatment \cite{gf}.
For the physical QFTs in the standard model, graphical functions in even dimensions $\geq4$ seem to suffice.

Graphical functions also led to the seven loop results for the renormalization constants in $\phi^4$ theory \cite{numfunct,7loops}.
Recently, the eight loop calculation of the gamma function was finalized \cite{Shlog}.
Renormalization demands an extension to non-integer dimensions as a regularization mechanism, see e.g.\ \cite{par,IZ}.
In the context of graphical functions this dimensional regularization is very convenient for calculations.
A precise definition of dimensional regularization, however, requires the use of the parametric representation which can be unwieldy for mathematical proofs.
Therefore, the theory of dimensionally regularized graphical functions mostly relies on a series of (well tested) conjectures.

In this article, we focus on graphical functions in even integer dimension $D\geq4$, where most results can be proved:
\begin{equation}\label{eqD}
D=2\lambda+2,\quad\lambda\in\{1,2,3,\ldots\}.
\end{equation}

We also restrict ourselves to scalar graphical functions. Graphical functions with positive spin can be expressed as tuples of scalar graphical functions.
By a dimension shift mechanism, calculations of graphical functions representing Feynman integrals of particles with positive spin demand the handling of scalar graphical
functions in higher even dimensions. Hence, this article also paves the ground for future calculations in QFTs with spin (such as Yang-Mills theories).

\subsection{The graphical function method}

Let $G$ be a graph with edge set $\sE_G$ and vertex set $\sV_G$ which splits into internal vertices $\sVGint$ and external vertices $\sVGext$.
We assume that $G$ has exactly three external vertices and label the vertices by vectors in $D$-dimensional Euclidean space, $\sVGext=\{z_0,z_1,z_2\} \subset \RR^D$.
The internal labels are $x_i\in\RR^D$, $i=1,2,\ldots,\VGint$.

Every edge $e\in\sE_G$ has a weight $\nu_e\in\RR$. We often assume $\lambda\nu_e\in\ZZ$.
Although Feynman graphs in scalar QFTs only have unit edge-weights, it is convenient to allow general weights.
Weights of multiple edges are additive: if $e,f\in\sE_G$ with weights $\nu_e,\nu_f$ have equal endpoints, the pair $e,f$ can be replaced by a single edge with the same endpoints
and weight $\nu_e+\nu_f$. Zero weight edges can be eliminated from the graph $G$. We do not allow $G$ to have self-loops.

To every edge $e=uv\in\sE_G$ we associate a quadric which is the square of the distance between the labels of its endpoints $u,v\in\sV_G$,
\begin{equation}\label{eqQe}
Q_e(u,v)=\|u-v\|^2=(u^1-v^1)^2+\ldots+(u^D-v^D)^2.
\end{equation}
With this data we write the position space three-point function $A_G$ as a positive integral \cite{IZ},
\begin{equation}\label{eqAdef}
A_G(z_0,z_1,z_2)=\left(\prod_{i=1}^\VGint\int_{\RR^D}\frac{\dd^Dx_i}{\pi^{D/2}}\right)\frac{1}{\prod_{e\in\sE_G}Q_e^{\lambda\nu_e}}.
\end{equation}
In this article we only consider graphs $G$ for which the above integral converges. This significantly restricts the set of admissible graphs, see Proposition \ref{propconv}.

\begin{figure}
\tdplotsetmaincoords{80}{120}

\begin{center}
\begin{tikzpicture}[tdplot_main_coords,scale=.85]
\draw[thin, ->,black!50] (0,0,0) -- (4.5,0,0) node[anchor=south east,opacity=1]{$x^1$};

\draw[thin, ->,black!50] (0,0,0) -- (0,4.5,0)  node[anchor=south west,opacity=1]{$x^2$};

\draw[thin, ->,black!50] (0,0,0) -- (0,0,4.5)  node[anchor=south west,opacity=1]{$x^D$};

\tdplotsetrotatedcoords{0}{90}{0}
\draw[dotted,black!50,tdplot_rotated_coords] (0,.8,0) arc (90:180:.8);

\pgfmathsetmacro{\Ox}{32}
\pgfmathsetmacro{\Oy}{14}
\pgfmathsetmacro{\Oz}{3}
\pgfmathsetmacro{\Onex}{-5}
\pgfmathsetmacro{\Oney}{3}
\pgfmathsetmacro{\Onez}{1}
\pgfmathsetmacro{\Zx}{-12}
\pgfmathsetmacro{\Zy}{1}
\pgfmathsetmacro{\Zz}{6}

\tdplotcrossprod(\Onex,\Oney,\Onez)(\Zx,\Zy,\Zz)
\pgfmathsetmacro{\rx}{\tdplotresx}
\pgfmathsetmacro{\ry}{\tdplotresy}
\pgfmathsetmacro{\rz}{\tdplotresz}
\tdplotcrossprod(\rx,\ry,\rz)(\Onex,\Oney,\Onez)
\pgfmathsetmacro{\tx}{\tdplotresx/100}
\pgfmathsetmacro{\ty}{\tdplotresy/100}
\pgfmathsetmacro{\tz}{\tdplotresz/100}

\pgfmathsetmacro{\dplane}{\rx*\Ox + \ry*\Oy + \rz*\Oz}

\pgfmathsetmacro{\OneLen}{sqrt(\Onex*\Onex+\Oney*\Oney+\Onez*\Onez)}
\pgfmathsetmacro{\iLen}{sqrt(\tx*\tx+\ty*\ty+\tz*\tz)}

\tikzset{perspective/.style= {canvas is plane={O(0,0,0)x(\Onex/\OneLen,\Oney/\OneLen,\Onez/\OneLen)y(\tx/\iLen,\ty/\iLen,\tz/\iLen)}} }
\pgfmathsetmacro{\axisscale}{2}
\tikzset{perspective2/.style= {canvas is plane={O(0,0,0)x(\axisscale*\Onex/\OneLen,\axisscale*\Oney/\OneLen,\axisscale*\Onez/\OneLen)y(\axisscale*\tx/\iLen,\axisscale*\ty/\iLen,\axisscale*\tz/\iLen)}} }

\pgfmathsetmacro{\axisovershoot}{2}
\pgfmathsetmacro{\axisundershoot}{.5}

\coordinate (v0) at (\Ox,\Oy,\Oz);
\coordinate (v1) at ([shift={(\Onex,\Oney,\Onez)}]v0);
\coordinate (vz) at ([shift={(\Zx,\Zy,\Zz)}]v0);
\coordinate (vi) at ([shift={(\tx,\ty,\tz)}]v0);

\filldraw[black!50] (0,0,0) circle (1.3pt);

\draw[dashed,-{Stealth[length=10pt, width=15pt]},black!50] (0,0,0) -- node[inner sep=1.5pt,below right] {$z_0$} (v0);
\draw[dashed,-{Stealth[length=10pt, width=15pt]},black!50] (0,0,0) -- node[inner sep=1.5pt,below left] {$z_1$} (v1);
\draw[dashed,-{Stealth[length=10pt, width=15pt]},black!50] (0,0,0) -- node[inner sep=1.5pt,below right] {$z_2$} (vz);
\draw[thin,->] ($(v0)!-\axisundershoot/\OneLen!(v1)$) -- ($(v0)!{1+2.7*\axisovershoot/\OneLen}!(v1)$) node[perspective,anchor=north west,opacity=1]{$\Re z$};
\draw[thin,->] ($(v0)!-\axisundershoot/\iLen!(vi)$) -- ($(v0)!{\OneLen/\iLen+1.4*\axisovershoot/\iLen}!(vi)$) node[perspective,anchor=south west,opacity=1]{$\Im z$};

\coordinate (C1) at ($(v0)!{2*\OneLen/\iLen}!(vi)$);
\node[perspective2] (C) at ($(C1)!{.5}!(v1)$) {$\CC$};

\pgfmathsetmacro{\xscale}{1.7}
\coordinate (R0) at (${1 + 2.1*\axisundershoot/\OneLen + 2*\axisundershoot/\iLen}*(v0) - 2.1*\axisundershoot/\OneLen*(v1)- 2*\axisundershoot/\iLen*(vi)$);
\coordinate (R1) at (${ - \xscale*2*\axisovershoot/\OneLen + 2*\axisundershoot/\iLen}*(v0) + {1 + 2*\xscale*\axisovershoot/\OneLen}*(v1)- 2*\axisundershoot/\iLen*(vi)$);
\coordinate (R2) at (${ - \OneLen/\iLen - 2*\xscale*\axisovershoot/\OneLen - 2*\axisovershoot/\iLen}*(v0) + {1 + 2*\xscale*\axisovershoot/\OneLen}*(v1)+ {\OneLen/\iLen + 2*\axisovershoot/\iLen}*(vi)$);
\coordinate (R3) at (${1 - \OneLen/\iLen + 2.1*\axisundershoot/\OneLen - 2*\axisovershoot/\iLen}*(v0) - {2.1*\axisundershoot/\OneLen}*(v1)+ {\OneLen/\iLen + 2*\axisovershoot/\iLen}*(vi)$);

\draw[opacity = .05,fill,black!50]  (R0) -- (R1) -- (R2) -- (R3);

\filldraw (v0) circle(1.3pt) node[below left,perspective] {$0$};
\filldraw (v1) circle(1.3pt) node[below,perspective] {$1$};
\filldraw (vz) circle(1.3pt) node[above right,perspective] {$z$};

\draw[thick] (v0) -- node[perspective,below]{$1$} (v1);
\draw[thick] (v0) -- node[perspective,above left]{$|z|$} (vz);
\draw[thick] (vz) -- node[perspective,below right]{$|z-1|$} (v1);

\end{tikzpicture}
\end{center}
\caption{The vectors $z_0,z_1,z_2\in\RR^D$ span a plane that is identified with the complex plane $\CC$ by requiring $z_0,z_1$ to coincide with $0,1\in \CC$.
The position of the vector $z_2$ inside the plane determines the value $z\in\CC$ (up to conjugation). Comparing ratios of squared side lengths of the triangles $z_0,z_1,z_2 \in \RR^D$
and $0,1,z \in \CC$ gives the relations in \eqref{eqinvs}.}
\label{fig:Ctriangle}
\end{figure}

The value of $A_G(z_0,z_1,z_2)$ is invariant under simultaneous translations and rotations of the vectors $z_0,z_1,z_2 \in \RR^D$ and it is homogeneous under scaling.
An efficient parametrization is obtained by the identification of the affine plane spanned by $z_0,z_1,z_2$ in $\RR^D$ with $\CC$.
We assume that $z_0$ and $z_1$ coincide with 0 and 1 in $\CC$ (setting the scale in $\CC$, see Figure~\ref{fig:Ctriangle}).
The third vector $z_2$ is associated with a variable $z\in\CC$ in one of the two possible ways (the ambiguity under complex conjugation leads to the symmetry (G1) in Theorem \ref{thm1}).
The graphical function $f_G(z)$ is the Feynman integral $A_G$ evaluated on the complex plane $\CC$.

We obtain the following relations between $z_0,z_1,z_2\in \RR^D$ and $z\in\CC$:
\begin{equation}\label{eqinvs}
z\zz =  \frac{\|z_2-z_0\|^2}{\|z_1-z_0\|^2},
\quad(z-1)(\zz-1) =\frac{\|z_2-z_1\|^2}{\|z_1-z_0\|^2},
\end{equation}
where $\zz$ is the complex conjugate of $z$. With these relations we get
\begin{equation}\label{eqfA}
f_G(z)=\|z_1-z_0\|^{(D-2)\sum_e\nu_e-D\VGint}A_G(z_0,z_1,z_2).
\end{equation}
By power-counting the pre-factor on the right hand side compensates for the scaling behavior of the Feynman integral $A_G$.

\subsection{Statement of results}
In Section \ref{sectgf} we define graphical functions and summarize their fundamental properties in Theorem \ref{thm1}: Graphical functions are single-valued real-analytic functions
on $\CC\backslash\{0,1\}$ (see \cite{par}) which admit series expansions at their singular points $0,1$ and $\infty$ of log-Laurent type (\ref{01expansion}) and (\ref{inftyexpansion}).
The expansion at infinity lifts graphical functions to objects on the (punctured) Riemann sphere $\CC\cup\{\infty\}$.

Theorem \ref{thm1}, in spite of its technical nature, is a cornerstone of the theory of graphical functions. The proof is based on the expansion (\ref{gfr}) of graphical functions
into Gegenbauer polynomials (see \cite{geg}) using radial and angular graphical functions (Sections \ref{sectrad}--\ref{sectpf1}). A full proof of (\ref{gfr}) is only provided in the
classical case of four dimensions and unit edge-weights, see Theorems \ref{D4posthm} and \ref{thmgegex0}. The general case is well tested but formally relies on the validity of
interchanging the Gegenbauer expansion with the position space integrals, see Conjecture \ref{congegex}.

Integration of $z$ over the complex plane $\int_\CC \dd^2z$ converts graphical functions into numbers: Feynman periods \cite{gf,numfunct}.
The calculation of Feynman periods is a classical subject in QFT (see e.g.\ \cite{BK,Census}).
Feynman periods are essential ingredients in the calculation of renormalization constants which, in turn, allow one to calculate critical exponents for statistical models \cite{ZJ}.
In Section \ref{sectper} we describe how Feynman periods can be obtained from graphical functions.

Note that graphical functions were originally devised in \cite{gf} for the purpose of computing Feynman periods.
Lists of Feynman periods in four-dimensional $\phi^4$ theory up to loop order eight can be found in \cite{Census,numfunct}. A list of periods in six-dimensional $\phi^3$ theory up to loop
order nine is in \cite{Shlog} which also extends the $\phi^4$ data to loop order eleven.

In Sections \ref{sectgfid1} to \ref{sectfish} we report on known identities for graphical functions. The efficiency of the theory of graphical functions relies on the existence of these
identities as tools for their computation.

The most important technique to calculate graphical functions is appending a single edge $e$ of weight $\nu_e=1$ to the vertex $z$ (thus creating a new internal vertex), 
\begin{align}
    \def\rad{.6}
\begin{tikzpicture}[baseline={([yshift=-.7ex]0,0)}]
    \coordinate (v) at (0,0);
    \draw[fill=black!20] (v) circle (\rad);
    \node (G) at (-.8,-\rad) {$G$};
    \coordinate[label=above:$1$] (v1) at ([shift=(90:\rad)]v);
    \coordinate[label=below:$0$] (v0) at ([shift=(-90:\rad)]v);
    \coordinate[label=right:$z$] (vz) at ([shift=(0:\rad)]v);
    \coordinate (vm) at ([shift=(0:1.5)]v);
    \filldraw (v0) circle (1.3pt);
    \filldraw (v1) circle (1.3pt);
    \filldraw (vz) circle (1.3pt);
    \coordinate (w) at (4,0);
    \draw[fill=black!20] (w) circle (\rad);
    \node (G1) at (3.2,-\rad) {$G_1$};
    \coordinate[label=above:$1$] (w1) at ([shift=(90:\rad)]w);
    \coordinate[label=below:$0$] (w0) at ([shift=(-90:\rad)]w);
    \coordinate (wx) at ([shift=(0:\rad)]w);
    \coordinate[label=right:$z$] (wz) at ([shift=(0:{2.5*\rad})]w);
    \coordinate (wm) at ([shift=(180:1.5)]w);
    \draw (wx) -- (wz);
    \draw[thick,->] (vm) -- (wm);
    \filldraw (w0) circle (1.3pt);
    \filldraw (w1) circle (1.3pt);
    \filldraw (wx) circle (1.3pt);
    \filldraw (wz) circle (1.3pt);
\end{tikzpicture}
\label{eqappend},
\end{align}
where we associate the labels $0,1,z$ to the external vertices $z_0,z_1,z_2$.

The main result of this article is the algorithm that performs the operation of appending an edge to a given graphical function in arbitrary even dimensions $\geq 4$
by taking single-valued primitives, see Section \ref{sectappalg}.

The three-point functions associated to $G_1$ and $G$ are related by an inhomogeneous Laplace equation,
\begin{equation}\label{eqlap}
\Delta_{z_2} A_{G_1}(z_0,z_1,z_2)=-\frac{4}{(\lambda-1)!}A_G(z_0,z_1,z_2).
\end{equation}
To verify \eqref{eqlap}, recall that for an edge $e$ with weight $\nu_e=1$ the relevant factor in the integrand of (\ref{eqAdef}) is a Green's functions of the $D$-dimensional Laplacian \cite{IZ},
\begin{equation}\label{eqLprop}
\Delta_u\frac{1}{Q_e(u,v)^\lambda}=-\frac{4\pi^{\lambda+1}}{(\lambda-1)!}\delta^{(D)}(u-v),
\end{equation}
where $\delta^{(D)}$ is the $D$-dimensional Dirac delta distribution.

In Lemma \ref{lemdiff} we derive from \eqref{eqlap} the following relation between the graphical functions $f_G(z)$ and $f_{G_1}(z)$,
\begin{equation}\label{eqdiff}
\frac{1}{(z-\zz)^\lambda}\Delta_{\lambda-1}(z-\zz)^\lambda f_{G_1}(z)=-\frac{1}{(\lambda-1)!}f_G(z),
\end{equation}
where $\Delta_{\lambda-1}$ is an effective Laplacian acting on complex functions 
\begin{equation}\label{eqdiffn}
\Delta_{\lambda-1}=\partial_z\partial_\zz+\frac{\lambda(\lambda-1)}{(z-\zz)^2}.
\end{equation}
Solving the partial differential equation (\ref{eqdiff}) for $f_{G_1}(z)$ is one of the main results of this article. 
We proceed in three steps: Firstly, a general solution has to be found. This is trivial for $D=4$, $\lambda=1$, where the effective Laplacian factors.
The general solution in even dimensions $\geq4$ is constructed in Theorem \ref{thm2}. The proof of Theorem \ref{thm2} is in Sections \ref{sectpf2a} and \ref{sectpf2}.

As second step towards the solution of \eqref{eqdiff} we need to control the kernel of (\ref{eqdiffn}). Remarkably, the kernel is trivial in the space of graphical functions:
\begin{thm}\label{thm01}
The differential equation (\ref{eqdiff}) has a unique solution in the space of graphical functions.
\end{thm}
The general properties of graphical functions in Theorem \ref{thm1} specify this unique solution. 
Formally, Theorem \ref{thm01} follows from Theorem \ref{thm3} which will be proved in Section \ref{sectpf3}.

With existence and uniqueness of the solution one obtains an algorithm that allows one to calculate the graphical function $f_{G_1}(z)$ from the graphical function $f_G(z)$.
This algorithm is presented in Section \ref{sectappalg} as the last step in the solution of \eqref{eqdiff}.
It relies on the construction of single-valued primitives in an underlying function space. In general, this is a difficult task. For many practical calculations, however,
the space of generalized single-valued hyperlogarithms (GSVHs) suffices \cite{GSVH}, see Section \ref{sectint}. Whenever $f_G(z)$ is a GSVH, then $f_{G_1}(z)$ is a GSVH which can be computed
(subject to mild constraints from time and memory consumption). The algorithm of appending an edge has been implemented by the second author in \cite{Shlog}.

The surprising fact that (\ref{eqdiff}) admits a closed solution may be related to the interpretation of the effective
Laplacian (\ref{eqdiffn}) in terms of the Casimir of a representation of the Lie algebra $\sL(2,\CC)$ whose eigenvalue is related to the dimension $D$, see Section \ref{sectrep}.

The core of the article are proofs of the results in Sections \ref{sectgfid1} to \ref{sectfish}. Because the proofs require some preparations and rely on independent techniques
they are moved to Sections \ref{sectrad} to \ref{sectpf3}.

The article also contains a generalization of the Gegenbauer identity in Section 3.6 of \cite{gf} to even dimensions $\geq4$, see Section \ref{sectgeg}.
The proof of the classical four-dimensional case is in Section \ref{sectpf4}.

%%%%%%%%%%%%%%%%%%%%%%%%%%%%%%%%%%%%%%%%%%%%%%%%%%%%%%%%%%%%%%%%%%%%%%%%%%%%%%%%%
\section*{Acknowledgements}
The first author is supported by NWO Vidi grant 680-47-551.
The second author is supported by DFG grant SCHN~1240/3. He is also grateful for valuable discussions on representation theory with F.~Knop and B.~Van Steirteghem.

%%%%%%%%%%%%%%%%%%%%%%%%%%%%%%%%%%%%%%%%%%%%%%%%%%%%%%%%%%%%%%%%%%%%%%%%%%%%%%%%%
\section{Graphical functions}\label{sectgf}
\subsection{Definition}\label{sectgfdef}
It is convenient to define graphical functions as position space Feynman integrals (\ref{eqAdef}) evaluated at specific vectors.
\begin{defn}\label{defgf}
Let $z\in\CC\backslash\{0,1\}$ with real part $\Re z=(z+\zz)/2$ and imaginary part $\Im z=(z-\zz)/2\ii$. Define the three vectors $z_0,z_1,z_2\in\RR^D$ as
\begin{equation}\label{eqzdef}
z_0=0=\left(\begin{array}{c}0\\0\\0\\ \vdots\\0\end{array}\right),\quad
z_1=e_1=\left(\begin{array}{c}1\\0\\0\\\vdots\\0\end{array}\right),\quad
z_2=\left(\begin{array}{c}\Re z\\\Im z\\0\\\vdots\\0\end{array}\right).
\end{equation}
Let $\lambda\in\{1,2,3,\ldots\}$ and $G$ be a weighted graph with external vertices $z_0,z_1,z_2$ such that the three-point function $A_G(z_0,z_1,z_2)$ in (\ref{eqAdef}) exists. Then
(compare (\ref{eqfA}) with $\|z_1-z_0\|=\|e_1\|=1$)
\begin{equation}\label{eqfGdef}
f_G^{(\lambda)}(z)=A_G(0,e_1,z_2(z))
\end{equation}
is the graphical function of $G$ in $D=2\lambda+2$ dimensions at the complex value $z$. In the graph $G$ we use labels 0, 1, $z$ for $z_0$, $z_1$, $z_2$ with the identification illustrated
in Figure \ref{fig:Ctriangle}. We drop the superscript $(\lambda)$ if we do not need to refer to the dimension of the ambient space.
\end{defn}
Definition \ref{defgf} provides an explicit formula for $f_G^{(\lambda)}(z)$ which is equivalent to (\ref{eqfA}).

A non-trivial graphical function is never holomorphic in $z$ but a function of $z$ and its complex conjugate $\zz$. Note that in some situations it can be convenient to consider $z$
and $\zz$ as independent analytic variables. This perspective is essential in QFT for the transition from Euclidean to Minkowski metric where $z$ and $\zz$ are independent real variables.

Occasionally we identify $f_G(z)$ with the graph $G$ together with the external labels $0,1,z$. In this case we write $G_{0,1,z}$ for the graph $G$.
\begin{ex}\label{exempty}
Let $\emptyset$ be the empty graph (with no edges) and $uv$ be the graph with single edge $uv$ of weight $\nu_{uv}$ between the external vertices $u,v\in\{0,1,z\}$. We get
$$
f_\emptyset(z)=f_{01}(z)=1,\quad f_{0z}(z)=(z\zz)^{-\lambda\nu_{0z}},\quad f_{1z}(z)=[(z-1)(\zz-1)]^{-\lambda\nu_{1z}}.
$$
\end{ex}

The first non-trivial graphical function is associated to the three-star $\smallclaw$.

\begin{figure}
\begin{align*}
& 
\begin{tikzpicture}
    \coordinate (v) ;
    \def \rad {1}
    \coordinate[label=above:$1$] (v1) at ([shift=(90:\rad)]v);
    \coordinate[label=below left:$0$] (v2) at ([shift=(210:\rad)]v);
    \coordinate[label=below right:$z$] (v3) at ([shift=(330:\rad)]v);
    \draw (v) -- node[inner sep=1pt,left] {\tiny{$1$}} (v1);
    \draw (v) -- node[inner sep=1pt,above left] {\tiny{$1$}} (v2);
    \draw (v) -- node[inner sep=1pt,above right] {\tiny{$1$}} (v3);
    \filldraw (v) circle (1.3pt);
    \filldraw (v1) circle (1.3pt);
    \filldraw (v2) circle (1.3pt);
    \filldraw (v3) circle (1.3pt);
    \node [below left=of v] {$\smallclaw$};
\end{tikzpicture}
\hspace{3cm}
\begin{tikzpicture}
    \coordinate (v) ;
    \def \rad {1}
    \coordinate[label=above:$1$] (v1) at ([shift=(90:\rad)]v);
    \coordinate[label=left:$0$] (v2) at ([shift=(180:\rad)]v);
    \coordinate[label=below:$\infty$] (v3) at ([shift=(270:\rad)]v);
    \coordinate[label=right:$z$] (v4) at ([shift=(0:\rad)]v);
    \draw (v) -- node[inner sep=1pt,left] {\tiny{$1$}} (v1);
    \draw (v) -- node[inner sep=1pt,below] {\tiny{$1$}} (v2);
    \draw (v) -- node[inner sep=1pt,right] {\tiny{$1$}} (v3);
    \draw (v) -- node[inner sep=1pt,above] {\tiny{$1$}} (v4);
    \draw (v1) -- node[inner sep=0pt,above left] {\tiny{$-1$}} (v2);
    \draw (v3) -- node[inner sep=0pt,below right] {\tiny{$-1$}} (v4);
    \filldraw (v) circle (1.3pt);
    \filldraw (v1) circle (1.3pt);
    \filldraw (v2) circle (1.3pt);
    \filldraw (v3) circle (1.3pt);
    \filldraw (v4) circle (1.3pt);
    \node [below left=of v] {$\overline{\smallclaw}$};
\end{tikzpicture}
\end{align*}
\caption{
The three-star $\smallclaw$ and its completion (see Section \ref{sectcomp}) $\overline{\smallclaw}$ with indicated weights.
}
\label{fig:3star}
\end{figure}
 
\begin{ex}\label{ex3star}
Let $z_1=(1,0,\ldots,0)^T$, $z_2=(\Re z,\Im z,0,\ldots,0)^T$ and
$$
f_\smallclaw^{(\lambda)}(z)=\int_{\RR^D}\frac{1}{\|x\|^{2\lambda}\|x-z_1\|^{2\lambda}\|x-z_2\|^{2\lambda}}\frac{\dd^Dx}{\pi^{D/2}}
$$
be the graphical function of the three-star in Figure \ref{fig:3star} with unit weights in $D=4, 6, 8$ dimensions ($\lambda=1,2,3$).
We find (see \cite{gf}, correcting a misprint in $f_\smallclaw^{(3)}(z)$)
\begin{equation*}
f_\smallclaw^{(1)}(z)=\frac{4\ii D(z)}{z-\zz},\quad f_\smallclaw^{(2)}(z)=\frac{1}{z\zz(z-1)(\zz-1)},\quad f_\smallclaw^{(3)}(z)=\frac{2z\zz-z-\zz+2}{8[z\zz(z-1)(\zz-1)]^2},
\end{equation*}
where
\begin{equation}\label{D}
D(z)=\Im(\Li_2(z)+\ln(1-z)\ln|z|)
\end{equation}
is the Bloch-Wigner dilogarithm (see e.g.\ \cite{Zagierdilog}). The graphical function of the three-star is rational in any even dimension $\geq6$.
This can be seen by using the integral representation (3.16) in \cite{gf},
$$
f_\smallclaw^{(\lambda)}(z)=\frac{1}{(\lambda-1)!(\lambda-1)}\int_0^1\frac{(z\zz)^{1-\lambda}-t^{2\lambda-2}}{(1-tz)^\lambda(1-t\zz)^\lambda}\,\dd t,\quad\text{if }\lambda=2,3,\ldots,
$$
and proving that the residues of the integrand at $t=1/z$ and at $t=1/\zz$ vanish. We, e.g., get for the residue at $t=1/z$
$$
\frac{1}{(-z)^\lambda}\left.\partial^{\lambda-1}_t\right|_{t=z^{-1}}\frac{(z\zz)^{1-\lambda}-t^{2\lambda-2}}{(1-t\zz)^\lambda}=
\left.\frac{(-1)^\lambda(2\lambda-2)!(1-(tz)^{\lambda-1})}{(\lambda-1)!z^{2\lambda-1}(1-t\zz)^{2\lambda-1}}\right|_{t=z^{-1}}=0.
$$
\end{ex}

\subsection{Fundamental properties}\label{sectfp}
We need the following generalization of the edge-weight to subgraphs: Let $g$ be a subgraph of $G$, then
\begin{equation}\label{eqNg}
N_g=\Big(\sum_{e\in\sE_g}\nu_e\Big)-\frac{\lambda+1}{\lambda}|\sV_g^\text{int}|
\end{equation}
is the weight of the graph $g$. Note that $N_g$ is related to the superficial degree of divergence in QFT \cite{IZ}.
Equation (\ref{eqfA}) becomes
\begin{equation}\label{eqfA1}
A_G(z_0,z_1,z_2)=\|z_1-z_0\|^{-2\lambda N_G}f_G^{(\lambda)}(z),
\end{equation}
subject to the relation (\ref{eqinvs}) between $z$ and $z_0,z_1,z_2$.

Following \cite{GSVH} we define $\sS\sV_{\{0,1,\infty\}}$ as the space of functions on $\CC\backslash{\{0,1\}}$ which are analytic in $z$ and $\zz$ and
have single-valued log-Laurent expansions (\ref{01expansion}) and (\ref{inftyexpansion}) at 0, 1, and $\infty$ (see Proposition 3.16 in \cite{gf} and Theorem 4.4 in \cite{numfunct}).
Graphical functions have the following three fundamental properties (note that (G1) and (G2) also hold in odd dimensions):

\begin{thm}\label{thm1}
Let $G$ be a weighted graph such that the graphical function $f_G(z)$ exists in $D=2\lambda+2\geq3$ (odd or even) dimensions.
For any vertex set $\sV$ we define $G[\sV]$ as the graph which is induced in $G$ by $\sV$ (i.e.\ $G[\sV]$ has the edges of $G$ which have both endpoints in $\sV$).
Then
\begin{enumerate}
\item[(G1)] $f_G(z)$ is a symmetric function, $f_G(z)=f_G(\zz)$.
\item[(G2)] $f_G(z)$ is a single-valued positive real-analytic function on $\CC\setminus\{0,1\}$.
\item[(G3)] $f_G(z)\in\sS\sV_{\{0,1,\infty\}}$, i.e.\ it admits single-valued log-Laurent expansions at $0$, $1$, and $\infty$, if
\begin{equation}\label{thm1cond}
D=4\qquad\text{and}\qquad\nu_e=1\text{ for all }e\in\sE_G.
\end{equation}
Explicitly, let $\nu_z$ ($\nu_z^>$, $\nu_z^<$) be the sum of (positive, negative) weights adjacent to the vertex $z$.
At the singular points $s=\sss=0,1$ there exist coefficients $c_{\ell,m,\mm}^s\in\RR$ such that
\begin{equation}\label{01expansion}
f_G(z)=\sum_{\ell=0}^\VGint\sum_{m,\mm=M_s}^\infty c_{\ell,m,\mm}^s[\log(z-s)(\zz-\sss)]^\ell(z-s)^m(\zz-\sss)^\mm\quad\text{if }|z-s|<1,
\end{equation}
where
\begin{equation}\label{Ma}
M_s=-\max_{\sV\subseteq\sVGint}\lambda N_{G[\sV\cup\{s,z\}]}\geq\min\{1-\lambda\nu_z^>,-\lambda\nu_{sz}\}.
\end{equation}
At infinity there exist coefficients $c_{\ell,m,\mm}^\infty\in\RR$ such that
\begin{equation}\label{inftyexpansion}
f_G(z)=\sum_{\ell=0}^\VGint\sum_{m,\mm=-\infty}^{M_\infty}c_{\ell,m,\mm}^\infty(\log z\zz)^\ell z^m\zz^\mm\quad\text{if }|z|>1,
\end{equation}
where
\begin{equation}\label{Minfty}
M_\infty=-\lambda N_G+\max_{\sV\subseteq\sVGint}\lambda N_{G[\sV\cup\{0,1\}]}\leq\max\{-1-\lambda\nu_z^<,-\lambda\nu_z\}.
\end{equation}
\end{enumerate}
\end{thm}
Property (G1) is implied by (\ref{eqinvs}) and \eqref{eqfA}. Property (G2) is proved in \cite{par}.
Property (G3) is proved in Sections \ref{sectrad}--\ref{sectpf1} using Gegenbauer expansions and the concept of radial and angular graphical functions.

\begin{con}\label{con1}
Statement (G3) in Theorem \ref{thm1} remains valid if Condition (\ref{thm1cond}) is weakened to
$$
D=2\lambda+2\geq4\text{ is even}\qquad\text{and}\qquad\lambda\nu_e\in\ZZ\text{ for all }e\in\sE_G.
$$
\end{con}
It is proved in Section \ref{sectpf1} that Conjecture \ref{con1} follows from Conjecture \ref{congegex} which posits that one can interchange
the sum in the Gegenbauer expansion with the position space integrals.

Another strategy to prove Conjecture~\ref{con1} is to shift the dimension $D$. It is known that Feynman integrals in dimension $D$ can be expressed as linear combinations of Feynman integrals
in dimension $D-2$ \cite{Tarasov,Lee:2009dh}. This suggests that the Conjecture might follow for all even dimensions from the result for $D=4$ in Theorem \ref{thm1}.
Unfortunately, such arguments not only require the existence of the Feynman integral in integer dimensions, but also the existence of Laurent expansions in a small parameter $\varepsilon$ in
the vicinity of an integer dimension. The Laurent expansion of a Feynman integral in, for instance, $6-2\varepsilon$ dimensions may be expressed in terms of a linear combination
(with $\varepsilon$-dependent coefficients) of $4-2\varepsilon$ dimensional Feynman integrals. It is possible to calculate the $\varepsilon$ expansion of graphical functions in such
non-integer valued dimensions \cite{numfunct,7loops}. Unfortunately, rigorous statements such as Theorem~\ref{thm1} on these $\varepsilon$ expansions are rare and the whole theory stands
on a much more conjectural basis. This strategy to prove Conjecture~\ref{con1} would therefore likely require further advances in this framework of dimensionally regularized graphical functions.
An advantage of the dimensional shift method would be that it is not limited to edge weights $1$.

A positive answer to Question \ref{posquest} also leads to a (possibly partial) proof of the conjecture.
In any case, a proof of Conjecture \ref{con1} requires methods in representation theory or in analysis which are beyond the scope of this article.
An alternative approach to a proof could be the parametric representation in Section \ref{sectpar}.

There also exist generalizations of (G3) to weights $\lambda\nu_e\notin\ZZ$ and an analog of (G3) for odd dimensions (with square roots).

\begin{cor}
If a graphical function $f_G(z)$ exists in a neighborhood of one value $z\in\CC\backslash\{0,1\}$, then it exists for all $z\in\CC\backslash\{0,1\}$.
\end{cor}
\begin{proof}
This follows from (G2) by real-analytic continuation.
\end{proof}
The existence of a log-Laurent expansion at infinity promotes graphical functions to objects on the punctured Riemann sphere.
Although (G3) is of technical nature it is an important property for the theory of graphical functions, see e.g.\ Section \ref{sectpf3}. Single-valued log-Laurent expansions
are also central in the theory of GSVHs which is vital for computations with graphical functions \cite{GSVH}.

\subsection{Completion}\label{sectcomp}
The weighted analog of the vertex degree (the number of edges attached to a given vertex) is the sum of the weights of all edges incident to the vertex (the vertex weight).

On the Riemann sphere it is natural to consider conformal transformations. We prepare the use of conformal symmetry by adding a fourth external vertex $\infty$ to the graph.
This vertex $\infty$ connects to all internal vertices such that their weights become $2(\lambda+1)/\lambda$.

\begin{defn}\label{defcompl}
Let $G$ be a graph with weighted edges and three external vertices $0,1,z$. A completion $\Gg$ of $G$ in $D=2\lambda+2$ dimensions is a graph $G$ with an extra external vertex $\infty$ and
weighted edges $\sE_G\cup\{01\}\cup\{x\infty: x\in\sV_G\}$ such that every internal vertex has weight $2(\lambda+1)/\lambda$ and every external vertex has weight zero.

Any graph which is completed up to edges between external vertices (i.e.\ all internal vertices have weight $2(\lambda+1)/\lambda$) is internally completed.
\end{defn}

The de-completion of a completed graphical function $\Gg$ is $\Gg\backslash\{\infty\}$. It may differ from $G$ by an extra edge between the external vertices 0 and 1.
Because this edge has quadric $Q_e=1$ in (\ref{eqQe}) it does not contribute to the graphical function,
\begin{equation}
f_{\Gg}(z)\equiv f_{\Gg\backslash\{\infty\}}(z)=f_G(z).
\end{equation}
Because only completed graphs have a vertex $\infty$, we can use the same notation for graphical functions of completed and uncompleted graphs.

\begin{prop}[Lemma 3.18 in \cite{gf}]\label{propcompl}
Completion is always possible and unique. If $G$ has weights in $\frac{1}{\lambda}\ZZ$, then $\Gg$ also has weights in $\frac{1}{\lambda}\ZZ$. A completed graph $\Gg$ has total weight $N_\Gg=0$
(see (\ref{eqNg})).
\end{prop}
\begin{proof}
Most of the proof is in \cite{gf}. Because of the importance of completion we reproduce the proof.

Firstly, we internally complete the graph $G$ by adding $\infty$ and weighted edges from $\infty$ to internal vertices such that all internal vertices have weight $2(\lambda+1)/\lambda$.
We denote the internally completed graph with $G'$. It is clear that internal completion is always possible.
Let $\nu_0$, $\nu_1$, $\nu_z$, $\nu_\infty$ be the total weights of the external vertices in $G'$.

Next, we add an edge $z\infty$ of weight $-\nu_z$ to $G'$. This gives $z$ the total weight zero whereas the weight of $\infty$ becomes $\nu_\infty-\nu_z$.

We are left with the triangle $01$, $0\infty$, $1\infty$ to be added to the graph in order to nullify the total weights of $0,1,\infty$.
The weights are given by the unique solution of a linear system,
\begin{equation}\label{eqpfcompl1}
\nu_{01}=\frac{-\nu_0-\nu_1-\nu_z+\nu_\infty}{2},\quad\nu_{0\infty}=\frac{-\nu_0+\nu_1+\nu_z-\nu_\infty}{2},\quad\nu_{1\infty}=\frac{\nu_0-\nu_1+\nu_z-\nu_\infty}{2}.
\end{equation}
We now assume that all weights of $G$ are in $\frac{1}{\lambda}\ZZ$. Because the total weights of internal vertices in $G$ are in $\frac{1}{\lambda}\ZZ$, all weights of $G'$
are in $\frac{1}{\lambda}\ZZ$. We also have $\nu_0,\nu_1,\nu_z,\nu_\infty\in\frac{1}{\lambda}\ZZ$. Therefore $\nu_{z\infty}\in\frac{1}{\lambda}\ZZ$.
To show that $\nu_{01},\nu_{0\infty},\nu_{1\infty}\in\frac{1}{\lambda}\ZZ$ we sum the weights of half-edges in $G'$ and obtain
\begin{equation}\label{eqpfcompl2}
\frac{2(\lambda+1)}{\lambda}|\sV_{G'}^\text{int}|+\nu_0+\nu_1+\nu_z+\nu_\infty=2\sum_{e\in\sE_{G'}}\nu_e\in\frac{2}{\lambda}\ZZ.
\end{equation}
Hence $\nu_\text{ext}=\nu_0+\nu_1+\nu_z+\nu_\infty\in\frac{2}{\lambda}\ZZ$ and $\nu_{01}=\nu_\infty-\frac{1}{2}\nu_\text{ext}$, $\nu_{0\infty}=\nu_1+\nu_z-\frac{1}{2}\nu_\text{ext}$,
$\nu_{1\infty}=\nu_0+\nu_z-\frac{1}{2}\nu_\text{ext}\in\frac{1}{\lambda}\ZZ$.

With (\ref{eqpfcompl1}) the total weight of the completion is $N_{G'}-\nu_z+\nu_{01}+\nu_{0\infty}+\nu_{1\infty}=N_{G'}-\frac{1}{2}\nu_\text{ext}$ which vanishes by (\ref{eqpfcompl2}).
\end{proof}

\begin{ex}
The completion of the three-star is depicted in Figure \ref{fig:3star}.
In Figure \ref{fig:gfcomp} the completions of the graphs $G_a$ and $G_b$ are depicted.
\end{ex}

By conformal symmetry a convergent Feynman integral of an (internally) completed graph $\Gg$ with four external vertices $A_\Gg(z_0,z_1,z_2,z_3)$ can be expressed in two
cross-ratios (compare (\ref{eqinvs}))
\begin{equation}\label{eqinvsconf}
\frac{\|z_2-z_0\|^2\|z_3-z_1\|^2}{\|z_1-z_0\|^2\|z_3-z_2\|^2}=z\zz,\quad\frac{\|z_2-z_1\|^2\|z_3-z_0\|^2}{\|z_1-z_0\|^2\|z_3-z_2\|^2}=(z-1)(\zz-1).
\end{equation}
Therefore the theory of graphical functions also covers convergent conformal Feynman integrals which are prominent in super Yang-Mills Theories, see e.g.\ \cite{SYM}.

We will see in Section \ref{sectperm} that a completed graphical function is invariant under a transformation of external labels which stabilizes a certain cross-ratio.
We can also use completion to formulate a combinatorial criterion for the existence of a graphical function.

\begin{prop}[Lemma 3.19 in \cite{gf}]\label{propconv}
The graphical function $f_\Gg(z)$ exists (for all $z\in\CC\backslash\{0,1\}$) if and only if (see (\ref{eqNg}))
\begin{equation}\label{eqconv}
N_{\Gg[\sV]}<(|\sV^{\mathrm{ext}}|-1)\frac{\lambda+1}{\lambda}
\end{equation}
for all $\sV\subset\sV_G$  with $|\sV|\geq2$ and $|\sV^{\mathrm{ext}}|\leq1$ external vertices.
\end{prop}
\begin{proof}
The proof is by power-counting \cite{powercount}. It is the weighted analog of the proof of Lemma 3.19 in \cite{gf} (the formulation of Lemma 3.19 in \cite{gf} is slightly imprecise,
see also \cite{par}).
\end{proof}
Note that it suffices to apply (\ref{eqconv}) to internally completed graphs. If $G$ fails to be convergent for some $\sV$, then $G$ is ultraviolet divergent if $\infty\notin\sV$
(Example \ref{uvdiv}) and infrared divergent if $\infty\in\sV$ (Example \ref{indiv}).

\begin{ex}\label{uvdiv}
Let $G$ be a graph with an edge $e=uv$ of weight $\nu_e=(\lambda+1)/\lambda$. We assume that not both $u$ and $v$ are external and set $\sV=\{u,v\}$ in (\ref{eqconv}).
If $u$ or $v$ is external, then $N_{\Gg[\sV]}=0$. If both $u$ and $v$ are internal, then $N_{\Gg[\sV]}=-(\lambda+1)/\lambda$. In both cases we conclude that the graphical
function $f_G(z)$ is (ultraviolet) divergent.
\end{ex}

\begin{ex}\label{indiv}
Let $G$ be a graph with an internal vertex $v$ of total weight $(\lambda+1)/\lambda$. In the completed graph $\Gg$ we have $\nu_{v\infty}=(\lambda+1)/\lambda$.
With $\sV=\{0,\infty\}$ the graphical function $f_G(z)$ is (infrared) divergent.
\end{ex}

\subsection{Permutation of external vertices}\label{sectperm}
By Theorem \ref{thm1} graphical functions live on the Riemann sphere with punctures at $0,1,\infty$. Let
\begin{equation}
\sM_{01\infty}=\Big\{z\mapsto\phi(z)\in\Big\{z,1-z,\frac{z-1}{z},\frac{z}{z-1},\frac{1}{1-z},\frac{1}{z}\Big\}\Big\}
\end{equation}
be the permutation group of these punctures by M\"obius transformations. Consider a graphical function $f_G(z)$ with external labels $0,1,z$ and $\phi\in\sM_{01\infty}$.
We define the graphical function of the graph $G'$ with external labels $0,1,\phi(z)$ as $f_{G'}(z)=f_G(\phi(z))$ and lift this definition to completed graphs.
Two completed graphical functions which differ only in external labels are equal if the cross-ratios of the external vertices are equal. In other words, a permutation of external vertices
changes the argument of a completed graphical function by a M\"obius transformation in $\sM_{01\infty}$. Extending the symmetry from the six-fold permutation of $0,1,z$ to the 24-fold
permutation of $0,1,z,\infty$ is the main benefit of completion.

\begin{thm}[Theorem 3.20 in \cite{gf}]\label{thmperm}
Let $\Gg_1$ and $\Gg_2$ be two completed graphical functions which differ only in their external labels $z_0^i,z_1^i,z_2^i,z_3^i\in\sV_{\Gg_i}^{\mathrm{ext}}$,
$z^i_j\in\{0,1,\phi_i(z),\infty\}$, $\phi_i\in\sM_{01\infty}$. If the two sets of external labels have equal cross-ratios of complex numbers
\begin{equation}\label{eqcross}
\frac{(z^1_2-z^1_0)(z^1_3-z^1_1)}{(z^1_2-z^1_1)(z^1_3-z^1_0)}=\frac{(z^2_2-z^2_0)(z^2_3-z^2_1)}{(z^2_2-z^2_1)(z^2_3-z^2_0)},
\end{equation}
then
\begin{equation}\label{eqperm}
f_{\Gg_1}(z)=f_{\Gg_2}(z).
\end{equation}
In particular, every completed graphical function is invariant under a double transposition of external labels.
\end{thm}
\begin{proof}
The symmetry under permutation of $0,1,z$ is evident from the definition of graphical functions by invariants (\ref{eqinvs}) and $N_{\Gg_1}=N_{\Gg_2}=0$ by Proposition \ref{propcompl}.
The general result follows from the conformal symmetry of the integrand in (\ref{eqAdef}). The proof is in \cite{gf}.
\end{proof}

\begin{figure}
\begin{center}
    \def\scale{.75}
\begin{tikzpicture}[baseline={([yshift=-.7ex]0,0)}]
\begin{scope}
    \coordinate[label=above:$0$] (v0) at (-\scale,\scale);
    \coordinate[label=above:$1$] (v1) at (0,\scale);
    \coordinate[label=above:$z$] (vz) at (\scale,\scale);
    \coordinate (v2) at (-\scale,0);
    \coordinate (v3) at ( \scale,0);
    \coordinate (voo) at (0,-\scale);
    \node (G) at (0,{-1.5*\scale}) {$G_a$};   
    \draw[preaction={draw, white, line width=3pt, -}] (v2) -- (vz);
    \draw[preaction={draw, white, line width=3pt, -}] (v3) -- (v0);
    \draw[preaction={draw, white, line width=3pt, -}] (v2) -- (v1);
    \draw[preaction={draw, white, line width=3pt, -}] (v3) -- (v1);
    \draw (v3) -- (vz);
    \draw (v2) -- (v0);
    \draw (v2) -- ($(v2)!.2!(vz)$);
    \draw (v3) -- ($(v3)!.2!(v0)$);
    \draw (v2) -- (v3);
    \filldraw (v0) circle (1.3pt);
    \filldraw (v1) circle (1.3pt);
    \filldraw (vz) circle (1.3pt);
    \filldraw (v2) circle (1.3pt);
    \filldraw (v3) circle (1.3pt);
\end{scope}
\begin{scope}[xshift={3.5cm*\scale}]
    \coordinate[label=above:$0$] (v0) at (-\scale,\scale);
    \coordinate[label=above:$1$] (v1) at (0,\scale);
    \coordinate[label=above:$z$] (vz) at (\scale,\scale);
    \coordinate (v2) at (-\scale,0);
    \coordinate (v3) at ( \scale,0);
    \coordinate[label=above right:$\infty$] (voo) at (0,-\scale);
    \node (G) at (0,{-1.5*\scale}) {$\Gg_a$};   
    \draw[preaction={draw, white, line width=3pt, -}] (v2) -- (vz);
    \draw[preaction={draw, white, line width=3pt, -}] (v3) -- (v0);
    \draw (v2) -- (v3);
    \draw[preaction={draw, white, line width=3pt, -}] (v1) --  (voo);
    \draw[preaction={draw, white, line width=3pt, -}] (v2) -- (v1);
    \draw[preaction={draw, white, line width=3pt, -}] (v3) -- (v1);
    \draw (v1) --  (voo);
    \draw (v3) -- (vz);
    \draw (v2) -- (v0);
    \draw (v2) -- ($(v2)!.2!(vz)$);
    \draw (v3) -- ($(v3)!.2!(v0)$);
    \draw (v2) -- ($(v2)!.2!(v3)$);
    \draw (v3) -- ($(v3)!.2!(v2)$);
    \draw (v0) -- node[inner sep=.5pt,above] {\tiny{$-3$}} (v1);
    \draw (voo) to[bend right=90,looseness=1.5] node[pos=.2,inner sep=1pt,below right] {\tiny{$-2$}} (vz);
    \draw (voo) to[bend left=90,looseness=1.5] (v0);
    \filldraw (v0) circle (1.3pt);
    \filldraw (v1) circle (1.3pt);
    \filldraw (vz) circle (1.3pt);
    \filldraw (v2) circle (1.3pt);
    \filldraw (v3) circle (1.3pt);
    \filldraw (voo) circle (1.3pt);
\end{scope}
\begin{scope}[xshift={7cm*\scale}]
    \coordinate[label=above:$1$] (v0) at (-\scale,\scale);
    \coordinate[label=above:$0$] (v1) at (0,\scale);
    \coordinate[label=above:$\infty$] (vz) at (\scale,\scale);
    \coordinate (v2) at (-\scale,0);
    \coordinate (v3) at ( \scale,0);
    \coordinate[label=above right:$z$] (voo) at (0,-\scale);
    \node (G) at (0,{-1.5*\scale}) {$\Gg_b$};
    \draw[preaction={draw, white, line width=3pt, -}] (v2) -- (vz);
    \draw[preaction={draw, white, line width=3pt, -}] (v3) -- (v0);
    \draw (v2) -- (v3);
    \draw[preaction={draw, white, line width=3pt, -}] (v1) --  (voo);
    \draw[preaction={draw, white, line width=3pt, -}] (v2) -- (v1);
    \draw[preaction={draw, white, line width=3pt, -}] (v3) -- (v1);
    \draw (v1) --  (voo);
    \draw (v3) -- (vz);
    \draw (v2) -- (v0);
    \draw (v2) -- ($(v2)!.2!(vz)$);
    \draw (v3) -- ($(v3)!.2!(v0)$);
    \draw (v2) -- ($(v2)!.2!(v3)$);
    \draw (v3) -- ($(v3)!.2!(v2)$);
    \draw (v0) -- node[inner sep=.5pt,above] {\tiny{$-3$}} (v1);
    \draw (voo) to[bend right=90,looseness=1.5] node[pos=.2,inner sep=1pt,below right] {\tiny{$-2$}} (vz);
    \draw (voo) to[bend left=90,looseness=1.5] (v0);
    \filldraw (v0) circle (1.3pt);
    \filldraw (v1) circle (1.3pt);
    \filldraw (vz) circle (1.3pt);
    \filldraw (v2) circle (1.3pt);
    \filldraw (v3) circle (1.3pt);
    \filldraw (voo) circle (1.3pt);
\end{scope}
\begin{scope}[xshift={10.5cm*\scale}]
    \coordinate[label=above:$1$] (v0) at (-\scale,\scale);
    \coordinate[label=above:$0$] (v1) at (0,\scale);
    \coordinate (v2) at (-\scale,0);
    \coordinate (v3) at ( \scale,0);
    \coordinate[label=above right:$z$] (voo) at (0,-\scale);
    \node (G) at (0,{-1.5*\scale}) {$G_b$};
    \draw[preaction={draw, white, line width=3pt, -}] (v3) -- (v0);
    \draw (v2) -- (v3);
    \draw[preaction={draw, white, line width=3pt, -}] (v1) --  (voo);
    \draw[preaction={draw, white, line width=3pt, -}] (v2) -- (v1);
    \draw[preaction={draw, white, line width=3pt, -}] (v3) -- (v1);
    \draw (v1) --  (voo);
    \draw (v2) -- (v0);
    \draw (v3) -- ($(v3)!.2!(v0)$);
    \draw (v2) -- ($(v2)!.2!(v3)$);
    \draw (v3) -- ($(v3)!.2!(v2)$);
    \draw (v0) -- node[inner sep=.5pt,above] {\tiny{$-3$}} (v1);
    \draw (voo) to[bend left=90,looseness=1.5] (v0);
    \filldraw (v0) circle (1.3pt);
    \filldraw (v1) circle (1.3pt);
    \filldraw (v2) circle (1.3pt);
    \filldraw (v3) circle (1.3pt);
    \filldraw (voo) circle (1.3pt);
\end{scope}
\begin{scope}[xshift={14cm*\scale}]
    \coordinate[label=above:$1$] (v0) at (-\scale,\scale);
    \coordinate[label=above:$0$] (v1) at (0,\scale);
    \coordinate (v2) at (-\scale,0);
    \coordinate (v3) at ( \scale,0);
    \node (G) at (0,{-1.5*\scale}) {$G_c$};
    \draw[preaction={draw, white, line width=3pt, -}] (v2) -- (v0);
    \draw[preaction={draw, white, line width=3pt, -}] (v3) -- (v0);
    \draw[preaction={draw, white, line width=3pt, -}] (v2) -- (v1);
    \draw[preaction={draw, white, line width=3pt, -}] (v3) -- (v1);
    \draw (v3) -- ($(v3)!.2!(v0)$);
    \draw (v2) -- (v3);
    \draw (v2) -- (v0);
    \filldraw (v0) circle (1.3pt);
    \filldraw (v1) circle (1.3pt);
    \filldraw (v2) circle (1.3pt);
    \filldraw (v3) circle (1.3pt);
\end{scope}
\begin{scope}[xshift={17.5cm*\scale}]
    \coordinate (v0) at (-\scale,\scale);
    \coordinate (v1) at (0,\scale);
    \coordinate (vz) at (\scale,\scale);
    \coordinate (v2) at (-\scale,0);
    \coordinate (v3) at ( \scale,0);
    \node (G) at (0,{-1.5*\scale}) {$\Gg_c=K_5$};
    \draw[preaction={draw, white, line width=3pt, -}] (v2) -- (vz);
    \draw[preaction={draw, white, line width=3pt, -}] (v3) -- (v0);
    \draw[preaction={draw, white, line width=3pt, -}] (v2) -- (v1);
    \draw[preaction={draw, white, line width=3pt, -}] (v3) -- (v1);
    \draw (v2) -- (v0);
    \draw (v3) -- (vz);
    \draw (v2) -- ($(v2)!.2!(vz)$);
    \draw (v3) -- ($(v3)!.2!(v0)$);
    \draw (v2) -- (v3);
    \draw (v0) --  (v1);
    \draw (v1) --  (vz);
    \draw (v0) to[bend left=50] (vz);
    \filldraw (v0) circle (1.3pt);
    \filldraw (v1) circle (1.3pt);
    \filldraw (vz) circle (1.3pt);
    \filldraw (v2) circle (1.3pt);
    \filldraw (v3) circle (1.3pt);
\end{scope}
\end{tikzpicture}
\end{center}
\caption{The calculation of the four-dimensional graphical function $G_a$ by completion. Edge-weights $\neq1$ are indicated. We complete to $\Gg_a$, double transpose
the external vertices to $\Gg_b$, de-complete to $G_b$, remove edges between external vertices to $G_c$, and period complete to $\Gg_c$.}
\label{fig:gfcomp}
\end{figure}

\begin{ex}[Example 3.21 in \cite{gf}]\label{exf31}
Consider the four-dimensional graphical function associated to $G_a$ in Figure \ref{fig:gfcomp}. The completion $\Gg_a$ becomes $\Gg_b$ by a double transposition of external vertices.
The graph $\Gg_b$ de-completes to $G_b$. We obtain
$$
f_{G_a}(z)=f_{\Gg_a}(z)=f_{\Gg_b}(z)=f_{G_b}(z).
$$
The graphical function $f_{G_b}(z)$ can easily be calculated, see Examples \ref{exf32} and \ref{exf33}.
\end{ex}

Internally completed graphs with no edges between external labels and no distinction of external labels are equivalence classes of graphical functions with equal complexity.
The graphs in Example \ref{exempty} are all in the equivalence class of the empty graph.

\subsection{Parametric representation}\label{sectpar}
The position space definition of a graphical function in (\ref{eqAdef}) and (\ref{eqfGdef}) can be replaced by a parametric representation.
The virtue of the parametric representation is that the dimension $D$ enters only as a parameter so that the generalization to $D\in\CC$ is possible.
Moreover, there exists the method of parametric integration by F. Brown and E. Panzer \cite{BrH1,BrH2,Panzer:HyperInt}, see Section \ref{sectparint}.

We follow \cite{par} for the subsequent definition and theorem.

\begin{defn}\label{spanningforest}
Let $P=\{P_1,\ldots ,P_N\}$ denote a partition of the external vertices $\{0,1,z\}$ of a graph $G$.
We write $\sF_G^P$ for the set of all spanning forests $T_1\cup\ldots\cup T_N$ consisting of exactly $N$ (pairwise disjoint) trees $T_i$ such that $P_i \subseteq T_i$.
The \emph{dual spanning forest polynomial} associated to $P$ is
\begin{equation}\label{eq:dual-forpol}
\widetilde\Psi_G^P(\alpha)=\sum_{F\in\sF_G^P}\;\prod_{e\in F}\alpha_e.
\end{equation}
We write $\widetilde\Psi_G$ for $\widetilde\Psi_G^{0,1,z}$ and define the polynomial
\begin{equation}\label{eq:dual-phi}
\widetilde\Phi_G(\alpha,z)=\sum_{abc\in\{01z,0z1,1z0\}}|a-b|^2\widetilde\Psi_G^{ab,c}(\alpha).
\end{equation}
\end{defn}

\begin{ex}
The three-star in Figure \ref{fig:3star} with edges $1, 2, 3$ attached to vertices $0, 1, z$ has
\begin{equation}
\widetilde\Psi_\smallclaw(\alpha)=\alpha_1+\alpha_2+\alpha_3,\quad\widetilde\Phi_\smallclaw(\alpha,z)=\alpha_1\alpha_2+z\zz\alpha_1\alpha_3+(z-1)(\zz-1)\alpha_2\alpha_3.
\end{equation}
\end{ex}

\begin{thm}\label{dualparam}
Let $G$ be a non-empty weighted graph such that the graphical function $f_G(z)$ exists. We label the edges by $1,2,\ldots,|\sE_G|$.
For any set of non-negative integers $n_e$ such that $n_e+\lambda\nu_e>0$ for all edges $e\in\sE_G$ we have the following \emph{dual parametric representation}
of $f_G(z)$ as the convergent projective integral ($\Gamma(x)=\int_0^\infty t^{x-1}\exp(-t)\dd t$ is the gamma function)
\begin{equation}\label{fdualparam}
f_G(z)=\frac{(-1)^{\sum_e\!n_e}\,\Gamma(\lambda N_G)}{\prod_e\Gamma(n_e+\lambda\nu_e)}\int_\Delta \Omega
\Big[\prod_e\alpha_e^{n_e+\lambda\nu_e-1}\partial_{\alpha_e}^{n_e}\Big]
\frac{1}{\widetilde\Phi_G^{\lambda N_G}\widetilde\Psi_G^{\lambda+1-\lambda N_G}},
\end{equation}
where
\begin{equation}\label{Omegadef}
\Omega=\sum_{e=1}^{|\sE_G|}(-1)^{e-1}\alpha_e\dd\alpha_1\wedge\ldots\wedge\widehat{\dd\alpha_e}\wedge\ldots\wedge \dd\alpha_{|\sE_G|}
\end{equation}
is the projective volume form. The integration domain is given by the positive coordinate simplex
\begin{equation}\label{Deltadef}
\Delta=\{(\alpha_1:\alpha_2:\ldots:\alpha_{|\sE_G|})\colon \alpha_e>0\text{ for all }e\in\{1,2,\ldots,|\sE_G|\}\}\subset\PP^{|\sE_G|-1}\RR.
\end{equation}
\end{thm}

If all weights $\nu_e$ are positive, the parametric representation can be dualized by a Cremona transformation, see Corollary 1.8 in \cite{par}.

%%%%%%%%%%%%%%%%%%%%%%%%%%%%%%%%%%%%%%%%%%%%%%%%%%%%%%%%%%%%%%%%%%%%%%%%%%%%%%%%%
\section{Periods}\label{sectper}
Feynman periods are constant graphical functions. If a graph $G$ has an isolated vertex $z$, then the integrand in the corresponding integral (\ref{eqfGdef})
has no dependence on $z$. The Feynman integral depends on the sole invariant $\|z_1-z_0\|$ which is 1 in the context of graphical functions.
A Feynman period is hence given by a graph with two external vertices 0 and 1, see e.g.\ $G_c$ in Figure \ref{fig:gfcomp}.

Feynman periods can be calculated by specializing graphical functions to $z=0$, $z=1$, or $z=\infty$, see e.g.\ \cite{ZZ}.
A more powerful method is to choose a convenient internal vertex in a period graph $G$ and promote it to the external vertex $z$. If the corresponding
graphical function can be calculated, the period can be obtained by integration of $z$ over $\CC$.

\begin{defn}
Let $G$ be a graph with two external labels $0,1$. The Feynman period $P_G$ is the constant graphical function $f_{G\cup\{z\}}$ with an isolated external vertex $z$.
\end{defn}

A combinatorial criterion for convergence of the Feynman period $P_G$ can be obtained from the corresponding result for the graphical function $f_{G\cup\{z\}}$, see Proposition \ref{propconv}.

\begin{prop}[Lemma 3.34 in \cite{gf}]
Let $G$ be a graph with two external labels 0 and 1 such that the Feynman period $P_G$ exists. In $G'$ we promote one internal vertex of $G$
to the external vertex $z$. Then ($\dd^2 z=\dd\Re z\wedge\dd\Im z$)
\begin{equation}\label{PGint}
P_G=(-1)^\lambda\frac{(\lambda-1)!}{(2\lambda-1)!}\int_\CC(z-\zz)^{2\lambda}f_{G'}(z)\frac{\dd^2 z}{2\pi}.
\end{equation}
\end{prop}
\begin{proof}
The proof uses spherical coordinates (\ref{angularcoords}). The volume form in these coordinates $Z^{D-1}\sin^{D-2}\phi^z_1\dd Z\dd\phi^z_1$ translates into $(-4)^{-\lambda}(z-\zz)^{2\lambda}\dd^2z/2$.
The pre-factor originates from the volume of the $(D-2)$-sphere, see \cite{gf}.
\end{proof}
For the evaluation of the integral over the complex plane we can take advantage of the single-valuedness of graphical functions
and use a residue theorem in \cite{gf} (Theorem 2.29).

Like graphical functions, Feynman periods have a completion \cite{Census}.
\begin{defn}\label{defperiodcompl}
Let $G$ be a graph with weighted edges and two external vertices $0,1$. The period completion $\Gg$ of $G$ is the graph obtained from the completion of the graphical function
$G\cup\{z\}$ (with isolated vertex $z$, Definition \ref{defcompl}) and the additional triangle $01$, $0\infty$, $1\infty$ with edge-weights $(\lambda+1)/\lambda$.
In $\Gg$ we delete $z$ and forget the labels $0,1,\infty$ so that $\Gg$ is an unlabeled weighted $2(\lambda+1)/\lambda$-regular graph.
\end{defn}

It is clear by Proposition \ref{propcompl} that the completed period has weights in $\frac{1}{\lambda}\ZZ$ if the uncompleted graph has weights in $\frac{1}{\lambda}\ZZ$.
The period is invariant under completion.
\begin{thm}[Theorem 2.7 in \cite{Census}]\label{thmpercompl}
Let $G_1$ and $G_2$ be two graphs with external vertices $0,1$ and equal period completion $\Gg$. Then
\begin{equation}\label{eqperiodcompl}
P_\Gg\equiv P_{G_1}=P_{G_2}.
\end{equation}
\end{thm}
\begin{proof}
The proof in \cite{Census} relies on conformal transformations of the integrand.
\end{proof}

Completion is a means to handle equivalence classes of graphs with equal Feynman period. Note that completion of periods is more powerful than completion of graphical functions.
One has the freedom to remove any vertex $\infty$ from the completed graph and thereafter pick any two vertices 0 and 1 for the calculation of the period.

\begin{figure}
\begin{align*}
\begin{tikzpicture}[baseline={([yshift=-1.3ex]0,0)}]
    \pgfmathsetmacro{\rad}{3}
    \coordinate (vm)  at (0,0);
    \coordinate[label=right:$1$] (v1)  at ([shift=(0:\rad)]vm);
    \coordinate[label=above right:$2$] (v2)  at ([shift=(45:\rad)]vm);
    \coordinate[label=above:$3$] (v3)  at ([shift=(90:\rad)]vm);
    \coordinate[label=above left:$4$] (v4)  at ([shift=(135:\rad)]vm);
    \coordinate[label=left:$5$] (v5)  at ([shift=(180:\rad)]vm);
    \coordinate[label=below left:$6$] (v6)  at ([shift=(225:\rad)]vm);
    \coordinate (v7)  at ([shift=(270:\rad)]vm);
    \coordinate[label=below right:$n$] (v8)  at ([shift=(315:\rad)]vm);
    \node[black!50] (vdots) at ([shift=(270:{.7*\rad})]vm) {$\bullet \bullet \bullet$};
    \draw (vm) circle (\rad);
    \draw[dashed,black!50] (vm) -- ([shift=(247.5:{.5*\rad})]vm);
    \draw[dashed,black!50] (vm) -- ([shift=(270:{.5*\rad})]vm);
    \draw[dashed,black!50] (vm) -- ([shift=(292.5:{.5*\rad})]vm);
    \node[fill=white] (lvm) at ([shift=(270:.3)]vm) {$0$};
    \filldraw (vm) circle (1.3pt);
    \filldraw (v1) circle (1.3pt);
    \filldraw (v2) circle (1.3pt);
    \filldraw (v3) circle (1.3pt);
    \filldraw (v4) circle (1.3pt);
    \filldraw (v5) circle (1.3pt);
    \filldraw (v6) circle (1.3pt);
    \filldraw (v8) circle (1.3pt);
    \filldraw[black!50] ([shift=(247.5:{\rad})]vm) circle (1.3pt);
    \filldraw[black!50] ([shift=(270:{\rad})]vm) circle (1.3pt);
    \filldraw[black!50] ([shift=(292.5:{\rad})]vm) circle (1.3pt);
    \draw[dashed] (vm) -- (v1);
    \draw[dashed] (vm) -- (v2);
    \draw[dashed] (vm) -- (v3);
    \draw[dashed] (vm) -- (v4);
    \draw[dashed] (vm) -- (v5);
    \draw[dashed] (vm) -- (v6);
    \draw[dashed] (vm) -- (v8);
\end{tikzpicture}
&&&
\begin{tikzpicture}[baseline={([yshift=-1.3ex]0,0)}]
    \pgfmathsetmacro{\rad}{3}
    \coordinate (vm)  at (0,0);
    \coordinate[label=right:$1$] (v1)  at ([shift=(0:\rad)]vm);
    \coordinate[label=above right:$2$] (v2)  at ([shift=(45:\rad)]vm);
    \coordinate[label=above:$3$] (v3)  at ([shift=(90:\rad)]vm);
    \coordinate[label=above left:$4$] (v4)  at ([shift=(135:\rad)]vm);
    \coordinate[label=left:$5$] (v5)  at ([shift=(180:\rad)]vm);
    \coordinate[label=below left:$6$] (v6)  at ([shift=(225:\rad)]vm);
    \coordinate (v7)  at ([shift=(270:\rad)]vm);
    \coordinate[label=below right:$n$] (v8)  at ([shift=(315:\rad)]vm);
    \node[black!50] (vdots) at ([shift=(270:{.7*\rad})]vm) {$\bullet \bullet \bullet$};
    \coordinate (vm0) at ([shift=(157.5:{\rad/2})]vm);
    \coordinate (vmoo) at ([shift=(-22.5:{\rad/2})]vm);
    \draw (vm) circle (\rad);
    \draw[dashed,black!50] (vm0) -- ([shift=(262.5:{.4*\rad})]vm0);
    \draw[dashed,black!50] (vm0) -- ([shift=(285:{.4*\rad})]vm0);
    \draw[dashed,black!50] (vm0) -- ([shift=(307.5:{.4*\rad})]vm0);
    \draw[dashed,black!50] (vmoo) -- ([shift=(232.5:{.4*\rad})]vmoo);
    \draw[dashed,black!50] (vmoo) -- ([shift=(255:{.4*\rad})]vmoo);
    \draw[dashed,black!50] (vmoo) -- ([shift=(277.5:{.4*\rad})]vmoo);
    \node (lvm0) at ([shift=(157.5:{\rad/2+.3})]vm) {$0$};
    \node (lvmoo) at ([shift=(-22.5:{\rad/2+.3})]vm) {$\infty$};
    \draw[dotted] (vm0) -- (vmoo);
    \filldraw (vm0) circle (1.3pt);
    \filldraw (vmoo) circle (1.3pt);
    \filldraw (v1) circle (1.3pt);
    \filldraw (v2) circle (1.3pt);
    \filldraw (v3) circle (1.3pt);
    \filldraw (v4) circle (1.3pt);
    \filldraw (v5) circle (1.3pt);
    \filldraw (v6) circle (1.3pt);
    \filldraw (v8) circle (1.3pt);
    \filldraw[black!50] ([shift=(247.5:{\rad})]vm) circle (1.3pt);
    \filldraw[black!50] ([shift=(270:{\rad})]vm) circle (1.3pt);
    \filldraw[black!50] ([shift=(292.5:{\rad})]vm) circle (1.3pt);
%[]
    \draw[dashed] (vm0) -- (v1);
    \draw[dashed] (vm0) -- (v2);
    \draw[dashed] (vm0) -- (v3);
    \draw[dashed] (vm0) -- (v4);
    \draw[dashed] (vm0) -- (v5);
    \draw[dashed] (vm0) -- (v6);
    \draw[dashed] (vm0) -- (v8);
    \draw[dashed] (vmoo) -- (v1);
    \draw[dashed] (vmoo) -- (v2);
    \draw[dashed] (vmoo) -- (v3);
    \draw[dashed] (vmoo) -- (v4);
    \draw[dashed] (vmoo) -- (v5);
    \draw[dashed] (vmoo) -- (v6);
    \draw[dashed] (vmoo) -- (v8);
\end{tikzpicture}
\end{align*}
\caption{The wheel with $n$ spokes $W\!S_{n,D}$ in $D=2\lambda+2$ dimensions and its completion $\overline{W\!S}_{n,D}$. 
The solid edges on the rim 
$\solidedge$  
have weight $1$,
the dashed spokes  
$\dashededge$
have weight $1/\lambda$ and the 
dotted edge 
$\dottededge$
between $0$ and $\infty$ in $\overline{W\!S}_{n,D}$ has weight $(D-n)/\lambda$.}
\label{fig:wheels}
\end{figure}

\begin{ex}\label{exf32}
The graph $G_c$ in Figure \ref{fig:gfcomp} completes to $K_5$, the complete graph with five vertices. In four dimensions its period is \cite{geg}
$$
P_{K_5}=6\zeta(3),
$$
where $\zeta$ is the Riemann zeta function.
\end{ex}

\begin{ex}[Complete graphs $\K_{n,D}$]\label{Kn}
We generalize Example \ref{exf32} to complete graphs $\K_{n,D}$ of uniform weight $\nu$ in $D$ dimensions. By $2(\lambda+1)/\lambda$-regularity the weight of all edges is
$$
\nu=\frac{D}{(n-1)\lambda}.
$$
The period of the graph $\K_{3,D}$ is 1.

The period of the graph $\K_{4,D}$ is given by (\ref{eqconvolute}),
$$
P_{\K_{4,D}}=\frac{\Gamma(D/6)^3}{\Gamma(D/3)^3}.
$$
It is rational for $D\in6\ZZ$ where the weight is in $\frac{1}{\lambda}\ZZ$.

For $n\geq5$ we impose the constraint $(n-1)|D$ so that $\lambda\nu\in\ZZ$. The graph $\K_{5,D}$ is constructible (see Section \ref{sectconstructible}).
By Conjecture \ref{concon} we expect that the period is a multiple zeta value (MZV) of weight $\leq3$,
i.e.\ a rational linear combination of 1 and $\zeta(3)$. With {\tt HyperlogProcedures} \cite{Shlog} we get
\begin{align*}
P_{\K_{5,4}}&=6\zeta(3)\qquad\qquad\qquad\qquad\text{(see Example \ref{exf32}),}\\
P_{\K_{5,8}}&=\frac{3}{5}\zeta(3)-\frac{7}{10},\\
P_{\K_{5,12}}&=\frac{1}{80}\zeta(3)-\frac{173}{11520},\\
P_{\K_{5,16}}&=\frac{47}{617760}\zeta(3)-\frac{73219}{800616960},\\
P_{\K_{5,20}}&=\frac{79}{448081920}\zeta(3)-\frac{984571}{4645713346560}.
\end{align*}
Because the numerical evaluation of $P_{\K_{5,D}}$ becomes very small for large $D$, periods of complete graphs with five vertices provide a sequence of rational approximations of $\zeta(3)$.

Beyond five vertices complete graphs are not constructible. For six vertices, double use of unique triangle reductions, see Section \ref{sectunique} and Example \ref{exK6}, gives \cite{Shlog}
\begin{align*}
P_{\K_{6,10}}&=-\frac{5}{7}\zeta(5)+\frac{19}{30}\zeta(3)-\frac{13}{630},\\
P_{\K_{6,20}}&=\frac{25}{5444195328}\zeta(5)-\frac{3886573}{1440534083788800}\zeta(3)-\frac{1417432087}{933466086295142400},\\
P_{\K_{6,30}}&=-\frac{109}{3737370802913280000000}\zeta(5)+\frac{3035142067981}{155555277863119316582400000000000}\zeta(3)\\
&\quad+\,\frac{684204291515294677}{100799820055301317145395200000000000000}.
\end{align*}

For $n\geq7$ the computation of $\K_{n,D}$ is complicated. With \cite{Shlog} we get (see Example \ref{K7ex})
$$
P_{\K_{7,6}}=360\zeta(3,5)+690\zeta(3)\zeta(5)-\frac{29}{315}\pi^8,
$$
where $\zeta(3,5)=\sum_{0<k<\ell}k^{-3}\ell^{-5}$ is an MZV of weight eight.
\end{ex}

\begin{ex}[Wheels $W\!S_{n,D}$]\label{exWSn}
Another generalization of Example \ref{exf32} is given by the wheels with $n$ spokes in Figure \ref{fig:wheels}
where the spokes have weight $1/\lambda$ and the rim has weight 1. In their completion the edge $0\infty$ has weight $(D-n)/\lambda$. The period
is convergent if $(D-n)/\lambda\leq1\Leftrightarrow D\leq2n-2$ (see Example \ref{indiv}). In Example \ref{exWSn1} we will obtain
a closed expression for the period of $W\!S_{n,D}$ using radial and angular graphical functions (the Gegenbauer method).

For even $D$ the result is a $\QQ$-linear combination of odd single zetas with weights ranging from $2n-2\lambda-1$ to $2n-3$.
The wheels $W\!S_{n,D}$ are constructible by appending weight 1 edges, only. The wheel has $n+1$ vertices, so that the maximum weight is consistent with Theorem \ref{conthm}.

For odd $D$ the period of $W\!S_{n,D}$ is a polynomial in $\pi^2$ of degree $n-1$.
\end{ex}

\subsection{Parametric representation}\label{sectparper}
In Section \ref{sectpar} we gave a (dual) parametric representation for graphical functions. As a corollary we obtain an analogous result for periods.
For an unlabeled graph $G$ we define the dual graph (Kirchhoff)-polynomial \cite{KIR} as (see (\ref{eq:dual-forpol}))
\begin{equation}\label{kirchdef}
\widetilde\Psi_G=\widetilde\Psi_G^\emptyset=\sum_T\prod_{e\in T}\alpha_e
\end{equation}
where the sum is over spanning trees.

\begin{cor}[Corollary of Theorem \ref{dualparam}]\label{dualparamper}
Let $\Gg$ be a weighted $2(\lambda+1)/\lambda$-regular graph such that the Feynman period $P_\Gg$ exists in $D=2\lambda+2\geq3$ dimensions.
Let $\infty\in\sV_\Gg$ be a vertex in $\Gg$ and $G=\Gg\backslash\{\infty\}$ be a de-completion of $\Gg$.
We label the edges of $G$ by $1,2,\ldots,|\sE_G|$. For any set of non-negative integers $n_e$ such that $n_e+\lambda\nu_e>0$ for all edges $e\in\sE_G$ we have the following
\emph{dual parametric representation} of $P_G\equiv P_\Gg$ as the convergent projective integral
\begin{equation}\label{fdualparamper}
P_G=\frac{(-1)^{\sum_e\!n_e}\,\Gamma(\lambda+1)}{\prod_{e\in\sE_G}\Gamma(n_e+\lambda\nu_e)}\int_\Delta\Omega\Big[\prod_{e\in\sE_G}\alpha_e^{n_e+\lambda\nu_e-1}\partial_{\alpha_e}^{n_e}\Big]
\frac{1}{\widetilde\Psi_G^{\lambda+1}},
\end{equation}
where $\Omega$ is defined in (\ref{Omegadef}). The integration domain is given by the positive coordinate simplex (\ref{Deltadef}).
\end{cor}

\begin{proof}
The theorem is trivial in the case that $G$ is a single edge of weight $(\lambda+1)/\lambda$. We may assume that $G$ has at least three vertices. Let 0, 1 be two vertices in $G$
such that the edge $01$ has positive weight $\nu_{01}$ (such an edge always exists). We define $G_{0,1,z}=G\cup\{z\}\backslash\{01\}$ as the graph which is obtained
from $G$ by removing the edge 01 and adding an isolated vertex $z$. The completion of $G_{0,1,z}$ has weight $N_{\Gg_{0,1,z}}=0$, see Proposition \ref{propcompl}.
The graph of the completed period $\Gg$ has the edges of $\Gg_{0,1,z}$ plus three edges of weight $(\lambda+1)/\lambda$, see Definition \ref{defperiodcompl}.
So, $N_G=(\lambda+1)/\lambda$ if one considers 0 and 1 as external vertices in $G=\Gg\backslash\{\infty\}$. We find
\begin{equation}\label{nu01}
N_{G_{0,1,z}}=\frac{\lambda+1}{\lambda}-\nu_{01}.
\end{equation}
Moreover, we obtain from Definition \ref{spanningforest} that
$$
\widetilde\Psi_G=\alpha_{01}\widetilde\Psi_{G_{0,1,z}}+\widetilde\Phi_{G_{0,1,z}}.
$$
Using integration by parts it suffices to prove (\ref{fdualparamper}) for $n_{01}=0$. In an affine coordinate patch where some $\alpha_e=1$ for $e\neq01$ the integral over $\alpha_{01}$
in (\ref{fdualparamper}) is
$$
\int_0^\infty\frac{\alpha_{01}^{\lambda\nu_{01}-1}\dd\alpha_{01}}{(\alpha_{01}\widetilde\Psi_{G_{0,1,z}}+\widetilde\Phi_{G_{0,1,z}})^{\lambda+1}}
=\int_0^\infty\frac{x^{\lambda\nu_{01}-1}\dd x}{(x+1)^{\lambda+1}}\,\cdot\,\frac{\widetilde\Phi_{G_{0,1,z}}^{\lambda\nu_{01}-\lambda-1}}{\widetilde\Psi_{G_{0,1,z}}^{\lambda\nu_{01}}}
$$
where we substituted $\alpha_{01}=\widetilde\Phi_{G_{0,1,z}}x/\widetilde\Psi_{G_{0,1,z}}$. We write the first factor on the right hand side as projective integral
$\int x^{\lambda\nu_{01}-1}y^{\lambda(1-\nu_{01})}/{(x+y)^{\lambda+1}}\Omega$ with coordinates $(x:y)$. The affine patch $x+y=1$ with $\Omega=\dd x$ gives
the definition of the Euler beta function. For the integral over $\alpha_{01}$ in (\ref{fdualparamper}) we hence get
$$
\frac{\Gamma(\lambda\nu_{01})\Gamma(\lambda+1-\lambda\nu_{01})}{\Gamma{(\lambda+1})}\frac{1}{\widetilde\Phi_{G_{0,1,z}}^{\lambda+1-\lambda\nu_{01}}\widetilde\Psi_{G_{0,1,z}}^{\lambda\nu_{01}}}.
$$
With this result and (\ref{nu01}) Equation (\ref{fdualparamper}) reduces to (\ref{fdualparam}) for $G_{0,1,z}$.
\end{proof}
In the case of positive weights $\nu_e$ a Cremona transformation provides a parametric representation of $P_G$ with integrand $(\prod_e\alpha_e^{\lambda(1-\nu_e)})/\Psi_G^{\lambda+1}$
where $\Psi_G=\sum_T\prod_{e\notin T}\alpha_e$ is the Kirchhoff polynomial.

\begin{remark}
Examples of parametrically represented Feynman periods in various even dimensions are invariant differential forms on complexes of graphs in \cite{Bidf}.
\end{remark}

%%%%%%%%%%%%%%%%%%%%%%%%%%%%%%%%%%%%%%%%%%%%%%%%%%%%%%%%%%%%%%%%%%%%%%%%%%%%%%%%%
\section{Single-valued integration and generalized single-valued hyperlogarithms}\label{sectint}
The theory of graphical functions reduces the calculation of corresponding Feynman integrals to single-valued integrations.

For rational functions it suffices to know that a single-valued primitive of $1/(z-a)$ is $\log[(z-a)(\zz-\aaa)]$ for any constant $a\in\CC$.
Finding single-valued primitives in general (or even proving their existence) is an open problem. One expects that under rather mild conditions single-valued functions have
single-valued primitives. First advances are due to F. Brown in the context of hyperlogarithms which generalize the above example to higher weights \cite{BrSVMP,BrSVMPII}.
Brown's approach was the mathematically appealing use of generating functions. It turned out that in the context of graphical functions, generating functions are
unwieldy because they carry the information of all hyperlogarithms up to a certain weight while one is only interested in a few very specific ones.
In \cite{gf} the second author overcame the difficulty by constructing a bootstrap algorithm. The essence of the algorithm is a commutative hexagon which allows one to
reduce the single-valued integration of hyperlogarithms to multi-valued (anti-)integrations and single-valued integrations of lower weights \cite{numfunct}.

The main problem for using single-valued integration in the context of graphical function, however, was that the function space of single-valued hyperlogarithms is too restrictive.
Only the simplest graphical functions can be expressed in terms of single-valued hyperlogarithms.
It turned out that it is necessary to generalize single-valued hyperlogarithms to include primitives of differential forms $\dd z/(az\zz+bz+c\zz+d)$, $a,b,c,d\in\CC$.
Examples are the primitives of $\log(z\zz)/(z-1/\zz)$ or of $D(z)/(z-\zz)$ (where $D$ is the Bloch-Wigner dilogarithm (\ref{D})). The latter example was already studied
in \cite{Duhr}. The construction of these generalized single-valued hyperlogarithms (GSVHs) also relies on the commutative hexagon but the theory is much more involved \cite{GSVH}.

In fact, the full theory of GSVHs is almost as comprehensive as the theory of graphical functions. It became the second pillar in the framework that uses graphical functions
for calculations in QFT. (The third pillar is the extension to non-integer dimensions in order to handle Feynman integrals with divergences \cite{7loops}.)
The main result of the theory of GSVHs is that in the function space of GSVHs single-valued (anti-)primitives always exist and that they can be constructed very efficiently
(for suitable sets of singularities). As it is not easy to summarize the theory in a few paragraphs we refer the interested reader to \cite{GSVH}.

There exist graphical functions which are not GSVHs (in four dimensions, e.g., the graph $\Gg_8$ in Figure~\ref{fig:irred7}).
For these graphical functions the general theory is still valid but there presently exists no handle on their computation. In Definition \ref{defintsv} we
merely assume the existence of single-valued primitives so that the theory of graphical function extends to any future developments in the theory of single-valued integration.

%%%%%%%%%%%%%%%%%%%%%%%%%%%%%%%%%%%%%%%%%%%%%%%%%%%%%%%%%%%%%%%%%%%%%%%%%%%%%%%%%
\section{Elementary identities for graphical functions}\label{sectgfid1}
In many cases graphical functions can be calculated. A particularly tractable case are graphical functions that emerge from the empty graph by a series of combinatorial
operations: adding edges between external vertices (Section \ref{sectextedge}), products and factors (Sections \ref{sectprod} and \ref{sectfactor}),
permutation of external vertices (Section \ref{sectperm}), and appending an edge to the vertex $z$ (Section \ref{sectappedge}). In this case the graphical function is {\em constructible}
(see Section \ref{sectconstructible}) and it can be computed in terms of GSVHs \cite{GSVH}. A Maple implementation is in {\tt HyperlogProcedures} \cite{Shlog}.

\subsection{Products}\label{sectprod}
Assume the graph of an (internally) completed graphical function $\Gg$ disconnects upon the removal of its external vertices into $\Gg_1$ and $\Gg_2$ where the removed edges
(adjacent to the removed external vertices) are added to corresponding graphs. Then the Feynman integral (\ref{eqAdef}) factors into two disjoint sets of integration variables.
Accordingly, the graphical function $f_\Gg(z)$ factors,
\begin{align}
f_{\overline{G}}(z) &= f_{\overline{G}_1}(z)f_{\overline{G}_2}(z),
\intertext{which can be depicted diagrammatically as}
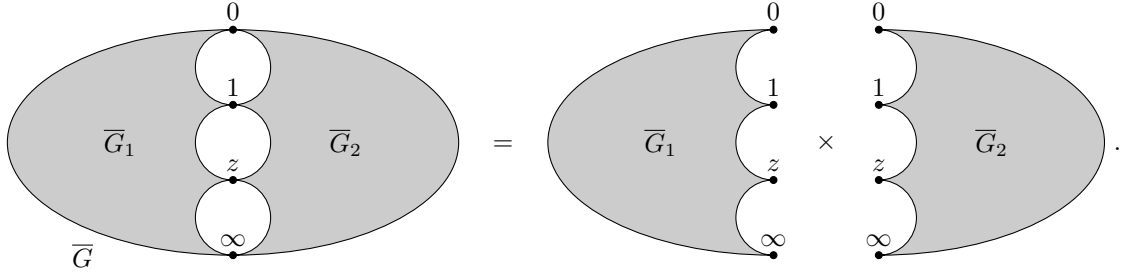
\begin{tikzpicture}[baseline={([yshift=-1.3ex]current bounding box.center)}]
    \coordinate[label=above:$0$] (v0) at (0,3);
    \coordinate[label=above:$1$] (v1) at (0,2);
    \coordinate[label=above:$z$] (vz) at (0,1);
    \coordinate[label=above:$\infty$] (voo) at (0,0);
    \filldraw (v0) circle (1.3pt);
    \filldraw (v1) circle (1.3pt);
    \filldraw (vz) circle (1.3pt);
    \filldraw (voo) circle (1.3pt);
    \draw[fill=black!20] (v0) arc (90:270:3 and 1.5) arc (270:90:.5) arc (270:90:.5) arc (270:90:.5);
    \draw[fill=black!20] (v0) arc (90:-90:3 and 1.5) arc (-90:90:.5) arc (-90:90:.5) arc (-90:90:.5);
    \node (G1) at (-1.5,1.5) {$\overline{G}_1$};
    \node (G2) at (+1.5,1.5) {$\overline{G}_2$};
    \node (G) at (-2,0) {$\overline{G}$};
\end{tikzpicture}
\quad &= \quad
\begin{tikzpicture}[baseline={([yshift=-1.3ex]current bounding box.center)}]
    \coordinate[label=above:$0$] (v0) at (.3,3);
    \coordinate[label=above:$1$] (v1) at (.3,2);
    \coordinate[label=above:$z$] (vz) at (.3,1);
    \coordinate[label=above:$\infty$] (voo) at (.3,0);
    \coordinate[label=above:$0$] (w0) at (1.7,3);
    \coordinate[label=above:$1$] (w1) at (1.7,2);
    \coordinate[label=above:$z$] (wz) at (1.7,1);
    \coordinate[label=above:$\infty$] (woo) at (1.7,0);
    \filldraw (v0) circle (1.3pt);
    \filldraw (v1) circle (1.3pt);
    \filldraw (vz) circle (1.3pt);
    \filldraw (voo) circle (1.3pt);
    \filldraw (w0) circle (1.3pt);
    \filldraw (w1) circle (1.3pt);
    \filldraw (wz) circle (1.3pt);
    \filldraw (woo) circle (1.3pt);
    \draw[fill=black!20] (v0) arc (90:270:3 and 1.5) arc (270:90:.5) arc (270:90:.5) arc (270:90:.5);
    \draw[fill=black!20] (w0) arc (90:-90:3 and 1.5) arc (-90:90:.5) arc (-90:90:.5) arc (-90:90:.5);
    \node (G1) at (-1.2,1.5) {$\overline{G}_1$};
    \node (G2) at (+3.2,1.5) {$\overline{G}_2$};
    \node (t) at (1,1.5) {$\times$};
    \node (G) at (-2,0) {$\phantom{\overline{G}}$};
\end{tikzpicture}~~. \notag
\end{align}

\subsection{Edges between external vertices}\label{sectextedge}
Edges between external vertices provide rational factors by the product reduction from the previous subsection and Example \ref{exempty}.

\begin{ex}\label{exf33}
For the graphical function $f_{G_b}(z)$ in Figure \ref{fig:gfcomp} we obtain in four dimensions
$$
f_{G_b}(z)=\frac{f_{G_c}(z)}{z\zz(z-1)(\zz-1)}=\frac{P_{K_5}}{z\zz(z-1)(\zz-1)}=\frac{6\zeta(3)}{z\zz(z-1)(\zz-1)}.
$$
Note that $f_{G_c}(z)$ is a constant graphical function, see Section \ref{sectper}. It equals the Feynman period of the completed graph $K_5$, see Example \ref{exf32}.
\end{ex}

\subsection{Period factors}\label{sectfactor}

We can interpret Example \ref{exf33} in the following way: An (internally) completed graphical function with isolated vertex $\infty$ evaluates to a Feynman period
times a weighted triangle between the external vertices $0,1,z$. This generalizes to the following proposition.

\begin{prop}\label{propfactor}
Let $\Gg$ be an (internally) completed graphical function with a three-vertex-split into $G_1$ and $G_2$.
We include the split vertices $a,b,c$ together with their corresponding edges to $G_1$ or $G_2$ (respectively) and assume that $G_1\backslash\{a,b,c\}$ has only internal vertices.
We add (unique) weighted triangles $(ab,ac,bc)$ to $G_1$ and to $G_2$ to obtain the (unlabeled) completed period $\Gg_1$ and the (internally) completed graph $\Gg_2$. Then
\begin{align}
\label{eqfactor}
f_\Gg(z)&=P_{\Gg_1}f_{\Gg_2}(z),
\intertext{which can be depicted diagrammatically as}
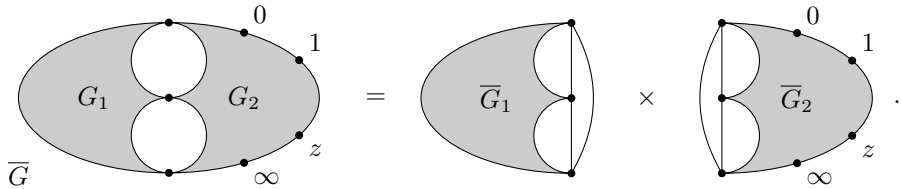
\begin{tikzpicture}[baseline={([yshift=-.7ex]current bounding box.center)}]
    \coordinate (va) at (0,2);
    \coordinate (vb) at (0,1);
    \coordinate (vc) at (0,0);
    \filldraw (va) circle (1.3pt);
    \filldraw (vb) circle (1.3pt);
    \filldraw (vc) circle (1.3pt);
    \draw[fill=black!20] (va) arc (90:270:2 and 1) arc (270:90:.5) arc (270:90:.5);
    \draw[fill=black!20] (va) arc (90:-90:2 and 1) arc (-90:90:.5) arc (-90:90:.5);
    \node (G1) at (-1,1) {${G}_1$};
    \node (G2) at (+1,1) {${G}_2$};
    \node (G) at (-2,0) {$\overline{G}$};
    \coordinate[label=above right:$0$] (v0)  at ([shift={(vb)}]60:2 and 1);
    \coordinate[label=above right:$1$] (v1)  at ([shift={(vb)}]30:2 and 1);
    \coordinate[label=below right:$z$] (vz)  at ([shift={(vb)}]-30:2 and 1);
    \coordinate[label=below right:$\infty$] (voo) at ([shift={(vb)}]-60:2 and 1);
    \filldraw (v0) circle (1.3pt);
    \filldraw (v1) circle (1.3pt);
    \filldraw (vz) circle (1.3pt);
    \filldraw (voo) circle (1.3pt);
\end{tikzpicture}
\quad
&=
\quad
\begin{tikzpicture}[baseline={([yshift=-.7ex]current bounding box.center)}]
    \coordinate (va) at (-1,2);
    \coordinate (vb) at (-1,1);
    \coordinate (vc) at (-1,0);
    \coordinate (wa) at (1,2);
    \coordinate (wb) at (1,1);
    \coordinate (wc) at (1,0);
    \filldraw (va) circle (1.3pt);
    \filldraw (vb) circle (1.3pt);
    \filldraw (vc) circle (1.3pt);
    \filldraw (wa) circle (1.3pt);
    \filldraw (wb) circle (1.3pt);
    \filldraw (wc) circle (1.3pt);
    \draw (va) -- (vb) -- (vc);
    \draw (wa) -- (wb) -- (wc);
    \draw (va) to[bend left] (vc);
    \draw (wa) to[bend right] (wc);
    \draw[fill=black!20] (va) arc (90:270:2 and 1) arc (270:90:.5) arc (270:90:.5);
    \draw[fill=black!20] (wa) arc (90:-90:2 and 1) arc (-90:90:.5) arc (-90:90:.5);
    \node (G1) at (-2,1) {$\overline{G}_1$};
    \node (G2) at (+2,1) {$\overline{G}_2$};
    \coordinate[label=above right:$0$] (v0)  at ([shift={(wb)}]60:2 and 1);
    \coordinate[label=above right:$1$] (v1)  at ([shift={(wb)}]30:2 and 1);
    \coordinate[label=below right:$z$] (vz)  at ([shift={(wb)}]-30:2 and 1);
    \coordinate[label=below right:$\infty$] (voo) at ([shift={(wb)}]-60:2 and 1);
    \filldraw (v0) circle (1.3pt);
    \filldraw (v1) circle (1.3pt);
    \filldraw (vz) circle (1.3pt);
    \filldraw (voo) circle (1.3pt);
    \node (t) at (0,1) {$\times$};
\end{tikzpicture}~~.
\notag
\end{align}
\end{prop}
\begin{proof}
We may add to $\Gg$ a pair of edges $ab$ with total weight zero without changing the graphical function $f_\Gg(z)$.
We move one edge of the pair to $G_1$ whereas the other one goes to $G_2$. Likewise we generate pairs of edges $ac$ and $bc$. We can adjust the weights of these pairs such that
the vertices $a$, $b$, $c$ have total weight zero in $G_1$. This procedure does not affect completion so that the right hand side of (\ref{eqfactor}) is insensitive
to the extra edges.

We thus may assume without restriction that the total weights of the vertices $a$, $b$, $c$ in $G_1$ are zero. Hence, the total weights of $a$, $b$, $c$ in $\Gg_2$ equal
the total weights of $a$, $b$, $c$ in the (internally) completed graph $\Gg$. Therefore, $\Gg_2$ is (internally) completed.

Completion adds an isolated vertex $\infty$ to $G_1$. By Proposition \ref{propcompl} we have $N_{G_1}=N_{\Gg_1}=0$ and by (\ref{eqfA1}) the Feynman integral $A_{G_1}$ with external
vertices $a$, $b$, $c$ equals the graphical function $f_{\Gg_1}(z)$ with invariants(\ref{eqinvs}). By Theorem \ref{thmperm} we can simultaneously swap the external vertices $0,1$ and $z,\infty$
without changing the graphical function. Now $z$ is an isolated vertex. This implies that the Feynman integral of $G_1$ is constant and factors from the integral.
By Theorem \ref{thmpercompl} we have $f_{\Gg_1}(z)=P_{\Gg_1}$ yielding (\ref{eqfactor}).
\end{proof}

If, after the removal of three vertices, a graphical function has a bridge of negative weight $-k/\lambda$, $k=1,2,\ldots$, (and one part is fully internal),
then a period factorization is still possible with a technique that is a straightforward generalization of the method in Section \ref{sectappendneg}.

%%%%%%%%%%%%%%%%%%%%%%%%%%%%%%%%%%%%%%%%%%%%%%%%%%%%%%%%%%%%%%%%%%%%%%%%%%%%%%%%%
\section{Appending an edge}\label{sectappedge}
In this section we consider graphs which have a single edge $e$ with weight $\lambda\nu_e\in\ZZ$ connecting the external vertex $z$ to an internal vertex.
By Theorem \ref{thmperm} this also covers the case that any other external vertex has a single edge.
By convergence we have $\nu_e\leq 1$ (see Example \ref{uvdiv}). We distinguish the cases $\nu_e=1$ (the main case, Section \ref{sectnu1}),
$0<\nu_e<1$ (Section \ref{sectnu01}), and $\nu_e<0$ (Section \ref{sectappendneg}).

\subsection{Weight \texorpdfstring{$\nu_e=1$}{nu(e)=1}}\label{sectnu1}
If a graph $G_1$ has a single edge $e$ of weight $\nu_e=1$ connecting the vertex $z$ to an internal vertex of $G_1$,
then the effective Laplace equation (\ref{eqdiff}) gives a relation between the graphical functions $f_{G_1}(z)$ and $f_G(z)$ in \eqref{eqappend}.

\begin{lem}[Proposition 3.22 in \cite{gf}]\label{lemdiff}
In the situation of \eqref{eqappend} where the edge $e$ that is attached to $z$ in $G_1$ has weight $\nu_e=1$ we have
\begin{equation}\label{eqdiff1}
\Delta_{\lambda-1}(z-\zz)^\lambda f_{G_1}(z)=-\frac{1}{(\lambda-1)!}(z-\zz)^\lambda f_G(z)
\end{equation}
with the effective Laplacian $\Delta_{\lambda-1}=\partial_z\partial_\zz+\lambda(\lambda-1)/(z-\zz)^2$.
\end{lem}
\begin{proof}
The proof uses spherical coordinates (\ref{angularcoords}), see \cite{gf}.
\end{proof}

We solve the differential equation in three steps: Firstly, we construct a general solution in the space of single-valued functions (Theorem \ref{thm2}).
Then, we prove in Theorem \ref{thm3} that the solution is unique in the space of functions with general property (G3) of graphical functions (Theorem \ref{thm1}).
Finally, we give an algorithm in Section \ref{sectappalg} that picks the unique graphical function in the general solution of (\ref{eqdiff1}).

With the solution of (\ref{eqdiff1}) we construct the graphical function $G_1$ from $G$. In practice, we need to work in a function space where single-valued primitives can
be computed. A good such function space are generalized single-valued hyperlogarithms (GSVHs) \cite{GSVH}.
If $f_G(z)$ is a GSVH, then also $f_{G_1}(z)$ is a GSVH, so that the algorithm can be iterated.

The algorithm of appending an edge is the essence of the theory of graphical functions. It is this highly non-trivial construction that makes so many graphical functions
calculable. With the rare exception of the Gegenbauer method in Section \ref{sectgeg}, appending an edge is the only situation where actual computations in terms of single-valued
integrations are performed. All other identities in Sections \ref{sect4split} to \ref{sectextdiff} have the purpose to reduce the calculation of a graphical function to a situation
where an edge can be appended to a simpler graph. Note that the algorithm of appending an edge extends to non-integer dimensions in the context of dimensional regularization \cite{numfunct,7loops}.

\begin{defn}\label{defintsv}
The single-valued integral $\intsv\dd z$ (or $\intsv\dd\zz$) with respect to $z$ (or $\zz$) is a fixed vector space endomorphism of $\sS\sV_{\{0,1,\infty\}}$ such that
$\partial_z\intsv \dd z=\partial_\zz\intsv \dd\zz=\mathrm{id}_{\sS\sV_{\{0,1,\infty\}}}$.
If single-valued integrals do not exist on the whole space $\sS\sV_{\{0,1,\infty\}}$, we restrict ourselves to a maximum subspace on which single-valued integrals exist, see Section \ref{sectint}.
Such a function space always contains the space of GSVHs \cite{GSVH}. If a restriction is necessary, the subsequent results inherit this restriction.
\end{defn}

\begin{remark}\mbox{}
\begin{enumerate}
\item Definition \ref{defintsv} does not uniquely specify $\intsv$. Single-valued integration with respect to $z$, e.g., is only defined up to the addition of rational functions in $\zz$
with poles at $0$ and $1$. For solving (\ref{eqdiff1}) we can use any version of $\intsv$.
\item In \cite{GSVH} efficient algorithms for single-valued integration of GSVHs with respect to $z$ and $\zz$ are presented.
\item The single-valued integration of GSVHs in \cite{Shlog} can produce anti-holomorphic poles in $\CC[(\zz-\sss)^{-1}]$ for $\sss\neq0,1$.
These poles need to be subtracted to obtain an endomorphism of $\sS\sV_{\{0,1,\infty\}}$.
\item It is natural to assume that single-valued integration in the sense of Definition \ref{defintsv} always exist. However, it is unclear how to construct single-valued primitives
for functions which are not GSVHs. By lack of a general construction, solving (\ref{eqdiff1}) (and hence appending edges) is practically limited to GSVHs.
\end{enumerate}
\end{remark}

For $\lambda-1=n=0,1,2,\ldots$ we define the following integral operators
\begin{align}\label{Ipm}
&I^+_n:(z-\zz)^n\sS\sV_{\{0,1,\infty\}}\rightarrow\sS\sV_{\{0,1,\infty\}},\quad I^-_n:\sS\sV_{\{0,1,\infty\}}\rightarrow(z-\zz)^{-n}\sS\sV_{\{0,1,\infty\}},\nonumber\\
&f(z)\mapsto I^{\pm}_nf(z)=\sum_{k=0}^n(-1)^{n-k}\frac{(n+k)!}{(n-k)!k!^2}(z-\zz)^{\pm k}\intsv(z-\zz)^{\mp k}f(z)\dd z.
\end{align}
Moreover, we define the differential operators $D_0=0$,
\begin{align}\label{Dn}
D_n&=\sum_{k=1}^n(-1)^{k-1}\frac{(n-k)!}{(n+k)!}(z-\zz)^k(\partial_z\partial_\zz)^{k-1}(z-\zz)^k,\quad n\geq1,\nonumber\\
d_n&=\sum_{k=0}^n(-1)^k\frac{(n+k)!}{(n-k)!k!}\frac{1}{(z-\zz)^k}\partial_z^{n-k}.
\end{align}
The complex conjugate of $d_n$ is $\ddd_n$. Because (anti-)differentiation lowers the degree in $z-\zz$ by one, we have for any $\ell\in\ZZ$
\begin{equation}\label{Dndef}
D_n:(z-\zz)^\ell\sS\sV_{\{0,1,\infty\}}\rightarrow(z-\zz)^{\ell+2}\sS\sV_{\{0,1,\infty\}},\quad d_n\,(\ddd_n):(z-\zz)^\ell\sS\sV_{\{0,1,\infty\}}\rightarrow(z-\zz)^{\ell-n}\sS\sV_{\{0,1,\infty\}}.
\end{equation}
With these definitions we obtain the following theorems.

\begin{thm}\label{thm2a}
The kernel of $\Delta_n$ on $\CC\backslash\RR$ is $d_nh(z)+\ddd_n\hh(\zz)$ for arbitrary (anti-)holomorphic functions $h(z)$ and $\hh(\zz)$.
\end{thm}
In the above theorem $\hh$ is not (necessarily) the complex conjugate of $h$. The proof of Theorem \ref{thm2a} is in Section \ref{sectpf2a}.

\begin{thm}\label{thm2}
Let $n=0,1,2,\ldots$ and
\begin{equation}\label{eqIndef}
\sI_n:(z-\zz)^n\sS\sV_{\{0,1,\infty\}}\rightarrow(z-\zz)^{-n}\sS\sV_{\{0,1,\infty\}},\quad\sI_nf(z)=I^-_n\intsv\partial_zI^+_nf(z)\dd\zz.
\end{equation}
Assume $f\in(z-\zz)^n\sS\sV_{\{0,1,\infty\}}$ such that $\sI_nf$ exists. For any $F\in(z-\zz)^{-n}\sS\sV_{\{0,1,\infty\}}$ with $\Delta_nF(z)=f(z)$
there exist (anti-)meromorphic functions $\phi\in\CC[z,z^{-1},(z-1)^{-1}]$, $\overline\phi\in\CC[\zz,\zz^{-1},(\zz-1)^{-1}]$ and polynomials $p_0$, $p_1$ of degrees $\leq2n$ such that
\begin{equation}\label{eqthm2}
F(z)=[\sI_n+D_n(1-\Delta_n\sI_n)]f(z)+d_nh(z)+\ddd_n\hh(\zz),
\end{equation}
where
\begin{equation}\label{eqhdef}
h(z)=\phi(z)+\sum_{s=0,1}p_s(z)\log(z-s),\quad\hh(\zz)=\overline{\phi}(\zz)+(-1)^n\sum_{s=0,1}p_s(\zz)\log(\zz-\sss).
\end{equation}
\end{thm}
The proof of Theorem \ref{thm2} is in Section \ref{sectpf2}.

\begin{ex}
For $n=0$ the differential operator $\Delta_0=\partial_z\partial_\zz$ factors.
We get $I_0^\pm=\intsv\dd z$, $D_0=0$, $d_0=\ddd_0=1$, and $\sI_0=\intsv\dd z\intsv\dd\zz$. Equations (\ref{eqthm2}) and (\ref{eqhdef}) become
$$
F(z)=\intsv\intsv f(z)\dd\zz\dd z+\phi(z)+\overline{\phi}(\zz)+p_0\log(z\zz)+p_1\log[(z-1)(\zz-1)]
$$
with constants $p_0$, $p_1$. We see that $F(z)$ is single-valued and that the branch of the logarithms in (\ref{eqhdef}) is insignificant.
The structure of $\phi,\overline{\phi},p_0,p_1$ is determined by the condition $F\in\sS\sV_{\{0,1,\infty\}}$.
\end{ex}

\begin{thm}[The case $n=0$ is Lemma 4.5 in \cite{numfunct}]\label{thm3}
Let $f_G(z)$ be a graphical function in even dimensions $\geq4$.
Then there exists exactly one function $f_{G_1}(z)$ in $\sS\sV_{\{0,1,\infty\}}$ with (G3) of Theorem \ref{thm1} such that (\ref{eqdiff1}) holds.
\end{thm}
The proof of uniqueness is in Section \ref{sectpf3}. If the general solution (\ref{eqthm2}) can be calculated, then the above uniqueness result provides
an algorithm to calculate $f_{G_1}(z)$ from $f_G(z)$. This algorithm is detailed in Section \ref{sectappalg}.

\subsection{Representation theory and hyperbolic space}\label{sectrep}
We define the differential operators
$$
L_k=z^k\partial_z+\zz^k\partial_\zz.
$$
With these operators we define an $\sL(2,\CC)$ Lie algebra representation on $\sS\sV_{\{0,1,\infty\}}$ by
\begin{equation}
X=-L_2,\quad Y=L_0,\quad H=2L_1
\end{equation}
with commutation relations $[H,X]=2X$, $[H,Y]=-2Y$, $[X,Y]=H$. The Casimir operator of this representation is
$$
C=(H-1)^2+4XY=-4(z-\zz)^2\partial_z\partial_\zz+1
$$
so that the effective Laplace operator (\ref{eqdiffn}) in $D=2\lambda+2$ dimensions becomes
$$
\Delta_{\lambda-1}=-\frac{C-(D-3)^2}{4(z-\zz)^2}
$$
Homogeneous solutions of (\ref{eqdiff}) for fixed $D$ are (sub-)representations of $\sL(2,\CC)$ on $\sS\sV_{\{0,1,\infty\}}$. It is possible that the connection to representation theory
underlies the existence of an explicit solution in Theorem \ref{thm2}.

\begin{remark}
In physics it is more common to use the operators $L_+=X$, $L_-=Y$, $L_z=H/2$ with commutation relations $[L_z,L_\pm]=\pm L_\pm$, $[L_+,L_-]=2L_z$ and Casimir operator
$C_{\mathrm{ph}}=L_z^2-L_z+L_+L_-=-(z-\zz)^2\partial_z\partial_\zz$. In this notation we get $\Delta_{\lambda-1}=-[C_{\mathrm{ph}}-\lambda(\lambda-1)]/(z-\zz)^2$.
\end{remark}

In Cartesian coordinates $z=x+\ii y$ we get $C_{\mathrm{ph}}=(C-1)/4=y^2(\partial_x^2+\partial_y^2)$ which is the Laplacian of the hyperbolic space with constant
curvature $-1$ in the Poincar\'e model on the half plane $H=\{z\in\CC: \Im z>0\}$. Because of reflection symmetry (Theorem \ref{thm1} (G1)) we can restrict graphical functions
to $H$ without loss of information. So, graphical functions can be viewed as functions on the Poincar\'e half plane where (\ref{eqdiff}) is related to a Klein-Gordon equation
in two-dimensional hyperbolic space with mass-square $\lambda(\lambda-1)$. Single-valuedness of graphical functions, however, is obscured in this picture.

\subsection{Weight \texorpdfstring{$0<\nu_e<1$}{0<nu(e)<1}}\label{sectnu01}
A chain of $\lambda+1-k$ edges of weight 1 is (up to a factor) equivalent to an edge of weight $\nu_e=k/\lambda$ (for $k=1,2,\ldots,\lambda$, see \eqref{eqappendk}).
Iterated use of the algorithm in the previous subsection hence allows one to append edges of weights $0<\nu_e<1$, $\nu_e\in\frac{1}{\lambda}\ZZ$.

\begin{lem}[see `Rule 2' in \cite{K2}]\label{lemconv}
An internal two-valent vertex can be eliminated by a convolution product:
For $x\in\RR^D$ and $\nu_1,\nu_2\in\RR$ we have diagrammatically,
\begin{align}\label{eqconvolute}
&\begin{tikzpicture}[baseline={([yshift=-2.5ex]current bounding box.center)}]
    \coordinate (v1) at (0,0);
    \coordinate (v2) at (1.5,0);
    \coordinate (v3) at (3,0);
    \filldraw[black!20] (v1) -- ([shift=(150:.2)]v1) arc (150:210:.2) -- (v1);
    \draw (v1) -- ([shift=(150:.2)]v1);
    \draw (v1) -- ([shift=(210:.2)]v1);
    \filldraw[black!20] (v3) -- ([shift=(30:.2)]v3) arc (30:-30:.2) -- (v3);
    \draw (v3) -- ([shift=(30:.2)]v3);
    \draw (v3) -- ([shift=(-30:.2)]v3);
    \filldraw (v1) circle (1.3pt);
    \filldraw (v2) circle (1.3pt);
    \filldraw (v3) circle (1.3pt);
    \draw[white] (v1) -- node[above] {$\phantom{\nu_1+\nu_2-\frac{\lambda+1}{\lambda}}$} (v3);
    \draw (v1) -- node[above] {$\nu_1$} (v2);
    \draw (v2) -- node[above] {$\nu_2$} (v3);
\end{tikzpicture}
\quad
=
\quad
c^{(\lambda)}_{\nu_1,\nu_2}
\quad
\begin{tikzpicture}[baseline={([yshift=-2.5ex]current bounding box.center)}]
    \coordinate (v1) at (0,0);
    \coordinate (v3) at (3,0);
    \filldraw[black!20] (v1) -- ([shift=(150:.2)]v1) arc (150:210:.2) -- (v1);
    \draw (v1) -- ([shift=(150:.2)]v1);
    \draw (v1) -- ([shift=(210:.2)]v1);
    \filldraw[black!20] (v3) -- ([shift=(30:.2)]v3) arc (30:-30:.2) -- (v3);
    \draw (v3) -- ([shift=(30:.2)]v3);
    \draw (v3) -- ([shift=(-30:.2)]v3);
    \filldraw (v1) circle (1.3pt);
    \filldraw (v3) circle (1.3pt);
    \draw (v1) -- node[above] {$\nu_1+\nu_2-\frac{\lambda+1}{\lambda}$} (v3);
\end{tikzpicture}\\
&\text{where } \quad
c^{(\lambda)}_{\nu_1,\nu_2} = \frac{\Gamma(\lambda(1-\nu_1)+1)\Gamma(\lambda(1-\nu_2)+1)\Gamma(\lambda(\nu_1+\nu_2-1)-1)}
{\Gamma(\lambda\nu_1)\Gamma(\lambda\nu_2)\Gamma(\lambda(2-\nu_1-\nu_2)+2)},\nonumber
\end{align}
provided that the integral on the left hand side of \eqref{eqconvolute} converges.
\end{lem}
\begin{proof}
An elementary calculation using spherical coordinates gives the following Fourier transform
\begin{equation}\label{eqfourier}
\int_{\RR^D}\frac{\ee^{\ii x\cdot p}}{\|x\|^\alpha}\frac{\dd^Dx}{\pi^{D/2}}=\frac{2^{D-\alpha}}{\|p\|^{D-\alpha}}\frac{\Gamma((D-\alpha)/2)}{\Gamma(\alpha/2)}
\quad\text{for }0<\alpha<D.
\end{equation}
The lemma follows from a commutative diagram that translates the convolution product into a pointwise product after Fourier transformation.
\end{proof}

\begin{prop}\label{propwtk1}
For $k=1,2,\ldots,\lambda$ we have
\begin{align}
\label{eqappendk}
\begin{tikzpicture}[baseline={([yshift=-2.5ex]current bounding box.center)}]
    \coordinate (v1) at (0,0);
    \coordinate (v2) at (1.5,0);
    \filldraw[black!20] (v1) -- ([shift=(150:.2)]v1) arc (150:210:.2) -- (v1);
    \draw (v1) -- ([shift=(150:.2)]v1);
    \draw (v1) -- ([shift=(210:.2)]v1);
    \filldraw[black!20] (v2) -- ([shift=(30:.2)]v2) arc (30:-30:.2) -- (v2);
    \draw (v2) -- ([shift=(30:.2)]v2);
    \draw (v2) -- ([shift=(-30:.2)]v2);
    \filldraw (v1) circle (1.3pt);
    \filldraw (v2) circle (1.3pt);
    \draw (v1) -- node[above] {$\frac{k}{\lambda}$} (v2);
\end{tikzpicture}
\quad
&=
\quad
\frac{\displaystyle (\lambda-1)!^{\lambda+1-k}(\lambda-k)!}{\displaystyle (k-1)!}
\quad
\underbrace{
\begin{tikzpicture}[baseline={([yshift=-2.5ex]current bounding box.center)}]
    \coordinate (v1) at (0,0);
    \coordinate (v2) at (1.5,0);
    \coordinate (v3) at (3,0);
    \coordinate (v3a) at (3.5,0);
    \node (v4) at (4.5,0) {$\cdots$};
    \coordinate (v5a) at (5.5,0);
    \coordinate (v5) at (6,0);
    \coordinate (v6) at (7.5,0);
    \filldraw[black!20] (v1) -- ([shift=(150:.2)]v1) arc (150:210:.2) -- (v1);
    \draw (v1) -- ([shift=(150:.2)]v1);
    \draw (v1) -- ([shift=(210:.2)]v1);
    \filldraw[black!20] (v6) -- ([shift=(30:.2)]v6) arc (30:-30:.2) -- (v6);
    \draw (v6) -- ([shift=(30:.2)]v6);
    \draw (v6) -- ([shift=(-30:.2)]v6);
    \filldraw (v1) circle (1.3pt);
    \filldraw (v2) circle (1.3pt);
    \filldraw (v3) circle (1.3pt);
    \filldraw (v5) circle (1.3pt);
    \filldraw (v6) circle (1.3pt);
    \draw[white] (v1) -- node[above] {$\phantom{\frac{k}{\lambda}}$} (v6);
    \draw (v1) -- node[above] {$1$} (v2) -- node[above] {$1$} (v3) -- (v3a);
    \draw (v5a) -- (v5) -- node[above] {$1$} (v6);
\end{tikzpicture}
}_{\lambda + 1 - k}.
\end{align}

\end{prop}
\begin{proof}
The proof is by induction from $k=\lambda$ down to $k=1$. The case $k=\lambda$ is trivial.
Fix $k$ in $1,\ldots,\lambda-1$. By induction an edge of weight $\nu_1=(k+1)/\lambda$ is $(\lambda-1)!^{\lambda-k}(\lambda-k-1)!/k!$ times a sequence of $\lambda-k$ edges.
We attach an edge of weight $\nu_2=1$ which, by Lemma \ref{lemconv}, lowers $\nu_1$ by $1/\lambda$ and divides by $(\lambda-1)!(\lambda\nu_1-1)(\lambda-\lambda\nu_1+1)$.
We obtain the result from
\begin{gather*}
\frac{(\lambda-1)!^{\lambda-k}(\lambda-k-1)!}{k!}(\lambda-1)![(k+1)-1][\lambda-(k+1)+1]=\frac{(\lambda-1)!^{\lambda+1-k}(\lambda-k)!}{(k-1)!}.
\qedhere
\end{gather*}
\end{proof}
Upon completion (Section \ref{sectcomp}) the above lemma becomes a special case of the factor identity in Section \ref{sectfactor}.

\subsection{Weight \texorpdfstring{$\nu_e<0$}{nu(e)<0}}\label{sectappendneg}
If we append an edge of negative weight $\nu_e<0$, the graphical function $f_{G_1}(z)$ in \eqref{eqappend} is a polynomial in $z$ and $\zz$.

The vector $z_2=(\Re z,\Im z,0,\ldots,0)^T$ in (\ref{eqfGdef}) is attached to an internal vertex $x$. If $\nu_e=-k/\lambda$ for $k=1,2,\ldots$, then the edge $xz_2$ contributes to the integrand with the numerator
\begin{equation}\label{eqnumer}
\|x-z_2\|^{2k}=(\|x\|^2-2x\cdot z_2+\|z_2\|^2)^k=\sum_{k_0+k_1+k_2=k}\frac{k!}{k_0!k_1!k_2!}\|x\|^{2k_0}(-2x\cdot z_2)^{k_1}\|z_2\|^{2k_2},
\end{equation}
where $x\cdot z_2$ is the Euclidean scalar product between $x$ and $z_2$. The sum on the right hand side is over all partitions of $k$ into three non-negative integers.

\begin{lem}\label{lemk}
Let $z_1,z_2\in\RR^D$ with $\|z_1\|=1$ and $A_G(0,z_1,x)$ be the Feynman integral \eqref{eqAdef} with external vectors $0,z_1,x\in\RR^D$.
Then
\begin{equation}\label{eqpropk}
\int_{\RR^D}(x\cdot z_2)^kA_G(0,z_1,x)\frac{\dd^Dx}{\pi^{D/2}}=(z_1\cdot z_2)^k\int_{\RR^D}(x\cdot z_1)^kA_G(0,z_1,x)\frac{\dd^Dx}{\pi^{D/2}}.
\end{equation}
\end{lem}
\begin{proof}
We expand the scalar product
$$
(x\cdot z_2)^k=\sum_{\mu_1,\ldots,\mu_k}\Big(\prod_{\ell=1}^kz_2^{\mu_\ell}\Big)\Big(\prod_{\ell=1}^kx^{\mu_\ell}\Big).
$$
The first factor is constant and can be extracted from the integral. The second factor gives a tensor of rank $k$. The remaining integral solely depends on the unit vector $z_1$.
Hence
$$
\int_{\RR^D}\Big(\prod_{\ell=1}^kx^{\mu_\ell}\Big)A_G(0,z_1,x)\frac{\dd^Dx}{\pi^{D/2}}=P\prod_{\ell=1}^kz_1^{\mu_\ell}
$$
for some constant $P\in\RR$. We multiply both sides with $\prod_{\ell=1}^kz_1^{\mu_\ell}$ and sum over the indices $\mu_\ell$ to obtain
$$
P=\sum_{\mu_1,\ldots,\mu_k}\int_{\RR^D}\Big(\prod_{\ell=1}^kx^{\mu_\ell}z_1^{\mu_\ell}\Big)A_G(0,z_1,x)\frac{\dd^Dx}{\pi^{D/2}}=\int_{\RR^D}(x\cdot z_1)^kA_G(0,z_1,x)\frac{\dd^Dx}{\pi^{D/2}}.
$$
This gives the desired result.
\end{proof}

The integral on the right hand side of (\ref{eqpropk}) is a real number which can be expressed in terms of periods (see Section \ref{sectper}) by the identity
\begin{equation}\label{eqscalar}
x\cdot z_1=\frac{1+\|x\|^2-\|x-z_1\|^2}{2}.
\end{equation}

\begin{thm}\label{thmappneg}
Consider a graphical function $f_{G_1(k)}(z)$ whose external vertex $z$ is attached to a single internal vertex $x$ by an edge of weight $\nu_{xz}=-k/\lambda$, $k=1,2,\ldots$ (see picture below).
Let $P_{G(k_0,k_1)}$ be the period (see Section \ref{sectper}) of the graph $G(k_0,k_1)=G_1(k)\backslash\{xz\}\cup\{x0,x1\}$ where the edges $x0$ and $x1$
have weights $\nu_{x0}=-k_0/\lambda$ and $\nu_{x1}=-k_1/\lambda$. Then $P_{G(k_0,k_1)}$ exists for $k_0,k_1\in\ZZ_{\geq0}$, $k_0+k_1\leq k$ and
\begin{align}\label{eqappneg}
f_{G_1(k)}(z)&=\sum_{\genfrac{}{}{0pt}{}{k_0+k_1+k_2=k}{0\leq k_0,k_1,k_2}} p^k_{k_0,k_1,k_2}(z,\zz) ~ P_{G(k_0,k_1)},\\
\intertext{with polynomials}
p^k_{k_0,k_1,k_2}(z,\zz) &= \frac{k!(2-z-\zz)^{k_0}(z+\zz)^{k_1}(2z\zz-z-\zz)^{k_2}}{2^kk_0!k_1!k_2!}.\nonumber\\
\intertext{Diagrammatically,}
\underbrace{
\begin{tikzpicture}[baseline={([yshift=-.7ex]current bounding box.center)}]
    \def \rad {.9}
    \coordinate (w) at (0,0);
    \draw[fill=black!20] (w) circle (\rad);
%    \node (G1) at (-1,-.9) {$G_1$};
%
    \coordinate[label=above:$1$] (w1) at ([shift=(90:\rad)]w);
    \coordinate[label=below:$0$] (w0) at ([shift=(-90:\rad)]w);
    \coordinate[label=below right:$x$] (wx) at ([shift=(0:\rad)]w);
    \coordinate[label=right:$z$] (wz) at ([shift=(0:2)]w);
    \draw (wx) -- node[above] {$-\frac{k}{\lambda}$} (wz);
    \filldraw (w0) circle (1.3pt);
    \filldraw (w1) circle (1.3pt);
    \filldraw (wx) circle (1.3pt);
    \filldraw (wz) circle (1.3pt);
\end{tikzpicture}
}_{
G_1(k)
}
&=
\sum_{\genfrac{}{}{0pt}{}{k_0+k_1+k_2=k}{0\leq k_0,k_1,k_2}}p^k_{k_0,k_1,k_2}(z,\zz)
\quad
\underbrace{
\begin{tikzpicture}[baseline={([yshift=-.7ex]current bounding box.center)}]
    \def \rad {.9}
    \coordinate (w) at (0,0);
    \draw[fill=black!20] (w) circle (\rad);
%    \node (G1) at (-1,-.9) {$G(k_0,k_1)$};
%
    \coordinate[label=above:$1$] (w1) at ([shift=(90:\rad)]w);
    \coordinate[label=below:$0$] (w0) at ([shift=(-90:\rad)]w);
    \coordinate[label=right:$x$] (wx) at ([shift=(0:\rad)]w);
    \filldraw (w0) circle (1.3pt);
    \filldraw (w1) circle (1.3pt);
    \filldraw (wx) circle (1.3pt);
    \draw (wx) arc (45:-135:0.63639610306) node[midway,inner sep=0pt,below right]{$-\frac{k_0}{\lambda}$};
    \draw (wx) arc (-45:135:0.63639610306) node[midway,inner sep=0pt,above right]{$-\frac{k_1}{\lambda}$};
\end{tikzpicture}
}_{
G(k_0,k_1)}
.
\notag
\end{align}
\end{thm}

\begin{proof}
We first prove the existence of the periods $P_{G(k_0,k_1)}$. To use Proposition \ref{propconv} we internally complete $G_1(k)$ to $\Gg_1(k)$ and $G(k_0,k_1)$ to $\Gg(k_0,k_1)$ (with isolated vertex $z$).
We identify corresponding vertices of $\Gg_1(k)$ and $\Gg(k_0,k_1)$ and set $k_2=k-k_0-k_1$.
We find that $\Gg(k_0,k_1)$ can be obtained from $\Gg_1(k)$ by adding edges $x0$, $x1$, $x\infty$, $xz$ with weights $\nu_{x0}=-k_0/\lambda$, $\nu_{x1}=-k_1/\lambda$,
$\nu_{x\infty}=-k_2/\lambda$, and $\nu_{xz}=k/\lambda$ (killing the edge $xz$ in $G_1(k)$).

Assume $P_{G(k_0,k_1)}$ were divergent. Then there exists a vertex subset $\sV$ with $|\sV^\text{ext}|\leq1$ such that $N_{\Gg(k_0,k_1)[\sV]}\geq(|\sV^\text{ext}|-1)(\lambda+1)/\lambda$.
Because (by convergence) $N_{\Gg_1(k)[\sV]}<(|\sV^\text{ext}|-1)(\lambda+1)/\lambda$ we need $N_{\Gg(k_0,k_1)[\sV]}>N_{\Gg_1(k)[\sV]}$. The only extra edge with positive weight is $xz$ which
implies that the vertices $x$ and $z$ need to be in $\sV$. By $|\sV^\text{ext}|\leq1$, the set $\sV$ has no external vertices 0, 1, $\infty$.
Hence
$$
N_{\Gg(k_0,k_1)[\sV]}=N_{\Gg_1(k)[\sV\backslash\{z\}]}<-\frac{\lambda+1}{\lambda}
$$
by (\ref{eqconv}) for $\Gg_1(k)$ and $\sV\backslash\{z\}$. This contradicts the divergence of the $P_{G(k_0,k_1)}$.

Let $G=G_1(k)\backslash\{z\}$. The internal vertex $x$ in $G_1(k)$ becomes external in $G$.
With (\ref{eqnumer}) and Lemma \ref{lemk} we get from (\ref{eqfGdef}),
$$
f_{G_1(k)}(z)=\sum_{k_0+k_1+k_2=k}\frac{k!}{k_0!k_1!k_2!}(z_1\cdot z_2)^{k_1}\|z_2\|^{2k_2}\int_{\RR^D}\|x\|^{2k_0}(-2x\cdot z_1)^{k_1}A_G(0,z_1,x)\frac{\dd^Dx}{\pi^{D/2}}.
$$
We substitute (\ref{eqscalar}) into the integrand and expand the term $(\|x-z_1\|^2-\|x\|^2-1)^{k_1}$. This amounts to partitioning $k_1$ into three non-negative integers
$k_1'$, $k_3$, $k_4$, corresponding to powers of $\|x-z_1\|^2$, $\|x\|^2$, and 1, respectively. Powers of $\|x-z_1\|^2$ and $\|x\|^2$ are edges of negative weights in (\ref{eqfGdef}).
We rename $k_1'$ back to $k_1$ and get
$$
f_{G_1(k)}(z)=\sum_{k_0+k_1+k_2+k_3+k_4=k}\frac{(-1)^{k_3+k_4}k!}{k_0!k_1!k_2!k_3!k_4!}(z_1\cdot z_2)^{k_1+k_3+k_4}\|z_2\|^{2k_2}P_{G(k_0+k_3,k_1)}.
$$
Now, we shift $k_0\mapsto k_0-k_3$ and perform the binomial sum over $k_3$ to obtain
$$
f_{G_1(k)}(z)=\sum_{k_0+k_1+k_2+k_4=k}\frac{(-1)^{k_4}k!}{k_0!k_1!k_2!k_4!}(1-z_1\cdot z_2)^{k_0}(z_1\cdot z_2)^{k_1+k_4}\|z_2\|^{2k_2}P_{G(k_0,k_1)}.
$$
Likewise, we shift $k_2\mapsto k_2-k_4$ and perform the binomial sum over $k_4$ yielding
$$
f_{G_1(k)}(z)=\sum_{k_0+k_1+k_2=k}\frac{k!}{k_0!k_1!k_2!}(1-z_1\cdot z_2)^{k_0}(z_1\cdot z_2)^{k_1}(\|z_2\|^2-z_1\cdot z_2)^{k_2}P_{G(k_0,k_1)}.
$$
From (\ref{eqzdef}) we get $\|z_2\|^2=z\zz$ and $z_1\cdot z_2=(z+\zz)/2$. This gives (\ref{eqappneg}).
\end{proof}

%%%%%%%%%%%%%%%%%%%%%%%%%%%%%%%%%%%%%%%%%%%%%%%%%%%%%%%%%%%%%%%%%%%%%%%%%%%%%%%%%
\section{Constructible graphical functions}\label{sectconstructible}
A graphical function in even dimensions $\geq4$ can always be computed (subject to mild constraints from time and memory consumption) if it is constructible in the sense of the following definition:

\begin{defn}[An extension of Section 3.7 in \cite{gf}]\label{defconstructible}
We consider the following set of commuting reduction steps for (internally) completed graphical functions in even dimensions $\geq4$:
\begin{enumerate}
\item[(R1)] Deletion of edges between external vertices, Section \ref{sectextedge},
\item[(R2)] Product factorization, Section \ref{sectprod},
\item[(R3)] Period factorization, Section \ref{sectfactor},
\item[(R4)] Contraction of single edges with weights in $\frac{1}{\lambda}\ZZ$ attached to external vertices, Sections \ref{sectperm} and \ref{sectappedge}.
\end{enumerate}
A graphical function is {\em irreducible} if it cannot be reduced by any of these steps.
Maximum use of the reduction steps maps a graph $G$ to a set of Feynman periods and irreducible graphical functions.
The {\em kernel} of $G$ is the unique \cite{Census} representation of this set in terms of completed Feynman periods and internally completed graphical functions with no edges between
external vertices (and no distinction of external vertices).

We inductively define constructible completed graphical functions and periods by
\begin{enumerate}
\item A completed graphical function is constructible if its kernel consists of the empty graphical function and constructible periods.
\item A completed Feynman period with three vertices is constructible. Its value is $1$.
\item A completed Feynman period $P_\Gg$ with $|\sV_\Gg|\geq4$ vertices is constructible if there exist four vertices $a,b,c,d\in\sV_\Gg$
such that the graphical function $\Gg|_{abcd=01z\infty}$ is constructible.
\end{enumerate}
Uncompleted graphical functions and periods are constructible if their completions are constructible.
\end{defn}

\begin{ex}
The graphical function $G_a$ in Figure \ref{fig:gfcomp} is constructible, see Examples \ref{exf31}, \ref{exf32}, \ref{exf33}.
\end{ex}

\begin{figure}
    \def\scale{1}
\begin{align*}
\begin{tikzpicture}[baseline={([yshift=-.7ex]0,0)}]
    \coordinate (vm) at  (0,0);
    \coordinate[label=above left:$z$] (vz) at  (-\scale, \scale);
    \coordinate[label=above right:$1$] (v1) at  ( \scale, \scale);
    \coordinate[label=below left:$\infty$] (voo) at (-\scale,-\scale);
    \coordinate[label=below right:$0$] (v0)  at (\scale, -\scale);
    \coordinate (v2)  at (-\scale, 0);
    \coordinate (v3)  at (\scale, 0);
    \coordinate (vm)  at (0, 0);
    \node (G) at (0,{-2*\scale}) {$\Gg_7$};
    \draw[preaction={draw, white, line width=5pt, -}] (voo) -- (v3);
    \draw[preaction={draw, white, line width=5pt, -}] (v0) -- (v2);
    \draw (voo) -- (v2);
    \draw (v0) -- (v3);
    \draw (vm) -- (v2);
    \draw (vm) -- (v3);
    \draw (vm) -- (vz);
    \draw (vm) -- (v1);
    \draw (v2) -- (vz);
    \draw (v3) -- (v1);
    \filldraw (v0) circle (1.3pt);
    \filldraw (v1) circle (1.3pt);
    \filldraw (vz) circle (1.3pt);
    \filldraw (voo) circle (1.3pt);
    \filldraw (v2) circle (1.3pt);
    \filldraw (v3) circle (1.3pt);
    \filldraw (vm) circle (1.3pt);
\end{tikzpicture}
&&
\begin{tikzpicture}[baseline={([yshift=-.7ex]0,0)}]
    \coordinate (vm) at  (0,0);
    \coordinate[label=above left:$z$] (vz) at  (-\scale, \scale);
    \coordinate[label=above right:$1$] (v1) at  ( \scale, \scale);
    \coordinate[label=below left:$\infty$] (voo) at (-\scale,-\scale);
    \coordinate[label=below right:$0$] (v0)  at (\scale, -\scale);
    \coordinate (v2)  at ({ \scale/3}, { \scale/3});
    \coordinate (v3)  at ({-\scale/3}, { \scale/3});
    \coordinate (v4)  at ({-\scale/3}, {-\scale/3});
    \coordinate (v5)  at ({ \scale/3}, {-\scale/3});
    \node (G) at (0,{-2*\scale}) {$\Gg_8$};
    \draw (vz) -- (v2);
    \draw (voo) -- (v3);
    \draw (v0) -- (v4);
    \draw (v1) -- (v5);
    \draw[preaction={draw, white, line width=3pt, -}] (vz) -- (v4);
    \draw[preaction={draw, white, line width=3pt, -}] (voo) -- (v5);
    \draw[preaction={draw, white, line width=3pt, -}] (v0) -- (v2);
    \draw[preaction={draw, white, line width=3pt, -}] (v1) -- (v3);
    \draw (vz) -- ($(vz)!.2!(v2)$);
    \draw (voo) -- ($(voo)!.2!(v3)$);
    \draw (v0) -- ($(v0)!.2!(v4)$);
    \draw (v1) -- ($(v1)!.2!(v5)$);
    \draw (v2) -- (v3) -- (v4) -- (v5) -- (v2);
    \filldraw (v0) circle (1.3pt);
    \filldraw (v1) circle (1.3pt);
    \filldraw (vz) circle (1.3pt);
    \filldraw (voo) circle (1.3pt);
    \filldraw (v2) circle (1.3pt);
    \filldraw (v3) circle (1.3pt);
    \filldraw (v4) circle (1.3pt);
    \filldraw (v5) circle (1.3pt);
\end{tikzpicture}
\end{align*}
\caption{Two irreducible internally completed graphical functions in four dimensions with all edges of weight $1$: While $f_{\Gg_7}(z)$ can be calculated by Gegenbauer factorization, the graphical function $f_{\Gg_8}(z)$ is conjecturally elliptic.}
\label{fig:irred7} 
\end{figure}

\begin{ex}
The internally completed graphical functions $\Gg_7$ (see Example \ref{G7ex}) and $\Gg_8$ in Figure \ref{fig:irred7} are irreducible.
An analysis with {\tt HyperInt} \cite{Panzer:HyperInt} suggests that $\Gg_8$ is elliptic in $D=4$.

In four dimensions (and in six dimensions, see Examples \ref{ex2int} and \ref{exgen}),
all completed graphical functions with $\leq7$ vertices are expressible in terms of GSVHs \cite{GSVH}. So, $\Gg_8$ is a minimal example of a four-dimensional graphical function
which conjecturally is not a GSVH.
\end{ex}

A key benefit of graphical functions (and their analogs at non-integer dimensions \cite{5loopphi3,numfunct,7loops})
is the reduction to kernels. At modest loop orders kernels are rare and have significantly less vertices than the original graph.

By Proposition \ref{propwtk1} we can replace any edge of weight $k/\lambda$, $k=1,\ldots,\lambda-1$ by a chain of $\lambda+1-k$ edges of weight 1.
So, we may alter a constructible graph such that the reduction uses appending of an edge (R4) only for weights in $\{1\}\cup\ZZ_{<0}/\lambda$.
In this case the weight of the graphical function or period is constrained by Theorem \ref{conthm}. All constructible periods are MZVs whereas all constructible graphical functions
are GSVHs in the alphabet 0, 1, $\zz$ (see Section 8.3 in \cite{GSVH}).

In the following theorem we consider the weight as a filtration where the weight of MZV constants is added to the weight of hyperlogarithms.
Alternatively one may choose to lift the notion of weight to the motivic setup \cite{Bcoact2}.
\begin{thm}\label{conthm}
Let $G$ be the graph of a constructible graphical function in even dimensions $\geq4$ whose reduction in Definition \ref{defconstructible} uses (R4) only
for edge-weights in $\{1\}\cup\frac1\lambda\ZZ_{<0}$. Then $f_G(z)$ is a GSVH of weight $\leq2\VGint$ in the letters $0, 1, \zz$.

Let $G$ be the uncompleted graph of a constructible period in even dimensions $\geq4$ whose reduction uses (R4) only for edge-weights in $\{1\}\cup\frac1\lambda\ZZ_{<0}$.
If $G$ has $|\sV_G|\geq3$ vertices, then $P_G$ is an MZV of weight $\leq2|\sV_G|-5$.
\end{thm}
\begin{proof}
We first prove that the result for periods follows from the statement about graphical functions:
$P_G$ is the integral (\ref{PGint}) of a constructible graphical function with $|\sV_G|-3$ internal vertices. After a single-valued integration of the integrand the integral is given by residues,
see Theorem 2.28 in \cite{gf}. The graphical function is a GSVH of weight $\leq2|\sV_G|-6$ in 0, 1, $\zz$. Multiplication by $(z-\zz)^{2\lambda}$ and integration preserves the alphabet
0, 1, $\zz$ while it brings the weight up to $\leq2|\sV_G|-5$. It is proved in Section 8.3 of \cite{GSVH} that GSVHs in 0, 1, $\zz$ evaluate to MZVs.

Now we prove the result for graphical functions by induction over the number of internal vertices. If $\VGint=0$, then $f_G(z)$ is a rational GSVH in the letters 0 and 1. It has weight zero.

For general $\VGint$ we have a reduction chain by cases (R1) to (R4) in Definition \ref{defconstructible} which leads to a graphical function with fewer internal vertices.
Reduction step (R1) is an endomorphism on GSVHs in the letters 0, 1, $\zz$ which does not change the weight.
The alphabet 0, 1, $\zz$ is stable under product factorization (R2) while the weight is additive. The number of internal vertices also adds, so that the condition on the weight stays intact.
In case of a factorization (R3) by a completed period $P_\Gg$ with $|\sV_\Gg|\geq4$ vertices the graphical function loses $|\sV_\Gg|-3$ internal vertices.
By induction the graphical function has weight $\leq 2(\VGint-|\sV_\Gg|+3)+2(|\sV_\Gg|-1)-5=2\VGint-1$ (graphical functions with period factors have weight drop).

The crucial step is appending an edge $e$ in (R4). If $e$ has weight $\nu_e=1$, then the algorithm in Section \ref{sectappalg} may be used with the integral operator
$\sI_n'$ in Lemma \ref{lem2} instead of $\sI_n$. In this case appending $e$ amounts to a double integration within the alphabet 0, 1, $\zz$. The weight is increased by two \cite{GSVH} while
one new internal vertex is created.

If $\nu_e<0$, we use Theorem \ref{thmappneg} to express the graphical function as a polynomial in $z$ and $\zz$ whose coefficients are given by periods of the graphs $G(k_0,k_1)$.
The graphs $G(k_0,k_1)$ may have edges of weights $k/\lambda$, $k=1,\ldots,\lambda-1$ between $x$ and 0 or 1. So, we cannot directly use induction. However, by (R4) we know that
the graphical function we obtain by setting $x=z$ in $G(k_0,k_1)$ is constructible. In this graphical function the edges between $x$ and 0 or 1
become external edges between $z$ and 0 or 1. Removal of these edges by (R1) does not affect the weight or the alphabet of the result. Hence $P_{G(k_0,k_1)}$ has weight
$\leq2|\sV_{G(k_0,k_1)}|-5$ by induction. Because $\VGint=|\sV_{G(k_0,k_1)}|-2$ the original graph has weight $\leq2\VGint-1$ (it has weight drop). The graphical function is a polynomial
in $z$, $\zz$ and hence trivially a GSVH in any alphabet.
\end{proof}

The resolution of edges of weight $k/\lambda$, $k=1,\ldots,\lambda-1$ into chains of weight 1 edges increases the number of internal vertices. Nonetheless, we expect that any reduction (R4)
increases the weight of the graphical function at most by two.
\begin{con}\label{concon}
Theorem \ref{conthm} holds for all constructible graphs and periods in even dimensions $\geq4$.
\end{con}

%%%%%%%%%%%%%%%%%%%%%%%%%%%%%%%%%%%%%%%%%%%%%%%%%%%%%%%%%%%%%%%%%%%%%%%%%%%%%%%%%
\section{Gegenbauer factorization}\label{sectgeg}
Let $\Gg$ be an (internally) completed graphical function with a three-vertex split $a,b,c$ into $G_1$ and $G_2$. We include the split vertices together with their corresponding edges
to $G_1$ or $G_2$, respectively.
In Section \ref{sectfactor} we saw that $f_\Gg(z)$ has a period factor if $G_1\backslash\{a,b,c\}$ or $G_2\backslash\{a,b,c\}$ have no external vertices.
If on the other hand one of these graphs has two external vertices (say $z,\infty$), factorization involves (after the removal of $\infty$ by conformal symmetry)
the Feynman integrals of $G_1$ or $G_2$ with four external vertices $a,b,c,z$. In general, Feynman integrals with four external vertices cannot be expressed in terms of two complex variables:
The four vectors span a three-dimensional space.

A special case arises if both $G_1\backslash\{a,b,c\}$ and $G_2\backslash\{a,b,c\}$ have exactly one external vertex. By permutation symmetry (Section \ref{sectperm})
the labels of the two external vertices are insignificant. We assume that 1 is in $G_1\backslash\{a,b,c\}$ and $z$ in $G_2\backslash\{a,b,c\}$. The other two external labels
0 and $\infty$ must be in the split vertices $a,b,c$. In this case there exists a factorization formula which is based on a convolution product in $\CC$, see
Theorem \ref{thm4}. This factorization can be depicted diagrammatically as 
\begin{align}
\label{eqdiaggegenbauer}
\begin{tikzpicture}[baseline={([yshift=-1.3ex]current bounding box.center)}]
    \coordinate[label=above:$0$] (va) at (0,2);
    \coordinate[label=above:$x$] (vb) at (0,1);
    \coordinate[label=above:$\infty$] (vc) at (0,0);
    \filldraw (va) circle (1.3pt);
    \filldraw (vb) circle (1.3pt);
    \filldraw (vc) circle (1.3pt);
    \draw[fill=black!20] (va) arc (90:270:2 and 1) arc (270:90:.5) arc (270:90:.5);
    \draw[fill=black!20] (va) arc (90:-90:2 and 1) arc (-90:90:.5) arc (-90:90:.5);
    \node (G1) at (-1,1) {$\overline{G}_1$};
    \node (G2) at (+1,1) {$\overline{G}_2$};
    \node (G) at (-2,0) {$\overline{G}$};
    \coordinate[label=left:$1$] (v1)  at ([shift={(vb)}]180:2 and 1);
    \coordinate[label=right:$z$] (vz)  at ([shift={(vb)}]0:2 and 1);
    \filldraw (v1) circle (1.3pt);
    \filldraw (vz) circle (1.3pt);
\end{tikzpicture}
\quad
&\rightarrow
\quad
\begin{tikzpicture}[baseline={([yshift=-1.3ex]current bounding box.center)}]
    \coordinate[label=above:$0$] (va) at (0,2);
    \coordinate[label=above:$z$] (vb) at (0,1);
    \coordinate[label=above:$\infty$] (vc) at (0,0);
    \coordinate[label=above:$0$] (wa) at (1,2);
    \coordinate[label=above:$1$] (wb) at (1,1);
    \coordinate[label=above:$\infty$] (wc) at (1,0);
    \coordinate (wa) at (1,2);
    \coordinate (wb) at (1,1);
    \coordinate (wc) at (1,0);
    \filldraw (va) circle (1.3pt);
    \filldraw (vb) circle (1.3pt);
    \filldraw (vc) circle (1.3pt);
    \filldraw (wa) circle (1.3pt);
    \filldraw (wb) circle (1.3pt);
    \filldraw (wc) circle (1.3pt);
    \draw[fill=black!20] (va) arc (90:270:2 and 1) arc (270:90:.5) arc (270:90:.5);
    \draw[fill=black!20] (wa) arc (90:-90:2 and 1) arc (-90:90:.5) arc (-90:90:.5);
    \node (G1) at (-1,1) {$\overline{G}_1$};
    \node (G2) at (+2,1) {$\overline{G}_2$};
    \coordinate[label=left:$1$] (v1)  at ([shift={(vb)}]180:2 and 1);
    \coordinate[label=right:$z$] (vz)  at ([shift={(wb)}]0:2 and 1);
    \filldraw (v1) circle (1.3pt);
    \filldraw (vz) circle (1.3pt);
    \node (t) at (.5,1) {$\star$};
    \node (G) at (-2,0) {$\phantom{\overline{G}}$};
\end{tikzpicture}~~,
\end{align}
where the $\star$ indicates the convolution of the graphical functions $f_{\Gg_1}(z)$ and $f_{\Gg_2}(z)$.

The factorization is an adaption of the Gegenbauer technique to graphical functions \cite{geg}.

\begin{defn}\label{defgegsplit}
Assume the edges of the (internally) completed graphical function $\Gg$ split into $\Gg_1$ and $\Gg_2$ with three common vertices $0,x,\infty$ where $x\in\sVGint$.
If $\sV_{\Gg_1}\ni 1\notin \sV_{\Gg_2}$ and $\sV_{\Gg_1}\not\ni z\in \sV_{\Gg_2}$, then the pair $\Gg_1$ with $x=z$ and $\Gg_2$ with $x=1$ is a Gegenbauer split of $\Gg$.
\end{defn}

Note that edges between 0, $x$, $\infty$ can be either moved to $\Gg_1$ or $\Gg_2$ which are both internally completed graphical functions.
\begin{ex}[Example 3.33 in \cite{gf}]\label{G7ex}
The graph $\Gg_7$ in Figure \ref{fig:irred7} has a Gegenbauer split along 0, $\infty$ and the central internal vertex.
In four dimensions the graphical function $f_{\Gg_7}(z)$ can be calculated by the subsequent theorem. The result is a single-valued multiple polylogarithm of weight six, divided by $(z-\zz)(1-z\zz)$ \cite{Shlog}.
\end{ex}

\begin{thm}[The four-dimensional case is the unproved Remark 3.32 in \cite{gf}. The general case assumes Conjecture \ref{congegex}]\label{thm4}
For every graphical function $f_G(z)$ in even $D=2\lambda+2\geq4$ dimensions there exists a unique anti-symmetric (under $z\leftrightarrow\zz$) function ${}^2\!f_G(z)$ with
\begin{equation}\label{eqthm4int}
\int_0^{2\pi}{}^2\!f_G(Z\ee^{\ii\phi})\sin k\phi\,\dd\phi=0\quad\text{for all $k=1,2,\ldots,\lambda-1$ and all $1\neq Z\in\RR_+$}
\end{equation}
such that
\begin{equation}\label{eq2av}
f_G(z)=\Big[\frac{1}{z-\zz}(z\partial_z-\zz\partial_\zz)\Big]^{\lambda-1}\,\frac{{}^2\!f_G(z)}{z-\zz}.
\end{equation}
For an (internally) completed graph $\Gg$ we define ${}^2\!f_\Gg(z)={}^2\!f_{\Gg\backslash\{\infty\}}(z)$.

If $\Gg$ has a Gegenbauer split into $\Gg_1$ and $\Gg_2$, then
\begin{equation}\label{geg2}
(z\partial_z-\zz\partial_\zz)\;{}^2\!f_\Gg(z)=\int_\CC\;{}^2\!f_{\Gg_1}(x)\;{}^2\!f_{\Gg_2}\Big(\frac{z}{x}\Big)\;(x\xx)^{\lambda(1-N_{\Gg_2\backslash\{\infty\}})}\frac{\dd^2x}{\pi},
\end{equation}
where $\dd^2x=\dd\Re(x)\wedge\dd\Im(x)$. Equation (\ref{geg2}) uniquely determines ${}^2\!f_\Gg(z)$.
\end{thm}
The proof of Theorem \ref{thm4} in Section \ref{sectpf4} uses radial and angular graphical functions which are defined in Sections \ref{sectrad} and \ref{sectang}.
It relies on Conjecture \ref{congegex} which is only proved in the classical case of four dimensions and unit edge-weights, see Theorems \ref{D4posthm} and \ref{thmgegex0}.
One may consider ${}^2\!f_G(z)$ as the two-dimensional avatar of $f_G(z)$ (hence the superscript 2). 

\begin{ex}
In four dimensions condition (\ref{eqthm4int}) is empty and ${}^2\!f_G(z)=(z-\zz)f_G(z)\in\sS\sV_{\{0,1,\infty\}}$ (which is anti-symmetric by (G1) in Theorem \ref{thm1}).
\end{ex}

We have ${}^2\!f_G(z)/(z-\zz)\in\sS\sV_{\{0,1,\infty\}}$ and therefore ${}^2\!f_G(z)\in\sS\sV_{\{0,1,\infty\}}$ in all even dimensions $\geq4$.
We will not need this result, so we prove it under the assumption that ${}^2\!f_G(z)/(z-\zz)$ has a real-analytic point ($\neq1$) on the unit circle, see Remark \ref{remarkgeg} (1).

\begin{prop}\label{propgeg}
Assume ${}^2\!f_G(z)/(z-\zz)$ is real-analytic at some point on the unit circle. Then ${}^2\!f_G(z)/(z-\zz)\in\sS\sV_{\{0,1,\infty\}}$.
\end{prop}
The proof of Proposition \ref{propgeg} is in Section \ref{sectpf4}.

\begin{remark}\label{remarkgeg}\mbox{}
\begin{enumerate}
\item One can derive an explicit parametric representation for $\,{}^2\!f_G(z)/(z-\zz)$, see Section \ref{sectpar}. To see this we write the differential operator
$D_z=(z-\zz)^{-1}(z\partial_z-\zz\partial_\zz)$ in terms of the invariants $s_0=z\zz$ and $s_1=(z-1)(\zz-1)$, see (\ref{eqinvs}) and (\ref{eq:dual-phi}).
We obtain $D_z=-\partial_{s_1}$, so that inverting $D_z$ is equivalent to integration with respect to $s_1$. In the right hand side of (\ref{fdualparam}) integration
is elementary because $\widetilde{\Phi}_G$ is linear in $s_1$ while $\widetilde{\Psi}_G$ is constant. Compatibility with (\ref{eqthm4int}) can be enforced by subtraction,
see Section \ref{sectpf4}.

The parametric representation ensures that ${}^2\!f_G(z)/(z-\zz)$ is real-analytic on $\CC\backslash\{0,1\}$.
\item In (\ref{2fG}) we provide an explicit formula for ${}^2\!f_G(z)$ in terms of radial and angular graphical functions.
\item The map (\ref{eq2av}) from ${}^2\!f_G$ to $f_G$ can be inverted in polar coordinates $z=Z\hat{z}$, $\zz=Z/\hat{z}$ with $\hat{z}=\ee^{\ii\phi}$
(one can alternatively use the coordinates $z$ and $Z$): In these coordinates the differential operator $z\partial_z-\zz\partial_\zz$ becomes $\hat{z}\partial_{\hat{z}}$
(see (\ref{proppf4eq})) and inversion is (multiple) integration with respect to $\hat{z}$. This provides a candidate $f(z)$ for ${}^2\!f_G(z)$.
Condition (\ref{eqthm4int}) becomes
$$
-\frac{1}{2}\int_{\partial E}{}^2\!f_G(Z\hat{z})(\hat{z}^k-\hat{z}^{-k})\frac{\dd\hat{z}}{\hat{z}}=0,
$$
where the integration is along the boundary of the unit disc $E$. With the residue theorem in the complex $\hat{z}$-plane one can calculate the left hand side for $f(Z\hat{z})$
instead of ${}^2\!f_G(Z\hat{z})$. This gives functions $g_k(Z)$ for $k\in\{1,2,\ldots,\lambda-1\}$. We get (see (\ref{sinort}))
$$
{}^2\!f_G(Z\hat{z})=f(Z\hat{z})-\sum_{k=1}^{\lambda-1}\frac{\hat{z}^k-\hat{z}^{-k}}{2\pi\ii}g_k(Z).
$$
Proposition \ref{proppf4} states that the subtraction is in the kernel of the differential operator on the right hand side of (\ref{eq2av}).
\item It is unclear in which cases the transformation from $f_G(z)$ to ${}^2\!f_G(z)$ maps GSVHs to GSVHs:
The transformation from $Z,\hat z$ back to $z$, $\zz$ requires taking square roots which, in general, does not close in the space of GSVHs.
Conversely, it can be proved that $f_G$ is a GSVH if ${}^2\!f_G(z)$ is a GSVH (Section 8 in \cite{GSVH}). By experiment, we find that ${}^2\!f_G(z)$ often is a GSVH.
The authors are not aware of an example where $f_G$ is a GSVH but ${}^2\!f_G(z)$ is not.
\item The convolution product in (\ref{geg2}) can be calculated by a residue theorem on $\sS\sV_{\{0,1,\infty\}}$ (Theorem 2.29 in \cite{gf}).
One, however, needs two sets of complex conjugated variables $x,\xx$ and $z,\zz$ which slows down the calculation and potentially leads out of the space of GSVHs.
\end{enumerate}
\end{remark}

Even if it closes in GSVHs, Gegenbauer factorization is time and memory consuming. In \cite{Shlog} it is only implemented for special cases in four dimensions.
Although Gegenbauer splits are rare, Theorem \ref{thm4} is a powerful tool to calculated some irreducible graphical functions.

%%%%%%%%%%%%%%%%%%%%%%%%%%%%%%%%%%%%%%%%%%%%%%%%%%%%%%%%%%%%%%%%%%%%%%%%%%%%%%%%%
\section{Four vertex splits}\label{sect4split}
Assume an (internally) completed graphical function $\Gg$ has a four-vertex split $a,b,c,d$ such that $\Gg_1\backslash\{a,b,c,d\}$ has only internal vertices in $\Gg$
(see the left hand side of \eqref{eqtwistdiagram}). Then the edges of $\Gg$ split into two parts $\Gg_1$ and $\Gg_2$ which meet in the vertices $a,b,c,d$.
The split vertices may be internal or external (see Sections \ref{sectfactor} and \ref{sectgeg} for a similar situation with a three-vertex cut).

We can interpret $\Gg_1$ with external vertices $a,b,c,d$ as an internally completed graphical function which
is nested into $\Gg$. We obtain an identity for $f_\Gg(z)$ if we are able to replace $\Gg_1$ by a different graph $\Gg_1'$ which evaluates to the same graphical function.
The two known identities on graphical functions in this context are the twist identity in Section \ref{secttwist} and the Fourier identity in Section \ref{sectdual}.

It is also possible to replace $\Gg_1$ by a sum of (internally) completed graphical functions which evaluates to $f_{\Gg_1}(z)$. An example for such identities is the integration
by parts technique in Section \ref{sectibp}. This technique is not restricted to four vertex cuts.

\subsection{Twist}\label{secttwist}
A completed graphical function is invariant under a double transposition of external labels, see Theorem \ref{thmperm}. This identity lifts to the twist identity for graphical
functions with an internal four-vertex split. The twist identity was first found and proved in the context of periods, see Theorem 2.11 in \cite{Census}.

\begin{prop}
Assume the edges of the (internally) completed graphical function $\Gg$ split into $\Gg_1$ and $\Gg_2$ with four common vertices $a,b,c,d$.
If all vertices of $\Gg_1\backslash\{a,b,c,d\}$ are internal in $\Gg$, then we obtain the twisted graph $\Gg'$ by
gluing the vertices $b,a,d,c$ of $\Gg_1$ to the vertices $a,b,c,d$ of $\Gg_2$ and thereafter (uniquely) moving weighted edges along opposite sides of the four-cycle $acbd$ in
such a way that $\Gg'$ becomes (internally) completed. We have
\begin{align}
\label{eqtwist}
f_\Gg(z)&=f_{\Gg'}(z).
\intertext{Diagrammatically,}
\label{eqtwistdiagram}
\begin{tikzpicture}[baseline={([yshift=-1.3ex]current bounding box.center)}]
    \coordinate[label=above:$a$] (va) at (0,3);
    \coordinate[label=above:$b$] (vb) at (0,2);
    \coordinate (vm) at (0,1.5);
    \coordinate[label=above:$c$] (vc) at (0,1);
    \coordinate[label=above:$d$] (vd) at (0,0);
    \filldraw (va) circle (1.3pt);
    \filldraw (vb) circle (1.3pt);
    \filldraw (vc) circle (1.3pt);
    \filldraw (vd) circle (1.3pt);
    \draw[fill=black!20] (va) arc (90:270:3 and 1.5) arc (270:90:.5) arc (270:90:.5) arc (270:90:.5);
    \draw[fill=black!20] (va) arc (90:-90:3 and 1.5) arc (-90:90:.5) arc (-90:90:.5) arc (-90:90:.5);
    \node (G1) at (-1.5,1.5) {$\overline{G}_1$};
    \node (G2) at (+1.5,1.5) {$\overline{G}_2$};
    \node (G) at (-2,0) {$\overline{G}$};
    \coordinate[label=above right:$0$] (v0)  at ([shift={(vm)}]60:3 and 1.5);
    \coordinate[label=above right:$1$] (v1)  at ([shift={(vm)}]30:3 and 1.5);
    \coordinate[label=below right:$z$] (vz)  at ([shift={(vm)}]-30:3 and 1.5);
    \coordinate[label=below right:$\infty$] (voo) at ([shift={(vm)}]-60:3 and 1.5);
    \filldraw (v0) circle (1.3pt);
    \filldraw (v1) circle (1.3pt);
    \filldraw (vz) circle (1.3pt);
    \filldraw (voo) circle (1.3pt);
\end{tikzpicture}
\quad &= \quad
\begin{tikzpicture}[baseline={([yshift=-1.3ex]current bounding box.center)}]
    \coordinate[label=above:$a$] (va) at (0,3);
    \coordinate[label=above:$b$] (vb) at (0,2);
    \coordinate (vm) at (0,1.5);
    \coordinate[label=below:$c$] (vc) at (0,1);
    \coordinate[label=above:$d$] (vd) at (0,0);
    \draw[fill=black!20] (vb) .. controls (-1,3) and (-3,3) .. (-3,1.5) .. controls (-3,0) and (-1,0) .. (vc) .. controls (-1,.5) .. (-1,.75) .. controls (-1,1) .. (vd) .. controls (-1,1.5) .. (va) .. controls (-1,2) .. (-1,2.25) .. controls (-1,2.5) .. (vb);
    \draw[fill=black!20] (-1,.75) .. controls (-1,1) .. (vd) .. controls (-1,1.5) .. (va) .. controls (-1,2) .. (-1,2.25);
    \draw[fill=black!20] (va) arc (90:-90:3 and 1.5) arc (-90:90:.5) arc (-90:90:.5) arc (-90:90:.5);
    \draw[dashed,black!50,thick] (vb) -- (vc) arc (-90:90:1) arc (90:-90:1.5) arc (-90:90:1);
    \node (G1) at (-1.75,1.5) {$\overline{G}_1$};
    \node (G2) at (+2,1.5) {$\overline{G}_2$};
    \node (G) at (-2,0) {$\overline{G}'$};
    \coordinate[label=above right:$0$] (v0)  at ([shift={(vm)}]60:3 and 1.5);
    \coordinate[label=above right:$1$] (v1)  at ([shift={(vm)}]30:3 and 1.5);
    \coordinate[label=below right:$z$] (vz)  at ([shift={(vm)}]-30:3 and 1.5);
    \coordinate[label=below right:$\infty$] (voo) at ([shift={(vm)}]-60:3 and 1.5);
    \filldraw (v0) circle (1.3pt);
    \filldraw (v1) circle (1.3pt);
    \filldraw (vz) circle (1.3pt);
    \filldraw (voo) circle (1.3pt);
    \filldraw (va) circle (1.3pt);
    \filldraw (vb) circle (1.3pt);
    \filldraw (vc) circle (1.3pt);
    \filldraw (vd) circle (1.3pt);
\end{tikzpicture}~~,
\end{align}
where the dashed lines indicate the four-cycle $acbd$.
\end{prop}
\begin{proof}
The proof is analogous to the twist identity for periods, Theorem 2.11 in \cite{Census}.
By solving a linear system it is clear that there exists a unique way to move weights along opposite sides of the four-cycle $acbd$ such that the twisted graph becomes
(internally) completed.

For $i=a,b,c,d$, let
$$
N_i=\sum_{i\sim e\in\sE_{\Gg_1}}\nu_e
$$
be the sum of the weights of the edges in $\Gg_1$ attached to $i$. Because both $\Gg$ and $\Gg'$ are internally completed, we have $N_a=N_b$ and $N_c=N_d$.
It is proved in Corollary~9 of \cite{5twist} that the four-point function $A_{\Gg_1}(a,b,c,d)=A_{\Gg_1}(b,a,d,c)$ is invariant under the twist.
Integration over $a,b,c,d$ (if internal) gives (\ref{eqtwist}).
\end{proof}

\begin{ex}
The magic identities in \cite{Magicid} are twist identities.
\end{ex}

\subsection{Planar duality}\label{sectdual}

Planar duality was already used in \cite{BK} to prove identities between periods. This duality of periods was systematically studied as Fourier identity in Section 2.7 of \cite{Census}.
The name Fourier reflects that the identity can be proved by Fourier transforming Feynman integrals. The situation translates to graphical functions in the following way.

\begin{thm}[Theorem 1.9 in \cite{par}]\label{thmdual}
An uncompleted graphical function $f_G(z)$ has a planar dual if it is externally planar, i.e.\ if $G\cup\{01,0z,1z\}$ is planar.
In the planar dual of $G\cup\{01,0z,1z\}$ we delete the vertex of the face $01z$ and define the external edges $0,1,z$ to be in the faces that are bounded by $G$ and $1z,0z,01$,
respectively. The edges of the resulting graph $G^\star$ are in one-to-one correspondence to the edges of $G$. If the edge-weights fulfill (see (\ref{eqNg}))
\begin{equation*}
\nu_e(G)>0\qquad\text{and}\qquad\lambda\nu_e(G)+\lambda\nu_e(G^\star)=\lambda N_G=\lambda+1\qquad\text{for all }e\in\sE_G\equiv\sE_{G^\star},
\end{equation*}
then the dual graph has a convergent graphical function $f_{G^\star}(z)$ with
\begin{equation}\label{plandual}
\Big[\prod_{e\in\sE_G}\Gamma(\lambda\nu_e)\Big]f_G(z)=\Big[\prod_{e\in\sE_{G^\star}}\Gamma(\lambda\nu_e)\Big]f_{G^\star}(z).
\end{equation}
\end{thm}

With this theorem one can construct an identity between graphical functions in a way which is analogous to the twist identity in the previous section.
One has to keep in mind that, before applying planar duality, one needs to choose a vertex $\infty$ in the split vertices to de-complete the split graph $\Gg_1$.
After dualizing, the graph $G_1^\star$ needs to be completed before it can be glued back into $G$.

Like the twist identity, planar duality in split graphs was first studied in the context of periods, see Theorem 2.8 in \cite{HSSY}.
Because the right (passive) split graph has all external vertices the result for periods trivially translates to graphical functions.

\begin{ex}
Consider the pair of graphs
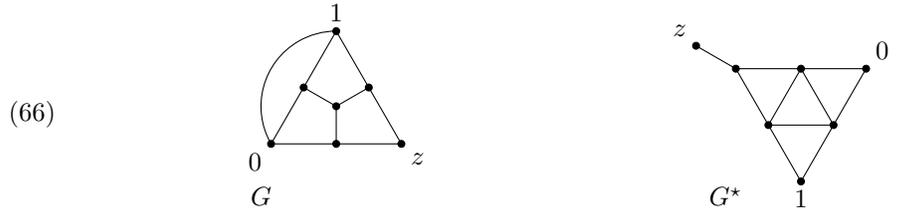
\def\scale{1}
\begin{align}
\begin{tikzpicture}[baseline={([yshift=-1.3ex]0,0)}]
    \coordinate (vm) at (0,0);
    \coordinate[label=above:$1$] (v1) at ([shift=(90:\scale)]vm);
    \coordinate[label=below left:$0$] (v0) at ([shift=(210:\scale)]vm);
    \coordinate[label=below right:$z$] (vz) at ([shift=(330:\scale)]vm);
    \coordinate (v2) at ($(v0)!.5!(v1)$);
    \coordinate (v3) at ($(v0)!.5!(vz)$);
    \coordinate (v4) at ($(v1)!.5!(vz)$);
    \node (G) at (-\scale,{-1.2*\scale}) {$G$};
    \draw (vm) -- (v2);
    \draw (vm) -- (v3);
    \draw (vm) -- (v4);
    \draw (v1) -- (v2);
    \draw (v1) -- (v4);
    \draw (vz) -- (v3);
    \draw (vz) -- (v4);
    \draw (v0) -- (v3);
    \draw (v0) -- (v2);
    \draw (v0) arc(210:90:{\scale});
    \filldraw (vm) circle (1.3pt);
    \filldraw (v1) circle (1.3pt);
    \filldraw (v0) circle (1.3pt);
    \filldraw (vz) circle (1.3pt);
    \filldraw (v2) circle (1.3pt);
    \filldraw (v3) circle (1.3pt);
    \filldraw (v4) circle (1.3pt);
\end{tikzpicture}
&&
\begin{tikzpicture}[baseline={([yshift=-1.3ex]0,0)}]
    \coordinate (vm) at (0,0);
    \coordinate (v1) at ([shift=(150:\scale)]vm);
    \coordinate[label=below:$1$] (v0) at ([shift=(270:\scale)]vm);
    \coordinate[label=above right:$0$] (vz) at ([shift=(390:\scale)]vm);
    \coordinate (v2) at ($(v0)!.5!(v1)$);
    \coordinate (v3) at ($(v0)!.5!(vz)$);
    \coordinate (v4) at ($(v1)!.5!(vz)$);
    \coordinate[label=above left:$z$] (vzz) at ([shift=(150:{1.61*\scale})]vm);
    \node (G) at (-\scale,{-1.2*\scale}) {$G^\star$};
    \draw (v3) -- (v2);
    \draw (v4) -- (v3);
    \draw (v2) -- (v4);
    \draw (v1) -- (v2);
    \draw (v1) -- (v4);
    \draw (vz) -- (v3);
    \draw (vz) -- (v4);
    \draw (v0) -- (v3);
    \draw (v0) -- (v2);
    \draw (v1) -- (vzz);
    %\draw (v0) arc(210:90:{\scale});
%
    %\filldraw (vm) circle (1.3pt);
    \filldraw (vzz) circle (1.3pt);
    \filldraw (v1) circle (1.3pt);
    \filldraw (v0) circle (1.3pt);
    \filldraw (vz) circle (1.3pt);
    \filldraw (v2) circle (1.3pt);
    \filldraw (v3) circle (1.3pt);
    \filldraw (v4) circle (1.3pt);
\end{tikzpicture}~~,
\label{eqduals}
\end{align}
with edge-weights 1. The graphs $G$ and $G^\star$ are planar duals and $N_{G}=2$ in four dimensions. By completion we obtain
$f_{G^\star}(z)=80\zeta(5)\ii D(z)/(z-\zz)$, where the Bloch-Wigner dilogarithm $D(z)$ is defined in (\ref{D}). Planar duality gives $f_G(z)=f_{G^\star}(z)$.
\end{ex}

Non-trivial transformations of graphical functions by planar duality in split graphs are relatively rare (compared to the twist).

%%%%%%%%%%%%%%%%%%%%%%%%%%%%%%%%%%%%%%%%%%%%%%%%%%%%%%%%%%%%%%%%%%%%%%%%%%%%%%%%%
\section{Integration by parts}\label{sectibp}
Integration by parts in position space follows the same idea as in momentum space \cite{ibp}. The difference is that in position space the Feynman rules are unambiguous and the graphical
interpretation is simpler. Position space integration by parts provides a local graph identity which substitutes a vertex of a certain structure with a sum over vertices of different structures.
The general formula is 
\begin{gather}
0=~
(2\!-\!\nu_1\!-\!\sum\limits_{i=1}^N\nu_i) ~~
\begin{tikzpicture}[scale=1,baseline={([yshift=-0.7ex]current bounding box.center)}]
    \coordinate (vm) at (0, 0);
    \coordinate (v1) at (0, .7);
    \coordinate (v2) at (-.7, 0);
    \coordinate (v3) at (0,-.7);
    \draw[black!50] (vm) -- ([shift=(30:.3)]vm);
    \draw[black!50] (vm) -- ([shift=(0:.3)]vm);
    \draw[black!50] (vm) -- ([shift=(-30:.3)]vm);
    \draw (vm) -- node[right] {\tiny{$\nu_1\!+\!\frac{1}{\lambda}$}} (v1);
    \draw (vm) -- node[below] {\tiny{$\nu_2$}} (v2);
    \draw (vm) -- node[right] {\tiny{$\nu_3$}} (v3);
    \filldraw[black!20] (v1) -- ([shift=(60:.2)]v1) arc (60:120:.2) -- (v1);
    \draw (v1) -- ([shift=(60:.2)]v1);
    \draw (v1) -- ([shift=(120:.2)]v1);
    \filldraw[black!20] (v2) -- ([shift=(150:.2)]v2) arc (150:210:.2) -- (v2);
    \draw (v2) -- ([shift=(150:.2)]v2);
    \draw (v2) -- ([shift=(210:.2)]v2);
    \filldraw[black!20] (v3) -- ([shift=(240:.2)]v3) arc (240:300:.2) -- (v3);
    \draw (v3) -- ([shift=(240:.2)]v3);
    \draw (v3) -- ([shift=(300:.2)]v3);
    \filldraw (v1) circle (1.3pt);
    \filldraw (v2) circle (1.3pt);
    \filldraw (v3) circle (1.3pt);
    \filldraw (vm) circle (1.3pt);
\end{tikzpicture}
\hspace{-3ex}
+
\nu_2~~
\begin{tikzpicture}[scale=1,baseline={([yshift=-0.7ex]current bounding box.center)}]
    \coordinate (vm) at (0, 0);
    \coordinate (v1) at (0, .7);
    \coordinate (v2) at (-.7, 0);
    \coordinate (v3) at (0,-.7);
    \draw[black!50] (vm) -- ([shift=(30:.3)]vm);
    \draw[black!50] (vm) -- ([shift=(0:.3)]vm);
    \draw[black!50] (vm) -- ([shift=(-30:.3)]vm);
    \draw (vm) -- node[right] {\tiny{$\nu_1\!+\!\frac{1}{\lambda}$}} (v1);
    \draw (vm) -- node[below] {\tiny{$\nu_2\!+\!\frac{1}{\lambda}$}} (v2);
    \draw (vm) -- node[right] {\tiny{$\nu_3$}} (v3);
    \draw (v1) -- node[inner sep=1pt,above left] {\tiny{$-\frac{1}{\lambda}$}} (v2);
    \filldraw[black!20] (v1) -- ([shift=(60:.2)]v1) arc (60:120:.2) -- (v1);
    \draw (v1) -- ([shift=(60:.2)]v1);
    \draw (v1) -- ([shift=(120:.2)]v1);
    \filldraw[black!20] (v2) -- ([shift=(150:.2)]v2) arc (150:210:.2) -- (v2);
    \draw (v2) -- ([shift=(150:.2)]v2);
    \draw (v2) -- ([shift=(210:.2)]v2);
    \filldraw[black!20] (v3) -- ([shift=(240:.2)]v3) arc (240:300:.2) -- (v3);
    \draw (v3) -- ([shift=(240:.2)]v3);
    \draw (v3) -- ([shift=(300:.2)]v3);
    \filldraw (v1) circle (1.3pt);
    \filldraw (v2) circle (1.3pt);
    \filldraw (v3) circle (1.3pt);
    \filldraw (vm) circle (1.3pt);
\end{tikzpicture}
\hspace{-3ex}
-
\nu_2~~
\begin{tikzpicture}[scale=1,baseline={([yshift=-0.7ex]current bounding box.center)}]
    \coordinate (vm) at (0, 0);
    \coordinate (v1) at (0, .7);
    \coordinate (v2) at (-.7, 0);
    \coordinate (v3) at (0,-.7);
    \draw[black!50] (vm) -- ([shift=(30:.3)]vm);
    \draw[black!50] (vm) -- ([shift=(0:.3)]vm);
    \draw[black!50] (vm) -- ([shift=(-30:.3)]vm);
    \draw (vm) -- node[right] {\tiny{$\nu_1$}} (v1);
    \draw (vm) -- node[below] {\tiny{$\nu_2\!+\!\frac{1}{\lambda}$}} (v2);
    \draw (vm) -- node[right] {\tiny{$\nu_3$}} (v3);
    \filldraw[black!20] (v1) -- ([shift=(60:.2)]v1) arc (60:120:.2) -- (v1);
    \draw (v1) -- ([shift=(60:.2)]v1);
    \draw (v1) -- ([shift=(120:.2)]v1);
    \filldraw[black!20] (v2) -- ([shift=(150:.2)]v2) arc (150:210:.2) -- (v2);
    \draw (v2) -- ([shift=(150:.2)]v2);
    \draw (v2) -- ([shift=(210:.2)]v2);
    \filldraw[black!20] (v3) -- ([shift=(240:.2)]v3) arc (240:300:.2) -- (v3);
    \draw (v3) -- ([shift=(240:.2)]v3);
    \draw (v3) -- ([shift=(300:.2)]v3);
    \filldraw (v1) circle (1.3pt);
    \filldraw (v2) circle (1.3pt);
    \filldraw (v3) circle (1.3pt);
    \filldraw (vm) circle (1.3pt);
\end{tikzpicture}
+
\nu_3
\hspace{-2ex}
\begin{tikzpicture}[scale=1,baseline={([yshift=-0.7ex]current bounding box.center)}]
    \coordinate (vm) at (0, 0);
    \coordinate (v1) at (0, .7);
    \coordinate (v2) at (-.7, 0);
    \coordinate (v3) at (0,-.7);
    \draw[black!50] (vm) -- ([shift=(30:.3)]vm);
    \draw[black!50] (vm) -- ([shift=(0:.3)]vm);
    \draw[black!50] (vm) -- ([shift=(-30:.3)]vm);
    \draw (vm) -- node[right] {\tiny{$\nu_1\!+\!\frac{1}{\lambda}$}} (v1);
    \draw (vm) -- node[below] {\tiny{$\nu_2$}} (v2);
    \draw (vm) -- node[right] {\tiny{$\nu_3\!+\!\frac{1}{\lambda}$}} (v3);
    \draw (v1) .. controls (-1.4,.8) and (-1.4,-.8) .. (v3) node[pos=.3,inner sep=1pt,above left] {\tiny{$-\frac{1}{\lambda}$}} (v2);
    \filldraw[black!20] (v1) -- ([shift=(60:.2)]v1) arc (60:120:.2) -- (v1);
    \draw (v1) -- ([shift=(60:.2)]v1);
    \draw (v1) -- ([shift=(120:.2)]v1);
    \filldraw[black!20] (v2) -- ([shift=(150:.2)]v2) arc (150:210:.2) -- (v2);
    \draw (v2) -- ([shift=(150:.2)]v2);
    \draw (v2) -- ([shift=(210:.2)]v2);
    \filldraw[black!20] (v3) -- ([shift=(240:.2)]v3) arc (240:300:.2) -- (v3);
    \draw (v3) -- ([shift=(240:.2)]v3);
    \draw (v3) -- ([shift=(300:.2)]v3);
    \filldraw (v1) circle (1.3pt);
    \filldraw (v2) circle (1.3pt);
    \filldraw (v3) circle (1.3pt);
    \filldraw (vm) circle (1.3pt);
\end{tikzpicture}
\hspace{-3ex}
-
\nu_3~~
\begin{tikzpicture}[scale=1,baseline={([yshift=-0.7ex]current bounding box.center)}]
    \coordinate (vm) at (0, 0);
    \coordinate (v1) at (0, .7);
    \coordinate (v2) at (-.7, 0);
    \coordinate (v3) at (0,-.7);
    \draw[black!50] (vm) -- ([shift=(30:.3)]vm);
    \draw[black!50] (vm) -- ([shift=(0:.3)]vm);
    \draw[black!50] (vm) -- ([shift=(-30:.3)]vm);
    \draw (vm) -- node[right] {\tiny{$\nu_1$}} (v1);
    \draw (vm) -- node[below] {\tiny{$\nu_2$}} (v2);
    \draw (vm) -- node[right] {\tiny{$\nu_3\!+\!\frac{1}{\lambda}$}} (v3);
    \filldraw[black!20] (v1) -- ([shift=(60:.2)]v1) arc (60:120:.2) -- (v1);
    \draw (v1) -- ([shift=(60:.2)]v1);
    \draw (v1) -- ([shift=(120:.2)]v1);
    \filldraw[black!20] (v2) -- ([shift=(150:.2)]v2) arc (150:210:.2) -- (v2);
    \draw (v2) -- ([shift=(150:.2)]v2);
    \draw (v2) -- ([shift=(210:.2)]v2);
    \filldraw[black!20] (v3) -- ([shift=(240:.2)]v3) arc (240:300:.2) -- (v3);
    \draw (v3) -- ([shift=(240:.2)]v3);
    \draw (v3) -- ([shift=(300:.2)]v3);
    \filldraw (v1) circle (1.3pt);
    \filldraw (v2) circle (1.3pt);
    \filldraw (v3) circle (1.3pt);
    \filldraw (vm) circle (1.3pt);
\end{tikzpicture}
\hspace{-3ex}
+~\ldots
\label{eqibpgeneral}
\end{gather}
which follows from Lemma~\ref{lemibp}. The gray cones indicate where each term has to be substituted into the same ambient graph. The dots indicate further terms
corresponding to edges $4,\ldots,N$ (two terms for each edge). The possible existence of these edges is indicated by the three small lines right to the central vertex.
An important locally completed example in six dimensions is
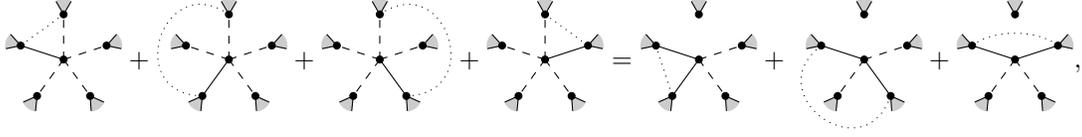
\begin{align}
\begin{tikzpicture}[baseline={([yshift=-0.7ex]0,0)}]
    \coordinate (vm) at (0, 0);
    \coordinate (v1) at ([shift=(90:.6)]vm);
    \coordinate (v2) at ([shift=(162:.6)]vm);
    \coordinate (v3) at ([shift=(234:.6)]vm);
    \coordinate (v4) at ([shift=(306:.6)]vm);
    \coordinate (v5) at ([shift=(18:.6)]vm);
    \filldraw[black!20] (v1) -- ([shift=(60:.2)]v1) arc (60:120:.2) -- (v1);
    \draw (v1) -- ([shift=(60:.2)]v1);
    \draw (v1) -- ([shift=(120:.2)]v1);
    \filldraw[black!20] (v2) -- ([shift=(132:.2)]v2) arc (132:192:.2) -- (v2);
    \draw (v2) -- ([shift=(132:.2)]v2);
    \draw (v2) -- ([shift=(192:.2)]v2);
    \filldraw[black!20] (v3) -- ([shift=(204:.2)]v3) arc (204:264:.2) -- (v3);
    \draw (v3) -- ([shift=(204:.2)]v3);
    \draw (v3) -- ([shift=(264:.2)]v3);
    \filldraw[black!20] (v4) -- ([shift=(276:.2)]v4) arc (276:336:.2) -- (v4);
    \draw (v4) -- ([shift=(274:.2)]v4);
    \draw (v4) -- ([shift=(336:.2)]v4);
    \filldraw[black!20] (v5) -- ([shift=(348:.2)]v5) arc (-12:48:.2) -- (v5);
    \draw (v5) -- ([shift=(348:.2)]v5);
    \draw (v5) -- ([shift=(48:.2)]v5);
    \filldraw (v1) circle (1.3pt);
    \filldraw (v2) circle (1.3pt);
    \filldraw (v3) circle (1.3pt);
    \filldraw (v4) circle (1.3pt);
    \filldraw (v5) circle (1.3pt);
    \filldraw (vm) circle (1.3pt);
    \draw[dashed] (vm) -- (v1);
    \draw (vm) -- (v2);
    \draw[dashed] (vm) -- (v3);
    \draw[dashed] (vm) -- (v4);
    \draw[dashed] (vm) -- (v5);
    \draw[dotted] (v1) -- (v2);
\end{tikzpicture}
+
\hspace{-3ex}
\begin{tikzpicture}[baseline={([yshift=-0.7ex]0,0)}]
    \coordinate (vm) at (0, 0);
    \coordinate (v1) at ([shift=(90:.6)]vm);
    \coordinate (v2) at ([shift=(162:.6)]vm);
    \coordinate (v3) at ([shift=(234:.6)]vm);
    \coordinate (v4) at ([shift=(306:.6)]vm);
    \coordinate (v5) at ([shift=(18:.6)]vm);
    \filldraw[black!20] (v1) -- ([shift=(60:.2)]v1) arc (60:120:.2) -- (v1);
    \draw (v1) -- ([shift=(60:.2)]v1);
    \draw (v1) -- ([shift=(120:.2)]v1);
    \filldraw[black!20] (v2) -- ([shift=(132:.2)]v2) arc (132:192:.2) -- (v2);
    \draw (v2) -- ([shift=(132:.2)]v2);
    \draw (v2) -- ([shift=(192:.2)]v2);
    \filldraw[black!20] (v3) -- ([shift=(204:.2)]v3) arc (204:264:.2) -- (v3);
    \draw (v3) -- ([shift=(204:.2)]v3);
    \draw (v3) -- ([shift=(264:.2)]v3);
    \filldraw[black!20] (v4) -- ([shift=(276:.2)]v4) arc (276:336:.2) -- (v4);
    \draw (v4) -- ([shift=(274:.2)]v4);
    \draw (v4) -- ([shift=(336:.2)]v4);
    \filldraw[black!20] (v5) -- ([shift=(348:.2)]v5) arc (-12:48:.2) -- (v5);
    \draw (v5) -- ([shift=(348:.2)]v5);
    \draw (v5) -- ([shift=(48:.2)]v5);
    \filldraw (v1) circle (1.3pt);
    \filldraw (v2) circle (1.3pt);
    \filldraw (v3) circle (1.3pt);
    \filldraw (v4) circle (1.3pt);
    \filldraw (v5) circle (1.3pt);
    \filldraw (vm) circle (1.3pt);
    \draw[dashed] (vm) -- (v1);
    \draw[dashed] (vm) -- (v2);
    \draw (vm) -- (v3);
    \draw[dashed] (vm) -- (v4);
    \draw[dashed] (vm) -- (v5);
    \draw[dotted] (v1) .. controls ([shift=(126:1.5)]vm) and ([shift=(198:1.5)]vm) .. (v3);
\end{tikzpicture}
+
\begin{tikzpicture}[baseline={([yshift=-0.7ex]0,0)}]
    \coordinate (vm) at (0, 0);
    \coordinate (v1) at ([shift=(90:.6)]vm);
    \coordinate (v2) at ([shift=(162:.6)]vm);
    \coordinate (v3) at ([shift=(234:.6)]vm);
    \coordinate (v4) at ([shift=(306:.6)]vm);
    \coordinate (v5) at ([shift=(18:.6)]vm);
    \filldraw[black!20] (v1) -- ([shift=(60:.2)]v1) arc (60:120:.2) -- (v1);
    \draw (v1) -- ([shift=(60:.2)]v1);
    \draw (v1) -- ([shift=(120:.2)]v1);
    \filldraw[black!20] (v2) -- ([shift=(132:.2)]v2) arc (132:192:.2) -- (v2);
    \draw (v2) -- ([shift=(132:.2)]v2);
    \draw (v2) -- ([shift=(192:.2)]v2);
    \filldraw[black!20] (v3) -- ([shift=(204:.2)]v3) arc (204:264:.2) -- (v3);
    \draw (v3) -- ([shift=(204:.2)]v3);
    \draw (v3) -- ([shift=(264:.2)]v3);
    \filldraw[black!20] (v4) -- ([shift=(276:.2)]v4) arc (276:336:.2) -- (v4);
    \draw (v4) -- ([shift=(274:.2)]v4);
    \draw (v4) -- ([shift=(336:.2)]v4);
    \filldraw[black!20] (v5) -- ([shift=(348:.2)]v5) arc (-12:48:.2) -- (v5);
    \draw (v5) -- ([shift=(348:.2)]v5);
    \draw (v5) -- ([shift=(48:.2)]v5);
    \filldraw (v1) circle (1.3pt);
    \filldraw (v2) circle (1.3pt);
    \filldraw (v3) circle (1.3pt);
    \filldraw (v4) circle (1.3pt);
    \filldraw (v5) circle (1.3pt);
    \filldraw (vm) circle (1.3pt);
    \draw[dashed] (vm) -- (v1);
    \draw[dashed] (vm) -- (v2);
    \draw[dashed] (vm) -- (v3);
    \draw (vm) -- (v4);
    \draw[dashed] (vm) -- (v5);
    \draw[dotted] (v1) .. controls ([shift=(54:1.5)]vm) and ([shift=(-18:1.5)]vm) .. (v4);
\end{tikzpicture}
\hspace{-3ex}
+
\begin{tikzpicture}[baseline={([yshift=-0.7ex]0,0)}]
    \coordinate (vm) at (0, 0);
    \coordinate (v1) at ([shift=(90:.6)]vm);
    \coordinate (v2) at ([shift=(162:.6)]vm);
    \coordinate (v3) at ([shift=(234:.6)]vm);
    \coordinate (v4) at ([shift=(306:.6)]vm);
    \coordinate (v5) at ([shift=(18:.6)]vm);
    \filldraw[black!20] (v1) -- ([shift=(60:.2)]v1) arc (60:120:.2) -- (v1);
    \draw (v1) -- ([shift=(60:.2)]v1);
    \draw (v1) -- ([shift=(120:.2)]v1);
    \filldraw[black!20] (v2) -- ([shift=(132:.2)]v2) arc (132:192:.2) -- (v2);
    \draw (v2) -- ([shift=(132:.2)]v2);
    \draw (v2) -- ([shift=(192:.2)]v2);
    \filldraw[black!20] (v3) -- ([shift=(204:.2)]v3) arc (204:264:.2) -- (v3);
    \draw (v3) -- ([shift=(204:.2)]v3);
    \draw (v3) -- ([shift=(264:.2)]v3);
    \filldraw[black!20] (v4) -- ([shift=(276:.2)]v4) arc (276:336:.2) -- (v4);
    \draw (v4) -- ([shift=(274:.2)]v4);
    \draw (v4) -- ([shift=(336:.2)]v4);
    \filldraw[black!20] (v5) -- ([shift=(348:.2)]v5) arc (-12:48:.2) -- (v5);
    \draw (v5) -- ([shift=(348:.2)]v5);
    \draw (v5) -- ([shift=(48:.2)]v5);
    \filldraw (v1) circle (1.3pt);
    \filldraw (v2) circle (1.3pt);
    \filldraw (v3) circle (1.3pt);
    \filldraw (v4) circle (1.3pt);
    \filldraw (v5) circle (1.3pt);
    \filldraw (vm) circle (1.3pt);
    \draw[dashed] (vm) -- (v1);
    \draw[dashed] (vm) -- (v2);
    \draw[dashed] (vm) -- (v3);
    \draw[dashed] (vm) -- (v4);
    \draw (vm) -- (v5);
    \draw[dotted] (v1) -- (v5);
\end{tikzpicture}
=
\begin{tikzpicture}[baseline={([yshift=-0.7ex]0,0)}]
    \coordinate (vm) at (0, 0);
    \coordinate (v1) at ([shift=(90:.6)]vm);
    \coordinate (v2) at ([shift=(162:.6)]vm);
    \coordinate (v3) at ([shift=(234:.6)]vm);
    \coordinate (v4) at ([shift=(306:.6)]vm);
    \coordinate (v5) at ([shift=(18:.6)]vm);
    \filldraw[black!20] (v1) -- ([shift=(60:.2)]v1) arc (60:120:.2) -- (v1);
    \draw (v1) -- ([shift=(60:.2)]v1);
    \draw (v1) -- ([shift=(120:.2)]v1);
    \filldraw[black!20] (v2) -- ([shift=(132:.2)]v2) arc (132:192:.2) -- (v2);
    \draw (v2) -- ([shift=(132:.2)]v2);
    \draw (v2) -- ([shift=(192:.2)]v2);
    \filldraw[black!20] (v3) -- ([shift=(204:.2)]v3) arc (204:264:.2) -- (v3);
    \draw (v3) -- ([shift=(204:.2)]v3);
    \draw (v3) -- ([shift=(264:.2)]v3);
    \filldraw[black!20] (v4) -- ([shift=(276:.2)]v4) arc (276:336:.2) -- (v4);
    \draw (v4) -- ([shift=(274:.2)]v4);
    \draw (v4) -- ([shift=(336:.2)]v4);
    \filldraw[black!20] (v5) -- ([shift=(348:.2)]v5) arc (-12:48:.2) -- (v5);
    \draw (v5) -- ([shift=(348:.2)]v5);
    \draw (v5) -- ([shift=(48:.2)]v5);
    \filldraw (v1) circle (1.3pt);
    \filldraw (v2) circle (1.3pt);
    \filldraw (v3) circle (1.3pt);
    \filldraw (v4) circle (1.3pt);
    \filldraw (v5) circle (1.3pt);
    \filldraw (vm) circle (1.3pt);
    \draw (vm) -- (v2);
    \draw (vm) -- (v3);
    \draw[dashed] (vm) -- (v4);
    \draw[dashed] (vm) -- (v5);
    \draw[dotted] (v2) -- (v3);
\end{tikzpicture}
+
\hspace{-3ex}
\begin{tikzpicture}[baseline={([yshift=-0.7ex]0,0)}]
    \coordinate (vm) at (0, 0);
    \coordinate (v1) at ([shift=(90:.6)]vm);
    \coordinate (v2) at ([shift=(162:.6)]vm);
    \coordinate (v3) at ([shift=(234:.6)]vm);
    \coordinate (v4) at ([shift=(306:.6)]vm);
    \coordinate (v5) at ([shift=(18:.6)]vm);
    \filldraw[black!20] (v1) -- ([shift=(60:.2)]v1) arc (60:120:.2) -- (v1);
    \draw (v1) -- ([shift=(60:.2)]v1);
    \draw (v1) -- ([shift=(120:.2)]v1);
    \filldraw[black!20] (v2) -- ([shift=(132:.2)]v2) arc (132:192:.2) -- (v2);
    \draw (v2) -- ([shift=(132:.2)]v2);
    \draw (v2) -- ([shift=(192:.2)]v2);
    \filldraw[black!20] (v3) -- ([shift=(204:.2)]v3) arc (204:264:.2) -- (v3);
    \draw (v3) -- ([shift=(204:.2)]v3);
    \draw (v3) -- ([shift=(264:.2)]v3);
    \filldraw[black!20] (v4) -- ([shift=(276:.2)]v4) arc (276:336:.2) -- (v4);
    \draw (v4) -- ([shift=(274:.2)]v4);
    \draw (v4) -- ([shift=(336:.2)]v4);
    \filldraw[black!20] (v5) -- ([shift=(348:.2)]v5) arc (-12:48:.2) -- (v5);
    \draw (v5) -- ([shift=(348:.2)]v5);
    \draw (v5) -- ([shift=(48:.2)]v5);
    \filldraw (v1) circle (1.3pt);
    \filldraw (v2) circle (1.3pt);
    \filldraw (v3) circle (1.3pt);
    \filldraw (v4) circle (1.3pt);
    \filldraw (v5) circle (1.3pt);
    \filldraw (vm) circle (1.3pt);
    \draw (vm) -- (v2);
    \draw[dashed] (vm) -- (v3);
    \draw (vm) -- (v4);
    \draw[dashed] (vm) -- (v5);
    \draw[dotted] (v2) .. controls ([shift=(198:1.5)]vm) and ([shift=(270:1.5)]vm) .. (v4);
\end{tikzpicture}
+
\begin{tikzpicture}[baseline={([yshift=-0.7ex]0,0)}]
    \coordinate (vm) at (0, 0);
    \coordinate (v1) at ([shift=(90:.6)]vm);
    \coordinate (v2) at ([shift=(162:.6)]vm);
    \coordinate (v3) at ([shift=(234:.6)]vm);
    \coordinate (v4) at ([shift=(306:.6)]vm);
    \coordinate (v5) at ([shift=(18:.6)]vm);
    \filldraw[black!20] (v1) -- ([shift=(60:.2)]v1) arc (60:120:.2) -- (v1);
    \draw (v1) -- ([shift=(60:.2)]v1);
    \draw (v1) -- ([shift=(120:.2)]v1);
    \filldraw[black!20] (v2) -- ([shift=(132:.2)]v2) arc (132:192:.2) -- (v2);
    \draw (v2) -- ([shift=(132:.2)]v2);
    \draw (v2) -- ([shift=(192:.2)]v2);
    \filldraw[black!20] (v3) -- ([shift=(204:.2)]v3) arc (204:264:.2) -- (v3);
    \draw (v3) -- ([shift=(204:.2)]v3);
    \draw (v3) -- ([shift=(264:.2)]v3);
    \filldraw[black!20] (v4) -- ([shift=(276:.2)]v4) arc (276:336:.2) -- (v4);
    \draw (v4) -- ([shift=(274:.2)]v4);
    \draw (v4) -- ([shift=(336:.2)]v4);
    \filldraw[black!20] (v5) -- ([shift=(348:.2)]v5) arc (-12:48:.2) -- (v5);
    \draw (v5) -- ([shift=(348:.2)]v5);
    \draw (v5) -- ([shift=(48:.2)]v5);
    \filldraw (v1) circle (1.3pt);
    \filldraw (v2) circle (1.3pt);
    \filldraw (v3) circle (1.3pt);
    \filldraw (v4) circle (1.3pt);
    \filldraw (v5) circle (1.3pt);
    \filldraw (vm) circle (1.3pt);
    \draw (vm) -- (v2);
    \draw[dashed] (vm) -- (v3);
    \draw[dashed] (vm) -- (v4);
    \draw (vm) -- (v5);
    \draw[dotted] (v2) .. controls ([shift=(126:.5)]vm) and ([shift=(54:.5)]vm) .. (v5);
\end{tikzpicture},
\label{eqibp6}
\end{align}
where solid edges 
$\solidedge$ 
have weight $1$, dashed edges 
$\dashededge$
have weight $\frac12$ and dotted edges
$\dottededge$
have weight $-\frac12$.

Consider the integration over an internal vertex $x$ in the Feynman integral (\ref{eqAdef}). Because the integration domain has no boundary we get from Stokes' theorem that
the integral over a closed differential form vanishes. In particular, we get the following lemma.

\begin{lem}\label{lemibp}
Let $\star$ be the $N$-star with $N\geq3$ edges $1,\dots,N$ of weights $\nu_i<1$ attached to the internal vertex $x\in\RR^D$.
The external vertices are ${\mathbf z}=z_1,\ldots,z_N$ so that the integrand in (\ref{eqAdef}) is $I_\star(x,{\mathbf z})=\prod_{i=1}^N\|x-z_i\|^{-2\lambda\nu_i}$. Then
\begin{equation}\label{eqibp1}
\sum_{\mu=1}^D\frac{\partial}{\partial x^\mu}\frac{(x^\mu-z_1^\mu)I_\star(x,{\mathbf z})}{\|x-z_1\|^2}=\lambda\Big(2-\nu_1-\sum_{i=1}^N\nu_i+\sum_{i=2}^N\nu_i
\frac{\|z_1-z_i\|^2-\|x-z_1\|^2}{\|x-z_i\|^2}\Big)\frac{I_\star(x,{\mathbf z})}{\|x-z_1\|^2}.
\end{equation}
\end{lem}
\begin{proof}
The proof is a straightforward calculation using (\ref{eqscalar}) to eliminate scalar products in the numerator.
\end{proof}

If $\sum_{i=1}^N\nu_i>1$, the integral over $x$ is convergent, see Example \ref{indiv}. The left hand side vanishes upon integration, providing the identity in \eqref{eqibpgeneral}.
The restrictions on the weights $\nu_i$ ensure that each individual Feynman integral is convergent.

It is convenient to internally complete the graphs in \eqref{eqibpgeneral} so that the identity can be substituted for internal stars in completed graphs.
To do this we add edges $x\infty$ to the individual graphs such that $x$ has weight $2(\lambda+1)/\lambda$. The completed graphs also need edges between
$\infty$ and the $z_i$ to ensure that the individual graphs have equal weights at corresponding external vertices (see \eqref{eqibp6}).

After completion \eqref{eqibpgeneral} has a very simple structure. We define the graph $f_{ij}$ as the internally completed $(N+1)$-star with weights
$\nu_1,\ldots,\nu_N,\nu_\infty$ and additional edges $xz_i$, $xz_j$ with weight $1/\lambda$ plus edge $z_iz_j$ with weight $-1/\lambda$. All $f_{ij}$ have equal weights at
all vertices. The weight $\nu_\infty$ is fixed by internal completeness to $\nu_\infty=2-\sum_{i=1}^N\nu_i$. Because a self-loop of negative weight nullifies the integrand of a graphical function,
it is natural to define $f_{ii}=0$. We identify the index $N+1$ with $\infty$ and get \eqref{eqibpgeneral} as the special case $i=N+1$, $j=1$ of
\begin{equation}\label{Idij}
\sum_{k=1}^{N+1}\nu_k(f_{ik}-f_{jk})=0,\qquad\text{with}\qquad\sum_{k=1}^{N+1}\nu_k=2.
\end{equation}
Equation (\ref{Idij}) is equivalent to the statement that $\sum_k\nu_kf_{ik}$ is independent of $i$.
\begin{gather*}
\text{For completed graphs}\quad
\nu_{i+1}
\begin{tikzpicture}[scale=1,baseline={([yshift=-0.7ex]current bounding box.center)}]
    \coordinate (vm) at (0, 0);
    \coordinate (v1) at (0, .7);
    \coordinate (v2) at (-.7, 0);
    \coordinate (v3) at (0,-.7);
    \draw[black!50] (vm) -- ([shift=(30:.3)]vm);
    \draw[black!50] (vm) -- ([shift=(0:.3)]vm);
    \draw[black!50] (vm) -- ([shift=(-30:.3)]vm);
    \draw (vm) -- node[right] {\tiny{$\nu_i\!+\!\frac{1}{\lambda}$}} (v1);
    \draw (vm) -- node[below] {\tiny{$\nu_{i\!+\!1}\!+\!\frac{1}{\lambda}\hspace{2ex}$}} (v2);
    \draw (vm) -- node[right] {\tiny{$\nu_{i\!+2}$}} (v3);
    \draw (v1) -- node[inner sep=1pt,above left] {\tiny{$-\frac{1}{\lambda}$}} (v2);
    \filldraw[black!20] (v1) -- ([shift=(60:.2)]v1) arc (60:120:.2) -- (v1);
    \draw (v1) -- ([shift=(60:.2)]v1);
    \draw (v1) -- ([shift=(120:.2)]v1);
    \filldraw[black!20] (v2) -- ([shift=(150:.2)]v2) arc (150:210:.2) -- (v2);
    \draw (v2) -- ([shift=(150:.2)]v2);
    \draw (v2) -- ([shift=(210:.2)]v2);
    \filldraw[black!20] (v3) -- ([shift=(240:.2)]v3) arc (240:300:.2) -- (v3);
    \draw (v3) -- ([shift=(240:.2)]v3);
    \draw (v3) -- ([shift=(300:.2)]v3);
    \filldraw (v1) circle (1.3pt);
    \filldraw (v2) circle (1.3pt);
    \filldraw (v3) circle (1.3pt);
    \filldraw (vm) circle (1.3pt);
\end{tikzpicture}
+\hspace{2ex}
\nu_{i+2}\hspace{-1ex}
\begin{tikzpicture}[scale=1,baseline={([yshift=-0.7ex]current bounding box.center)}]
    \coordinate (vm) at (0, 0);
    \coordinate (v1) at (0, .7);
    \coordinate (v2) at (-.7, 0);
    \coordinate (v3) at (0,-.7);
    \draw[black!50] (vm) -- ([shift=(30:.3)]vm);
    \draw[black!50] (vm) -- ([shift=(0:.3)]vm);
    \draw[black!50] (vm) -- ([shift=(-30:.3)]vm);
    \draw (vm) -- node[right] {\tiny{$\nu_i\!+\!\frac{1}{\lambda}$}} (v1);
    \draw (vm) -- node[below] {\tiny{$\nu_{i\!+\!1}$}} (v2);
    \draw (vm) -- node[right] {\tiny{$\nu_{i\!+2}\!+\!\frac{1}{\lambda}$}} (v3);
    \draw (v1) .. controls (-1.4,.8) and (-1.4,-.8) .. (v3) node[pos=.3,inner sep=1pt,above left] {\tiny{$-\frac{1}{\lambda}$}} (v2);
    \filldraw[black!20] (v1) -- ([shift=(60:.2)]v1) arc (60:120:.2) -- (v1);
    \draw (v1) -- ([shift=(60:.2)]v1);
    \draw (v1) -- ([shift=(120:.2)]v1);
    \filldraw[black!20] (v2) -- ([shift=(150:.2)]v2) arc (150:210:.2) -- (v2);
    \draw (v2) -- ([shift=(150:.2)]v2);
    \draw (v2) -- ([shift=(210:.2)]v2);
    \filldraw[black!20] (v3) -- ([shift=(240:.2)]v3) arc (240:300:.2) -- (v3);
    \draw (v3) -- ([shift=(240:.2)]v3);
    \draw (v3) -- ([shift=(300:.2)]v3);
    \filldraw (v1) circle (1.3pt);
    \filldraw (v2) circle (1.3pt);
    \filldraw (v3) circle (1.3pt);
    \filldraw (vm) circle (1.3pt);
\end{tikzpicture}
+~\ldots\quad\text{does not depend on }i.
%\label{eqibpcompleted}
\end{gather*}
Accordingly, the system of equations (\ref{Idij}) for all $i$, $j$ has rank $N$.

It is non-trivial to efficiently use integration by parts. In QFT (\ref{Idij}) is particularly useful if all internally completed graphs have edge-weights $\lambda\nu_i\in\ZZ$,
with $0\leq\nu_i<1$ for $i=1,\ldots,N+1$.
In four dimensions we get no such configuration with convergent individual terms (one may use dimensional regularization \cite{numfunct,7loops} to derive non-trivial identities).
In six dimensions we want to look at $\nu_i\in\{0,1/2\}$ where we can ignore cases with more than two $\nu_i=0$. We get three cases with edge-weights
$(\frac{1}{2},\frac{1}{2},\frac{1}{2},\frac{1}{2})$, $(0,\frac{1}{2},\frac{1}{2},\frac{1}{2},\frac{1}{2})$, and $(0,0,\frac{1}{2},\frac{1}{2},\frac{1}{2},\frac{1}{2})$.
The first case reduces to a twist identity, see Section \ref{secttwist}, whereas the second and the third case are equivalent.
The unique new formula is depicted in \eqref{eqibp6}. Insertions of graphs on the right hand side of \eqref{eqibp6} typically lead to simpler graphical
functions (compared to the left hand side).
Moreover, a twist identity acts on the right hand side, so that one effectively obtains five independent equations by cyclically permuting the external labels.
Regretfully, the five equations cannot be solved for one of the five possible internal configurations on the left hand side without generating denominators which
are not propagators. With no systematic way to use this integration by parts formula one needs to try a plethora of configurations. This makes the use of integration by parts tedious.
Moreover---because in subgraphs the external vertices may become internal---one has to check convergence of each individual term.
Only relations with convergent individual terms can be used without regularization. Still, integration by parts is a powerful technique in six dimensions.
It allowed the authors to compute all primitive Feynman periods in $\phi^3$ up to loop order six (and many beyond six loops) \cite{phi3,Shlog}.

In general, integration by parts becomes more powerful (but harder to handle) in higher dimensions.

%%%%%%%%%%%%%%%%%%%%%%%%%%%%%%%%%%%%%%%%%%%%%%%%%%%%%%%%%%%%%%%%%%%%%%%%%%%%%%%%%
\section{Unique triangles}\label{sectunique}
A {\em unique} three-star is a star with three edges whose weights sum up to $2(\lambda+1)/\lambda$ (the notion of uniqueness was introduced in \cite{K2}).
One can eliminate an insertion of a unique three-star with internal central vertex by the factor identity in Section \ref{sectfactor}.
The three-star is replaced by a unique triangle between its external vertices (which may be internal in the ambient graph) times
a completed period of $K_4$ topology which is a product of gamma functions, see (\ref{eqconvolute}).
The unique triangle has edge-weights that sum up to $(\lambda+1)/\lambda$. In some situations it can be beneficial to reverse the process and replace a unique triangle
by a unique three-star. This is evident if the unique triangle has an external vertex and weights in $\frac{1}{\lambda}\ZZ$. Then the unique three-star has a single external edge.
By Section \ref{sectappedge} the three-star can be calculated from the graph with contracted external edge. Diagrammatically, this amounts to the local operation
\begin{align}
\begin{tikzpicture}[baseline={([yshift=-1.3ex]current bounding box.center)}]
    \pgfmathsetmacro{\rad}{1}
    \coordinate (v0) at (0,0);
    \coordinate (v1) at ([shift=(90:\rad)]v0);
    \coordinate (v2) at ([shift=(210:\rad)]v0);
    \coordinate (v3) at ([shift=(330:\rad)]v0);
    \filldraw (v0) circle (1.3pt);
    \filldraw (v1) circle (1.3pt);
    \filldraw (v2) circle (1.3pt);
    \filldraw (v3) circle (1.3pt);
    \draw (v0) -- node[inner sep=0pt,pos=.4,right] {\tiny{$\nu_1\!+\!\nu_2$}} (v1);
    \draw (v0) -- node[inner sep=0pt,pos=.1,above left] {\tiny{$\nu_1\!+\!\nu_3$}} (v2);
    \draw (v0) -- node[inner sep=0pt,pos=.7,below left] {\tiny{$\nu_2\!+\!\nu_3$}} (v3);
    \draw (v1) arc (90:210:\rad) node[pos=.5,left] {\tiny{$\nu_2\!+\!\nu_3$}};
    \draw (v2) arc (210:330:\rad) node[pos=.5,below] {\tiny{$\nu_1\!+\!\nu_2$}};
    \draw (v3) arc (330:450:\rad) node[pos=.5,right] {\tiny{$\nu_1\!+\!\nu_3$}};
\end{tikzpicture}
\times
\quad
\begin{tikzpicture}[baseline={([yshift=-1.3ex]current bounding box.center)}]
    \pgfmathsetmacro{\rad}{1}
    \coordinate (v0) at (0,0);
    \coordinate[label=above:$z$] (v1) at ([shift=(90:\rad)]v0);
    \coordinate (v2) at ([shift=(210:\rad)]v0);
    \coordinate (v3) at ([shift=(330:\rad)]v0);
    \draw (v1) -- node[left] {\tiny{$\nu_1$}} (v2);
    \draw (v2) -- node[below] {\tiny{$\nu_3$}} (v3);
    \draw (v3) -- node[right] {\tiny{$\nu_2$}} (v1);
    \filldraw[black!20] (v2) -- ([shift=(240:.2)]v2) arc (240:300:.2) -- (v2);
    \draw (v2) -- ([shift=(240:.2)]v2);
    \draw (v2) -- ([shift=(300:.2)]v2);
    \filldraw[black!20] (v3) -- ([shift=(240:.2)]v3) arc (240:300:.2) -- (v3);
    \draw (v3) -- ([shift=(240:.2)]v3);
    \draw (v3) -- ([shift=(300:.2)]v3);
    \filldraw (v1) circle (1.3pt);
    \filldraw (v2) circle (1.3pt);
    \filldraw (v3) circle (1.3pt);
\end{tikzpicture}
\quad
=
\quad
\begin{tikzpicture}[baseline={([yshift=-1.3ex]current bounding box.center)}]
    \pgfmathsetmacro{\rad}{1}
    \coordinate (v0) at (0,0);
    \coordinate[label=above:$z$] (v1) at ([shift=(90:\rad)]v0);
    \coordinate (v2) at ([shift=(210:\rad)]v0);
    \coordinate (v3) at ([shift=(330:\rad)]v0);
    \draw (v0) -- node[inner sep=0pt,pos=.6,right] {\tiny{$\nu_1\!+\!\nu_2$}} (v1);
    \draw (v0) -- node[inner sep=0pt,pos=.4,above left] {\tiny{$\nu_1\!+\!\nu_3$}} (v2);
    \draw (v0) -- node[inner sep=0pt,pos=.8,below left] {\tiny{$\nu_2\!+\!\nu_3$}} (v3);
    \filldraw[black!20] (v2) -- ([shift=(240:.2)]v2) arc (240:300:.2) -- (v2);
    \draw (v2) -- ([shift=(240:.2)]v2);
    \draw (v2) -- ([shift=(300:.2)]v2);
    \filldraw[black!20] (v3) -- ([shift=(240:.2)]v3) arc (240:300:.2) -- (v3);
    \draw (v3) -- ([shift=(240:.2)]v3);
    \draw (v3) -- ([shift=(300:.2)]v3);
    \filldraw (v0) circle (1.3pt);
    \filldraw (v1) circle (1.3pt);
    \filldraw (v2) circle (1.3pt);
    \filldraw (v3) circle (1.3pt);
\end{tikzpicture}
\quad
\mathop{
\longleftrightarrow
}_{\text{(append edge)}}
\quad
\begin{tikzpicture}[baseline={([yshift=-1.3ex]current bounding box.center)}]
    \pgfmathsetmacro{\rad}{1}
    \coordinate (v0) at (0,0);
    \coordinate[label=above:$z$] (v1) at ([shift=(90:\rad)]v0);
    \coordinate (v2) at ([shift=(210:\rad)]v0);
    \coordinate (v3) at ([shift=(330:\rad)]v0);
    \draw (v1) -- node[inner sep=0pt,pos=.4,above left] {\tiny{$\nu_1\!+\!\nu_3$}} (v2);
    \draw (v1) -- node[inner sep=0pt,pos=.4,above right] {\tiny{$\nu_2\!+\!\nu_3$}} (v3);
    \filldraw[black!20] (v2) -- ([shift=(240:.2)]v2) arc (240:300:.2) -- (v2);
    \draw (v2) -- ([shift=(240:.2)]v2);
    \draw (v2) -- ([shift=(300:.2)]v2);
    \filldraw[black!20] (v3) -- ([shift=(240:.2)]v3) arc (240:300:.2) -- (v3);
    \draw (v3) -- ([shift=(240:.2)]v3);
    \draw (v3) -- ([shift=(300:.2)]v3);
    \filldraw (v1) circle (1.3pt);
    \filldraw (v2) circle (1.3pt);
    \filldraw (v3) circle (1.3pt);
\end{tikzpicture}~~,
\label{eqexttriangle}
\end{align}
where $\nu_1+\nu_2+\nu_3=(\lambda+1)/\lambda$.
Effectively, the $\nu_3$ edge in the unique triangle on the left is deleted on the right graph while the weights of the other edges change to preserve the weights on the bottom vertices.

Unique triangles can be generated between any three vertices by adding pairs of edges with zero total weights. In general, this makes graphs more complicated. It is however possible
that a graph simplifies (with other techniques) if one completes two connected edges $ab$, $bc$ to a unique triangle $abc$.
If $b$ is external, one can use this technique to effectively reduce the weight of the edge $ac$.

We can also reverse the process in \eqref{eqexttriangle} to increase the weight of the edge opposite to the external vertex $z$. This is useful if the weight of this edge is negative.

Note that unique three-stars and triangles with weights in $\frac{1}{\lambda}\ZZ$ only exist in $\geq6$ dimensions.

\begin{ex}\label{ex2int}
Consider the generic internally completed graphical function with two internal vertices 
    \def\scale{1.1}
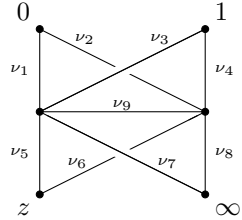
\begin{align*}
\begin{tikzpicture}[baseline={([yshift=-.7ex]0,0)}]
    \coordinate (vm) at  (0,0);
    \coordinate[label=above left:$0$] (v0) at  (-\scale, \scale);
    \coordinate[label=above right:$1$] (v1) at  ( \scale, \scale);
    \coordinate[label=below left:$z$] (vz) at (-\scale,-\scale);
    \coordinate[label=below right:$\infty$] (voo)  at (\scale, -\scale);
    \coordinate (v2)  at (-\scale, 0);
    \coordinate (v3)  at (\scale, 0);
    \coordinate (vm)  at (0, 0);
    \draw[preaction={draw, white, line width=5pt, -}] (vz) -- node[pos=.3,inner sep=.5pt,above left] {\tiny{$\nu_6$}} (v3);
    \draw[preaction={draw, white, line width=5pt, -}] (voo) -- node[pos=.3,inner sep=.5pt,above right] {\tiny{$\nu_7$}} (v2);
    \draw[preaction={draw, white, line width=5pt, -}] (v0) -- node[pos=.2,inner sep=.5pt,above right] {\tiny{$\nu_2$}} (v3);
    \draw[preaction={draw, white, line width=5pt, -}] (v1) -- node[pos=.2,inner sep=.5pt,above left] {\tiny{$\nu_3$}} (v2);
    \draw (v3) -- ($(v3)!.1!(vz)$);
    \draw (vz) -- node[left] {\tiny{$\nu_5$}} (v2);
    \draw (voo) -- node[right] {\tiny{$\nu_8$}} (v3);
    \draw (v1) -- (v2);
    \draw (v2) -- (voo);
    \draw (v2) -- node[inner sep=.5pt,above] {\tiny{$\nu_9$}} (v3);
    \draw (v0) -- node[left] {\tiny{$\nu_1$}} (v2);
    \draw (v3) -- node[right] {\tiny{$\nu_4$}} (v1);
    \filldraw (voo) circle (1.3pt);
    \filldraw (v1) circle (1.3pt);
    \filldraw (v0) circle (1.3pt);
    \filldraw (vz) circle (1.3pt);
    \filldraw (v2) circle (1.3pt);
    \filldraw (v3) circle (1.3pt);
\end{tikzpicture}~~.
\end{align*}
We assume that $1\geq\nu_i\in\frac1\lambda\ZZ$ for all $i$. In four dimensions the internal vertices have total weight four. Therefore all $\nu_i\in\{0,1\}$ with at least
one $\nu_i=0$. If one $\nu_i=0$, the graphical function is constructible in any even dimension $\geq4$, see Section \ref{sectconstructible}.

In six dimensions the internal vertices have total weight three. Assume that $\nu_i\in\{\frac{1}{2},1\}$ for all external edges $1,2,\ldots,8$ (see Example \ref{exgen}).
Unique triangles or stars give the following reductions of the weights $(\nu_1,\nu_2,\ldots,\nu_9)$:
\begin{align*}
\arraycolsep=2.4pt
\def\arraystretch{1.2}
\begin{array}{rrrrrrrrrrrrrrrrrrr}
(\frac{1}{2},&\frac{1}{2},&\frac{1}{2},&\frac{1}{2},&\frac{1}{2},&\frac{1}{2},&\frac{1}{2},&\frac{1}{2},&1)
&\rightarrow&(1,&1,&\frac{1}{2},&\frac{1}{2},&\frac{1}{2},&\frac{1}{2},&\frac{1}{2},&\frac{1}{2},&\frac{1}{2}),\\
(1,&1,&\frac{1}{2},&\frac{1}{2},&\frac{1}{2},&\frac{1}{2},&\frac{1}{2},&\frac{1}{2},&\frac{1}{2})
&\rightarrow&(1,&1,&1,&1,&\frac{1}{2},&\frac{1}{2},&\frac{1}{2},&\frac{1}{2},&0),\\
(1,&\frac{1}{2},&\frac{1}{2},&1,&\frac{1}{2},&\frac{1}{2},&\frac{1}{2},&\frac{1}{2},&\frac{1}{2})
&\rightarrow&(1,&\frac{1}{2},&\frac{1}{2},&1,&1,&1,&\frac{1}{2},&\frac{1}{2},&0),\\
(1,&1,&1,&1,&1,&1,&\frac{1}{2},&\frac{1}{2},&-\frac{1}{2})
&\leftarrow&(1,&1,&1,&1,&\frac{1}{2},&\frac{1}{2},&\frac{1}{2},&\frac{1}{2},&0),\\
(1,&1,&1,&1,&1,&\frac{1}{2},&\frac{1}{2},&1,&-\frac{1}{2})
&\leftarrow&(1,&1,&\frac{1}{2},&\frac{1}{2},&1,&\frac{1}{2},&\frac{1}{2},&1,&0),\\
(1,&1,&1,&1,&1,&1,&1,&1,&-1)
&\leftarrow&(1,&1,&1,&1,&1,&1,&\frac{1}{2},&\frac{1}{2},&-\frac{1}{2}),
\end{array}
\end{align*}
where the orientation of the arrows indicates the direction in which we use \eqref{eqexttriangle}. Note that by permutation symmetry, Section \ref{sectperm}, we
can use uniqueness reductions at any external vertex. Every configuration of weights thus reduces to a constructible product.
\end{ex}

\begin{ex}\label{exK6}
The calculation of the period of the complete graph $\K_{6,D}$ with six vertices in $D$ dimensions leads to the generic graphical function with two internal vertices.
It corresponds to the case $\nu_1=\ldots=\nu_9=2(\lambda+1)/5\lambda$ in the previous example, see Example \ref{Kn}. We may consider the internal edge as two edges of weight $(\lambda+1)/5\lambda$
belonging to two unique triangles with (any) two external vertices. Reduction of the two unique triangles internally disconnects the graph. The result is the square of
a star with two edges of weight $3(\lambda+1)/5\lambda$ and two edges of weight $2(\lambda+1)/5\lambda$ (which is constructible).
\end{ex}

%%%%%%%%%%%%%%%%%%%%%%%%%%%%%%%%%%%%%%%%%%%%%%%%%%%%%%%%%%%%%%%%%%%%%%%%%%%%%%%%%
\section{External differentiation}\label{sectextdiff}
In this section we derive an external version of the integration by parts method in Section \ref{sectibp}.
Consider the situation in Lemma \ref{lemibp} with $x$ replaced by the external vertex $z\in\RR^D$. We set $z_1=0$ and $\nu_1=-1/\lambda$ in (\ref{eqibp1}).
Because $\sum_\mu z^\mu\frac{\partial}{\partial z^\mu}=-D+\sum_\mu\frac{\partial}{\partial z^\mu}z^\mu$ we obtain
\begin{equation}\label{eqextibp}
\sum_{\mu=1}^Dz^\mu\frac{\partial}{\partial z^\mu}I_\star(z,{\mathbf z})=\Big(-\sum_{i=1}^N\lambda\nu_i+\sum_{i=1}^N\lambda\nu_i
\frac{\|z_i\|^2-\|z\|^2}{\|z-z_i\|^2}\Big)I_\star(z,{\mathbf z}),
\end{equation}
where we relabeled the external vertices so that they are $z_1,\ldots,z_N$ again. We consider $z_1,\ldots,z_N$ as internal vertices in some ambient graph.
The differential operator $\sum_\mu z^\mu\frac{\partial}{\partial z^\mu}$ translates to $z\partial_z+\zz\partial_\zz$. To see this it suffices to check that
the operators have corresponding actions on the invariants $\|z\|^2\cong z\zz$ and $\|z-1\|^2\cong (z-1)(\zz-1)$---see (\ref{eqinvs})---where on the left hand sides $z$ is in $\RR^D$ whereas
on the right hand sides $z\in\CC$.
After internal completion we obtain the identity,
\begin{align}
\sum\limits_{i=1}^N\lambda\nu_i
~~
\begin{tikzpicture}[baseline={([yshift=-0.7ex]0,0)}]
    \pgfmathsetmacro{\rad}{1.2}
    \coordinate[label=above:$z$] (vz) at (0, \rad);
    \coordinate (v1) at (-\rad, 0);
    \coordinate (v2) at (0, 0);
    \coordinate (v3) at (\rad, 0);
    \coordinate[label=left:$0$] (v0) at (-{\rad/2}, {-2*\rad/3});
    \coordinate[label=right:$\infty$] (voo) at ({\rad/2}, {-2*\rad/3});
    \node (v1a) at ({-\rad/2},0) {$...$};
    \node (v3a) at ({\rad/2},0) {$...$};
    \filldraw[black!20] (v0) -- ([shift=(60:.2)]v0) arc (60:120:.2) -- (v0);
    \draw (v0) -- ([shift=(60:.2)]v0);
    \draw (v0) -- ([shift=(120:.2)]v0);
    \filldraw[black!20] (voo) -- ([shift=(60:.2)]voo) arc (60:120:.2) -- (voo);
    \draw (voo) -- ([shift=(60:.2)]voo);
    \draw (voo) -- ([shift=(120:.2)]voo);
    \filldraw[black!20] (v1) -- ([shift=(240:.2)]v1) arc (240:300:.2) -- (v1);
    \draw (v1) -- ([shift=(240:.2)]v1);
    \draw (v1) -- ([shift=(300:.2)]v1);
    \filldraw[black!20] (v2) -- ([shift=(240:.2)]v2) arc (240:300:.2) -- (v2);
    \draw (v2) -- ([shift=(240:.2)]v2);
    \draw (v2) -- ([shift=(300:.2)]v2);
    \filldraw[black!20] (v3) -- ([shift=(240:.2)]v3) arc (240:300:.2) -- (v3);
    \draw (v3) -- ([shift=(240:.2)]v3);
    \draw (v3) -- ([shift=(300:.2)]v3);
    \draw[black!50] (vz) -- ([shift=(240:{\rad/2})]vz);
    \draw[black!50] (vz) -- ([shift=(255:{\rad/2})]vz);
    \draw[black!50] (vz) -- ([shift=(285:{\rad/2})]vz);
    \draw[black!50] (vz) -- ([shift=(300:{\rad/2})]vz);
    \filldraw (vz) circle (1.3pt);
    \filldraw (v1) circle (1.3pt);
    \filldraw (v2) circle (1.3pt);
    \filldraw (v3) circle (1.3pt);
    \filldraw (v0) circle (1.3pt);
    \filldraw (voo) circle (1.3pt);
    \draw (vz) -- node[inner sep=0pt,above left] {\tiny{$\nu_1$}} (v1);
    \draw (vz) -- node[inner sep=0pt,pos=.8,right] {\tiny{$\nu_i\!+\!\frac{1}{\lambda}$}} (v2);
    \draw (vz) -- node[inner sep=0pt,above right] {\tiny{$\nu_N$}} (v3);
    \draw (v2) to[out=180,in=0] node[inner sep=0pt,below right] {\tiny{$-\frac{1}{\lambda}$}} (v0);
\end{tikzpicture}
&
=
\Big(z\partial_z+\zz\partial_\zz+\sum\limits_{i=1}^N\lambda\nu_i\Big)
~~
\begin{tikzpicture}[baseline={([yshift=-0.7ex]0,0)}]
    \pgfmathsetmacro{\rad}{1.2}
    \coordinate[label=above:$z$] (vz) at (0, \rad);
    \coordinate (v1) at (-\rad, 0);
    \coordinate (v2) at (0, 0);
    \coordinate (v3) at (\rad, 0);
    \coordinate[label=left:$0$] (v0) at (-{\rad/2}, {-2*\rad/3});
    \coordinate[label=right:$\infty$] (voo) at ({\rad/2}, {-2*\rad/3});
    \coordinate (v1a) at ({-\rad/2},0);
    \node (v3a) at ({\rad/4},0) {$...$};
    \filldraw[black!20] (v0) -- ([shift=(60:.2)]v0) arc (60:120:.2) -- (v0);
    \draw (v0) -- ([shift=(60:.2)]v0);
    \draw (v0) -- ([shift=(120:.2)]v0);
    \filldraw[black!20] (voo) -- ([shift=(60:.2)]voo) arc (60:120:.2) -- (voo);
    \draw (voo) -- ([shift=(60:.2)]voo);
    \draw (voo) -- ([shift=(120:.2)]voo);
    \filldraw[black!20] (v1) -- ([shift=(240:.2)]v1) arc (240:300:.2) -- (v1);
    \draw (v1) -- ([shift=(240:.2)]v1);
    \draw (v1) -- ([shift=(300:.2)]v1);
    \filldraw[black!20] (v1a) -- ([shift=(240:.2)]v1a) arc (240:300:.2) -- (v1a);
    \draw (v1a) -- ([shift=(240:.2)]v1a);
    \draw (v1a) -- ([shift=(300:.2)]v1a);
    \filldraw[black!20] (v3) -- ([shift=(240:.2)]v3) arc (240:300:.2) -- (v3);
    \draw (v3) -- ([shift=(240:.2)]v3);
    \draw (v3) -- ([shift=(300:.2)]v3);
    \draw[black!50] (vz) -- ([shift=(270:{\rad/2})]vz);
    \draw[black!50] (vz) -- ([shift=(292.5:{\rad/2})]vz);
    \filldraw (vz) circle (1.3pt);
    \filldraw (v1) circle (1.3pt);
    \filldraw (v1a) circle (1.3pt);
    \filldraw (v3) circle (1.3pt);
    \filldraw (v0) circle (1.3pt);
    \filldraw (voo) circle (1.3pt);
    \draw (vz) -- node[inner sep=0pt,above left] {\tiny{$\nu_1$}} (v1);
    \draw (vz) -- node[inner sep=.5pt,below right] {\tiny{$\nu_2$}} (v1a);
    \draw (vz) -- node[inner sep=0pt,above right] {\tiny{$\nu_N$}} (v3);
\end{tikzpicture}
+
z \zz 
\sum\limits_{i=1}^N\lambda\nu_i
~~
\begin{tikzpicture}[baseline={([yshift=-0.7ex]0,0)}]
    \pgfmathsetmacro{\rad}{1.2}
    \coordinate[label=above:$z$] (vz) at (0, \rad);
    \coordinate (v1) at (-\rad, 0);
    \coordinate (v2) at (0, 0);
    \coordinate (v3) at (\rad, 0);
    \coordinate[label=left:$0$] (v0) at (-{\rad/2}, {-2*\rad/3});
    \coordinate[label=right:$\infty$] (voo) at ({\rad/2}, {-2*\rad/3});
    \node (v1a) at ({-\rad/2},0) {$...$};
    \node (v3a) at ({\rad/2},0) {$...$};
    \filldraw[black!20] (v0) -- ([shift=(60:.2)]v0) arc (60:120:.2) -- (v0);
    \draw (v0) -- ([shift=(60:.2)]v0);
    \draw (v0) -- ([shift=(120:.2)]v0);
    \filldraw[black!20] (voo) -- ([shift=(60:.2)]voo) arc (60:120:.2) -- (voo);
    \draw (voo) -- ([shift=(60:.2)]voo);
    \draw (voo) -- ([shift=(120:.2)]voo);
    \filldraw[black!20] (v1) -- ([shift=(240:.2)]v1) arc (240:300:.2) -- (v1);
    \draw (v1) -- ([shift=(240:.2)]v1);
    \draw (v1) -- ([shift=(300:.2)]v1);
    \filldraw[black!20] (v2) -- ([shift=(240:.2)]v2) arc (240:300:.2) -- (v2);
    \draw (v2) -- ([shift=(240:.2)]v2);
    \draw (v2) -- ([shift=(300:.2)]v2);
    \filldraw[black!20] (v3) -- ([shift=(240:.2)]v3) arc (240:300:.2) -- (v3);
    \draw (v3) -- ([shift=(240:.2)]v3);
    \draw (v3) -- ([shift=(300:.2)]v3);
    \draw[black!50] (vz) -- ([shift=(240:{\rad/2})]vz);
    \draw[black!50] (vz) -- ([shift=(255:{\rad/2})]vz);
    \draw[black!50] (vz) -- ([shift=(285:{\rad/2})]vz);
    \draw[black!50] (vz) -- ([shift=(300:{\rad/2})]vz);
    \filldraw (vz) circle (1.3pt);
    \filldraw (v1) circle (1.3pt);
    \filldraw (v2) circle (1.3pt);
    \filldraw (v3) circle (1.3pt);
    \filldraw (v0) circle (1.3pt);
    \filldraw (voo) circle (1.3pt);
    \draw (vz) -- node[inner sep=0pt,above left] {\tiny{$\nu_1$}} (v1);
    \draw (vz) -- node[inner sep=0pt,pos=.8,right] {\tiny{$\nu_i\!+\!\frac{1}{\lambda}$}} (v2);
    \draw (vz) -- node[inner sep=0pt,above right] {\tiny{$\nu_N$}} (v3);
    \draw (v2) to[in=180,out=0] node[inner sep=0pt,below left] {\tiny{$-\frac{1}{\lambda}$}} (voo);
\end{tikzpicture},
\label{eqextdiff}
\end{align}
where the graphs have to be considered as subgraph insertions into a larger ambient graph which may have extra edges attached to 0, $\infty$, or the $N$ unlabeled internal vertices.

Note that in polar coordinates $z=Z\ee^{\ii\phi}$ the differential operator $z\partial_z+\zz\partial_\zz+\Lambda=Z\partial_Z+\Lambda$ can be inverted ($\Lambda=\sum_i\lambda\nu_i$).
In fact, one best changes the coordinate $\zz$ to $zy$. In the new coordinates $(z,y)$ the differential operator is $z\partial_z+\Lambda$ with an inverse $z^{-\Lambda}\int\dd zz^{\Lambda-1}$.
The ubiquitous denominator $z-\zz=z(1-y)$ remains linear and the inverse transformation is $y\rightarrow\zz/z$.
The solution is unique in the space of symmetric (under $z\leftrightarrow\zz$) functions and often maps GSVHs to GSVHs if $\Lambda\in\ZZ$ (see Section \ref{sectpf4} for the analogous angular case).

\begin{ex}\label{exgen}
Consider the generic graph in six dimensions with two internal vertices in Example \ref{ex2int}.
Here, we assume there exists an external edge with negative weight. Because a triangle with three weight one edges is divergent, see (\ref{eqconv}), we are left with the following five cases:
\begin{eqnarray*}
(\nu_1,\nu_2,\nu_3,\nu_4,\nu_5,\nu_6,\nu_7,\nu_8,\nu_9):&(-1,\frac{1}{2},1,\frac{1}{2},1,\frac{1}{2},1,\frac{1}{2},1),&
\textstyle(-\frac{1}{2},\frac{1}{2},\frac{1}{2},\frac{1}{2},1,\frac{1}{2},1,\frac{1}{2},1),\\
\textstyle(-\frac{1}{2},1,1,-\frac{1}{2},1,1,1,1,\frac{1}{2}),&(-\frac{1}{2},\frac{1}{2},1,\frac{1}{2},1,\frac{1}{2},1,1,\frac{1}{2}),&
\textstyle(-\frac{1}{2},1,1,\frac{1}{2},1,\frac{1}{2},1,\frac{1}{2},\frac{1}{2}).
\end{eqnarray*}
There always exists a (sequence of) reduction(s) by external differentiation into constructible graphical functions (i.e.\ graphs with an external edge of weight zero).
The reduction may produce pairs of graphs with triangles of weight one edges where only the sums of the pairs are finite. The divergent graphical
functions can be dimensionally regularized \cite{numfunct,7loops} and calculated in terms of GSVHs \cite{Shlog}. We find that in six dimensions every graphical function with
two internal vertices is a GSVH.
\end{ex}

\begin{ex}\label{K7ex}
Consider the complete graph $\K_{7,6}$ with seven vertices and edge-weights $1/2$ in six dimensions (see Example \ref{Kn}). Choose any four vertices as external vertices $0,1,z,\infty$ and delete
the edges between the external vertices to obtain the graph
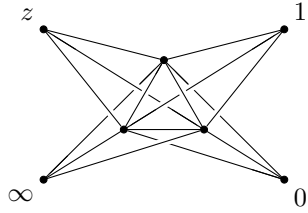
\begin{align*}
    \def\scale{1.6}
\begin{tikzpicture}[baseline={([yshift=-.7ex]0,0)}]
    \coordinate (vm) at  (0,0);
    \coordinate[label=above left:$z$] (vz) at  (-\scale, {\scale/2});
    \coordinate[label=above right:$1$] (v1) at  ( \scale, {\scale/2});
    \coordinate[label=below left:$\infty$] (voo) at (-\scale,{-\scale/4*3});
    \coordinate[label=below right:$0$] (v0)  at (\scale, {-\scale/4*3});
    \coordinate (va)  at (0, { \scale/3*(sqrt(3)-1)});
    \coordinate (vb)  at ({-\scale/3}, {-\scale/3});
    \coordinate (vc)  at ({ \scale/3}, {-\scale/3});
    \draw (voo) -- (va);
    \draw (v0) -- (va);
    \draw[preaction={draw, white, line width=3pt, -}] (vz) -- (vc);
    \draw[preaction={draw, white, line width=3pt, -}] (v1) -- (vb);
    \draw (v0) -- (vb);
    \draw[preaction={draw, white, line width=3pt, -}] (vz) -- (vb);
    \draw[preaction={draw, white, line width=3pt, -}] (v1) -- (vc);
    \draw[preaction={draw, white, line width=3pt, -}] (voo) -- (vc);
    \draw[preaction={draw, white, line width=3pt, -}] (va) -- (vb);
    \draw[preaction={draw, white, line width=3pt, -}] (va) -- (vc);
    \draw (vz) -- ($(vz)!.2!(vc)$);
    \draw (v1) -- ($(v1)!.2!(vb)$);
    \draw (va) -- ($(va)!.2!(v0)$);
    \draw (v0) -- ($(v0)!.2!(vb)$);
    \draw (voo) -- ($(voo)!.2!(va)$);
    \draw (va) -- ($(va)!.2!(voo)$);
    \draw (va) -- ($(va)!.2!(vb)$);
    \draw (vb) -- ($(vb)!.2!(v1)$);
    \draw (vc) -- ($(vc)!.2!(vz)$);
    \draw (vc) -- (vb);
    \draw (vz) -- (va);
    \draw (v1) -- (va);
    \draw (voo) -- (vb);
    \draw (v0) -- (vc);
    \filldraw (v0) circle (1.3pt);
    \filldraw (v1) circle (1.3pt);
    \filldraw (vz) circle (1.3pt);
    \filldraw (voo) circle (1.3pt);
    \filldraw (va) circle (1.3pt);
    \filldraw (vb) circle (1.3pt);
    \filldraw (vc) circle (1.3pt);
\end{tikzpicture}.
\end{align*}
One can use external differentiation for $N=3$ in various ways to reduce the number of edges.
The most symmetric way is to set $\nu_1=\nu_2=\nu_3=1/2$ and invert the differential operator $z\partial_z+\zz\partial_\zz+3$. It is also possible to
set $\nu_1=0$, $\nu_2=\nu_3=1/2$ and solve for the first summand on the left hand side of \eqref{eqextdiff}. Together with other identities it is possible to calculate $f_{\K_{7,6}}(z)$.
The result is a GSVH of weight seven divided by $z\zz(z-1)(\zz-1)(z-\zz)^3$ \cite{Shlog}.

In four and six dimensions all graphical functions with $\leq3$ internal vertices are GSVHs.
\end{ex}

\begin{ex}
Particularly useful is the case $N=2$ if the edge between $0$ and the internal vertex $x_1$ (say) has weight $-1/\lambda$.
Using (\ref{eqextdiff}) eliminates that edge and the external vertex 0 only connects to $x_2$.
We can use the method of appending an edge to reduce the graphical function, see Section \ref{sectappedge}.

The case $\nu_1=0$ in (\ref{eqextdiff}) is singular.
To handle this case we need to regularize the integrals by shifting $\lambda$ to $\lambda_\varepsilon=\lambda-\varepsilon$.
With the (mostly conjectural) theory of dimensionally regularized graphical functions \cite{numfunct,7loops} we obtain
\begin{gather}
\begin{gathered}
\begin{tikzpicture}[baseline={([yshift=-.7ex]0,0)}]
    \pgfmathsetmacro{\rad}{1.2}
    \coordinate[label=above:$z$] (vz) at (0,\rad);
    \coordinate (v1) at ({-\rad/2},0);
    \coordinate (v2) at ({\rad/2},0);
    \coordinate[label=below:$0$] (v0) at (0,-\rad);
    \coordinate (vp) at ({-\rad/4},{\rad/2});
    \draw[fill=black!20] (v1) to[bend left=50] (v2) to[bend left=50] (v1);
    \filldraw (vz) circle (1.3pt);
    \filldraw (v1) circle (1.3pt);
    \filldraw (v2) circle (1.3pt);
    \filldraw (v0) circle (1.3pt);
    \draw (vz) -- node[inner sep=1pt,pos=.4,left] {\tiny{$\frac{1-\varepsilon}{\lambda_\varepsilon}$}}(v1);
    \draw (vz) -- node[inner sep=2pt,right] {\tiny{$\frac{m}{\lambda_\varepsilon}$}}(v2);
    \draw (v0) -- node[inner sep=1pt,pos=.4,right] {\tiny{$\frac{\ell-\varepsilon}{\lambda_\varepsilon}$}}(v2);
    \draw (v0) -- node[inner sep=2pt,left] {\tiny{$-\frac{1}{\lambda_\varepsilon}$}}(v1);
\end{tikzpicture}
=
\tfrac{m}{\varepsilon}
~
\begin{tikzpicture}[baseline={([yshift=-.7ex]0,0)}]
    \pgfmathsetmacro{\rad}{1.2}
    \coordinate[label=above:$z$] (vz) at (0,\rad);
    \coordinate (v1) at ({-\rad/2},0);
    \coordinate (v2) at ({\rad/2},0);
    \coordinate[label=below:$0$] (v0) at (0,-\rad);
%    \coordinate (vp) at ({-\rad/4},{\rad/2});
%
    \draw[fill=black!20] (v1) to[bend left=50] (v2) to[bend left=50] (v1);
    \filldraw (vz) circle (1.3pt);
    \filldraw (v1) circle (1.3pt);
    \filldraw (v2) circle (1.3pt);
    \filldraw (v0) circle (1.3pt);
%    \filldraw (vp) circle (1.3pt);
%
    \draw (vz) -- node[inner sep=1pt,pos=.4,left] {\tiny{$-\frac{\varepsilon}{\lambda_\varepsilon}$}}(v1);
%    \draw (vp) -- node[inner sep=2pt,pos=.5,left] {\tiny{$1$}}(vz);
    \draw (vz) -- node[inner sep=2pt,right] {\tiny{$\frac{m+1}{\lambda_\varepsilon}$}}(v2);
    \draw (v0) -- node[inner sep=1pt,pos=.4,right] {\tiny{$\frac{\ell-1-\varepsilon}{\lambda_\varepsilon}$}}(v2);
\end{tikzpicture}
-
\tfrac{z\partial_z+\zz\partial_\zz+m-\varepsilon}{\varepsilon}
~
\begin{tikzpicture}[baseline={([yshift=-.7ex]0,0)}]
    \pgfmathsetmacro{\rad}{1.2}
    \coordinate[label=above:$z$] (vz) at (0,\rad);
    \coordinate (v1) at ({-\rad/2},0);
    \coordinate (v2) at ({\rad/2},0);
    \coordinate[label=below:$0$] (v0) at (0,-\rad);
%    \coordinate (vp) at ({-\rad/4},{\rad/2});
%
    \draw[fill=black!20] (v1) to[bend left=50] (v2) to[bend left=50] (v1);
    \filldraw (vz) circle (1.3pt);
    \filldraw (v1) circle (1.3pt);
    \filldraw (v2) circle (1.3pt);
    \filldraw (v0) circle (1.3pt);
%    \filldraw (vp) circle (1.3pt);
%
    \draw (vz) -- node[inner sep=1pt,pos=.4,left] {\tiny{$-\frac{\varepsilon}{\lambda_\varepsilon}$}}(v1);
%    \draw (vp) -- node[inner sep=2pt,pos=.5,left] {\tiny{$1$}}(vz);
    \draw (vz) -- node[inner sep=2pt,right] {\tiny{$\frac{m}{\lambda_\varepsilon}$}}(v2);
    \draw (v0) -- node[inner sep=1pt,pos=.4,right] {\tiny{$\frac{\ell-\varepsilon}{\lambda_\varepsilon}$}}(v2);
\end{tikzpicture}
+
%[-1ex]
%\hspace{-1.1cm}+~
z \zz ~
\begin{tikzpicture}[baseline={([yshift=-.7ex]0,0)}]
    \pgfmathsetmacro{\rad}{1.2}
    \coordinate[label=above:$z$] (vz) at (0,\rad);
    \coordinate (v1) at ({-\rad/2},0);
    \coordinate (v2) at ({\rad/2},0);
    \coordinate[label=below:$0$] (v0) at (0,-\rad);
%    \coordinate (vp) at ({-\rad/4},{\rad/2});
%
    \draw[fill=black!20] (v1) to[bend left=50] (v2) to[bend left=50] (v1);
    \filldraw (vz) circle (1.3pt);
    \filldraw (v1) circle (1.3pt);
    \filldraw (v2) circle (1.3pt);
    \filldraw (v0) circle (1.3pt);
    \draw (vz) -- node[inner sep=1pt,pos=.4,left] {\tiny{$\frac{1-\varepsilon}{\lambda_\varepsilon}$}}(v1);
    \draw (vz) -- node[inner sep=2pt,right] {\tiny{$\frac{m}{\lambda_\varepsilon}$}}(v2);
    \draw (v0) -- node[inner sep=1pt,pos=.4,right] {\tiny{$\frac{\ell-\varepsilon}{\lambda_\varepsilon}$}}(v2);
\end{tikzpicture}
-
\tfrac{mz\zz}{\varepsilon}
~
\begin{tikzpicture}[baseline={([yshift=-.7ex]0,0)}]
    \pgfmathsetmacro{\rad}{1.2}
    \coordinate[label=above:$z$] (vz) at (0,\rad);
    \coordinate (v1) at ({-\rad/2},0);
    \coordinate (v2) at ({\rad/2},0);
    \coordinate[label=below:$0$] (v0) at (0,-\rad);
%    \coordinate (vp) at ({-\rad/4},{\rad/2});
%
    \draw[fill=black!20] (v1) to[bend left=50] (v2) to[bend left=50] (v1);
    \filldraw (vz) circle (1.3pt);
    \filldraw (v1) circle (1.3pt);
    \filldraw (v2) circle (1.3pt);
    \filldraw (v0) circle (1.3pt);
%    \filldraw (vp) circle (1.3pt);
%
    \draw (vz) -- node[inner sep=1pt,pos=.4,left] {\tiny{$-\frac{\varepsilon}{\lambda_\varepsilon}$}}(v1);
%    \draw (vp) -- node[inner sep=2pt,pos=.5,left] {\tiny{$1$}}(vz);
    \draw (vz) -- node[inner sep=2pt,right] {\tiny{$\frac{m+1}{\lambda_\varepsilon}$}}(v2);
    \draw (v0) -- node[inner sep=1pt,pos=.4,right] {\tiny{$\frac{\ell-\varepsilon}{\lambda_\varepsilon}$}}(v2);
\end{tikzpicture}.
\label{eqextdiffN2}
\end{gathered}
\end{gather}
It is shown in \cite{7loops} that one can append edges of weights $(k-\varepsilon)/\lambda_\varepsilon$, $\lambda\geq k\in\ZZ$ in dimensionally regularized graphical functions.
Hence, all graphical functions on the right hand side emerge from the residual graph
$
\begin{tikzpicture}[baseline={([yshift=-.7ex]0,0)}]
    \pgfmathsetmacro{\rad}{.8}
    \coordinate[label=left:$z$] (v1) at ({-\rad/2},0);
    \coordinate[label=right:$0$] (v2) at ({\rad/2},0);
    \draw[fill=black!20] (v1) to[bend left=50] (v2) to[bend left=50] (v1);
    \filldraw (v1) circle (1.3pt);
    \filldraw (v2) circle (1.3pt);
\end{tikzpicture}
$
by appending edges and adding edges between external vertices (Section \ref{sectextedge}).

Note that the reduction can be iterated in the case that the edge between 0 and $x_1$ has weight $-k/\lambda$ for $k=2,3,\ldots$.
\end{ex}

%%%%%%%%%%%%%%%%%%%%%%%%%%%%%%%%%%%%%%%%%%%%%%%%%%%%%%%%%%%%%%%%%%%%%%%%%%%%%%%%%
\section{Algebraic identities}\label{sectalgint}
By linear dependence of the $D+1$ vectors $x_1$, \ldots, $x_{D+1}$ in $D$ dimensions the Gram determinant vanishes,
$$
\det(x_i\cdot x_j)_{1\leq i,j\leq D+1}=0.
$$
With $x_i\cdot x_j=(\|x_i\|^2+\|x_j\|^2-\|x_i-x_j\|^2)/2$ (see (\ref{eqscalar})) the above identity can be used to derive equations between graphical functions. The number of terms
in the resulting equations, however, seems too large for any practical use. In a future extension of graphical functions which incorporates numerators $x_i\cdot x_j$,
algebraic identities may play a more prominent role.

%%%%%%%%%%%%%%%%%%%%%%%%%%%%%%%%%%%%%%%%%%%%%%%%%%%%%%%%%%%%%%%%%%%%%%%%%%%%%%%%%
\section{Parametric integration}\label{sectparint}
If a graphical function with not too many edges is `linearly reducible', see \cite{BrH1,BrH2}, parametric integration can be used to perform one integral after the other.
The method has been implemented in Maple by E. Panzer ({\tt HyperInt} \cite{Panzer:HyperInt}). It works best in four dimensions or with numerator edges of weight $-1/\lambda$.
Sometimes parametric integration allows one to calculate a graphical function which is not amenable to any other method.
Still, the method has a brute-force character which makes it time-consuming. There is room for speed improvements in the implementation by using faster computer algebra
systems or by parallelization. But even the way it is now, it is a very valuable last resort if all other methods fail.

%%%%%%%%%%%%%%%%%%%%%%%%%%%%%%%%%%%%%%%%%%%%%%%%%%%%%%%%%%%%%%%%%%%%%%%%%%%%%%%%%
\section{Fishnets}\label{sectfish}
The fishnet graph $\Gg_{m,n}$ is a square lattice with external edges attached \cite{fishnet},
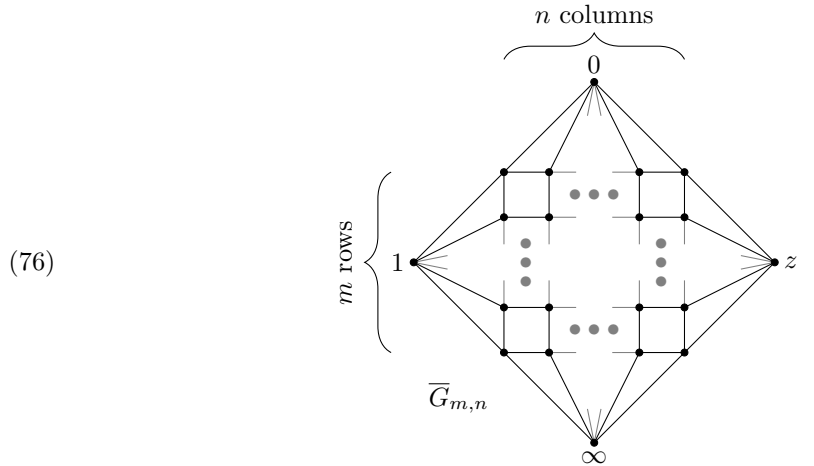
\begin{align}
\def\scale{.6}
\begin{tikzpicture}[baseline={([yshift=-.7ex]0,0)}]
    \coordinate[label=above:$0$] (v0) at  (0, {4*\scale});
    \coordinate[label=left:$1$] (v1) at  ({-4*\scale}, 0);
    \coordinate[label=right:$z$] (vz) at ({4*\scale},0);
    \coordinate[label=below:$\infty$] (voo)  at (0, {-4*\scale});
    \coordinate (vlt0)  at ({-2*\scale}, {2*\scale});
    \coordinate (vlt1)  at ({-1*\scale}, {2*\scale});
    \coordinate (vlt2)  at ({-2*\scale}, {1*\scale});
    \coordinate (vlt3)  at ({-1*\scale}, {1*\scale});
    \coordinate (vrt0)  at ({2*\scale}, {2*\scale});
    \coordinate (vrt1)  at ({1*\scale}, {2*\scale});
    \coordinate (vrt2)  at ({2*\scale}, {1*\scale});
    \coordinate (vrt3)  at ({1*\scale}, {1*\scale});
    \coordinate (vlb0)  at ({-2*\scale}, {-2*\scale});
    \coordinate (vlb1)  at ({-1*\scale}, {-2*\scale});
    \coordinate (vlb2)  at ({-2*\scale}, {-1*\scale});
    \coordinate (vlb3)  at ({-1*\scale}, {-1*\scale});
    \coordinate (vrb0)  at ({2*\scale}, {-2*\scale});
    \coordinate (vrb1)  at ({1*\scale}, {-2*\scale});
    \coordinate (vrb2)  at ({2*\scale}, {-1*\scale});
    \coordinate (vrb3)  at ({1*\scale}, {-1*\scale});
    \node (G) at ({-3*\scale}, {-3*\scale}) {$\Gg_{m,n}$};
    \draw (v0) -- (vlt0);
    \draw (v0) -- (vlt1);
    \draw (v0) -- (vrt0);
    \draw (v0) -- (vrt1);
    \draw (v1) -- (vlt0);
    \draw (v1) -- (vlt2);
    \draw (v1) -- (vlb0);
    \draw (v1) -- (vlb2);
    \draw (voo) -- (vlb0);
    \draw (voo) -- (vlb1);
    \draw (voo) -- (vrb0);
    \draw (voo) -- (vrb1);
    \draw (vz) -- (vrt0);
    \draw (vz) -- (vrt2);
    \draw (vz) -- (vrb0);
    \draw (vz) -- (vrb2);
    \draw (vlt0) -- (vlt1) -- (vlt3) -- (vlt2) -- (vlt0);
    \draw (vrt0) -- (vrt1) -- (vrt3) -- (vrt2) -- (vrt0);
    \draw (vlb0) -- (vlb1) -- (vlb3) -- (vlb2) -- (vlb0);
    \draw (vrb0) -- (vrb1) -- (vrb3) -- (vrb2) -- (vrb0);
    \draw[black!50] (vlt2) -- ($(vlt2)!.3!(vlb2)$);
    \draw[black!50] (vlb2) -- ($(vlb2)!.3!(vlt2)$);
    \draw[black!50] (vlt3) -- ($(vlt3)!.3!(vlb3)$);
    \draw[black!50] (vlb3) -- ($(vlb3)!.3!(vlt3)$);
    \draw[black!50] (vrt2) -- ($(vrt2)!.3!(vrb2)$);
    \draw[black!50] (vrb2) -- ($(vrb2)!.3!(vrt2)$);
    \draw[black!50] (vrt3) -- ($(vrt3)!.3!(vrb3)$);
    \draw[black!50] (vrb3) -- ($(vrb3)!.3!(vrt3)$);
    \draw[black!50] (vlt1) -- ($(vlt1)!.3!(vrt1)$);
    \draw[black!50] (vrt1) -- ($(vrt1)!.3!(vlt1)$);
    \draw[black!50] (vlt3) -- ($(vlt3)!.3!(vrt3)$);
    \draw[black!50] (vrt3) -- ($(vrt3)!.3!(vlt3)$);
    \draw[black!50] (vlb1) -- ($(vlb1)!.3!(vrb1)$);
    \draw[black!50] (vrb1) -- ($(vrb1)!.3!(vlb1)$);
    \draw[black!50] (vlb3) -- ($(vlb3)!.3!(vrb3)$);
    \draw[black!50] (vrb3) -- ($(vrb3)!.3!(vlb3)$);
    \draw[black!50] (v0) -- ($(v0)!.15!(vlb3)$);
    \draw[black!50] (v0) -- ($(v0)!.15!(vrb3)$);
    \draw[black!50] (v1) -- ($(v1)!.15!(vrt3)$);
    \draw[black!50] (v1) -- ($(v1)!.15!(vrb3)$);
    \draw[black!50] (vz) -- ($(vz)!.15!(vlt3)$);
    \draw[black!50] (vz) -- ($(vz)!.15!(vlb3)$);
    \draw[black!50] (voo) -- ($(voo)!.15!(vlt3)$);
    \draw[black!50] (voo) -- ($(voo)!.15!(vrt3)$);
    \node[black!50] (BT) at (0,{\scale*1.5}) {$\bullet \bullet \bullet$};
    \node[black!50] (BL) at ({-\scale*1.5},0) {\rotatebox{90}{$\bullet \bullet \bullet$}};
    \node[black!50] (BB) at (0,{-\scale*1.5}) {$\bullet \bullet \bullet$};
    \node[black!50] (BR) at ({\scale*1.5},0) {\rotatebox{90}{$\bullet \bullet \bullet$}};
    \filldraw (voo) circle (1.3pt);
    \filldraw (v1) circle (1.3pt);
    \filldraw (v0) circle (1.3pt);
    \filldraw (vz) circle (1.3pt);
    \filldraw (vlt0) circle (1.3pt);
    \filldraw (vlt1) circle (1.3pt);
    \filldraw (vlt2) circle (1.3pt);
    \filldraw (vlt3) circle (1.3pt);
    \filldraw (vrt0) circle (1.3pt);
    \filldraw (vrt1) circle (1.3pt);
    \filldraw (vrt2) circle (1.3pt);
    \filldraw (vrt3) circle (1.3pt);
    \filldraw (vlb0) circle (1.3pt);
    \filldraw (vlb1) circle (1.3pt);
    \filldraw (vlb2) circle (1.3pt);
    \filldraw (vlb3) circle (1.3pt);
    \filldraw (vrb0) circle (1.3pt);
    \filldraw (vrb1) circle (1.3pt);
    \filldraw (vrb2) circle (1.3pt);
    \filldraw (vrb3) circle (1.3pt);
\draw [decorate,decoration={brace,amplitude=10pt}] ({-2*\scale},{4.5*\scale}) -- ({2*\scale},{4.5*\scale}) node [black,midway,yshift=0.6cm] {$n$ \text{columns}};
\draw [decorate,decoration={brace,amplitude=10pt}] ({-4.5*\scale},{-2*\scale}) -- ({-4.5*\scale},{2*\scale}) node [black,midway,xshift=-0.6cm] {\rotatebox{90}{$m$ \text{rows}}};
\end{tikzpicture}~~.
\end{align}
For all $m,n\geq1$ the fishnets are convergent graphical functions in four dimensions. By permutation symmetry (Theorem \ref{thmperm}) it is sufficient to study the case $m\leq n$.
The case $m=1$ is constructible, see Section \ref{sectconstructible}. It was first calculated by N. Ussyukina and A. Davydychev in 1993 (the case $m=n=1$ is $f_{\overline{\smallclaw}}(z)$ in
Figure \ref{fig:3star}, Example \ref{ex3star}):

\begin{prop}[N. Ussyukina, A. Davydychev \cite{Ladder}]
In four dimensions we have
\begin{equation}
f_{\Gg_{1,n}}(z)=\sum_{k=0}^n\binom{n+k}{k}\frac{(-\log z\zz)^{n-k}}{(n-k)!}\frac{\Li_{n+k}(z)-\Li_{n+k}(\zz)}{z-\zz},
\end{equation}
where $\Li_n(z)=\sum_{k=1}^\infty z^k/k^n$ is the polylogarithm.
\end{prop}
\begin{proof}
A proof using graphical functions is in Example 3.31 of \cite{gf}.
\end{proof}

For $n\geq m\geq2$ the only known relation for fishnets is planar duality, Section \ref{sectdual}, which is insufficient to solve the fishnets.
(The graph $\Gg_{2,2}$ can be computed using parametric integration in Section \ref{sectparint}.) Nevertheless, there exists a conjecture for all $f_{\Gg_{m,n}}(z)$, $n\geq m\geq2$ in terms
of Hankel determinants. The result was found by B. Basso and L. J. Dixon in 2017 in the context of super Yang-Mills theory.

\begin{thm}[B. Basso, L. J. Dixon \cite{fishnet} with a proof in \cite{fishnet2}]\label{confish}
For $n\geq m\geq2$ define the Hankel matrix $H=(H_{i,j})_{i,j=1,\ldots,m}$ by
\begin{equation}
H_{i,j}(z)=(n-m+i+j-2)!(n-m+i+j-1)!f_{\Gg_{1,n-m+i+j-1}}(z).
\end{equation}
Then
\begin{equation}\label{eqfishnet}
f_{\Gg_{m,n}}(z)=\frac{\det H(z)}{\prod_{k=n-m}^{n+m-1}k!}.
\end{equation}
\end{thm}

\begin{ex}
For $m=2$ we obtain
$$
f_{\Gg_{2,n}}(z)=f_{\Gg_{1,n-1}}(z)f_{\Gg_{1,n+1}}(z)-\frac{n-1}{n+1}f_{\Gg_{1,n}}(z)^2.
$$
\end{ex}
Note that only very few graphical functions can be expressed in terms of simple polylogarithms. This simplicity seems in conflict with the complexity of the fishnet graphs.
The proof of the Fishnet Theorem in \cite{fishnet2} uses a detour over relations that are derived within the context of super Yang-Mills theory. It would be desirable to have a proof
which is closer to the technologies used in this article.

It is unclear if the Fishnet Theorem generalizes in some way. Experiments show that there exist further unknown relations between graphical functions.

%%%%%%%%%%%%%%%%%%%%%%%%%%%%%%%%%%%%%%%%%%%%%%%%%%%%%%%%%%%%%%%%%%%%%%%%%%%%%%%%%
\section{Radial graphical functions}\label{sectrad}
With this section we begin the second part of the article which contains detailed proofs.

In the next sections we revisit a classical method in perturbative QFT, the Gegenbauer expansion \cite{geg}. In our framework the Gegenbauer method is a
decomposition of graphical functions into radial and angular parts, see Theorem \ref{thmgegex0} and Conjecture \ref{congegex}.

In any (odd or even) dimension $D=2\lambda+2\geq3$ we define the spherical coordinates of a vector $x\in\RR^D$ as
\begin{equation}\label{angularcoords}
x=\left(\begin{array}{c}X\cos(\phi^x_1)\\
X\sin(\phi^x_1)\cos(\phi^x_2)\\
\vdots\\
X\sin(\phi^x_1)\cdots\sin(\phi^x_{D-2})\cos(\phi^x_{D-1})\\
X\sin(\phi^x_1)\cdots\sin(\phi^x_{D-2})\sin(\phi^x_{D-1})\end{array}\right),
\end{equation}
where $\phi^x_1,\ldots,\phi^x_{D-2}\in[0,\pi)$ and $\phi^x_{D-1}\in[0,2\pi)$. We use capitals for the moduli, i.e.\ $X=\|x\|\in[0,\infty)$.
In this section we only work with the moduli.

Consider a graph $G$ that has a pair of weights $(\nu_e,n_e)$ with $\nu_e\in\RR$, $n_e\in\ZZ_{\geq0}$ at every edge $e\in\sE_G$. The additional weights $n_e$ refer to
the labels used in the expansion of the propagators into Gegenbauer polynomials, see e.g.\ (\ref{propagator}).

The vertices of $G$ split into three external vertices $0,1,Z\geq0$ and internal vertices $X_i\geq0$, $i=1,\ldots,\VGint$ (we identify labels with numbers or variables).
Every edge $e$ that is adjacent to the external label 0 has $n_e=0$. To every edge $e=XY$ we associate the propagator
\begin{equation}\label{propdef}
p^R_{XY}=\frac{1}{(XY)^{\lambda\nu_{XY}}}\genfrac(){}{}{X}{Y}^{n_{XY}+\lambda\nu_{XY}}_<,
\end{equation}
where $(x)_<=x$ if $x<1$ and $x^{-1}$ otherwise. It is convenient to initially consider the dimension $D=2\lambda+2$ as a complex parameter. In the end we will set $D$ to the wanted integer value.
For $X=0$ the propagator is $p^R_{0Y}=Y^{-2\lambda\nu_{0Y}}$ (the weighted position space Feynman propagator from 0 to $Y$).

We define the radial graphical function of $G$ as
\begin{equation}\label{fGRdef}
f_G^R(Z)=\Big(\prod_{i=1}^\VGint\int_0^\infty X_i^{D-1}\dd X_i\Big)\prod_{e\in\sE_G}p^R_e,
\end{equation}
whenever the integral converges.

\begin{ex}\label{exWSnR}
Consider the wheel with $n$ spokes $W\!S_{n,D}$ in Figure \ref{fig:wheels}. The hub has label 0 and the first vertex on the rim has label 1.
We assign the label $Z$ to the last label on the rim (label $n$ in Figure \ref{fig:wheels}). The spokes have propagators $X_i^{-2}$ where $X_1=1$ and $X_n=Z$.
The edge $i,i+1$ on the rim has the propagator $(X_iX_{i+1})^{-\lambda}(X_i/X_{i+1})^{n_{i,i+1}+\lambda}_<$ where we set $X_{n+1}=X_1=1$. We get
\begin{equation}\label{WSR}
f_{W\!S_{n,D}}^R(Z)=\frac{(Z)^{n_{1Z}+\lambda}_<}{Z^{2\lambda+2}}\Big(\prod_{i=2}^{n-1}\int_0^\infty\frac{\dd X_i}{X_i}\Big)\prod_{i=1}^{n-1}\genfrac(){}{}{X_i}{X_{i+1}}^{n_{i,i+1}+\lambda}_<.
\end{equation}
\end{ex}

\begin{lem}
We write $f_{G_{Z_0,Z_1,Z_2}}^R$ for the radial graphical function $f_G^R(Z)$ of the graph $G$ with external labels $Z_0,Z_1,Z_2$. Then
\begin{equation}\label{Zrez}
f_{G_{0,1,Z^{-1}}}^R=Z^{2\lambda N_G}f_{G_{0,Z,1}}^R,
\end{equation}
where the weight $N_G$ of the graphical function $G$ is defined in (\ref{eqNg}).
\end{lem}
\begin{proof}
We scale all integration variables in $f_{G_{0,Z,1}}^R$ by $Z$ to obtain the result from the definition of the radial graphical function.
\end{proof}

Now, we assume $Z<1$. With the previous lemma we can translate the results to the case $Z>1$.
Let $S_{Z1}$ be the set of orderings of $1,Z,X_i$, $i=1,\ldots,\VGint$ which preserve $Z<1$,
i.e.\ $S_{Z1}=\{\sigma:\{1,2,\ldots,\VGint+2\}\to\{1,Z,X_i\},$ with $\sigma^{-1}(Z)<\sigma^{-1}(1)\}$.
Note that $S_{Z1}$ depends on the internal vertices $X_1,\ldots,X_\VGint$.
\begin{ex}
For $\VGint=1$ the set $S_{Z1}$ consists of the three elements that map $1,2,3$ to $X_1,Z,1$, to $Z,X_1,1$, or to $Z,1,X_1$.
\end{ex}
In general, $S_{Z1}$ has $(\VGint+2)!/2$ elements.
The domain of integration $\{X_i>0\}$ partitions into the $|S_{Z1}|$ sectors $0<\sigma(1)<\ldots<\sigma(\VGint+2)$. Accordingly, the radial graphical function decomposes as
$$
f_G^R(Z)=\sum_{\sigma\in S_{Z1}}f_\sigma(Z),\quad f_\sigma(Z)=\int_{0<\sigma(1)<\ldots<\sigma(\VGint+2)}\Big(\prod_{i=1}^\VGint X_i^{D-1}\dd X_i\Big)\prod_{e\in\sE_G}p^R_e.
$$

We fix $\sigma\in S_{Z1}$ and use (without restriction) labels with $\{0<\sigma(1)<\ldots<\sigma(\VGint+2)\}=\{0<X_1<\ldots<X_s<Z<X_{s+1}<\ldots<X_t<1<X_{t+1}<\ldots<X_\VGint\}=\Sigma$
for $0\leq s<t\leq\VGint$.
The propagator is additive in the weights so that we can replace multiple edges by single edges with added weights. For vanishing weights $\nu_e=n_e=0$ the propagator $p^R_e$ is one.
To simplify the notation we may hence assume that $G$ is the complete graph (with zero weights at unwanted edges).
From (\ref{propdef}) we get
$$
f_\sigma(Z)=\int_\Sigma\prod_{X\in\{Z,X_1,\ldots,X_\VGint\}}\omega_X,
$$
with
\begin{align*}
\omega_X&=X^{D-1-(\sum_{Y<X}n_{XY}+2\lambda\nu_{XY})+\sum_{Y>X}n_{XY}}\dd X,\quad\text{for }X=X_1,\ldots,X_\VGint,\\
\omega_Z&=Z^{-(\sum_{Y<Z}n_{YZ}+2\lambda\nu_{YZ})+\sum_{Y>Z}n_{YZ}}.
\end{align*}

The path of the iterated integral over $\Sigma$ splits at $Z$ and 1. Therefore $f_\sigma(Z)$ factors,
$$
f_\sigma(Z)=f^{0Z}_\sigma(Z)\omega_Zf^{Z1}_\sigma(Z)f^{1\infty}_\sigma,
$$
where the $f^{XY}_\sigma$ are iterated integrals from $X$ to $Y$. Note that $f^{1\infty}_\sigma$ does not depend on $Z$.

For $f^{0Z}_\sigma$ we re-scale all variables $X_1,\ldots,X_s$ by $Z$ and obtain
$$
f^{0Z}_\sigma(Z)=f^{0Z}_\sigma(1)Z^{\alpha_s},
$$
where we used the notation
\begin{equation}\label{alphak}
\alpha_k=kD-\Big(\sum_{X<Y\leq X_k}n_{XY}+2\lambda\nu_{XY}\Big)+\sum_{Y>X\leq X_k}n_{XY}=kD-\sum_{X<Y\leq X_k}2\lambda\nu_{XY}+\sum_{X\leq X_k<Y}n_{XY}
\end{equation}
for $k=1,\ldots,\VGint$. In the above formula the sums are over $X$ and $Y$. The second equation holds because the contributions of the weights $n_{XY}$ cancel if both $X,Y\leq X_k$.

The factor $f^{0Z}_\sigma(1)$ can be calculated by integrating over $X_1,X_2,\ldots,X_s$. Each integration provides the reciprocal of a linear form in $D$ and in the weights.

More complicated is the term $f^{Z1}_\sigma(Z)$: For generic $D$ we are allowed to split the iterated integral from $Z$ to 1 at 0 (i.e.\ we integrate from $Z$ to 0 and then from 0 to 1).
We reverse the orientation of the path from $Z$ to 0 and obtain by path concatenation and path reversal of iterated integrals (see e.g.\ \cite{gf})
$$
f^{Z1}_\sigma(Z)=\sum_{k=s}^t(-1)^{k-s}f^{0Z}_k(Z)f^{01}_k,
$$
with
$$
f^{0Z}_k(Z)=\int_{0<X_k<X_{k-1}<\ldots<X_{s+1}<Z}\prod_{i=s+1}^k\omega_{X_i},\quad f^{01}_k=\int_{0<X_{k+1}<\ldots<X_t<1}\prod_{i=k+1}^t\omega_{X_i},
$$
where empty integrals are 1. We have
$$
f^{0Z}_k(Z)=f^{0Z}_k(1)Z^{\beta_k},\quad\text{with }\beta_k=(k-s)D-\Big(\sum_{X<Y;X_s<Y\leq X_k}n_{XY}+2\lambda\nu_{XY}\Big)+\sum_{Y>X;X_s<X\leq X_k}n_{XY}.
$$
The term $f^{0Z}_k(1)f^{01}_k$ gives $t-s$ reciprocals of linear forms. With (\ref{alphak}) we obtain (including $\omega_Z$)
$$
f_\sigma(Z)=\sum_{k=s}^t\frac{Z^{\alpha_k}}{P_k(D)},
$$
where $P_k$ is a polynomial in $D$ of degree $\VGint$ which factors into linear forms.

Consider the induced subgraph $G_k=G[\{0,Z,X_1,\ldots,X_k\}]$. The first sum in (\ref{alphak}) is over edges of $G_k$ whereas the second sum is over edges that cut
$G$ into $G_k$ and the graph $G[\{1,X_{k+1},\ldots,X_\VGint\}]$ which is induced by $\sV_G\backslash\sV_{G_k}$.
Summing over $\sigma\in S_{Z1}$ the graph $G_k$ can contain any subset of $\sVGint$. Using (\ref{eqNg}) we obtain the generic result
\begin{equation}\label{fRres}
f_G^R(Z)=\sum_{\sV\subseteq\sVGint}c_\sV Z^{-2\lambda N_{G[\sV\cup\{0,z\}]}+\sum_{\!\!\genfrac{}{}{0pt}{}{u\in\sV\cup\{0,z\}}{v\in\sVGint\cup\{1\}\backslash\sV}}n_{uv}},
\end{equation}
where $c_\sV$ is a finite sum over reciprocals of products of $\VGint$ linear forms in $D$ and the weights. It is possible to give a formula for $c_\sV$. Here, we only need its general shape.

\begin{thm}\label{radialthm}
Let $G$ be the graph of a radial graphical function in $D=2\lambda+2$ dimensions with external vertices $0,1,Z<1$. Then
\begin{equation}\label{Req}
f_G^R(Z)=\sum_{\ell=0}^{\VGint}\sum_{\sV\subseteq\sVGint}c_{\ell,\sV}(\log Z)^\ell Z^{-2\lambda N_{G[\sV\cup\{0,z\}]}+\sum_{e\in\sC^G_{\sV\cup\{0,z\}}}n_e},
\end{equation}
where $c_{\ell,\sV}\in\RR$ are constants (which depend on $\lambda,\nu_e,n_e$) and $\sC^G_{\sV}=\sE_G\backslash(\sE_{G[\sV]}\cup\sE_{G[\sV_G\backslash\sV]})$ is the set of edges which connect
the induced subgraph $G[\sV]$ to the induced subgraph of the complement $\sV_G\backslash\sV$.
\end{thm}
\begin{proof}
We use the generic result (\ref{fRres}). If we approach $D$ from generic values, we may encounter vanishing denominators in $c_\sV$.
We use l'Hospital to calculate the limit. The coefficient becomes a polynomial in $\log Z$ of degree $\leq\VGint$ (the maximum number of factors which can vanish in the denominators).
\end{proof}

\begin{ex}\label{exWSnR1}
We continue Example \ref{exWSnR}. The integral in (\ref{WSR}) is a multiple convolution which can be solved by Fourier transformation (this trick was used in Equation (81) of
\cite{Snatren} to calculate the period of $W\!S_{3,4}$). With the residue theorem we get for any positive number $x$,
$$
\frac{\mu}{\pi}\int_{-\infty}^\infty \frac{x^{\ii P}\dd P}{P^2+\mu^2}=(x)_<^\mu.
$$
Substitution of the right hand side into (\ref{WSR}) gives
$$
f_{W\!S_{n,D}}^R(Z)=\frac{2^{n-1}(Z)_<^{n_{1Z}+\lambda}}{2\pi Z^{2\lambda+2}}\Big(\prod_{i=1}^{n-1}\int_{-\infty}^\infty\frac{Z^{-iP_{n-1}}(n_{i,i+1}+\lambda)\dd P_i}
{P_i^2+(n_{i,i+1}+\lambda)^2}\Big)\cdot\prod_{i=2}^{n-1}\int_0^\infty\frac{X_i^{\ii(P_i-P_{i-1})}\dd X_i}{2\pi X_i}.
$$
The integrals over the $X_i$ identify all $P_i$s (substitute $X_i=\exp{\xi_i}$) and we obtain
\begin{equation}\label{WSReq0}
f_{W\!S_{n,D}}^R(Z)=\frac{2^{n-1}(Z)_<^{n_{1Z}+\lambda}}{Z^{2\lambda+2}}\int_{-\infty}^\infty\frac{Z^{-\ii P}\dd P}{2\pi}\prod_{i=1}^{n-1}\frac{n_{i,i+1}+\lambda}{P^2+(n_{i,i+1}+\lambda)^2}.
\end{equation}
Using the residue theorem again we close the contour in the upper half-plane and obtain in the special case of mutually distinct weights $n_{i,i+1}$ and $Z<1$,
\begin{align}\label{WSReq}
f_{W\!S_{n,D}}^R(Z)&=2^{n-1}Z^{n_{1Z}-\lambda-2}\ii\sum_{j=1}^{n-1}\mathrm{res}_{P=\ii(n_{j,j+1}+\lambda)}Z^{-\ii P}\prod_{i=1}^{n-1}
 \frac{n_{i,i+1}+\lambda}{P^2+(n_{i,i+1}+\lambda)^2}\nonumber\\
&=2^{n-2}Z^{n_{1Z}-2}\sum_{j=1}^{n-1}Z^{n_{j,j+1}}\prod_{1=i\neq j}^{n-1}\frac{n_{i,i+1}+\lambda}{(n_{i,i+1}+\lambda)^2-(n_{j,j+1}+\lambda)^2}.
\end{align}
To compare the result with Theorem \ref{radialthm} we consider the vertex set $\sV=\{0,Z,X_{j+1},\ldots,X_{n-1}\}$. The induced subgraph has $n-j+1$ spokes of weight $1/\lambda$,
$n-j$ edges on the rim with weight 1 and $n-j$ internal vertices. Hence $N_{G[\sV]}=1/\lambda$. The spokes in the cut $\sC_{\sV}$ have $n_e=0$. Only the two edges $1Z$ and $j,j+1$
contribute to the sum in the exponent, yielding the term $c_{0,\sV}Z^{-2+n_{1Z}+n_{j,j+1}}$ which we find in (\ref{WSReq}). Note that most constants $c_{\ell,\sV}$ in (\ref{Req}) are zero.
\end{ex}

%%%%%%%%%%%%%%%%%%%%%%%%%%%%%%%%%%%%%%%%%%%%%%%%%%%%%%%%%%%%%%%%%%%%%%%%%%%%%%%%%
\section{Angular graphical functions}\label{sectang}
In this section we focus on the angular part ($X=1$) of the spherical coordinates defined in (\ref{angularcoords}).

We consider a graph $G$ with a pair of edge-weights $(\nu_e,n_e)$ where $\nu_e\in\RR$, $n_e\in\ZZ_{\geq0}$ for every edge $e\in\sE_G$.
In general, the Feynman propagator between two vectors in $\RR^D$ depends on their relative angle.
If, however, an edge $0x$ of $G$ is adjacent to 0, the propagator $\|x\|^{-2\lambda\nu_{0x}}$ has no angular dependence. So, in the angular graph $G$ we will have no external vertex 0.
In this section we generalize to any number of external vertices $z_1,\ldots,z_\VGext$.
Because only angles matter we assume that internal and external vertices are unit vectors. In the context of standard graphical functions we will have $z_1=1$ and $z_2=z/|z|\in\CC$.

We normalize the integral over the unit sphere $S_{D-1}$ in (odd or even) $D$ dimensions, i.e.\ $\int_{S_{D-1}}\Omega_{D-1}^x=1$,
\begin{equation}\label{SD1}
\int_{S_{D-1}}\Omega_{D-1}^x=\frac{\Gamma(\lambda+1)}{2\pi^{\lambda+1}}\int_0^\pi\dd\phi^x_1\cdots\int_0^\pi\dd\phi^x_{D-2}\int_0^{2\pi}\dd\phi^x_{D-1}\prod_{i=1}^{D-1}\sin^{D-i-1}\phi^x_i.
\end{equation}

For $\alpha\in\RR$ the Gegenbauer polynomials $C^\alpha_k(x)$ are defined by the generating function (see e.g.\ \cite[Chapter~22]{AS})
\begin{equation}\label{Cdef}
\frac{1}{(1-2tx+t^2)^\alpha}=\sum_{k=0}^\infty C^\alpha_k(x)t^k,\quad\text{for }t<1.
\end{equation}
Every Gegenbauer polynomial $C^\alpha_k(x)$ is a polynomial of degree $k$ in $x$. If $k$ is even, then $C^\alpha_k(x)$ is symmetric, otherwise it is antisymmetric.
In the coordinate $x=(\ee^{\ii\phi}+\ee^{-\ii\phi})/2$ the generating function factors and one obtains
\begin{equation}\label{Cdef1}
C^\alpha_k(\cos\phi)=\sum_{\ell=0}^k\genfrac(){0pt}{}{\alpha+k-\ell-1}{\alpha-1}\genfrac(){0pt}{}{\alpha+\ell-1}{\alpha-1}\ee^{(k-2\ell)\ii\phi}.
\end{equation}
This identity implies
\begin{equation}\label{C1}
|C^\alpha_k(\cos\phi)|\leq C^\alpha_k(1)=\genfrac(){0pt}{}{k+2\alpha-1}{k}.
\end{equation}
Gegenbauer polynomials are orthogonal with respect to the measure $\Omega_{D-1}$ (see e.g.\ \cite{DStanton}),
\begin{equation}\label{Gegorth}
\int_{S_{D-1}}\dd\Omega^y_{D-1}C^\lambda_k(\cos\phi_{xy})C^\lambda_\ell(\cos\phi_{yz})=\frac{\lambda\delta_{k,\ell}}{\lambda+k}C^\lambda_k(\cos\phi_{xz}),
\end{equation}
where $x,y,z\in\RR^D$ and $\phi_{xy}$ is the angle between $x$ and $y$.

\begin{ex}
We have
\begin{equation}\label{Cex}
C^\alpha_0(x)=1,\quad C^\alpha_1(x)=2\alpha x,\quad C^\alpha_2(x)=2\alpha(\alpha+1)x^2-\alpha,\quad C^\alpha_3(x)=\frac{4(\alpha+2)!}{3(\alpha-1)!}x^3-\frac{2(\alpha+1)!}{3(\alpha-1)!}x.
\end{equation}
\end{ex}

To every edge $e\in\sE_G$ of the graph $G$ we associate the angle $\phi_e$ between its vertices. The angular propagator is (note that $p^\angle_e$ is not additive in the weights)
\begin{equation}\label{angprop}
p^\angle_e=C^{\lambda\nu_e}_{n_e}(\cos\phi_e).
\end{equation}
The angular graphical function of $G$ is defined as
\begin{equation}\label{fGangledef}
f_G^\angle(z_1,\ldots,z_\VGext)=\Big(\prod_{i=1}^\VGint\int_{S_{D-1}}\Omega^{x_i}_{D-1}\Big)\prod_{e\in\sE_G}p^\angle_e.
\end{equation}

\begin{ex}\label{exWSnA}
Let $C_n$ be the cycle with $n$ vertices $1,\ldots,n$ and edge-weights $(1,n_{i,i+1})$ where we identify the labels $n+1$ and 1.
We set $1=z_1$ and $n=z_2$ as the two external labels of $C_n$. All other labels are internal. By orthogonality we get
\begin{align}\label{WSA}
f_{C_n}^\angle(z_1,z_2)&=\Big(\prod_{i=2}^{n-1}\int_{S_{D-1}}\Omega^{x_i}_{D-1}\Big)\prod_{i=1}^nC^\lambda_{n_{i,i+1}}(\cos\phi_{i,i+1})\nonumber\\
&=\Big(\frac{\lambda}{\lambda+n_{z_1,2}}\Big)^{n-2}\Big(\prod_{i=2}^{n-1}\delta_{n_{i-1,i},n_{i,i+1}}\Big)C^\lambda_{n_{z_1,2}}(\cos\phi_{z_1,z_2})C^\lambda_{n_{z_1,z_2}}(\cos\phi_{z_1,z_2}).
\end{align}
\end{ex}

To express the connection between degree and symmetry of Gegenbauer polynomials we introduce the relation $\leq_2$ according to
\begin{equation}\label{leq2def}
a\leq_2b\quad\Leftrightarrow\quad(a\leq b\quad\text{and}\quad a\equiv b\mod 2)\quad\text{for }a,b\in\ZZ.
\end{equation}
With this notation we get
\begin{equation}\label{pexp}
p^\angle_e=\sum_{k\leq_2 n_e}c_k\cos^k\phi_e\quad\text{with }c_k\in\RR.
\end{equation}
Using (\ref{Cex}) we interpret $\cos\phi_e$ as angular propagator of an edge with weights $(\nu_e,n_e)=(1/2\lambda,1)$. The propagator $p^\angle_e$ becomes an $\RR$-linear-combination of
propagators for multiple edges of this type. Therefore it suffices to consider angular graphs with multiple edges where each edge $e$ has weights $\nu_e=1/2\lambda$ and $n_e=1$.

The general result for angular graphical functions can be derived from the $N$-star $\star_{\mathbf k}=\star_{k_1,\ldots,k_N}$ where a single internal vertex $x_1=x$ connects
to the external vertex $z_i$, $i=1,\ldots,N$ with $k_i$ parallel edges of weight $(1/2\lambda,1)$. We calculate the generating function of $f^\angle_{\star_{\mathbf k}}$
in the variables $\mathbf t=t_1,\ldots,t_N$,
$$
f_\star(\mathbf t)=\sum_{\mathbf k}\frac{(\sum_{i=1}^Nk_i)!}{\prod_{i=1}^Nk_i!}f_{\star_{\mathbf{k}}}^\angle(z_1,\ldots,z_N)\prod_{i=1}^Nt_i^{k_i}
=\int_{S_{D-1}}\frac{\Omega^x_{D-1}}{1-\sum_{i=1}^N(x\cdot z_i)t_i}.
$$
The denominator of the integrand is
$$
1-x\cdot z_{\mathbf t}=1-\cos(\phi)Z_{\mathbf t},\quad\text{with}\quad z_{\mathbf t}=\sum_{i=1}^Nz_it_i\quad\text{and}\quad\phi=\phi_{xz_t},\;Z_{\mathbf t}=|z_{\mathbf t}|.
$$
We orient the coordinate system such that $z_{\mathbf t}$ points into the 1-direction and obtain
$$
f_\star(\mathbf t)=\frac{\Gamma(\lambda+1)}{\sqrt{\pi}\Gamma(\lambda+\frac{1}{2})}\int_0^\pi\frac{\sin^{2\lambda}\phi\,\dd\phi}{1-\cos\phi Z_{\mathbf t}},
$$
where we performed the integral over the $S_{D-2}$ sphere.
For sufficiently small $t_i$ the integral convergences absolutely and the integrand can be expanded in $Z_{\mathbf t}$.
Consider the integral transformation $\phi\mapsto\pi-\phi$. Because $\cos\phi$ changes sign whereas $\sin\phi$ does not, only even powers in the expansion give non-zero results. We obtain
\begin{equation}\label{genstar}
f_\star(\mathbf t)=\frac{\Gamma(\lambda+1)}{\sqrt{\pi}\Gamma(\lambda+\frac{1}{2})}\sum_{m=0}^\infty\int_0^\pi\sin^{2\lambda}\phi\cos^{2m}\phi\,\dd\phi Z_{\mathbf t}^{2m}
=\sum_{m=0}^\infty\frac{\Gamma(\lambda+1)\Gamma(m+\frac{1}{2})}{\Gamma(\lambda+m+1)\Gamma(\frac{1}{2})}\Big(\sum_{i,j=1}^Nz_i\cdot z_j t_it_j\Big)^m.
\end{equation}
We have $z_i\cdot z_i=1$ whereas $z_i\cdot z_j=\cos\phi_e$ for the edge $e=z_iz_j$ between the external vertices $z_i$ and $z_j$.
Note that the coefficients of the generating series are rational functions in $\lambda$.

\begin{ex}\label{angularex1}
For the three-star $\star_{k_1,k_2,k_3}$ in $D=2\lambda+2\geq3$ dimensions we get
\begin{align*}
f^\angle_{\star_{0,0,0}}&=1,\quad f^\angle_{\star_{2,0,0}}=\frac{1}{2(\lambda+1)},\quad f^\angle_{\star_{4,0,0}}=\frac{3}{4(\lambda+1)(\lambda+2)},\quad
 f^\angle_{\star_{1,1,0}}=\frac{z_1\cdot z_2}{2(\lambda+1)},\\
f^\angle_{\star_{3,1,0}}&=\frac{3 z_1\cdot z_2}{4(\lambda+1)(\lambda+2)},\quad f^\angle_{\star_{2,2,0}}=\frac{2(z_1\cdot z_2)^2+1}{4(\lambda+1)(\lambda+2)},\quad
 f^\angle_{\star_{2,1,1}}=\frac{2(z_1\cdot z_2)(z_1\cdot z_3)+(z_2\cdot z_3)}{4(\lambda+1)(\lambda+2)}.
\end{align*}
Up to permutations these are the only non-zero three-stars with $k_1+k_2+k_3\leq4$.
\end{ex}

\begin{ex}\label{angularex2}
Consider the angular graphical function of the three-star $\smallclaw$ in Figure \ref{fig:3star} with external vertices $z_1$, $z_2$, $z_3$.
We assign the weights $(\nu_1,n_1)$, $(\nu_2,n_2)$, $(\nu_3,n_3)$ to the edges attached to $z_1$, $z_2$, $z_3$, respectively. In $D=2\lambda+2\geq3$ dimensions we get the following values
for $f_\smallclaw^\angle(z_1,z_2,z_3)$ with $(n_1,n_2,n_3)=$
\begin{align*}
&(0,0,0):1,\quad(2,0,0):-\frac{\lambda^2\nu_1(1-\nu_1)}{\lambda+1},\quad(4,0,0):\frac{\lambda^2\nu_1(\lambda\nu_1+1)(1-\nu_1)(\lambda(1-\nu_1)-1)}{2(\lambda+1)(\lambda+2)},\\
&(1,1,0):\frac{\lambda\nu_1\nu_2}{\lambda+1}C^\lambda_1(z_1\cdot z_2),\quad
 (3,1,0):-\frac{\lambda^2\nu_1\nu_2(\lambda\nu_1+1)(1-\nu_1)}{(\lambda+1)(\lambda+2)}C^\lambda_1(z_1\cdot z_2),\\
&(2,2,0):\frac{\lambda\nu_1\nu_2(\lambda\nu_1+1)(\lambda\nu_2+1)}{(\lambda+1)^2(\lambda+2)}C^\lambda_2(z_1\cdot z_2)+\frac{\lambda^4\nu_1\nu_2(1-\nu_1)(1-\nu_2)}{(\lambda+1)^2},\\
&(2,1,1):\frac{\lambda\nu_1\nu_2\nu_3(\lambda\nu_1+1)}{(\lambda+1)(\lambda+2)}C^\lambda_1(z_1\cdot z_2)C^\lambda_1(z_1\cdot z_3)
 -\frac{\lambda^2\nu_1\nu_2\nu_3(\lambda(1-\nu_1)+1)}{(\lambda+1)(\lambda+2)}C^\lambda_1(z_2\cdot z_3).
\end{align*}
Up to permutations these are the only non-zero three-stars with $n_1+n_2+n_3\leq4$.
\end{ex}

\begin{thm}\label{angularthm}
Let $G$ be the graph of an angular graphical function with external vertices $\sVGext=\{z_1,$ $\ldots,$ $z_\VGext\}$ and edge-weights $(\nu_e(G),n_e(G))$, $e\in\sE_G$.
Let $\sG$ be the set of purely external subgraphs $g$---i.e.\ $\sV_g=\sV_g^{\mathrm{ext}}\subseteq\sVGext$---whose weights $n_e(g)$, $e\in \sE_g$, have the property that
for every $\sV\subseteq\sV_G$
\begin{equation}\label{gammaG}
\sum_{e\in\sC^g_{\sV}}n_e(g)\leq_2\sum_{e\in\sC^G_{\sV}}n_e(G)\qquad\text{(see Theorem \ref{radialthm} and (\ref{leq2def}))}.
\end{equation}
Then the angular graphical function $f^\angle_G$ has a (non-unique) representation as
\begin{equation}\label{faG}
f^\angle_G(z_1,\ldots,z_\VGext)=\sum_{g\in\sG}c_gf^\angle_g(z_1,\ldots,z_\VGext),
\end{equation}
where the sum has a finite number of non-zero constants $c_g\in\RR$. The empty sum is zero.
\end{thm}

\begin{ex}\label{angularex3}
Consider the three-star $\smallclaw$ of Example \ref{angularex2}. If $G$ has an isolated vertex $z_3$, then $\sV=\{1,2\}$ in (\ref{gammaG}) gives $n_{13}(g)+n_{23}(g)\leq_20$.
We get $n_{13}(g)=n_{23}(g)=0$. The result depends on $z_1\cdot z_2$ only (as is inferred by the orthogonality of the Gegenbauer polynomials).

If $(n_1,n_2,n_3)=(2,1,1)$, then (\ref{gammaG}) gives the three conditions $n_{12}(g)+n_{13}(g)\leq_22$, $n_{12}(g)+n_{23}(g)\leq_21$, and $n_{13}(g)+n_{23}(g)\leq_21$.
These conditions have the two solutions $n_{12}(g)=n_{13}(g)=1$, $n_{23}(g)=0$ and $n_{12}(g)=n_{13}(g)=0$, $n_{23}(g)=1$.
\end{ex}

\begin{ex}
For the cycle $C_n$ in Example \ref{exWSnA} the sum in (\ref{faG}) can be represented by a single term which corresponds to a graph $g$ with two edges $e_1$, $e_2$ between $z_1$ and $z_2$.
The edge $e_1$ has weights $(1,n_{z_1,z_2})$ and is also an edge in $G$. The edge $e_2$ has weights $(1,n_{z_1,2})$. We want to verify that (\ref{gammaG}) holds for $g$.

We bi-partition the vertices of $C_n$ into $\sV\subseteq\{1,\ldots,n\}$ and its complement in $C_n$. The partitions are connected by the even number of edges in $\sC\equiv\sC^{C_n}_{\sV}$.
If $z_1z_2\in\sC$ (i.e.\ $\sV$ has exactly one of the vertices $z_1,z_2$), then (\ref{gammaG}) becomes $n_{z_1,z_2}+n_{z_1,2}\leq_2 n_{z_1,z_2}+\sum_{e\in\sC\backslash\{z_1z_2\}}n_e$.
This holds true because the product of Kronecker deltas in (\ref{WSA}) ensures that all $n_e$ in the sum equal $n_{z_1,2}$.
If $z_1z_2\notin\sC$, then (\ref{gammaG}) becomes $0\leq_2\sum_{e\in\sC}n_e$ which is true for the same reason.

It is clear that (\ref{gammaG}) is not strong enough to deduce that $f_{C_n}^\angle=0$ unless all edges weights $n_{i,i+1}$, $1\leq i\leq n-1$, are equal.
\end{ex}

\begin{proof}[Proof of Theorem \ref{angularthm}]
We prove the theorem by induction over $\VGint$. If $\VGint=0$, then (\ref{faG}) is trivially true for $c_G=1$ and $c_g=0$, otherwise.

Next, we assume that $G$ is the $N$-star $\star_{\mathbf k}$. In (\ref{genstar}) we interpret the scalar products $z_i\cdot z_j$ as edges $z_iz_j$ with weights $(1/2\lambda,1)$
in some (multi-)graph $g$. So, $\star_{\mathbf k}$ has a finite expansion (\ref{faG}). We have to prove (\ref{gammaG}) for $\star_{\mathbf k}$.

We may assume that $\sV=\{z_i:i\in\sI\}$, $\sI\subseteq\{1,\ldots,N\}$, has only external vertices (otherwise we replace $\sV$ by its complement $\sV_{\star_{\mathbf k}}\backslash\sV$).
The vertices in $\sV$ contribute with the monomial $M_\sV({\mathbf t})=\prod_{i\in\sI}t_i^{k_i}$ to the generating function $f_\star({\mathbf t})$.
The edges in $\sC^g_{\sV}$ have exactly one vertex $z_i$ in $\sV$ and contribute with a factor of $t_i$ to $M_\sV({\mathbf t})$.
Edges of $g$ with no vertices in $\sV$ do not contribute to $M_\sV({\mathbf t})$ whereas edges of $g$ with both vertices $z_i,z_j$ in $\sV$ contribute with a factor of $t_it_j$.
The scalar products $z_i\cdot z_i=1$, $i\in\sI$ in (\ref{genstar}) contribute with a factor of $t_i^2$ to $M_\sV({\mathbf t})$.
Because every edge in $g$ has weight $n_e(g)=1$ we get
$$
\sum_{e\in\sC^g_{\sV}}n_e(g)\leq_2\sum_{e\in\sC^g_{\sV}}1+\sum_{z_iz_j\in g;\,i,j\in\sI}2\;\leq_2\sum_{e\in\sC^{\star_{\mathbf k}}_{\sV}}n_e(\star_{\mathbf k}).
$$
By transitivity of $\leq_2$ the inequality (\ref{gammaG}) follows.

If $G$ has a single internal vertex $x$ and no edges between external vertices, we use (\ref{pexp}) to expand $f_G^\angle$ into a finite sum over $N$-stars with
multiple edges of weight $(\lambda\nu_e,n_e)=(1/2,1)$. The multiplicity $k_i$ of an edge $xz_i$ in any such $N$-star fulfills
$$
k_i\leq_2\sum_{e=xz_i\in\sE_G}n_e(G),
$$
where the sum takes into account that $G$ may have multiple edges between $x$ and $z_i$. We expand the $N$-stars according to $(\ref{faG})$ and get from (\ref{gammaG}) for any
$N$-star,
$$
\sum_{e\in\sC^g_{\sV}}n_e(g)\leq_2\sum_{e\in\sC^{\star_{\mathbf k}}_{\sV}}n_e(\star_{\mathbf{k}})=\sum_{xz_i\in\sC^{\star_{\mathbf k}}_{\sV}}k_i\leq_2\sum_{e\in\sC^{\star_{\mathbf k}}_{\sV}}n_e(G).
$$

Finally we observe that (\ref{gammaG}) and (\ref{faG}) are trivially stable under adding edges between external vertices. This establishes the case $\VGint=1$.

For $\VGint\geq2$ we integrate over one vertex $x\in\sVGint$, i.e.\ we consider all other vertices as external and evaluate the integral over $S_{D-1}^x$ using the result for $\VGint=1$.
We obtain a sum over graphs $g_x$ with no vertex $x$. For these $g_x$ we use induction providing an expansion in terms of external graphs $g$ with
$$
\sum_{e\in\sC^g_{\sV}}n_e(g)\leq_2\sum_{e\in\sC^{g_x}_{\sV}}n_e(g_x)\leq_2\sum_{e\in\sC^G_{\sV}}n_e(G).
$$
The inequality (\ref{gammaG}) follows for every intermediate graph $g_x$.
\end{proof}

\begin{remark}\mbox{}
\begin{enumerate}
\item The proof of Theorem \ref{angularthm} is constructive. It establishes an algorithm to calculate angular graphical functions.
\item By rotational invariance the cases $\VGext=0$ and $\VGext=1$ are identical, see Section \ref{sectangper}.
If $\sVGext=\{z_1,z_2\}$, then $f_G^\angle(z_1,z_2)$ is a polynomial in $z_1\cdot z_2$.
\item It is possible to restrict the expansion (\ref{faG}) to simple graphs with weights $\nu_e(g)=1$ (or any other weight). In this case the coefficients $c_g$ are
well defined rational functions in the weights $\nu_e(G)$, see Example \ref{angularex2}.
\item It is possible to derive constraints on the weights $n_e$ for non-vanishing angular graphical functions.
For example, one may consider an internal edge cut $\sC$ (i.e.\ one of the cut subgraphs has only internal vertices).
Then $f_G^\angle=0$ unless for every edge $e\in\sC$
$$
n_e\leq_2\sum_{e\neq f\in\sC}n_f.
$$
We do not need this generalized triangle identity here and therefore leave it unproved.
\end{enumerate}
\end{remark}

%%%%%%%%%%%%%%%%%%%%%%%%%%%%%%%%%%%%%%%%%%%%%%%%%%%%%%%%%%%%%%%%%%%%%%%%%%%%%%%%%
\section{Angular periods}\label{sectangper}
If the graph $G$ has no external vertices, $\sV_G=\sVGint$, (or $G$ has one external vertex), its angular graphical function is a constant angular period,
\begin{equation}\label{Pangdef}
P^\angle_G=f^\angle_G(\emptyset).
\end{equation}
By the previous section it is clear that for fixed $n_e$, $e\in\sE$, angular periods in dimension $D=2\lambda+2\geq3$ are rational functions in $\lambda$ and $\nu_e$.
\begin{ex}\label{angularex4}
Consider the tetrahedron with vertices 1, 2, 3, 4. For tuples of integer edge-weights $(n_{12},n_{13},n_{14},n_{34},n_{24},n_{23})$ we get the angular periods (compare Example \ref{angularex2})
\begin{align*}
 \displaystyle(0,0,0,0,0,0)&:1,\\
 \displaystyle\!(2,0,0,0,0,0)&:-\frac{\lambda^2\nu_{12}(1-\nu_{12})}{\lambda+1},\\
 \displaystyle(4,0,0,0,0,0)&:\frac{\lambda^2\nu_{12}(\lambda\nu_{12}+1)(1-\nu_{12})(\lambda(1-\nu_{12})-1)}{2(\lambda+1)(\lambda+2)},\\
 \displaystyle\!(2,2,0,0,0,0)&:\frac{\lambda^2\nu_{12}\nu_{13}(1-\nu_{12})(1-\nu_{13})}{(\lambda+1)^2},\\
 \displaystyle(1,1,0,0,0,1)&:\frac{2\lambda^3\nu_{12}\nu_{13}\nu_{23}}{(\lambda+1)^2},\\
 \displaystyle\!(1,1,0,1,1,0)&:\frac{2\lambda^4\nu_{12}\nu_{13}\nu_{34}\nu_{24}}{(\lambda+1)^3}.
\end{align*}
Up to permutations these are the only non-zero tetrahedra periods with $\sum_e n_e\leq4$.
\end{ex}
In the above example we see that for general weights $\nu_e$ angular periods can be negative. If $D=4$ and $\nu_e=1$ for all $e\in\sE_G$, this is not the case.

\begin{thm}\label{D4posthm}
Let $P_G^\angle$ be the angular period of a graph $G$ with edge-weights $\nu_e=1$ for all $e\in\sE_G$ in $D=4$ dimensions. Then $P_G^\angle\geq0$.
\end{thm}
\begin{proof}
We use that a unit vector $x$ in $\RR^4$ can be identified with an $SU(2)$ group element via
$$
x\sim\Big(\begin{array}{cc}x_1&x_2\\
-\xx_2&\xx_1\end{array}\Big),\quad\text{where }x_1,x_2\in\CC\quad\text{and}\quad x=(\Re x_1,\Im x_1,\Re x_2,\Im x_2)^T.
$$
We obtain the relation
$$
\mathrm{Tr}\,xy^{-1}=2\Re\,(x_1\overline{y_1}+x_2\overline{y_2})=2\cos\phi_{xy}=C^1_1(\cos\phi_{xy}),
$$
where $y_1$ and $y_2$ are the complex coordinates of $y$. The matrix of $x$ is the two-dimensional representation $R^{(2)}(x)$ of the $SU(2)$ element $x$ with character
$\mathrm{Tr}\,R^{(2)}(x)=\chi_2(x)$. The group $SU(2)$ has one irreducible representation $R^{(n+1)}$ in each finite dimension $n+1\geq1$. Tensor products of irreducible $SU(2)$ representations
reduce to irreducible representations according to
$$
R^{(n_1+1)}(x)\otimes R^{(n_2+1)}(x)=\sum_{|n_1-n_2|\leq_2n_3\leq_2n_1+n_2}SR^{(n_3+1)}(x)S^{-1}
$$
for some invertible matrix $S$. By unitarity of the representations we can choose $S$ to be unitary, $S^{-1}=S^\dagger$ (in fact one can choose $S$ to be orthogonal). Using indices this gives
\begin{equation}\label{RRR}
R^{(n_1+1)}_{m_1,m_1'}(x)R^{(n_2+1)}_{m_2,m_2'}(x)=\sum_{|n_1-n_2|\leq_2n_3\leq_2n_1+n_2}\;\sum_{m_3,m_3'=1}^{n_3+1}S^{n_1,n_2,n_3}_{m_1,m_2,m_3}R^{(n_3+1)}_{m_3,m_3'}(x)
\overline S^{n_1,n_2,n_3}_{m_1',m_2',m_3'},
\end{equation}
where we made the dependence of $S$ on $n_1,n_2,n_3$ explicit. The matrix index of $S$ is the bi-label $m_1,m_2$ and the label $m_3$.
The $S^{n_1,n_2,n_3}_{m_1,m_2,m_3}$ are related to the 3-$j$ symbols in quantum mechanics (see e.g.\ \cite{Smultij} and the references therein).
By taking traces on both sides we obtain the addition formula for the $SU(2)$ characters
$$
\chi_{n_1+1}(x)\chi_{n_2+1}(x)=\sum_{|n_1-n_2|\leq_2n_3\leq_2n_1+n_2}\chi_{n_3+1}(x).
$$
This formula mirrors the addition formula for the Gegenbauer polynomials $C^1_n$,
$$
C^1_{n_1}(x)C^1_{n_2}(x)=\sum_{|n_1-n_2|\leq_2n_3\leq_2n_1+n_2}C^1_{n_3}(x),
$$
which follows from (\ref{Cdef1}). Moreover, we have $\chi_1(xy^{-1})=C^1_0(\cos\phi_{xy})=1$. With the initial conditions for the first two cases
the addition formulae fully determine the characters and Gegenbauer polynomials for higher indices. We get
$$
C^1_n(\cos\phi_{xy})=\chi_{n+1}(xy^{-1})=\sum_{m,m'=1}^{n+1}R^{(n+1)}_{m,m'}(x)\overline R^{(n+1)}_{m,m'}(y).
$$
The left hand side is the angular propagator (\ref{angprop}) for weights $(\nu_e,n_e)=(1,n)$. To handle the formal asymmetry in $x$ and $y$ we orient the edges of the angular graph
$G$ by the convention that $ij$ is the edge from $i$ to $j$. From the definition (\ref{fGangledef}), (\ref{Pangdef}) of the angular period we get
$$
P_G^\angle=\Big(\prod_{i=1}^{|\sV_G|}\int_{S_{D-1}}\Omega^{x_i}_{D-1}\Big)\prod_{e=ij\in\sE_G}\sum_{m,m'=1}^{n_e+1}R^{(n_e+1)}_{m,m'}(x_i)\overline R^{(n_e+1)}_{m,m'}(x_j).
$$
We interchange the finite sums with the integral and obtain
\begin{equation}\label{angperthmpf}
P_G^\angle=\sum_{m_1,m_1'=1}^{n_1+1}\cdots\sum_{m_{|\sE_G|},m_{|\sE_G|}'=1}^{n_{|\sE_G|}+1}\prod_{i=1}^{|\sV_G|}\int_{S_{D-1}}\Omega^{x_i}_{D-1}
\Big(\prod_{j:ij\in\sE_G}R^{(n_{ij}+1)}_{m_{ij},m_{ij}'}(x_i)\Big)\Big(\prod_{j:ji\in\sE_G}\overline R^{(n_{ji}+1)}_{m_{ji},m_{ji}'}(x_i)\Big),
\end{equation}
where we sum over $m_1,m_1' \in \{1,\ldots,n_1+1\}$, over $m_2,m_2' \in \{1,\ldots,n_2+1\},$
$\ldots,$
over 
$m_{|\sE_G|},m_{|\sE_G|}' \in \{ 1, n_{|\sE_G|}+1\}$ and 
where the empty product is 1. Note that we label the edges by $1,\ldots,|\sE_G|$ and also by the directed pairs of their adjacent vertices.
Iterated use of (\ref{RRR}) gives
$$
\prod_{j:ij\in\sE_G}R^{(n_{ij}+1)}_{m_{ij},m_{ij}'}(x_i)=\sum_{M,M',N}S^{(n_{ij})_j,N}_{(m_{ij})_j,M}R^{(N+1)}_{M,M'}(x_i)\overline S^{(n_{ij})_j,N}_{(m_{ij}')_j,M'},
$$
where we have defined multi-$j$-type symbols $S$, see \cite{Smultij}.
If on the left hand side the product over $j$ is empty, the right hand side is defined as $1=R^{(1)}_{1,1}$ whereas in the case of a single factor the formula is trivial.
Otherwise we explicitly have
$$
S^{(n_{ij})_j,N}_{(m_{ij})_j,M}=\sum_{M_1,\ldots,M_k}\sum_{N_1,\ldots,N_k}S^{n_1,n_2,N_2}_{m_1,m_2,M_2}S^{N_2,n_3,N_3}_{M_2,m_3,M_3}\cdots S^{N_{k-1},n_k,N}_{M_{k-1},m_k,M},
$$
where we assumed that the edges in the product are labeled from 1 to $k$.

With the orthogonality relation of representations (see e.g.\ \cite{DStanton})
$$
\int_{S_{D-1}}\Omega^x_{D-1}R^{(n_1+1)}_{m_1,m_1'}(x)\overline R^{(n_2+1)}_{m_2,m_2'}(x)=\frac{\delta_{n_1,n_2}\delta_{m_1,m_2}\delta_{m_1',m_2'}}{n_1+1}
$$
the integral over $x_i$ in (\ref{angperthmpf}) gives
$$
\sum_{M,M',N}\frac{S^{(n_{ij})_j,N}_{(m_{ij})_j,M}\overline S^{(n_{ij})_j,N}_{(m_{ij}')_j,M'}\overline S^{(n_{ji})_j,N}_{(m_{ji})_j,M}S^{(n_{ji})_j,N}_{(m_{ji}')_j,M'}}{N+1}.
$$
With this result the terms with unprimed and primed indices factor in (\ref{angperthmpf}),
$$
P_G^\angle=\frac{1}{\prod_{i=1}^{|\sV_G|}(N_i+1)}\sum_{N_1,\ldots,N_{|\sV_G|}}\Big|\sum_{m_1,\ldots,m_{|\sE_G|}}\sum_{M_1,\ldots,M_{|\sV_G|}}\prod_{i=1}^{|\sV_G|}
S^{(n_{ij})_j,N_i}_{(m_{ij})_j,M_i}\overline S^{(n_{ji})_j,N_i}_{(m_{ji})_j,M_i}\Big|^2.
$$
This expression for $P_G^\angle$ is explicitly non-negative.
\end{proof}

\begin{ex}\label{exWSnAper}
For the $n$-cycle $C_n$ with only internal vertices (and edge-weights $(1,n_{i,i+1})$) we get the angular period from integrating (\ref{WSA}) in Example \ref{exWSnA} over $z_2$.
By orthogonality (\ref{Gegorth}) of the Gegenbauer polynomials we get
\begin{align}
\begin{aligned}
\label{exWSnApereq}
P_{C_n}^\angle&=\Big(\frac{\lambda}{\lambda+n_{z_1,2}}\Big)^{n-1}\!\Big(\prod_{i=2}^n\delta_{n_{i-1,i},n_{i,i+1}}\Big)C^\lambda_{n_{z_1,2}}(1)
 =\\
&=\genfrac(){0pt}{}{n_{z_1,2}+2\lambda-1}{n_{z_1,2}}\Big(\frac{\lambda}{\lambda+n_{z_1,2}}\Big)^{n-1}\!\Big(\prod_{i=2}^n\delta_{n_{i-1,i},n_{i,i+1}}\Big),
\end{aligned}
\end{align}
where we have used (\ref{C1}).
\end{ex}

\begin{ex}\label{exneg}
Consider the tetrahedron with vertices $1,2,3,4$ and edge-weights $\nu_e=1$ for all edges $e$.
For integer edge-weights $(n_{12},n_{13},n_{14},n_{34},n_{24},n_{23})$ we get the angular periods (compare Example \ref{angularex4})
$$
(0,0,0,0,0,0):1,\quad(1,1,0,0,0,1):\frac{2\lambda^3}{(\lambda+1)^2},\quad(1,1,0,1,1,0):\frac{2\lambda^4}{(\lambda+1)^3}.
$$
Up to permutations these are the only non-zero tetrahedra periods with $\nu_e=1$ and $\sum_e n_e\leq4$. For $n_e=2$ we get
$$
(2,2,2,2,2,2):\frac{8\lambda^7(2\lambda+1)(\lambda^2+4\lambda-3)}{(\lambda+2)^3(\lambda+3)^3}
$$
which is $-\frac{6}{42875}$ in three dimensions. For $D\geq4$ the angular period of the tetrahedron with $n_e=2$ is positive.
\end{ex}

The proof of Theorem \ref{D4posthm} explicitly uses the group structure of the unit-sphere in four dimensions.
It does not generalize to higher dimensions. Still, angular periods may be positive for $D>4$.
\begin{quest}\label{posquest}
Does Theorem \ref{D4posthm} hold for (some) dimensions $D>4$?
\end{quest}
Example \ref{exneg} admits the possibility that Theorem \ref{D4posthm} holds for all $D\geq4$ and, in particular, for even $D\geq4$. So, non-negativity of angular periods may still be
a viable path to extend the validity of (G3) in Theorem \ref{thm1}. Note that Theorem \ref{thm1} has been excessively tested by the calculation of many Feynman periods with graphical functions.

%%%%%%%%%%%%%%%%%%%%%%%%%%%%%%%%%%%%%%%%%%%%%%%%%%%%%%%%%%%%%%%%%%%%%%%%%%%%%%%%%
\section{Gegenbauer expansion}\label{sectgegexp}
We want to prove that it is possible to express graphical functions as multiple sums over radial and angular graphical functions \cite{geg}.
The obstacle is to legitimate interchanging the Gegenbauer sum with the position space integrals. We expect that this is always possible but
we only prove a special case with a positivity condition on angular periods, see Theorem \ref{D4posthm} and Question \ref{posquest}.

\begin{thm}\label{thmgegex0}
Let $G$ be a graph such that the graphical function $f_G(z)$ exists in $D=2\lambda+2\geq3$ (odd or even) dimensions.
Let $G(\mathbf n)$ be the graph $G$ with additional non-negative integer edge-weights $\mathbf n=n_1,\ldots,n_{|\sE_G|}$
which vanish for every edge that is incident to the external vertex 0. Moreover, for $n_{1z}\in\ZZ_{\geq0}$ we define $G(\mathbf n,n_{1z})$ as the graph that is obtained
from $G(\mathbf n)\backslash\{0\}$ by adding an edge $1z$ with weights $\nu_{1z}=1$ and $n_{1z}$. We assume that the angular period (\ref{Pangdef}) fulfills
\begin{equation}\label{gegthmcond}
P_{G(\mathbf n,n_{1z})}^\angle\geq0\quad\text{for all }\mathbf n,n_{1z}\in\ZZ_{\geq0}.
\end{equation}
Then the graphical function $f_G(z)$ admits an absolutely convergent Gegenbauer expansion for $z=Z\ee^{\ii\phi}$, $Z\neq1$,
\begin{equation}\label{gfr}
f_G(z)=\Big(\frac{2}{\Gamma(\lambda+1)}\Big)^\VGint\sum_{\mathbf n}f_{G(\mathbf n)}^R(Z)f_{G(\mathbf n)\backslash\{0\}}^\angle(\cos\phi).
\end{equation}
The radial and angular graphical functions $f_{G(\mathbf n)}^R$ and $f_{G(\mathbf n)\backslash\{0\}}^\angle$ are defined in Sections \ref{sectrad} and \ref{sectang}.
\end{thm}
\begin{proof}
We use spherical coordinates (\ref{angularcoords}). From (\ref{eqzdef}) we get $z_1=(1,0,\ldots,0)^T$ and $z_2=(Z\cos\phi,$ $Z\sin\phi,$ $0,\ldots,0)^T$.
Propagators in (\ref{eqAdef}) which are not attached to 0 expand into Gegenbauer polynomials, see (\ref{Cdef}),
\begin{equation}\label{propagator}
\frac{1}{\|x-y\|^{2\lambda\nu_{xy}}}=\frac{1}{(XY)^{\lambda\nu_{xy}}}\sum_{n_{xy}=0}^\infty\left(\frac{X}{Y}\right)_<^{n_{xy}+\lambda\nu_{xy}}C^{\lambda\nu_{xy}}_{n_{xy}}(\cos\phi_{xy}),
\end{equation}
where $\phi_{xy}$ is the angle between $x$ and $y$ and $(x)_<=x$ if $x<1$ and $x^{-1}$ otherwise. Propagators of edges which are attached to 0 equal their radial propagators (\ref{propdef}).

The integration measure splits into radial and angular coordinates according to (see (\ref{SD1}))
$$
\int_{\RR^D}\frac{\dd^Dx_i}{\pi^{D/2}}=\frac{\text{vol}\,S_{D-1}}{\pi^{D/2}}\int_0^\infty\dd X X^{D-1}\int_{S_{D-1}}\Omega^x_{D-1}
 =\frac{2}{\Gamma(\lambda+1)}\int_0^\infty\dd X X^{D-1}\int_{S_{D-1}}\Omega^x_{D-1}.
$$
By Fubini we can do all angular integrals first. The domain of integration is compact and the integral is absolutely convergent because by (\ref{C1}) we get for $X\neq Y$,
$$
\frac{1}{(XY)^{\lambda\nu_{xy}}}\sum_{n_{xy}=0}^\infty\left(\frac{X}{Y}\right)_<^{n_{xy}+\lambda\nu_{xy}}|C^{\lambda\nu_{xy}}_{n_{xy}}(\cos\phi_{xy})|
\leq\frac{1}{|X-Y|^{2\lambda\nu_{xy}}}<\infty.
$$
The angular integration over the Gegenbauer polynomials gives the angular graphical function $f_{G(\mathbf n)\backslash\{0\}}^\angle(\cos\phi)$.
It admits an expansion into Gegenbauer polynomials
\begin{equation}\label{eqpfpf4}
f_{G(\mathbf n)\backslash\{0\}}^\angle(\cos\phi)=\sum_{k\geq0}f_{G(\mathbf n)\backslash\{0\},k}^\angle C^\lambda_k(\cos\phi).
\end{equation}
If we multiply $f_{G(\mathbf n)\backslash\{0\}}^\angle$ with $C^\lambda_{n_{1z}}(\cos\phi)$ and integrate over $\int_{S_{D-1}}\Omega^z_{D-1}$, orthogonality
(\ref{Gegorth}) of Gegenbauer polynomials provides an expansion into angular periods,
$$
f_{G(\mathbf n)\backslash\{0\}}^\angle(\cos\phi)=\sum_{n_{1z}\geq0}\frac{\lambda+n_{1z}}{\lambda C^\lambda_{n_{1z}}(1)}P_{G(\mathbf n,n_{1z})}^\angle C^\lambda_{n_{1z}}(\cos\phi).
$$
Because $P_{G(\mathbf n,n_{1z})}^\angle\geq0$ we get from (\ref{C1}) the estimate
$$
|f_{G(\mathbf n)\backslash\{0\}}^\angle(\cos\phi)|\leq\sum_{n_{1z}\geq0}\frac{\lambda+n_{1z}}{\lambda}P_{G(\mathbf n,n_{1z})}^\angle=f_{G(\mathbf n)\backslash\{0\}}^\angle(1).
$$
The integral over the radial variables is absolutely convergent,
$$
\Big(\prod_{i=1}^\VGint\int_0^\infty\dd X_i X_i^{D-1}\Big)\Big|\Big(\prod_{e\in\sE_G}p^R_e\Big)f_{G(\mathbf n)\backslash\{0\}}^\angle(\cos\phi)\Big|\leq f_G(Z)<\infty
$$
for $1\neq Z>0$. In the above estimate we reversed the interchange of the sum and the angular integral to obtain $f_G(Z\ee^{\ii\phi})$ for $\phi=0$. Finiteness follows from
(G2) in Theorem \ref{thm1}.

By absolute convergence we can interchange the sum with the radial integrals which gives (\ref{gfr}) for $Z\neq1$.
\end{proof}

\begin{cor}
Let $G$ be an uncompleted graph such that the Feynman period $P_G$ exists in $D=2\lambda+2\geq3$ dimensions (see Section \ref{sectper}).
Let $G(\mathbf n)\cup\{Z\}$ be the graph of the corresponding constant radial graphical function with isolated external vertex $Z$.
If the angular period $P_{G(\mathbf n)\backslash\{0\}}^\angle\geq0$ for all $\mathbf n$, then
\begin{equation}\label{gfrper}
P_G=\Big(\frac{2}{\Gamma(\lambda+1)}\Big)^{|\sV_G|-2}\sum_{\mathbf n}f_{G(\mathbf n)\cup\{Z\}}^RP_{G(\mathbf n)\backslash\{0\}}^\angle.
\end{equation}
\end{cor}

\begin{proof}
Let $G_{0,1,z}$ be the constant graphical function which arises from the period $G$ by identifying any two vertices with 0 and 1 and adding an isolated vertex $z$.
We add an edge $1z$ with weight $(1,n_{1z})$ and consider the period $P_{G(\mathbf n,n_{1z})}^\angle$ as in Theorem \ref{thmgegex0}. By orthogonality (\ref{Gegorth}) the integration over
the vertex $z$ gives zero unless $n_{1z}=0$. In this case the edge $1z$ is absent and $P_{G(\mathbf n,n_{1z})}^\angle\geq0$ follows from $P_{G(\mathbf n)\backslash\{0\}}^\angle\geq0$.
Theorem \ref{thmgegex0} gives the result.
\end{proof}

\begin{con}\label{congegex}
The Gegenbauer expansion (\ref{gfr}) is valid for all graphical functions.
\end{con}
In physics the Gegenbauer $x$-space technique is considered a well established tool \cite{geg}. It tacitly assumes that the Gegenbauer expansion in the integrand interchanges with the
$x$-space integration as in Conjecture \ref{congegex}.

\begin{ex}\label{exWSn1}
Consider the wheel with $n$ spokes $W\!S_{n,D}$ in Figure \ref{fig:wheels} and assume that the (odd or even) dimension $D\leq2n-2$ for convergence.
We want to calculate the period $P_{W\!S_{n,D}}$. The hub has label 0 so that $W\!S_{n,D}\backslash\{0\}=C_n$
whose angular period is calculated in Example \ref{exWSnAper}. The result (\ref{exWSnApereq}) is non-negative.
The radial period can be derived from Example \ref{exWSnR1} by integrating (\ref{WSReq0}) over $Z$ with the measure
$Z^{2\lambda+1}\dd Z$. Alternatively one can consider the wheel with $n+1$ spokes at $Z=1$ (the triangle $01Z$ contributes with a trivial factor 1 to the radial graphical function for $Z=1$).
We get
$$
f_{W\!S_{n,D}\cup\{Z\}}^R=2^n\int_{-\infty}^\infty\frac{\dd P}{2\pi}\prod_{i=1}^n\frac{n_{i,i+1}+\lambda}{P^2+(n_{i,i+1}+\lambda)^2}.
$$
With (\ref{exWSnApereq}) all the $n_{i,i+1}$ are identified and we obtain
\begin{align*}
P_{W\!S_{n,D}}&=\frac{2^{2n-1}}{\Gamma(\lambda+1)^{n-1}}\sum_{k=0}^\infty\genfrac(){0pt}{}{k+2\lambda-1}{k}\Big(\frac{\lambda}{\lambda+k}\Big)^{n-1}
 \int_{-\infty}^\infty\frac{\dd P}{2\pi}\Big(\frac{k+\lambda}{P^2+(k+\lambda)^2}\Big)^n\\
&=\frac{2^{2n-1}}{\Gamma(\lambda)^{n-1}}\int_{-\infty}^\infty\frac{\dd P}{2\pi(P^2+1)^n}\sum_{k=0}^\infty\frac{\genfrac(){0pt}{}{k+2\lambda-1}{k}}{(k+\lambda)^{2n-2}},
\end{align*}
where we scaled the integration variable $P$ with $k+\lambda$. The residue theorem gives
$$
\int_{-\infty}^\infty\frac{\dd P}{2\pi(P^2+1)^n}=\ii\,\mathrm{res}_\ii\frac{1}{(P^2+1)^n}=\frac{\ii}{(n-1)!}\frac{\partial^{n-1}}{\partial_P^{n-1}}\Big|_\ii\frac{1}{(P+\ii)^n}\\
=\frac{1}{2^{2n-1}}\genfrac(){0pt}{}{2n-2}{n-1}.
$$
We obtain the formula
$$
P_{W\!S_{n,D}}=\frac{\genfrac(){0pt}{}{2n-2}{n-1}}{(2\lambda-1)!\Gamma(\lambda)^{n-1}}\sum_{k=\lambda}^\infty\frac{(k+\lambda-1)(k+\lambda-2)\cdots(k-\lambda+2)(k-\lambda+1)}{k^{2n-2}}.
$$

For even $D$ we extend the sum to $k=1$ (the surplus terms have zero numerator). We find that a factor of $k$ cancels and get by pairing the factors in the numerator
\begin{equation}\label{evendwheels}
P_{W\!S_{n,D}}=\frac{\genfrac(){0pt}{}{2n-2}{n-1}}{(2\lambda-1)!(\lambda-1)!^{n-1}}\sum_{k=1}^\infty\frac{\prod_{\ell=1}^{\lambda-1}(k^2-\ell^2)}{k^{2n-3}},
 \quad\text{if}\;D\leq 2n-2\;\text{is even}.
\end{equation}
For odd $D$ the sum over $k$ ranges over half-integers from $\lambda$ to infinity. We are free to extend the sum to $k=1/2$. If we add and subtract the sum over the integers from
1 to infinity, we get
\begin{gather}
\begin{gathered}
\label{odddwheels}
P_{W\!S_{n,D}}=\frac{\genfrac(){0pt}{}{2n-2}{n-1}}{(2\lambda-1)!(\lambda-1)!^{n-1}}\sum_{k=1}^\infty
 \frac{2^{2n-2\lambda-1}\prod_{\ell=1}^{\lambda-1/2}(k^2-(2\ell-1)^2)-\prod_{\ell=1}^{\lambda-1/2}(k^2-\ell^2)}{k^{2n-2}},\; 
\\ \text{if}\;D\leq 2n-3\;\text{is odd}.
\end{gathered}
\end{gather}
The products on the right hand sides expand to $\QQ$-linear combinations of single zetas of odd weights for even $D$ and even weights for odd $D$.
\end{ex}

%%%%%%%%%%%%%%%%%%%%%%%%%%%%%%%%%%%%%%%%%%%%%%%%%%%%%%%%%%%%%%%%%%%%%%%%%%%%%%%%%
\section{Proof of Theorem \ref{thm1}}\label{sectpf1}
In this section we prove (G3) in Theorem \ref{thm1}. The main task will be to prove the case $s=0$ in (\ref{01expansion}) with the identity in (\ref{Ma}).
We will keep the dimension and the weights general when possible to show that Conjecture \ref{con1} follows from Conjecture \ref{congegex}.

We start proving that the case $s=1$ in (\ref{01expansion}) follows from the case $s=0$. We use (\ref{eqfA1}) for the graph $G$ with external vertices $a,b,c$ which
we first interpret as $a=z_0$, $b=z_1$, $c=z_2$ yielding the function of the graph $G_{0,1,z}$ (indexed by the external labels).
Then we interpret $a,b,c$ as $z_1,z_0,z_2$ interchanging the role of $z$ and $1-z$ in (\ref{eqinvs}). We hence obtain the graph $G_{1,0,1-z}$.
From (\ref{eqfA1}) we get (alternatively one can use completion, see Sections \ref{sectcomp} and \ref{sectperm})
$$
\|z_1-z_0\|^{-2\lambda N_G}f_{G_{0,1,z}}=\|z_0-z_1\|^{-2\lambda N_G}f_{G_{1,0,1-z}}.
$$
This implies $f_{G_{1,0,1-z}}=f_{G_{0,1,z}}$ and (\ref{01expansion}) follows for $s=1$ by swapping $0\leftrightarrow1$ and $z\leftrightarrow1-z$.

Next we interpret $a,b,c$ as $z_0,z_2,z_1$ which is equivalent to swapping $z\leftrightarrow z^{-1}$ in (\ref{eqinvs}). From (\ref{eqfA1}) we get
$$
\|z_1-z_0\|^{-2\lambda N_G}f_{G_{0,1,z}}=\|z_2-z_0\|^{-2\lambda N_G}f_{G_{0,z^{-1},1}}.
$$
With (\ref{eqinvs}) this implies $f_{G_{0,z^{-1},1}}=(z\zz)^{\lambda N_G}f_{G_{0,1,z}}$ (see (\ref{Zrez})).
We expand $f_{G_{0,1,z}}$ at $s=0$ and substitute $z\mapsto z^{-1}$ to get for $|z|>1$,
$$
f_{G_{0,z,1}}=\sum_{\ell=0}^\VGint\sum_{m,\mm=M_0(G_{0,1,z})}^\infty (-1)^\ell c_{\ell,m,\mm}^0(G_{0,1,z})(\log z\zz)^\ell z^{-m-\lambda N_G}\zz^{-\mm-\lambda N_G}.
$$
With the new summation indices $m'=-m-\lambda N_G$, $\mm'=-\mm-\lambda N_G$ we get (\ref{inftyexpansion}) where
$$
M_\infty(G_{0,z,1})=-M_0(G_{0,1,z})-\lambda N_G.
$$
This gives the identity in (\ref{Minfty}) from (\ref{Ma}) for $s=0$ by interchanging 1 and $z$.

We now prove the inequality in (\ref{Ma}). For $s=0,1$ and any $\emptyset\neq\sV\subseteq\sVGint$ let $G_s=G[\sV\cup\{s,z\}]$.
Let $\nu_z(G_s)$ be the sum of weights adjacent to $z$ in $G_s$. Because $G_s$ is a subgraph of $G$ we have $\nu_z(G_s)\leq\nu_z^>$.
We consider the graph $G[\sV\cup\{s\}]=G_s\backslash\{z\}$. Note that $G[\sV\cup\{s\}]$ has at least one internal vertex.
By convergence (\ref{eqconv}) we get $N_{G[\sV\cup\{s\}]}<0$. Integrality of the weights ($\lambda\nu_e\in\ZZ$) implies $\lambda N_{G[\sV\cup\{s\}]}\leq-1$.
We multiply
$$
\sum_{e\in G_s}\nu_e=\sum_{e\in G[\sV\cup\{s\}]}\nu_e+\nu_z(G_s)\leq\sum_{e\in G[\sV\cup\{s\}]}\nu_e+\nu_z^>
$$
with $\lambda$ and subtract $(\lambda+1)|\sV|$ to obtain
$$
\lambda N_{G_s}\leq\lambda N_{G[\sV\cup\{s\}]}+\lambda\nu_z^>\leq-1+\lambda\nu_z^>.
$$
This implies $-\max_{\emptyset\neq\sV\subseteq\sVGint}\lambda N_{G_s}\geq1-\lambda\nu_z^>$. For $\sV=\emptyset$ we get $G_s=sz$ and the inequality in (\ref{Ma}) follows.

To prove the inequality in (\ref{Minfty}) we choose any proper subset $\sV$ of $\sVGint$. Depending on the number of vertices in $\sVGint\backslash\sV$
we split the edges of $G\backslash G[\sV\cup\{0,1\}]$ into the three sets $\sE_2=G[\sVGint\backslash\sV]$, $\sE_1=\{uv,\,u\in\sVGint\backslash\sV,\,v\in\sV\cup\{0,1,z\}\}$,
and $\sE_0=\{vz,\,v\in\sV\cup\{0,1\}\}$. We get
\begin{equation}\label{pf1eq1}
\lambda N_G-\lambda N_{G[\sV\cup\{0,1\}]}+(\lambda+1)(\VGint-|\sV|)=\sum_{e\in G\backslash G[\sV\cup\{0,1\}]}\lambda\nu_e\geq\sum_{e\in\sE_1\cup\sE_2}\lambda\nu_e+\lambda\nu_z^<.
\end{equation}
We consider the graph $\Gg[\sVGint\cup\{\infty\}\backslash\sV]$ which has at least one internal vertex. Convergence (\ref{eqconv}) gives
$$
\sum_{e\in\Gg[\sVGint\cup\{\infty\}\backslash\sV]}\nu_e-\frac{\lambda+1}{\lambda}(\VGint-|\sV|)<0.
$$
Edges in $\Gg[\sVGint\cup\{\infty\}\backslash\sV]$ are either in $\sE_2$ or they have one vertex $\infty$. By completion, see Definition \ref{defcompl}, we have
$\lambda\nu_{v\infty}=D-\lambda\nu_v$ where $\nu_v$ is the sum of all weights adjacent in $G$ to the vertex $v\in\sVGint\backslash\sV$. Altogether, convergence implies
$$
\sum_{v\in\sVGint\backslash\sV}(2\lambda+2-\lambda\nu_v)+\sum_{e\in\sE_2}\lambda\nu_e-(\lambda+1)(\VGint-|\sV|)<0.
$$
Because $\sum_{v\in\sVGint\backslash\sV}\nu_v=2\sum_{e\in\sE_2}\nu_e+\sum_{e\in\sE_1}\nu_e$ we have
$$
(\lambda+1)(\VGint-|\sV|)-\sum_{e\in\sE_1\cup\sE_2}\lambda\nu_e<0.
$$
With integrality of the weights and (\ref{pf1eq1}) we obtain
$$
-\lambda N_G+\lambda N_{G[\sV\cup\{0,1\}]}\leq-1-\lambda\nu_z^<.
$$
For $\sV=\sVGint$ we get $-\lambda N_G+\lambda N_{G[\sV\cup\{0,1\}]}=-\lambda\nu_z$ and the inequality in (\ref{Minfty}) follows.

We are left with the task to prove (G3) for $s=0$. We assume $Z<1$ and use the Gegenbauer expansion into radial and angular graphical functions in Theorem \ref{thmgegex0}.
Here, we have to specialize to $D=4$ and $\nu_e=1$ in (\ref{thm1cond}) because we need Theorem \ref{D4posthm} to establish condition (\ref{gegthmcond}).

In (\ref{Req}) we have $\log Z=\log(z\zz)/2$. We split the power of $Z$ into two parts: One part is
$$
Z^{-2\lambda N_{G[\sV\cup\{0,z\}]}}=(z\zz)^{-\lambda N_{G[\sV\cup\{0,z\}]}}=(z\zz)^M\quad\text{with }M\geq M_0,
$$
where $M_0$ is defined in (\ref{Ma}). Because $\lambda\nu_e\in\ZZ$ for all edges $e$ we get $M\in\ZZ$.

We combine the second part with the angular graphical function in (\ref{faG}) which specializes to
$$
f^\angle_G(1,z)=\sum_{g\in\sG}c_gf^\angle_g(1,z).
$$
The vertices of the external graph $g$ are a subset of $\{1,z\}$.
If $g$ has the vertices 1 and $z$, then $f_g^\angle(1,z)$ is a polynomial in $\cos(\phi)$ of degree $\sum_en_e(g)$ which is symmetric if $\sum_en_e(g)$ is even and anti-symmetric otherwise.
We have $Z\cos(\phi)=(z+\zz)/2$ and $Z^2=z\zz$. From (\ref{gammaG}) for $\sV\cup\{0,z\}\subseteq\sV_G$ we get with the second part of the power of $Z$ in (\ref{Req}),
$$
Z^{\sum_{e\in\sC^G_{\sV\cup\{0,z\}}}n_e}f^\angle_g(1,z)=P(z,\zz),\quad\text{with }P\in\RR[z,\zz].
$$
If $g$ has one or zero vertices, then $f_g^\angle(1,z)$ is constant and $\sum_{e\in\sC^G_{\sV\cup\{0,z\}}}n_e(G)$ is even.
In this case $P\in\RR[z\zz]\subset\RR[z,\zz]$.

Altogether we get sums of the three factors
$$
\log^\ell (z\zz) (z\zz)^MP(z,\zz)
$$
for $\ell=0,1,\ldots,\VGint$, $M_0\leq M\in\ZZ$ and $P\in\RR[z,\zz]$. This gives (\ref{01expansion}) and completes the proof of (G3) in Theorem \ref{thm1}.

\begin{remark}
With a cutoff in the weights ${\mathbf n}$ (like $\sum_en_e\leq N$) one can use (\ref{gfr}) to obtain approximations of graphical functions (or Feynman periods).
Particularly efficient is the case when $G\backslash\{0\}$ is series-parallel.
With accelerated convergence this gives an algorithm to approximate Feynman periods to many digits. This method was used by Broadhurst and Kreimer in \cite{BK}
and later by the second author in \cite{Census} to determine some Feynman periods with an exact numerical approach.
For general graphs there exists the new method of tropical Monte Carlo quadrature by the first author to obtain good approximations \cite{TMCQ}.
\end{remark}

%%%%%%%%%%%%%%%%%%%%%%%%%%%%%%%%%%%%%%%%%%%%%%%%%%%%%%%%%%%%%%%%%%%%%%%%%%%%%%%%%
\section{Proof of Theorem \ref{thm4}}\label{sectpf4}
\begin{lem}
In polar coordinates $z=Z\ee^{\ii\phi}$ we have
\begin{align}\label{proppf4eq}
&D_z\equiv\frac{1}{z-\zz}(z\partial_z-\zz\partial_\zz)=-\frac{1}{2Z\sin\phi}\partial_\phi=\frac{1}{Z}\partial_{2\!\cos\phi},\nonumber\\
&\frac{1}{(\lambda-1)!}\partial_{2\!\cos\phi}^{\lambda-1}\frac{\sin(k+\lambda)\phi}{\sin\phi}=C^\lambda_k(\cos\phi),
\end{align}
where $\lambda=1,2,3,\ldots$, $k=0,1,2,\ldots$, and $C^\lambda_k$ are the Gegenbauer polynomials (\ref{Cdef}).
\end{lem}
\begin{proof}
The first identity follows from $z-\zz=2\ii Z\sin\phi$ and
$\partial_\phi=\frac{\partial z}{\partial\phi}\partial_z+\frac{\partial\zz}{\partial\phi}\partial_\zz=\ii(z\partial_z-\zz\partial_\zz)$.

We prove the second identity with generating functions. We multiply both sides with $t^{k+\lambda}$ and sum over $k$. It is convenient to extend the sum to negative values
of $k$. To do this we define $C^\lambda_k=0$ for $k<0$. The left hand side gives
$$
\frac{1}{(\lambda-1)!}\partial_{2\!\cos\phi}^{\lambda-1}\sum_{k=-\lambda}^\infty\frac{(\ee^{\ii\phi}t)^{k+\lambda}-(\ee^{-\ii\phi}t)^{k+\lambda}}{\ee^{\ii\phi}-\ee^{-\ii\phi}}.
$$
The geometric series is $t/(1-(2\cos\phi)t+t^2)$. Differentiation with $\partial_{2\!\cos\phi}^{\lambda-1}$ gives $(\lambda-1)!t^\lambda/(1-(2\cos\phi)t+t^2)^\lambda$.
The factor $(\lambda-1)!$ cancels and we obtain the generating function of Gegenbauer polynomials (\ref{Cdef}).
\end{proof}

\begin{prop}\label{proppf4}
The kernel of the differential operator $D_z^{\lambda-1}$ is
\begin{equation}\label{corpf4eq}
\ker D_z^{\lambda-1}=\sum_{k=0}^{\lambda-2}h_k(z\zz)(z^k+\zz^k)
\end{equation}
for arbitrary functions $h_k(z\zz)$. In particular, every function in the kernel is symmetric under $z\leftrightarrow\zz$.
\end{prop}
\begin{proof}
From $2Z\cos\phi=z+\zz$ and the first identity in (\ref{proppf4eq}) we obtain that the kernel are polynomials in $z+\zz$ of degree $\leq\lambda-2$ with $z\zz$-dependent coefficients.
A linear basis change gives (\ref{corpf4eq}).
\end{proof}

\begin{proof}[Proof of Theorem \ref{thm4} assuming Conjecture \ref{congegex}]
We start by constructing the two-dimensional avatar ${}^2\!f_G(z)$. We do this in spherical coordinates and use the representation (\ref{gfr}) of $f_G(z)$ in terms
of radial and angular graphical functions. The angular graphical function $f_{G(\mathbf n)\backslash\{0\}}^\angle$ is a polynomial in $\cos\phi$ which can be expanded into
Gegenbauer polynomials by (\ref{eqpfpf4}) where the sum is finite for fixed values of the weights $\mathbf n$. We get (for $Z\neq1$ we have absolute convergence)
$$
f_G(z)=\Big(\frac{2}{\lambda!}\Big)^\VGint\sum_{k=0}^\infty\sum_{\mathbf n}f_{G(\mathbf n)}^R(Z)f_{G(\mathbf n)\backslash\{0\},k}^\angle C^\lambda_k(\cos\phi).
$$
We use this expansion to define
\begin{equation}\label{2fG}
{}^2\!f_G(z)=\Big(\frac{2}{\lambda!}\Big)^\VGint\frac{2\ii Z^\lambda}{(\lambda-1)!}\sum_{k=0}^\infty\sum_{\mathbf n}f_{G(\mathbf n)}^R(Z)f_{G(\mathbf n)\backslash\{0\},k}^\angle
\sin(k+\lambda)\phi.
\end{equation}
Anti-symmetry is provided by anti-symmetry in $\phi$. Equation (\ref{eqthm4int}) is fulfilled by the orthogonality of the sine,
\begin{equation}\label{sinort}
\int_0^{2\pi}\sin k\phi\sin\ell(\theta-\phi)\dd\phi=-\pi\delta_{k,\ell}\cos\theta\quad\text{for }k,\ell=1,2,3,\ldots,
\end{equation}
where we set $\theta=0$. To check that ${}^2\!f_G(z)$ fulfills the differential equation (\ref{eq2av}) we divide (\ref{2fG}) by $z-\zz=2\ii Z\sin\phi$ and use (\ref{proppf4eq}).

For uniqueness we observe that by Proposition \ref{proppf4} Equation (\ref{eq2av}) determines ${}^2\!f^{(\lambda)}_G(z)$ up to a function
$$
\sum_{k=0}^{\lambda-2}h_k(z\zz)(z^k+\zz^k)(z-\zz)=4\ii h_0(Z^2)Z\sin\phi+2\ii\sum_{k=1}^{\lambda-2}h_k(Z^2)Z^{k+1}[\sin(k+1)\phi-\sin(k-1)\phi].
$$
Condition (\ref{eqthm4int}) and orthogonality of the sine (\ref{sinort}) gives $Z^kh_{k-1}(Z^2)=Z^{k+2}h_{k+1}(Z^2)$ for $k=2,\ldots,\lambda-1$, where we set $h_\lambda(Z^2)=h_{\lambda-1}(Z^2)=0$.
This gives $h_{\lambda-2}(Z^2)=\ldots=h_1(Z^2)=0$. Orthogonality for $k=1$ gives $h_0(Z^2)=0$. This trivializes the kernel.

Now we prove (\ref{geg2}). The graph $G$ has a Gegenbauer split at an internal vertex $x$, see Definition \ref{defgegsplit}. The graphical function $f_G(z)$ is the
Feynman integral $A_G$ evaluated at special vertices $z_0=0$, $z_1$ and $z_2$, see (\ref{eqzdef}) and (\ref{eqfGdef}). It factors in the following sense,
\begin{equation}\label{gegfactor}
A_G(z_0,z_1,z_2)=\int_{\RR^D}\frac{\dd^Dx}{\pi^{D/2}}A_{G_1}(z_0,z_1,x)A_{G_2}(z_0,x,z_2),
\end{equation}
where $G_i=\Gg_i\backslash\{\infty\}$, $i=1,2$ are the de-completed split graphs.
With (\ref{eqfA1}) we express the Feynman integrals $A_{G_i}$ in terms of graphical functions evaluated at invariants (\ref{eqinvs}):
\begin{align*}
A_{G_1}(z_0,z_1,x)&=f_{G_1}(x_1),&\text{with $\|x\|^2=x_1\xx_1\;$ and $\;\;\|x-1\|^2=(x_1-1)(\xx_1-1)$},\\
A_{G_2}(z_0,x,z_2)&=\|x\|^{-2\lambda N_{G_2}}f_{G_2}(x_2),&\text{with $\frac{z\zz}{\|x\|^2}=x_2\xx_2$ and $\frac{\|z_2-x\|^2}{\|x\|^2}=(x_2-1)(\xx_2-1)$}.
\end{align*}
We want to express $f_{G_i}(x_i)$ in terms of radial and angular graphical functions. To to this we need polar representations of $x_1$ and $x_2$. From the above identities we find
(up to an irrelevant sign ambiguity in the angle)
$$
x_1=X\ee^{\ii\phi_{x1}},\quad\text{and}\quad x_2=\frac{Z}{X}\ee^{\ii\phi_{xz}},
$$
where $\phi_{x1}$, $\phi_{xz}$ are the angles between $x$ and $z_1$, $z_2$, respectively. From (\ref{gfr}) we get for $f_G(z)$ the expression
\begin{align*}
\Big(\frac{2}{\lambda!}\Big)^{\VGint-1}\!\!\!\int\frac{\dd^Dx}{\pi^{\lambda+1}} &\Big[\sum_{\mathbf n_1}f^R_{G_1(\mathbf n_1)}(X)f^\angle_{G_1(\mathbf n_1)\backslash\{0\}}(\cos\phi_{x1})\Big]
X^{-2\lambda N_{G_2}}
\times
\\
&\times
\Big[\sum_{\mathbf n_2}f^R_{G_2(\mathbf n_2)}\Big(\!\frac{Z}{X}\!\Big)f^\angle_{G_2(\mathbf n_2)\backslash\{0\}}(\cos\phi_{xz})\Big].
\end{align*}
We expand the angular graphical functions into Gegenbauer polynomials, see (\ref{eqpfpf4}).
With spherical coordinates for $x$ we use orthogonality of Gegenbauer polynomials (\ref{Gegorth}) to simplify the expression for $f_G(z)$ to
$$
\Big(\frac{2}{\lambda!}\Big)^\VGint\sum_{k=0}^\infty\frac{\lambda C^\lambda_k(\cos\phi)}{\lambda+k}
\int_0^\infty\dd X X^\alpha\sum_{\mathbf n_1,\mathbf n_2}f^R_{G_1(\mathbf n_1)}(X)f^\angle_{G_1(\mathbf n_1)\backslash\{0\},k}
f^R_{G_2(\mathbf n_2)}\Big(\!\frac{Z}{X}\!\Big)f^\angle_{G_2(\mathbf n_2)\backslash\{0\},k},
$$
where $\alpha=2\lambda+1-2\lambda N_{G_2}$ and $\phi$ is the angle between $z$ and 1. The two-dimensional avatar (\ref{2fG}) is obtained by replacing $C^\lambda_k(\cos\phi)$ with
$[2\ii Z^\lambda\sin(k+\lambda)\phi]/(\lambda-1)!$. By $z\partial_z-\zz\partial_\zz=-\ii\partial_\phi$ from the first identity in (\ref{proppf4eq}) we obtain for the
left hand side of (\ref{geg2}),
\begin{equation}\label{lhs}
\Big(\frac{2}{\lambda!}\Big)^{\VGint}\sum_{k=0}^\infty\frac{2\lambda Z^\lambda\cos(k+\lambda)\phi}{(\lambda-1)!}\int_0^\infty\dd X X^\alpha
\sum_{\mathbf n_1,\mathbf n_2} f^R_{G_1(\mathbf n_1)}(X)f^\angle_{G_1(\mathbf n_1)\backslash\{0\},k}f^R_{G_2(\mathbf n_2)}\Big(\frac{Z}{X}\Big)f^\angle_{G_2(\mathbf n_2)\backslash\{0\},k}.
\end{equation}
Convergence of the right hand side of (\ref{geg2}) is inherited from (\ref{gegfactor}). We use polar coordinates $x=X\ee^{\ii\theta}$, $\dd^2x=2X\dd X\wedge\dd\theta$,
and evaluate the integral over $\theta$ with (\ref{sinort}) (and the role of $\phi$ and $\theta$ interchanged). The result reproduces (\ref{lhs}).

For uniqueness of ${}^2\!f_G(z)$ in (\ref{geg2}) we observe that the kernel of $z\partial_z-\zz\partial_\zz$ are radial functions $h(z\zz)$, see the case $\lambda=2$ in
Proposition \ref{proppf4}. Asymmetry of ${}^2\!f_G(z)$ therefore kills the kernel.
\end{proof}

\begin{proof}[Proof of Proposition \ref{propgeg}]
From uniqueness of ${}^2\!f_G(z)$ and (\ref{2fG}) it is clear that ${}^2\!f_G(z)$ is analytic in $z$ and $\zz$ for $1\neq Z\in\RR_+$ with single-valued log-Laurent expansions at 0 and $\infty$.

We fix a point $a$ on the unit circle. We want to use (\ref{eq2av}) to determine ${}^2\!f_G(z)$ by $(\lambda-1)$ times inverting $D_z$.
Assume $g(z)=g(\zz)\in\sS\sV_{\{0,1,\infty\}}$. We show that in the neighborhood of $z=a$ there exists a symmetric solution of $D_zf(z)=g(z)$ which is analytic in $z$ and $\zz$
if $a\neq1$ or has a single-valued log-Laurent expansion if $a=1$. To do this we locally invert the differential operator $D_z=-(2Z\sin\phi)^{-1}\partial_\phi$ at $z=a$.

We first consider the case $a\neq1$. By linearity it suffices to handle the monomial $(z-a)^m(\zz-\aaa)^\mm$.
We multiply the monomial with $\ii(z-\zz)=-2Z\sin\phi$ and expand the expression. We obtain a sum of terms
$\ii Z^{k+\kk}\ee^{\ii(k-\kk)\phi}$ with $k\leq m+1$, $\kk\leq\mm+1$. To invert the differential operator $\partial_\phi$ we integrate over the variable $\phi$.
We get $z^k\zz^\kk/(k-\kk)$ if $k\neq\kk$ and $(z\zz)^k\log z$ otherwise (fix any branch of the logarithm). Both terms are analytic in $z$ and $\zz$ at $z=a$, $\zz=\aaa$.
The symmetry of $f(z)$ follows from the symmetry of $D_z$ (see (\ref{corpf4eq}) for $\lambda=2$).

We are left with the case $a=1$. By the symmetry of $g(z)$ the expansion (\ref{01expansion}) at $z=1$ can be organized in terms of symmetric monomials
\begin{equation}\label{lempf4eq}
\frac{\log^\ell[(z-1)(\zz-1)]}{[(z-1)(\zz-1)]^{-M_1}}\left[(z-1)^k(\zz-1)^\kk+(\zz-1)^k(z-1)^\kk\right]
\end{equation}
for $M_1\leq0$ and $k,\kk\geq0$. The differential operator $D_z$ lowers the total degree of symmetric polynomials by 1:
\begin{equation}\label{Dzsym}
D_z(z^k\zz^\kk+\zz^kz^\kk)=(k-\kk)\frac{z^k\zz^\kk-\zz^kz^\kk}{z-\zz}\in\ZZ[z,\zz],
\end{equation}
which is of total degree $k+\kk-1$. On functions of $|z-1|^2$ the operator $D_z$ is a negative derivative:
$$
D_zF((z-1)(\zz-1))=-F'((z-1)(\zz-1)).
$$
We combine both facts to invert $D_z$ on (\ref{lempf4eq}) using integration by parts. The first factor is integrated whereas the second factor is differentiated.
We repeat integrating by parts until the second factor is zero. The result provides the single-valued log-Laurent expansion of $f(z)$ at $z=1$.

By repeated application of this result and Proposition \ref{proppf4} we obtain that ${}^2\!f_G(z)/(z-\zz)$ differs from a function with good expansions at the unit circle by
a function $h(z)=\sum_{k=0}^{\lambda-2}h_k(z\zz)(z^k+\zz^k)$. By real-analyticity of ${}^2\!f_G(z)/(z-\zz)$ at some point $a\neq1$ on the unit circle we obtain
that $h(z)$ is analytic in $z$ and $\zz$ at $z=a$, $\zz=\aaa$. Hence, all $h_k(z\zz)$ are regular at 1. Therefore $h(z)$ is analytic in $z$ and $\zz$ on the whole unit circle and the result follows.
\end{proof}

%%%%%%%%%%%%%%%%%%%%%%%%%%%%%%%%%%%%%%%%%%%%%%%%%%%%%%%%%%%%%%%%%%%%%%%%%%%%%%%%%
\section{Proof of Theorem \ref{thm2a}}\label{sectpf2a}
Let $\lambda-1=n=0,1,2,\ldots$. In this section we determine the kernel of $\Delta_n=\partial_z\partial_\zz+n(n+1)/(z-\zz)^2$ in (\ref{eqdiffn}).

\begin{lem}
With
\begin{equation}\label{sDdef}
\sD_n=\sum_{k=0}^n\sum_{\ell=0}^k\frac{(-1)^kn!(k+\ell)!}{(n-k)!(k-\ell)!\ell!}(\partial_z\partial_\zz)^{n-k}(\partial_z-\partial_\zz)^{k-\ell}(z-\zz)^{n-k-\ell}
\end{equation}
and
$$
d_n=\sum_{k=0}^n(-1)^k\frac{(n+k)!}{(n-k)!k!}\frac{1}{(z-\zz)^k}\partial_z^{n-k}
$$
(see (\ref{Dn})) we get
\begin{align}\label{eqprop4}
\sD_n\Delta_n&=(\partial_z\partial_\zz)^{n+1}(z-\zz)^n,\nonumber\\
\Delta_nd_n&=\partial_zd_n\partial_\zz.
\end{align}
\end{lem}

\begin{proof}
For the first identity in (\ref{eqprop4}) we use
\begin{equation}\label{pfpropeq1}
(z-\zz)^k\partial_z\partial_\zz=\partial_z\partial_\zz(z-\zz)^k+k(\partial_z-\partial_\zz)(z-\zz)^{k-1}-k(k-1)(z-\zz)^{k-2},
\end{equation}
which follows by applying the Leibniz rule to the right hand side.

Substituting the above equation into $\sD_n\Delta_n$ gives three terms:
\begin{align}\label{pfpropeq2}
&\sum_{k=-1}^{n-1}\sum_{\ell=-1}^k\frac{(-1)^{k+1}n!(k+\ell+2)!}{(n-k-1)!(k-\ell)!(\ell+1)!}(\partial_z\partial_\zz)^{n-k}(\partial_z-\partial_\zz)^{k-\ell}(z-\zz)^{n-k-\ell-2}\nonumber\\
&+\;\sum_{k=0}^n\sum_{\ell=-1}^{k-1}\frac{(-1)^kn!(k+\ell+1)!(n-k-\ell-1)}{(n-k)!(k-\ell-1)!(\ell+1)!}(\partial_z\partial_\zz)^{n-k}(\partial_z-\partial_\zz)^{k-\ell}(z-\zz)^{n-k-\ell-2}\nonumber\\
&+\;\sum_{k=0}^n\sum_{\ell=0}^k\frac{(-1)^kn!(k+\ell)!(2n-k-\ell)(k+\ell+1)}{(n-k)!(k-\ell)!\ell!}(\partial_z\partial_\zz)^{n-k}(\partial_z-\partial_\zz)^{k-\ell}(z-\zz)^{n-k-\ell-2},
\end{align}
where in the first term we shifted $k\mapsto k+1$, $\ell\mapsto \ell+1$, in the second term we shifted $\ell\mapsto \ell+1$, and in the last term we combined the third term on the right hand
side of (\ref{pfpropeq1}) with the second term of $\Delta_n$.

The three summands in (\ref{pfpropeq2}) add up to zero. We are left with boundary terms $k,\ell=-1$ and $k=n$, $\ell=k$. The boundary terms are
\begin{align*}\label{pfpropeq3}
&(\partial_z\partial_\zz)^{n+1}(z-\zz)^n+\sum_{k=0}^n\frac{(-1)^{k+1}n!}{(n-k-1)!}(\partial_z\partial_\zz)^{n-k}(\partial_z-\partial_\zz)^{k+1}(z-\zz)^{n-k-1}\\
&+\;\sum_{k=0}^n\frac{(-1)^kn!}{(n-k-1)!}(\partial_z\partial_\zz)^{n-k}(\partial_z-\partial_\zz)^{k+1}(z-\zz)^{n-k-1}\quad=\quad(\partial_z\partial_\zz)^{n+1}(z-\zz)^n
\end{align*}
which proves the first identity in (\ref{eqprop4}). For the second identity in (\ref{eqprop4}) we calculate
$$
\partial_\zz d_n=\sum_{k=1}^n(-1)^k\frac{(n+k)!}{(n-k)!(k-1)!}\frac{1}{(z-\zz)^{k+1}}\partial_z^{n-k}+d_n\partial_\zz.
$$
Applying $\partial_z$ to the first term on the right hand side gives two terms, one when $\partial_z$ hits $z-\zz$ and the second one when $\partial_z$ commutes to the right.
In the second term we shift $k$ to $k+1$ and obtain for the sum of both terms
$$
\sum_{k=0}^n(-1)^{k+1}\frac{(n+k)![k(k+1)+(n+k+1)(n-k)]}{(n-k)!k!}\frac{1}{(z-\zz)^{k+2}}\partial_z^{n-k}=-\frac{n(n+1)}{(z-\zz)^2}d_n.
$$
This term cancels the contribution from the second term in $\Delta_n$. We obtain the second identity in (\ref{eqprop4}).
\end{proof}

\begin{proof}[Proof of Theorem \ref{thm2a}]
From the second identity in (\ref{eqprop4}) we get $\Delta_nd_nh(z)=\partial_zd_n\partial_\zz h(z)=0$. So, $d_nh(z)$ is in the kernel of $\Delta_n$ for any holomorphic function $h(z)$.
By complex conjugation we get $\Delta_n\ddd_n\hh(\zz)=0$.

We need to show that $d_nh(z)+\ddd_n\hh(\zz)$ exhausts the kernel of $\Delta_n$.
To do this we use the first identity in (\ref{eqprop4}) to see that any function $F$ in the kernel of $\Delta_n$ fulfills the identity
$$
(\partial_z\partial_\zz)^{n+1}(z-\zz)^nF=0.
$$
By repeated integration the kernel of $(\partial_z\partial_\zz)^{n+1}$ are polynomials in $z$ of degree $\leq n$ with anti-holomorphic coefficients plus polynomials in $\zz$ of degree $\leq n$
with holomorphic coefficients. Therefore, there are (independent) (anti-)holomorphic functions $h_k(z)$ ($\hh_k(\zz)$) such that $(z-\zz)^nF(z,\zz)=\sum_{k=0}^nh_k(z)\zz^k+\hh_k(\zz)z^k$.
By a linear transformation we equivalently have (anti-)holomorphic functions $f_k(z)$ ($\ff_k(\zz)$) with
$(z-\zz)^nF(z,\zz)=\sum_{k=0}^n(f_{n-k}(z)+\ff_{n-k}(\zz))(z-\zz)^k$. We get
$$
F=\sum_{k=0}^n\frac{f_k(z)+\ff_k(\zz)}{(z-\zz)^k}.
$$
With
$$
h(z)=\frac{(-1)^nn!}{(2n)!}f_n(z)\quad\text{and}\quad\hh(\zz)=\frac{n!}{(2n)!}\ff_n(\zz)
$$
we obtain
$$
F=d_nh(z)+\ddd_n\hh(\zz)+\sum_{k=0}^n\frac{f_k(z)+\ff_k(\zz)-\frac{(n+k)!n!}{(n-k)!k!(2n)!}[(-1)^{n-k}\partial_z^{n-k}f_n(z)+\partial_\zz^{n-k}\ff_n(\zz)]}{(z-\zz)^k}.
$$
The term $k=n$ is zero. For $n=0$ this proves the theorem, so we assume $n\geq1$ in the following.
Introducing new (anti-)homomorphic functions $g_k(z)$ ($\gggg_k(\zz)$) we get that
$$
F=d_nh(z)+\ddd_n\hh(\zz)+G_{n-1}
$$
where
\begin{equation}\label{pf2a1}
G_m=\sum_{k=0}^m\frac{g_k(z)+\gggg_k(\zz)}{(z-\zz)^k}.
\end{equation}
We have $\Delta_nF=\Delta_nG_{n-1}=0$. We show by induction that $\Delta_nG_m=0$ implies $G_m=0$ for all $m=0,\ldots,n-1$.

From $\Delta_nG_m=0$ we get
$$
\sum_{k=0}^m\left(k\frac{g_k'(z)-\gggg_k'(\zz)}{(z-\zz)^{k+1}}+[n(n+1)-k(k+1)]\frac{g_k(z)+\gggg_k(\zz)}{(z-\zz)^{k+2}}\right)=0,
$$
with $n(n+1)-k(k+1)=(n-k)(n+k+1)$. For $m=0$ the above equation implies $G_0=0$. We assume $m>0$ and multiply the above equation with $(z-\zz)^2/(n-m)(n+m+1)$.
In the first term we also shift the summation index $k$ to $k+1$ and obtain
$$
\frac{g_m(z)+\gggg_m(\zz)}{(z-\zz)^m}=-\sum_{k=0}^{m-1}\frac{(k+1)(g_{k+1}'(z)-\gggg_{k+1}'(\zz))+(n-k)(n+k+1)(g_k(z)+\gggg_k(\zz))}{(n-m)(n+m+1)(z-\zz)^k}.
$$
Substitution into (\ref{pf2a1}) gives
$$
G_m=\sum_{k=0}^{m-1}\frac{-(k+1)(g_{k+1}'(z)-\gggg_{k+1}'(\zz))-(m-k)(m+k+1)(g_k(z)+\gggg_k(\zz))}{(n-m)(n+m+1)(z-\zz)^k}.
$$
The right hand side is $G_{m-1}$ for some new functions $\widetilde{g}_k(z)$, $\widetilde{\gggg}_k(\zz)$ and vanishes by induction.
\end{proof}

%%%%%%%%%%%%%%%%%%%%%%%%%%%%%%%%%%%%%%%%%%%%%%%%%%%%%%%%%%%%%%%%%%%%%%%%%%%%%%%%%
\section{Proof of Theorem \ref{thm2}}\label{sectpf2}
\begin{lem}
For non-negative integers $a,b,c$ we have
\begin{align}
\begin{aligned}
\label{eqprop1a}
\sum_{k=0}^c(-1)^k\frac{(b+k)!}{(c-k)!k!(k+b-a)!}&=(-1)^c\frac{a!b!}{(a-c)!(b+c-a)!c!},\\
\sum_{k=0}^c(-1)^k\frac{(b+k)!}{(c-k)!k!(k+b+a)!}&=\frac{(a+c-1)!b!}{(a-1)!(a+b+c)!c!},
\end{aligned}
\end{align}
where terms with negative factorials in the denominator vanish. For $m\in\{0,1,\ldots,n\}$ we have
\begin{equation}\label{eqprop1b}
\sum_{k=1}^n(-1)^k\frac{(n+k)!}{(n-k)!k!(k-1)!(k+m)}=(-1)^n-\delta_{m,0},
\end{equation}
where $\delta_{k,\ell}$ is the Kronecker delta.
\end{lem}
\begin{proof}
For the first identity in (\ref{eqprop1a}) we calculate the $x^a$-coefficient of
$$
\sum_{k=0}^c(-1)^k\genfrac(){0pt}{}ck(1+x)^{b+k}=(1-(1+x))^c(1+x)^b=(-x)^c(1+x)^b.
$$
For the second identity in (\ref{eqprop1a}) we replace factorials which depend on $a$ by the gamma function according to $x!=\Gamma(x+1)$.
We analytically continue to negative $a$. The reflection formula $\Gamma(x)\Gamma(1-x)=\pi/\sin(\pi x)$ gives the second identity because
$\sin\pi(a-c+1)=(-1)^c\sin\pi(a+1)$.

For (\ref{eqprop1b}) we calculate
$$
\int_0^1x^m\frac{\partial^{n+1}}{\partial x^{n+1}}x^n(1-x)^n\dd x=\sum_{k=0}^n(-1)^k\genfrac(){0pt}{}nk\int_0^1x^m\frac{\partial^{n+1}}{\partial x^{n+1}}x^{n+k}\dd x=
\sum_{k=0}^n(-1)^k\genfrac(){0pt}{}nk\frac{(n+k)!}{(k-1)!(k+m)}.
$$
The left hand side can also be calculated by using integration by parts $n+1$ times. Because $n+1>m$ the final integral is zero. Only the first boundary term gives a non-zero
result which is
\begin{gather*}
x^m\frac{\partial^n}{\partial x^n}x^n(1-x)^n\Big|_0^1=(-1)^nn!-\delta_{m,0}\,n!
\qedhere
\end{gather*}
\end{proof}

\begin{lem}
We have (see (\ref{Ipm}), $\sum_{k=1}^0\equiv0$)
\begin{equation}\label{eqprop2a}
\Delta_nI^-_n=\partial_zI^-_n\partial_\zz-\sum_{k=1}^n(-1)^{n-k}\frac{(n+k)!}{(n-k)!k!(k-1)!}\frac{1}{(z-\zz)^{k+1}}\left[\partial_\zz,\intsv\dd z\right](z-\zz)^k
\quad\text{on }\sS\sV_{\{0,1,\infty\}}.
\end{equation}
If $n\geq1$, there exist anti-holomorphic functions $\gggg_k(\zz)$ on $\CC\backslash\RR$ for $k=1,\ldots,n$ such that
\begin{equation}\label{eqprop2b}
\partial_zI^-_n\partial_zI^+_n=1+\sum_{k=1}^n\frac{\gggg_k(\zz)}{(z-\zz)^{k+1}}\quad\text{on }(z-\zz)^n\sS\sV_{\{0,1,\infty\}}.
\end{equation}
\end{lem}
\begin{proof}
Differentiation $\partial_z$ of $I^\pm_n$ in (\ref{Ipm}) gives two terms. If $\partial_z$ hits $\intsv\dd z$, we obtain 1 from the first identity in (\ref{eqprop1a}) for $a=b=c=n$.
The second term comes from differentiating $(z-\zz)^{\pm k}$. We get
\begin{equation}\label{eqpf2prop2pf}
\partial_zI^\pm_n=1\pm\sum_{k=1}^n(-1)^{n-k}\frac{(n+k)!}{(n-k)!k!(k-1)!}(z-\zz)^{\pm k-1}\intsv\dd z\;(z-\zz)^{\mp k}
\end{equation}
on the corresponding function spaces.

For (\ref{eqprop2a}) we commute $\partial_\zz$ with $\partial_zI^-_n$.
We obtain two contributions from differentiating $z-\zz$ before and after the integral $\intsv\dd z$. In the latter term we shift the summation index $k$ to $k+1$ to obtain
for the sum of both terms
$$
-\sum_{k=0}^n(-1)^{n-k}\frac{(n+k)![k(k+1)+(n+k+1)(n-k)]}{(n-k)!k!^2}\frac{1}{(z-\zz)^{k+2}}\intsv\dd z\;(z-\zz)^k=-\frac{n(n+1)}{(z-\zz)^2}I^-_n.
$$
This gives the result.

For (\ref{eqprop2b}) we first note that $I^+_n$ maps $(z-\zz)^n\sS\sV_{\{0,1,\infty\}}$ to $\sS\sV_{\{0,1,\infty\}}$ which is stable under $\partial_z$.
Hence the left hand side is well-defined. We use (\ref{eqpf2prop2pf}) to express $-(\partial_zI^-_n-1)(\partial_zI^+_n-1)$ as the double sum
$$
\sum_{k,\ell=1}^n(-1)^{k+\ell}\frac{(n+k)!}{(n-k)!k!(k-1)!}\frac{(n+\ell)!}{(n-\ell)!\ell!(\ell-1)!}\frac{1}{(z-\zz)^{k+1}}\intsv\dd z\;(z-\zz)^{k+\ell-1}
\intsv\dd z\frac{1}{(z-\zz)^\ell}.
$$
Integration by parts gives for $k+\ell\neq0$,
$$
\intsv\dd z\;(z-\zz)^{k+\ell-1}\intsv\dd z\frac{1}{(z-\zz)^\ell}=\frac{1}{k+\ell}\left[(z-\zz)^{k+\ell},\intsv\dd z\right]\frac{1}{(z-\zz)^\ell}+\hh_{k+\ell}(\zz)
$$
with an anti-holomorphic function $\hh_{k+\ell}(\zz)$ on $\CC\backslash\RR$.
We substitute this into the previous equation and use (\ref{eqprop1b}) for $m=\ell\geq1$ and $m=k\geq1$.
With (\ref{eqpf2prop2pf}) this gives
\begin{gather*}
-(\partial_zI^-_n-1)(\partial_zI^+_n-1)= 
\\
=
\partial_zI^+-1+\partial_zI^--1+
\sum_{k,\ell=1}^n(-1)^{k+\ell}\frac{(n+k)!}{(n-k)!k!(k-1)!}\frac{(n+\ell)!}{(n-\ell)!\ell!(\ell-1)!}\frac{\hh_{k+\ell}(\zz)}{(z-\zz)^{k+1}}.
\end{gather*}
A linear transformation of the functions $\hh_\bullet(\zz)$ leads to (\ref{eqprop2b}).
\end{proof}

\begin{lem}\label{lem1pf2}
Let $f(z)\in(z-\zz)^n\sS\sV_{\{0,1,\infty\}}$. We have $\Delta_0\sI_0f(z)=f(z)$ ($\sI_n$ is defined in (\ref{eqIndef})).
If $n\geq1$, there exist anti-holomorphic functions $\gggg_k(\zz)$ on $\CC\backslash\RR$ for $k=1,\ldots,n$ such that
\begin{equation}\label{lem1pf2eq}
\Delta_n\sI_nf(z)=f(z)+\sum_{k=1}^n\frac{\gggg_k(\zz)}{(z-\zz)^{k+1}}.
\end{equation}
\end{lem}
\begin{proof}
We have $\sI_0=\intsv\dd z\intsv\dd\zz$. Hence $\Delta_0\sI_0=1$. Now, assume $n\geq1$.
Because $\partial_z[\partial_\zz,\intsv\dd z]=\partial_z\partial_\zz\intsv\dd z-\partial_\zz=0$ the commutator $[\partial_\zz,\intsv\dd z]$ projects onto
anti-holomorphic functions. We define
$$
\hh_k(\zz)=(-1)^{n-k+1}\frac{(n+k)!}{(n-k)!k!(k-1)!}\left[\partial_\zz,\intsv\dd z\right](z-\zz)^k\intsv\dd\zz\,\partial_zI^+_nf(z).
$$
With (\ref{eqprop2a}) we obtain
$$
\Delta_n\sI_nf(z)=\partial_zI^-_n\partial_\zz\intsv\dd\zz\,\partial_zI^+_nf(z)+\sum_{k=1}^n\frac{\hh_k(\zz)}{(z-\zz)^{k+1}}.
$$
The result follows from (\ref{eqprop2b}) after a re-definition of $\gggg_k(\zz)+\hh_k(\zz)$ as $\gggg_k(\zz)$.
\end{proof}

\begin{lem}\label{pf2lem3}
For any anti-holomorphic function $\gggg_k(\zz)$ on $\CC\backslash\RR$ and $k=1,\ldots,n$ we have (see (\ref{Dn}))
\begin{equation}\label{eqlem2}
\Delta_nD_n\frac{\gggg_k(\zz)}{(z-\zz)^{k+1}}=\frac{\gggg_k(\zz)}{(z-\zz)^{k+1}}.
\end{equation}
\end{lem}

\begin{proof}
The differential operator $\partial_z$ in $D_n$ can only act on $z-\zz$,
$$
D_n\frac{\gggg_k(\zz)}{(z-\zz)^{k+1}}=\sum_{\ell=1}^n(-1)^{\ell-1}\frac{(n-\ell)!}{(n+\ell)!}(z-\zz)^\ell\partial_\zz^{\ell-1}\gggg_k(\zz)\partial_z^{\ell-1}(z-\zz)^{\ell-k-1}.
$$
Terms with $\ell\geq k+1\geq2$ are nullified by $\partial_z^{\ell-1}$. This gives
\begin{equation}\label{eqprop3a}
D_n\frac{\gggg_k(\zz)}{(z-\zz)^{k+1}}=\sum_{\ell=1}^k\frac{(n-\ell)!(k-1)!}{(n+\ell)!(k-\ell)!}(z-\zz)^\ell\partial_\zz^{\ell-1}\frac{\gggg_k(\zz)}{(z-\zz)^k}.
\end{equation}
We get
$$
\partial_zD_n\frac{\gggg_k(\zz)}{(z-\zz)^{k+1}}=\sum_{\ell=1}^k\frac{(n-\ell)!(k-1)!}{(n+\ell)!(k-\ell)!}
\left[\ell(z-\zz)^{\ell-1}\partial_\zz^{\ell-1}(z-\zz)-k(z-\zz)^\ell\partial_\zz^{\ell-1}\right]\frac{\gggg_k(\zz)}{(z-\zz)^{k+1}}.
$$
The square bracket equals $-(k-\ell)(z-\zz)^\ell\partial_\zz^{\ell-1}-\ell(\ell-1)(z-\zz)^{\ell-1}\partial_\zz^{\ell-2}$.
We substitute this into the previous equation and shift $\ell\mapsto\ell-1$ in the first term. From the identity $-(n+\ell)(n-\ell+1)-\ell(\ell-1)=-n(n+1)$ we get
\begin{equation}\label{eqprop3b}
\partial_zD_n\frac{\gggg_k(\zz)}{(z-\zz)^{k+1}}=-n(n+1)\sum_{\ell=2}^k\frac{(n-\ell)!(k-1)!}{(n+\ell)!(k-\ell)!}(z-\zz)^{\ell-1}\partial_\zz^{\ell-2}\frac{\gggg_k(\zz)}{(z-\zz)^{k+1}}.
\end{equation}
We have
$$
\partial_\zz(z-\zz)^{\ell-1}\partial_\zz^{\ell-2}=(z-\zz)^{\ell-2}\partial_\zz^{\ell-1}(z-\zz)
$$
because both sides equal $-(\ell-1)(z-\zz)^{\ell-2}\partial_\zz^{\ell-2}+(z-\zz)^{\ell-1}\partial_\zz^{\ell-1}$.
We use this identity to calculate $\partial_\zz\partial_zD_n\gggg_k(\zz)(z-\zz)^{-k-1}$ from (\ref{eqprop3b}) and get with (\ref{eqprop3a})
$$
\partial_\zz\partial_zD_n\frac{\gggg_k(\zz)}{(z-\zz)^{k+1}}=-\frac{n(n+1)}{(z-\zz)^2}\left[D_n-\frac{(z-\zz)^2}{n(n+1)}\right]\frac{\gggg_k(\zz)}{(z-\zz)^{k+1}},
$$
where the second term in the square bracket cancels the $\ell=1$ summand in (\ref{eqprop3a}). Equation (\ref{eqlem2}) follows.
\end{proof}

\begin{lem}\label{lem1}
For any (anti-)holomorphic functions $h(z)$, $\hh(\zz)$ on $\CC\backslash\RR$ we have (see (\ref{Dn}))
\begin{equation}\label{prop1eq1}
\partial_\zz^n\partial_z^{n+1}(z-\zz)^n[d_nh(z)+\ddd_n\hh(\zz)]=(-1)^nn!\partial_z^{2n+1}h(z).
\end{equation}
For $m\in\{0,1,\ldots,2n\}$ we have
\begin{equation}\label{prop1eq2}
d_nz^m=(-1)^n\ddd_n\zz^m.
\end{equation}
Let $p\in\CC[z]$ be of degree $\leq2n$. Then $(z-\zz)^nd_np(z)$ is a symmetric polynomial of degree $\leq n$ in $z$ and $\zz$.
\end{lem}
\begin{proof}
Because $\partial_z^{n+1}(z-\zz)^{n-k}=0$ at least one derivative $\partial_z$ hits $\hh(\zz)$, nullifying it. So, we can assume $\hh(\zz)=0$.
Then, the only anti-holomorphic term in $d_k$ is $(z-\zz)^{n-k}$ and $\partial_\zz^n(z-\zz)^{n-k}=(-1)^nn!\delta_{k,0}$. The sum over $k$ becomes trivial and (\ref{prop1eq1}) follows.

To prove (\ref{prop1eq2}) we calculate
\begin{align*}
(z-\zz)^nd_nz^m&=\sum_{k=0}^n(-1)^k\frac{(n+k)!m!(z-\zz)^{n-k}z^{m-n+k}}{(n-k)!k!(m-n+k)!}
=\sum_{k=0}^n\sum_{\ell=0}^{n-k}\frac{(-1)^{k+\ell}(n+k)!m!z^{m-\ell}\zz^\ell}{k!(m-n+k)!(n-k-\ell)!\ell!}\\
&=\sum_{\ell=0}^n(-1)^\ell\frac{m!z^{m-\ell}\zz^\ell}{\ell!}\sum_{k=0}^{n-\ell}(-1)^k\frac{(n+k)!}{(n-\ell-k)!k!(k+m-n)!},
\end{align*}
where terms with negative factorials in the denominator vanish.
We substitute $a=2n-m$, $b=n$, $c=n-\ell$ into the first identity of (\ref{eqprop1a}) and obtain
\begin{gather*}
(z-\zz)^nd_nz^m=\sum_{\ell=0}^n(-1)^\ell\frac{m!z^{m-\ell}\zz^\ell}{\ell!}\frac{(-1)^{n-\ell}(2n-m)!n!}{(n-m+\ell)!(m-\ell)!(n-\ell)!}
=\\
=(-1)^nn!\sum_{\ell=0}^m\genfrac(){0pt}{}{m}{\ell}\genfrac(){0pt}{}{2n-m}{n-\ell}z^{m-\ell}\zz^\ell.
\end{gather*}
By substituting $\ell\mapsto m-\ell$ we see that the result is symmetric under complex conjugation. Upon division by $(z-\zz)^n$ we obtain (\ref{prop1eq2}).

For the last statement we observe that $P=(z-\zz)^nd_np(z)$ is a polynomial in $z$ and $\zz$. Because $p$ has degree $\leq2n$ we obtain from (\ref{prop1eq2}) that $P$ is symmetric
under $z\leftrightarrow\zz$. Obviously $P$ is of degree $\leq n$ in $\zz$. By symmetry it is also of degree $\leq n$ in $z$.
\end{proof}

\begin{proof}[Proof of Theorem \ref{thm2}]
With $F(z)$ in (\ref{eqthm2}) we calculate $\Delta_nF(z)$:

By Theorem \ref{thm2a} we have $\Delta_n(d_nh(z)+\ddd_n\hh(\zz))=0$.

From Lemma \ref{lem1pf2} we get $\Delta_n\sI_nf(z)=f(z)+\sum_{k=1}^n\gggg_k(\zz)/(z-\zz)^{k+1}$ for some anti-holomorphic functions $\gggg_k(\zz)$.
Hence $(1-\Delta_n\sI_n)f(z)=-\sum_{k=1}^n\gggg_k(\zz)/(z-\zz)^{k+1}$ and by Lemma \ref{pf2lem3} we get
$$
\Delta_nD_n(1-\Delta_n\sI_n)f(z)=f(z)-\Delta_n\sI_nf(z).
$$
This gives $\Delta_nF(z)=f(z)$.

We need to show that for a suitable choice of $\phi$, $\overline{\phi}$, $p_0$, $p_1$ every such $F\in(z-\zz)^{-n}\sS\sV_{\{0,1,\infty\}}$ is given by the right hand side of (\ref{eqthm2}).
It follows from Theorem \ref{thm2a} that for any function $F(z)$ with $\Delta_nF(z)=f(z)$ there exist (anti-)holomorphic functions $h(z)$, $\hh(\zz)$ on $\CC\backslash\RR$ such
that (\ref{eqthm2}) holds. By (\ref{Dndef}) we get
\begin{equation}\label{eqpf2}
(z-\zz)^n[d_nh(z)+\ddd_n\hh(\zz)]\in\sS\sV_{\{0,1,\infty\}}.
\end{equation}
Because $\sS\sV_{\{0,1,\infty\}}$ is stable under (anti-)differentiation we get from Lemma \ref{lem1} that $\partial_z^{2n+1}h(z)\in\sS\sV_{\{0,1,\infty\}}$.
Hence $\partial_z^{2n+1}h(z)$ has single-valued log-Laurent expansions (\ref{01expansion}), (\ref{inftyexpansion}) at $0,1,\infty$. Because $h(z)$ is analytic we get that $\partial_z^{2n+1}h(z)$ is
meromorphic on the Riemann sphere with poles at $0,1,\infty$. By the theory of Riemann surfaces we get $\partial_z^{2n+1}h(z)\in\CC[z,z^{-1},(z-1)^{-1}]$.
We use partial fraction decomposition to see that $h(z)=h_0(z)\log(z)+h_1(z)\log(z-1)+\phi(z)$ with $h_0,h_1\in\CC[z]$ of degrees $\leq2n$ and $\phi\in\CC[z,z^{-1},(z-1)^{-1}]$.

We get the complex conjugated result for $\hh(\zz)$ and substitute this into (\ref{eqpf2}). By single-valuedness the result is of the form
$H_0\log(z\zz)+H_1\log[(z-1)(\zz-1)]+\Phi$ with $H_0,H_1\in\CC[z,\zz]$ and $\Phi\in\CC[z,z^{-1},(z-1)^{-1},\zz,\zz^{-1},(\zz-1)^{-1}]$.
The polynomials $H_0$ and $H_1$ are obtained when no differential operator acts on logarithms in $h$ and $\hh$. We get
\begin{equation}\label{eqpf3}
H_i=(z-\zz)^nd_nh_i(z)=(z-\zz)^n\ddd_n\hh_i(\zz),\quad\text{for }i=0,1.
\end{equation}
If $\hh_i(\zz)=(-1)^nh_i(\zz)$, then, by (\ref{prop1eq2}), we have $\ddd_n\hh_i(\zz)=(-1)^n\ddd_nh_i(\zz)=d_nh_i(z)$ and (\ref{eqpf3}) is fulfilled.

We need to show that (\ref{eqpf3}) fixes the choice of $\hh_i(\zz)$. Then (\ref{eqhdef}) follows and the proof of Theorem \ref{thm2} is complete.
We consider $z$ and $\zz$ as independent variables. The pole of order $n$ at $z=\zz$ in $\ddd_n\hh_i(\zz)$ is $(-1)^n(2n)!\hh_i(\zz)/(n!(z-\zz)^n)$.
This allows one to determine $\hh_i(\zz)$ from $\ddd_n\hh_i(\zz)$ and the claim follows.
\end{proof}

In Theorem \ref{thm2} the definition of the integral operator $\sI_n$ has three integrations: Two with respect to $\dd z$ in $I_n^-,I_n^+$ and one explicit integration with respect to $\dd\zz$.
We close the section with an alternative representation $\sI_n'$ which has only two integrations: one holomorphic and one anti-holomorphic.
We need the result for the proof of Theorem \ref{conthm} on the weights of constructible graphical functions. For implementing the algorithm of appending an edge
the operator $\sI_n$ is more efficient.

\begin{lem}\label{lem2}
Let $\sI_n':(z-\zz)^n\sS\sV_{\{0,1,\infty\}}\rightarrow(z-\zz)^{-n}\sS\sV_{\{0,1,\infty\}}$ be defined by
\begin{equation}\label{Inprimedef}
\sI_n'=\sum_{k,\ell=0}^n(-1)^{n+k+\ell}\frac{(n+k)!(n+\ell)!}{(n-k)!(n-\ell)!(k+\ell)!k!\ell!}\frac{1}{(z-\zz)^k}\intsv\dd z(z-\zz)^{k+\ell}\intsv\dd\zz\frac{1}{(z-\zz)^\ell}.
\end{equation}
Then Lemma \ref{lem1pf2} holds if $\sI_n$ is replaced by $\sI_n'$.
\end{lem}
\begin{proof}
We have $\sI_0'=\sI_0$ and hence assume that $n\geq1$.

The derivative $\partial_z$ generates two terms on the right hand of (\ref{Inprimedef}). In one term $(z-\zz)^{-k}$ is differentiated and in the second term the integral operator $\intsv\dd z$
is annihilated. After the annihilation of $\intsv\dd z$ the sum over $k$ can be evaluated with the first identity in (\ref{eqprop1a}) for $a=n-\ell$, $b=c=n$.
The result $(-1)^n\delta_{\ell,0}$ trivializes the sum over $\ell$ and the second term becomes $\intsv\dd\zz$,
\begin{equation}\label{lem2pfeq}
\partial_z\sI_n'=\sum_{k=1}^n\sum_{\ell=0}^n\frac{-(-1)^{n+k+\ell}(n+k)!(n+\ell)!}{(n-k)!(n-\ell)!(k+\ell)!(k-1)!\ell!}\frac{1}{(z-\zz)^{k+1}}
 \intsv\dd z(z-\zz)^{k+\ell}\intsv\dd\zz\frac{1}{(z-\zz)^\ell}+\intsv\dd\zz.
\end{equation}
The derivative $\partial_\zz$ on the first term of the right hand side gives four terms. Term one and three hit powers of $z-\zz$.
The second term has the commutator $[\partial_\zz,\intsv\dd z]$ which gives rise to the anti-holomorphic functions $\gggg_k(\zz)$ on the right hand side of (\ref{lem1pf2eq}).
The fourth term annihilates $\intsv\dd\zz$. Like above we get in this term $(-1)^n\delta_{k,0}$ from evaluating the sum over $\ell$.
The $k$-sum starts at 1, so the fourth term is zero. After a shift $k\mapsto k+1$ the third term adds to the first term yielding
$$
-\sum_{k,\ell=0}^n\frac{(-1)^{n+k+\ell}[k(k+1)+(n+k+1)(n-k)](n+k)!(n+\ell)!}{(n-k)!(n-\ell)!(k+\ell)!k!\ell!\;(z-\zz)^{k+2}}\intsv\dd z(z-\zz)^{k+\ell}\intsv\dd\zz\frac{1}{(z-\zz)^\ell}.
$$
This is $-n(n+1)(z-\zz)^{-2}\sI_n'$ and the result follows because the second term in (\ref{lem2pfeq}) differentiates to the identity.
\end{proof}

%%%%%%%%%%%%%%%%%%%%%%%%%%%%%%%%%%%%%%%%%%%%%%%%%%%%%%%%%%%%%%%%%%%%%%%%%%%%%%%%%
\section{Proof of Theorem \ref{thm3}}\label{sectpf3}
Let $n=0,1,2,\ldots$. Assume $f(z)$ with (G3) for the graph $G_1$ is in the kernel of $\Delta_n(z-\zz)^{n+1}$. We need to show $f(z)=0$.

We consider $(z-\zz)^{2n+1}f(z)$. From Theorem \ref{thm2a} we know that there exist (anti-)holomorphic functions $h(z)$, $\hh(\zz)$ on $\CC\backslash\RR$ such that
\begin{equation}\label{pf3eq}
(z-\zz)^{2n+1}f(z)=(z-\zz)^n[d_nh(z)+\ddd_n\hh(\zz)].
\end{equation}
By (G3) we have $f\in\sS\sV_{\{0,1,\infty\}}$. We are thus in the situation of (\ref{eqpf2}) in the proof of Theorem \ref{thm2}.
We follow the proof to conclude that there exist (anti-)meromorphic functions $\phi\in\CC[z,z^{-1},(z-1)^{-1}]$, $\overline{\phi}\in\CC[\zz,\zz^{-1},(\zz-1)^{-1}]$
and polynomials $p_0$, $p_1$ of degrees $\leq2n$ such that (see (\ref{eqhdef}))
$$
(z-\zz)^{n+1}f(z)=d_n\Big[\phi(z)+\sum_{s=0,1}p_s(z)\log(z-s)\Big]+\ddd_n\Big[\overline{\phi}(\zz)+(-1)^n\sum_{s=0,1}p_s(\zz)\log(\zz-\sss)\Big].
$$
We consider the $\log(z\zz)$ term in the expansion (\ref{01expansion}) of $f(z)$ at $z=0$. The only contribution comes from
$[d_np_0(z)]\log(z)+(-1)^n[\ddd_np_0(\zz)]\log(\zz)$. By (\ref{prop1eq2}) this equals $[d_np_0(z)]\log(z\zz)$.
By Lemma \ref{lem1} there exists a symmetric polynomial $P$ of degrees $\leq n$ in $z$ and $\zz$ such that
$$
\text{$\log(z\zz)$ coefficient of $f(z)$ }=\frac{P(z,\zz)}{(z-\zz)^{2n+1}}.
$$
In the following we consider $z$ and $\zz$ as independent variables. If $P(z,\zz)\neq0$, the coefficient has singularity at $z=\zz$. This is ruled out by (G3).
(Because the right hand side is anti-symmetric, $P(z,\zz)=0$ also follows from (G1).)
Therefore $d_np_0(z)=0$ and because $d_n$ is injective in the space of analytic functions (this is the complex conjugate of the last statement in the proof
of Theorem \ref{thm2}) we have $p_0(z)=0$. Likewise we get $p_1(z)=0$.

Assume $\phi(z)$ has a pole of order $m>0$ at $z=0$. Then $d_n\phi(z)$ has a pole of order $n+m$ at $z=0$ (from the summand $k=0$ in $d_n$).
The term $\ddd_n\overline{\phi}(\zz)$ is analytic at $z=0$ (it may have poles at $\zz=0$). Therefore $(z-\zz)^{n+1}f(z)$ (and hence also $f(z)$) has a pole of order $n+m$ at $z=0$.
On the other hand the graph $G_1$ has a single edge of weight 1 connected to $z$ so that $\nu_z^>=1$ and $\nu_{0z}=0$ in (G3).
We get $M_0\geq-n$ which rules out poles of order $>n$ in $f(z)$. Therefore $\phi(z)$ is analytic at $z=0$. Likewise $\phi(z)$ is analytic at $z=1$ which makes $\phi(z)$ a polynomial in $z$.

Let $m$ be the degree of the polynomial $\phi(z)$ and assume $m>2n$. We consider the expansion of $f(z)$ at $z=\infty$.
From $d_n\phi(z)$ we get a leading order $z^{m-n}$ with coefficient
$$
\sum_{k=0}^n(-1)^k\frac{(n+k)!m!}{(n-k)!k!(m-n+k)!}=\frac{(m-n-1)!}{(m-2n-1)!},
$$
where we used the second identity in (\ref{eqprop1a}) for $a=m-2n$, $b=c=n$. The term $\ddd_n\overline{\phi}(\zz)$ is of order $\leq n<m-n$ for $z\to\infty$.
We conclude that in the expansion of $f(z)$ at $z=\infty$ we have a term of order $z^{m-2n-1}$. 
On the other hand we have $\nu_z^<=0$ and $\nu_z=1$ which implies $M_\infty\leq-1$ in (G3). This rules out the case $m>2n$. So, $\phi(z)$ is a polynomial of degree $\leq2n$ in $z$.
Likewise $\overline{\phi}(\zz)$ is a polynomial of degree $\leq2n$ in $\zz$. Like for the coefficient of $\log(z\zz)$ there exists a symmetric polynomial $P\in\CC[z,\zz]$
of degrees $\leq n$ in $z$ and $\zz$ such that $f(z)=P(z,\zz)/(z-\zz)^{2n+1}$. By (G3) (or by (G1)) we get $P(z,\zz)=0$.

%%%%%%%%%%%%%%%%%%%%%%%%%%%%%%%%%%%%%%%%%%%%%%%%%%%%%%%%%%%%%%%%%%%%%%%%%%%%%%%%%
\section{The algorithm for appending an edge of weight 1}\label{sectappalg}
Let $n=D/2-2$. By repeated single-valued integration we calculate
\begin{equation}
g(z)=(z-\zz)^n[\sI_n+D_n(1-\Delta_n\sI_n)](z-\zz)^{n+1}f_G(z)\in\sS\sV_{\{0,1,\infty\}}
\end{equation}
(if possible). By Lemma \ref{lemdiff} and Theorem \ref{thm2} we get
\begin{equation}\label{solution1}
-n!(z-\zz)^{2n+1}f_{G_1}(z)=g(z)-h(z),
\end{equation}
with
\begin{equation}\label{solution2}
h(z)=(z-\zz)^nd_n\Big[\phi(z)+\sum_{s=0,1}p_s(z)\log(z-s)\Big]+(z-\zz)^n\ddd_n\Big[\overline{\phi}(\zz)+(-1)^n\sum_{s=0,1}p_s(\zz)\log(\zz-\sss)\Big],
\end{equation}
where $\phi\in\CC[z,z^{-1},(z-1)^{-1}]$, $\overline{\phi}\in\CC[\zz,\zz^{-1},(\zz-1)^{-1}]$, and $p_0$, $p_1$, are polynomials of degrees $\leq2n$.
By partial fraction decomposition we write
\begin{align}\label{solution3}
\phi(z)&=\sum_{s=0,1}\phi_s(z)+\phi_\text{reg}(z)+\phi_\infty(z),\quad\text{where}\nonumber\\
\phi_s(z)&=\sum_{k=K_s}^{-1}c_k^s(z-s)^k,\quad\phi_\text{reg}(z)=\sum_{k=0}^{2n}c_k^\infty z^k,\quad\phi_\infty(z)=\sum_{k=2n+1}^{K_\infty}c_k^\infty z^k,
\end{align}
with an analogous decomposition for $\overline{\phi}(\zz)$.

\subsection{Solving for \texorpdfstring{$\phi_s(z)$}{phi(s,z)} and \texorpdfstring{$\overline{\phi}_s(\zz)$}{phibar(s,zbar)}}
For $s=0,1$ we expand $g(z)$ at $s$ in $z$ and $\zz$ up to order $n$ yielding the truncated single-valued log-Laurent series $g_s(z)$.

By (G3) the graphical function $f_{G_1}(z)$ admits an expansion at $z=s$ with terms $\log^\ell(|z-s|^2)(z-s)^m(\zz-\sss)^\mm$ with $m,\mm\geq-n$.
By (\ref{solution1}) the total degree (i.e.\ $m+\mm$) of any term in the expansion of $g(z)-h(z)$ at $z=s$ is thus $\geq1$.
Because the differential operators $(z-\zz)^nd_n$ and $(z-\zz)^n\ddd_n$ do not alter the total degree, terms of negative total degree in $h(z)$ at $z=s$
can only come from $\phi_s(z)$ and $\overline{\phi}_s(\zz)$. Therefore
$$
(z-\zz)^n[d_n\phi_s(z)+\ddd_n\overline{\phi}_s(\zz)]=g_s(z)\quad\text{for negative total degree.}
$$
Note that the left hand side has maximal degree $n$ in $z$ and $\zz$ (from the term $k=0$ in the sums of $d_n$ and $\ddd_n$).
Therefore all pole-terms are contained in $g_s(z)$. We read off $(z-\zz)^nd_n\phi_s(z)$ from the negative powers of $z$ in $g_s(z)$.
The coefficient of $\zz^n$ in $(z-\zz)^nd_n\phi_s(z)$ is $(-1)^n\partial_z^n\phi_s(z)$. This gives $\partial_z^n\phi_s(z)$ which can trivially be integrated $n$ times. Integration constants
are insignificant because they do not affect the pole terms in $\phi_s(z)$.

The procedure to obtain $\overline{\phi}_s(\zz)$ is analogous starting from negative powers of $\zz$ in $g_s(z)$.

\subsection{Solving for \texorpdfstring{$\phi_\infty(z)$}{phi(infty,z)} and \texorpdfstring{$\overline{\phi}_\infty(\zz)$}{phibar(infty,zbar)}}
We expand $g(z)$ at $\infty$ to order $-n-1$ in $z^{-1}$ and $-n$ in $\zz^{-1}$ yielding the truncated single-valued log-Laurent series $g_\infty(z)$.

By (G3) the graphical function $f_{G_1}(z)$ admits an expansion at $z=\infty$ with terms $\log^\ell(|z|^2)z^m\zz^\mm$ with $m,\mm\leq-1$.
The total degree of any term in the expansion of $g(z)-h(z)$ at $z=\infty$ is thus $\geq-2n+1$. Terms of total degree $\leq-2n-1$ in $h(z)$ at $z=\infty$
come from $\phi_\infty(z)$ and $\overline{\phi}_\infty(\zz)$. Therefore
$$
(z-\zz)^n[d_n\phi_\infty(z)+\ddd_n\overline{\phi}_\infty(\zz)]=g_\infty(z)\quad\text{for total degree}\leq-2n-1\text{ in }z^{-1}.
$$
In $(z-\zz)^nd_n\phi_\infty(z)$ the coefficient of $\zz^n$ is $(-1)^n\partial_z^n\phi_\infty(z)$ whereas in $(z-\zz)^n\ddd_n\overline{\phi}_\infty(\zz)$ all terms have degree
$\leq-n-1$ in $\zz^{-1}$. From the $\zz^n$ coefficient of $g_\infty(z)$ we read off $\partial_z^n\phi_\infty(z)$ which can trivially be integrated to
obtain $\phi_\infty(z)$ (integration constants are insignificant).

To obtain $\overline{\phi}_\infty(\zz)$ we expand $g(z)$ at $\infty$ to order $-n$ in $z^{-1}$ and $-n-1$ in $\zz^{-1}$. The analog of the above algorithm gives
$\overline{\phi}_\infty(\zz)$. Because expansions at infinity can be time consuming, it is faster to compute the orders $-n-1,-n$ and $-n,-n-1$ than a single expansion
to order $-n,-n$ in $z^{-1},\zz^{-1}$.

\subsection{Solving for \texorpdfstring{$p_s(z)$}{p(s,z)}}
Terms with $\log(|z-s|^2)$ in $h(z)$ come from (see (\ref{prop1eq2}))
$$
(z-\zz)^n[d_np_s(z)\log(z-s)+(-1)^n\ddd_np_s(\zz)\log(\zz-\sss)]=(z-\zz)^n\log(|z-s|^2)d_np_s(z)+\text{non-log terms.}
$$
The function $g\in\sS\sV_{\{0,1,\infty\}}$ has a single-valued log-Laurent expansion at $z=s$ of the form (\ref{01expansion}),
$$
g(z)=\sum_\ell\sum_{m,\mm}g_{\ell,m,\mm}^s\log^\ell(|z-s|^2)(z-s)^m(\zz-\sss)^\mm.
$$
We set $z=\zz$ in the $\log(|z-s|^2)$-coefficient of this expansion. Because of the factor $(z-\zz)^{2n+1}$ on the left hand side of (\ref{solution1}) we have $h(z)=g(z)$ in this limit.
In $(z-\zz)^nd_n$ only the summand $k=n$ survives. We get
\begin{equation}\label{solution4}
(-1)^n\frac{(2n)!}{n!}p_s(z)=\sum_{m,\mm}g_{1,m,\mm}^s(z-s)^{m+\mm}.
\end{equation}
On the other hand we get from (G1) and (\ref{solution1}) that $g(z)-h(z)$ is anti-symmetric under $z\leftrightarrow\zz$.
From Lemma \ref{lem1} we know that $(z-\zz)^nd_np_s(z)$ is a symmetric polynomial of degree $\leq n$ in $z$ and $\zz$.
Therefore, the expansion coefficients $g_{1,m,\mm}^s$ with $m$ or $\mm$ not in $\{0,1,\ldots,n\}$ are anti-symmetric,
$$
g_{1,m,\mm}^s=-g_{1,\mm,m}^s\quad\text{if not }m,\mm\in\{0,1,\ldots,n\}.
$$
This anti-symmetry lets the sum over $m$ and $\mm$ on the right hand side of (\ref{solution4}) collapse to values in $\{0,1,\ldots,n\}$,
$$
p_s(z)=(-1)^n\frac{n!}{(2n)!}\sum_{m,\mm=0}^ng_{1,m,\mm}^s(z-s)^{m+\mm}.
$$
With this restriction the coefficients $g_{1,m,\mm}^s$ are in $g_s(z)$.

\subsection{Solving for \texorpdfstring{$\phi_\text{reg}(z)$}{phiref(z)} and \texorpdfstring{$\overline{\phi}_\text{reg}(\zz)$}{phirefbar(zbar)}}
We expand
\begin{align*}
\widetilde{g}(z)=g(z)&-(z-\zz)^nd_n\Big[\phi_\infty(z)+\sum_{s=0,1}\phi_s(z)+p_s(z)\log(z-s)\Big]\\
&-\,(z-\zz)^n\ddd_n\Big[\overline{\phi}_\infty(\zz)+\sum_{s=0,1}\overline{\phi}_s(\zz)+(-1)^np_s(\zz)\log(\zz-\sss)\Big]
\end{align*}
in $z=0$ up to order $n$ yielding
$$
\widetilde{g}_\text{reg}(z)=\sum_\ell\sum_{m,\mm\leq n}\widetilde{g}^\text{reg}_{\ell,m,\mm}\log^\ell(z\zz)z^m\zz^\mm.
$$
By (\ref{prop1eq2}) we have
$$
(z-\zz)^n\ddd_n\overline{\phi}_\text{reg}(\zz)=(z-\zz)^nd_n\overline{\phi}_\text{reg}(z).
$$
So, the contributions of $\phi_\text{reg}(z)$ and $\overline{\phi}_\text{reg}(\zz)$ to $h(z)$ can be combined to $(z-\zz)^nd_n[\phi_\text{reg}(z)+\overline{\phi}_\text{reg}(z)]$.
We can hence set
$$
\overline{\phi}_{\mathrm{reg}}(\zz)=0
$$
without restriction. We follow the construction of $p_s(z)$ in the previous subsection, now using the logarithm-free part of the expansion of $\widetilde{g}_\text{reg}(z)$.
In complete analogy we get
$$
\phi_\text{reg}(z)=(-1)^n\frac{n!}{(2n)!}\sum_{m,\mm=0}^n\widetilde{g}^\text{reg}_{0,m,\mm}z^{m+\mm}.
$$

\bibliographystyle{plain}
\renewcommand\refname{References}

\end{document}